\documentclass[11pt,a4paper]{amsart} 
\usepackage{amssymb,bm,mathrsfs,epsfig,color,multirow,hhline,url,dsfont}
\usepackage[all]{xy}
\usepackage{hyperref}

\newtheorem{theorem}{Theorem}[section] 
\newtheorem{lemma}[theorem]{Lemma}
\newtheorem{proposition}[theorem]{Proposition}

\newtheorem{corollary}[theorem]{Corollary}
\newtheorem{definition}[theorem]{Definition}

\newtheorem{example}[theorem]{Example}

\def\1{{\mathds{1}}}
\def\triv{{\mathbf 1}}
\def\A{{\mathcal A}}
\def\Aa{{\bar{\sf A}}}
\def\Ad{{\rm Ad}}

\def\Aff{{\sf Aff}}

\def\Amp{{\sf A}}
\def\area{{\nabla_{\!\!\sf a}}}
\def\Aut{{\rm Aut}}
\def\B{{\mathfrak B}}
\def\Br{{\rm B}}
\def\build#1_#2^#3{\mathrel{\mathop{\kern 0pt#1}\limits_{#2}^{#3}}}
\def\C{{\mathbb C}}
\def\cn{{\sf cn}}
\def\Co{{\sf W}}

\def\Conf{{\mathscr C}}
\def\D{{\mathcal L}}

\def\e{{\sf e}}
\def\epsilon{{\varepsilon}}
\def\E{{\mathbb E}}
\def\ES{{\mathcal E}}
\def\EL{{\sf EL}}

\def\End{{\rm End}}

\def\f{{\sf f}}
\def\F{{\mathbb F}}
\def\Fr{{\rm F}}
\def\Fun{{\mathcal F}}
\def\G{{\mathbb G}}
\def\g{{\mathfrak g}}
\def\gar{{\mathcal G}}
\def\ggar{{\sf G}}
\def\GL{{\rm GL}}
\def\H{{\mathbb H}}
\def\hol{{\rm hol}}
\def\Hom{{\rm Hom}}
\def\inc{{\sf inc}}
\def\ii{{\sf i}}
\def\j{{\sf j}}
\def\k{{\sf k}}
\def\I{{\sf I}}
\def\id{{\rm id}}
\def\K{{\mathbb K}}
\def\L{{\mathcal L}}
\def\Lop{{\bf L}}
\def\Loop{{\sf L}}
\def\M{{\mathcal M}}
\def\Mat{{M}}
\def\Mr{{\rm M}}
\def\n{{\sf n}}
\def\N{{\mathbb N}}
\def\na{{\bar{\sf n}}}
\def\nat{{\sf nat}}
\def\nc{{\sf nc}}
\def\NC{{\rm NC}}
\def\Nn{{\sf N}}
\def\O{{\mathcal O}}

\def\Out{{\sf Out}}
\def\ptkn{p^{\scriptscriptstyle\K,N}_{t}}
\def\ptrn{p^{\scriptscriptstyle\R,N}_{t}}
\def\prn{p^{\scriptscriptstyle\R,N}}
\def\ptcn{p^{\scriptscriptstyle\C,N}_{t}}
\def\pcn{p^{\scriptscriptstyle\C,N}}
\def\pthn{p^{\scriptscriptstyle\H,N}_{t}}
\def\phn{p^{\scriptscriptstyle\H,N}}
\def\P{{\mathbb P}}
\def\phi{{\varphi}}

\def\Path{{\sf P}}
\def\R{{\mathbb R}}
\def\RAff{{\sf RAff}}
\def\RL{{\sf RL}}
\def\RP{{\sf RP}}
\def\S{{\mathfrak S}}
\def\Sk{{\sf Sk}}
\def\sk{{\mathcal S}}
\def\ssk{{\sf S}}
\def\SO{{\rm SO}}
\def\so{{\mathfrak{so}}}
\def\Sp{{\rm Sp}}
\def\sp{{\mathfrak{sp}}}
\def\SU{{\rm SU}}
\def\su{{\mathfrak{su}}}
\def\SL{{\rm SL}}

\def\src{\underline}

\def\t#1{{{}^{t} \! #1}}
\def\tgt{\overline}
\def\tauknt{\tau^{\scriptscriptstyle\K,N}_{t}}

\def\tr{{\rm tr}}
\def\T{{\sf T}}
\def\T{{\sf T}}
\def\Tr{{\rm Tr}}
\def\U{{\rm U}}
\def\UC{{\mathbb U}}
\def\u{{\mathfrak u}}
\def\v{{\sf v}}
\def\V{{\mathbb V}}
\def\Var{{\rm Var}}

\def\W{{\mathscr W}}
\def\Wg{{\rm Wg}}
\def\wE{{\widehat\E}}
\def\wG{{\widehat\G}}
\def\wT{{\widehat \T}}
\def\wV{{\widehat\V}}
\def\YM{{\sf YM}}
\def\Z{{\mathbb Z}}

\newcommand{\Rmnum}[1]{\expandafter\@slowromancap\romannumeral #1@}

\title[The Master Field on the plane]{The master field on the plane}
\author{Thierry L\'evy}
\address{\hspace{-4.2mm}Thierry L\'evy \newline Universit\'e Pierre et Marie Curie (Paris 6)\newline Laboratoire de Probabilit\'es et Mod\`eles Al\'eatoires \newline 4, place Jussieu \newline  F-75252 Paris Cedex 05 \newline \url{http://www.proba.jussieu.fr/pageperso/levy/}}

\begin{document}

\begin{abstract} We study the large $N$ asymptotics of the Brownian motions on the orthogonal, unitary and symplectic groups, extend the convergence in non-commutative distribution originally obtained by Biane for the unitary Brownian motion to the orthogonal and symplectic cases, and derive explicit estimates for the speed of convergence in non-commutative distribution of arbitrary words in independent Brownian motions.

Using these results, we construct and study the large $N$ limit of the Yang-Mills measure on the Euclidean plane with orthogonal, unitary and symplectic structure groups. We prove that each Wilson loop converges in probability towards a deterministic limit, and that its expectation converges to the same limit at a speed which is controlled explicitly by the length of the loop. In the course of this study, we reprove and mildly generalise a result of Hambly and Lyons on the set of tree-like rectifiable paths. 

Finally, we establish rigorously, both for finite $N$ and in the large $N$ limit, the Schwinger-Dyson equations for the expectations of Wilson loops, which in this context are called the Makeenko-Migdal equations. We study how these equations allow one to compute recursively the expectation of a Wilson loop as a component of the solution of a differential system with respect to the areas of the faces delimited by the loop. \end{abstract}

\maketitle

{\small \tableofcontents}

\newpage 
\setcounter{section}{0}

\section*{Introduction}

\subsection{The Yang-Mills measure}

The Euclidean two-dimensional Yang-Mills measure is a probability measure which was defined, first by A. Sengupta \cite{SenguptaAMS} and later in a different way by the author \cite{LevyAMS,LevySMF}, as a mathematically rigorous version of one of the functional integrals which occur in quantum field theory, more precisely in quantum gauge theories. 

The two-dimensional Yang-Mills measure is specified by the choice of a compact surface $\Sigma$, which plays the role of space-time and which we shall assume to be oriented, a compact connected Lie group $G$, and a principal $G$-bundle $\pi: P\to \Sigma$. The surface $\Sigma$ is endowed with a volume form, and the Lie algebra $\g$ of the group $G$ is endowed with an invariant scalar product $\langle \cdot,\cdot \rangle$. 

This data allows one to define, on the affine space $\A(P)$ of connections on $P$, the Yang-Mills functional, which is the real-valued function $S_{\YM}:\A(P)\to \R_{+}$ defined as follows. For all connection $1$-form $\omega\in \Omega^{1}(P)\otimes \g$, let $\Omega$ be the curvature of $\omega$. It is a $2$-form on $\Sigma$ with values in $\Ad(P)$, the bundle associated with $P$ through the adjoint representation of $G$. Dividing $\Omega$ by the volume form of $\Sigma$ yields a section of $\Ad(P)$ which we denote by $*\Omega$. The invariant scalar product on $\g$ endows each fibre of $\Ad(P)$ with a Euclidean structure, and $\langle \Omega \wedge *\Omega \rangle$ is a real-valued $2$-form on $\Sigma$ which can be integrated to produce
\[S_{\YM}(\omega)=\int_{\Sigma} \langle \Omega \wedge *\Omega \rangle,\]
the Yang-Mills action of $\omega$.

For the purposes of physics, the Yang-Mills measure is described by the formula
\begin{equation}\label{def mu intro}
\mu_{\YM}(d\omega)=\frac{1}{Z} e^{-\frac{1}{2}S_{\YM}(\omega)} \; D\omega,
\end{equation}
where $D\omega$ is a regular Borel measure on $\A(P)$ invariant by all translations and $Z$ is the normalisation constant which makes $\mu$ a probability measure. Unfortunately, $D\omega$ does not exist and, even pretending that it does, $Z$ appears to be infinite.

The mathematical approach to this formula consists in constructing, rather than a probability measure on $\A(P)$, the probability measure which should be its image under the holonomy mapping. Indeed, each element $\omega$ of $\A(P)$ determines a holonomy, or parallel transport, which for each suitably regular path $c:[0,1]\to \Sigma$ is an equivariant map $/\hspace{-1mm}/\!{}_{\omega,c}:P_{c(0)} \to P_{c(1)}$. Choosing an origin $o\in \Sigma$ and a point $p\in P_{o}$, the holonomy of each loop $l$ based at $o$ is completely described by the unique element $h$ of $G$ such that $/\hspace{-1mm}/\!{}_{\omega,l}(p)=ph$. All choices being understood, we shall denote this element $h$ by $\hol(\omega,l)$. 

A class of loops which turns out to be appropriate is the class of rectifiable loops, that is, the class of continuous paths with finite length. We denote by $\Loop_{o}(\Sigma)$ the set of rectifiable loops on $\Sigma$ based at $o$, taken up to reparametrisation. When parametrised by arc length, each element of $\Loop_{o}(\Sigma)$ admits a derivative at almost every time, which is essentially bounded, and of which it is the primitive. As long as the connection is smooth, the differential equation which defines the parallel transport along such a loop is well defined and has a unique solution. 

Rectifiable loops can be concatenated and the holonomy is anti-multiplicative, in the sense that for all $l_{1}, l_{2}\in \Loop_{o}(\Sigma)$, one has $\hol(\omega,l_{1}^{-1})=\hol(\omega,l_{1})^{-1}$ and $\hol(\omega,l_{1}l_{2})=\hol(\omega,l_{2})\hol(\omega,l_{1})$. The monoid $\Loop_{o}(\Sigma)$ can be turned into a group, which we still denote by $\Loop_{o}(\Sigma)$, by taking the quotient by the sub-monoid of tree-like loops, which is the closure in the topology of $1$-variation of the normal sub-monoid generated by all loops of the form $cc^{-1}$, where $c:[0,1]\to \Sigma$ is an arbitrary rectifiable path starting at $o$ and $c^{-1}(t)=c(1-t)$. The holonomy of a smooth connection is then a group homomorphism $\hol(\omega,\cdot) : \Loop_{o}(\Sigma)^{op} \to G$, where $\Loop_{o}(\Sigma)^{op}$ is the opposite group of $\Loop_{o}(\Sigma)$.

The definition of the homomorphism $\hol(\omega,\cdot)$ depends on the choice of the point $p$ in $P_{o}$. In order to get a more invariant picture, one must take into account the action on $\A(P)$ of the gauge group $\Aut(P)$, which is the group of bundle automorphisms of $P$. The group $G$ acts by conjugation on $\Hom(\Loop_{o}(\Sigma)^{op},G)$ and the holonomy mapping is the following injective map:
\begin{align*}
\hol : \A(P)/\Aut(P) & \longrightarrow \Hom(\Loop_{o}(\Sigma)^{op},G)/G\\
[\omega] & \longmapsto [\hol(\omega,\cdot)].
\end{align*}

Since the Yang-Mills action $S_{\YM}$ is invariant under the action of the gauge group $\Aut(P)$, and since this group acts by affine transformations on $\A(P)$, the Yang-Mills measure $\mu_{\YM}$ should be invariant under the action of $\Aut(P)$. There is thus, in principle, no information lost if one looks at the Yang-Mills measure through the holonomy map. 

The mathematical Yang-Mills measure is a probability measure on $\Hom(\Loop_{o}(\Sigma)^{op},G)$ which is invariant under the action of $G$. We denote this measure by $\YM$ and we think of it as being related to the physical Yang-Mills measure $\mu_{\YM}$ by the relation $\YM=\mu_{\YM}\circ \hol^{-1}$. The measure $\YM$ is thus the distribution of a collection $(H_{l})_{l\in \Loop_{o}(\Sigma)}$ of $G$-valued random variables.

In the present work, we consider the case where the surface $\Sigma$ is the Euclidean plane $\R^{2}$. This is not a compact surface, but we can think of it as the increasing limit of a sequence of disks of large radius, and nothing changes in the picture which we have described so far. We naturally take the point $o$ to be the origin of $\R^{2}$, denoted by $0$.

In this case, which is in fact the simplest case, the distribution of the collection $(H_{l})_{l\in \Loop_{0}(\R^{2})}$ of $G$-valued random variables is fully characterised by the following properties. \\

\indent ${\rm YM}_{1}$. It is anti-multiplicative in the sense that the equalities $H_{l_{1}^{-1}}=H_{l_{1}}^{-1}$ and $H_{l_{1}l_{2}}=H_{l_{2}}H_{l_{1}}$ hold almost surely for any two loops $l_{1}$ and $l_{2}$.\\
\indent ${\rm YM}_{2}$. It is stochastically continuous, in the sense that if $(l_{n})_{n\geq 0}$ is a sequence of loops which converges in $1$-variation towards a loop $l$, then the sequence $(H_{l_{n}})_{n\geq 0}$ of $G$-valued random variables converges in probability towards $H_{l}$.\\
\indent ${\rm YM}_{3}$. Its finite-dimensional marginal distributions can be described as follows. Consider a graph $\G$ traced on $\R^{2}$ such that $0$ is one of the vertices of $\G$. Let $\Loop_{0}(\G)$ denote the subgroup of $\Loop_{0}(\R^{2})$ consisting of the loops which can be formed by concatenating edges of $\G$. It is a free group of rank equal to the number of bounded faces delimited by $\G$. This free group admits particular bases, indexed by the set $\F^{b}$ of bounded faces of $\G$, and which we call lasso bases. 

A lasso basis is a set of loops $\{\lambda_{F} : F\in \F^{b}\}$ such that for each bounded face $F$, the loop $\lambda_{F}$ follows a path from $0$ to a vertex located on the boundary of $F$, then goes once along the boundary of the face $F$, and finally goes back to $0$ backwards along the same path. Moreover, we insist that it is possible to order the loops of the basis in such a way that their product taken in this order is equal to the boundary of the unbounded face of $\G$, oriented negatively. 

Let us choose a lasso basis $\{\lambda_{F} : F\in \F^{b}\}$ of $\Loop_{0}(\G)$. The distribution of the collection $(H_{l})_{l\in\Loop_{0}(\G)}$ is completely described by the distribution of the finite collection $(H_{\lambda_{F}})_{F\in \F^{b}}$, which is a collection of independent $G$-valued random variables such that for each bounded face $F$, $H_{\lambda_{F}}$ has the distribution of the Brownian motion on $G$ stopped at the time equal to the Euclidean area of $F$.\\

By the Brownian motion on $G$, we mean here the Markov process started from the unit element and whose generator is $\frac{1}{2}\Delta$, where $\Delta$ is the Laplace-Beltrami operator on $G$ corresponding to the bi-invariant Riemannian metric induced by the invariant scalar product on $\g$.

\subsection{Large $N$ limits}

Quantum gauge theories are used in the description of three of the four fundamental interactions between elementary particles, and from this perspective, the group $G$ characterises the kind of interaction which one is describing. Corresponding to the electromagnetic, weak and strong interaction, are respectively the groups $\U(1)$, $\SU(2)$ or $\U(2)$, and $\SU(3)$. In the paper \cite{tHooft} published in 1974, G. 't Hooft, who was trying to understand better quark confinement, considered gauge theories with larger structure groups, namely the unitary groups $\U(N)$, and observed that many quantities of interest become simpler in the limit where $N$ tends to infinity.

After the publication of this seminal work, the large $N$ behaviour of gauge theories was extensively studied by physicists (see for example \cite{Kazakov,KazakovKostov,MakeenkoMigdal,Polyakov}), and the idea emerged that there should be a universal deterministic large $N$ limit to a broad class of matrix models (see \cite{GopakumarGross} and the references therein). This limit was named the {\em master field} and it is the main object of study of the present paper.

In the mathematical literature, there are very few papers devoted to the master field on the Euclidean plane. The first (\cite{Singer}) was published by I. Singer in 1995. In this paper, Singer described conjecturally the master field as a deterministic object whose properties would be naturally expressed in the language of free probability, and which would give rise, through a universal geometric construction sketched by Kobayashi, to a connection on a principal bundle over $\R^{2}$ with structure group the group of the unitaries of a II${}_{1}$ factor. He also gave, without proof, a correct explicit expression of the limit of the expectation of the trace of the holonomy along the loop which goes $n$ times along the boundary of a disk of area $t$, for all $t\geq 0$ and all $n\in \Z$.

The other mathematical contributions to the study of the master field are due to A. Sengupta, who started to investigate the problem in \cite{SenguptaFN,SenguptaMF}, and, during the preparation of the present work, gave with M. Anshelevitch in \cite{AnshelevitchSengupta} the first construction at a mathematical level of rigour of the master field on the plane. Their approach is based on the use of free white noise and of free stochastic calculus. It differs from the one which we follow here pretty much in the same way Sengupta's original construction of the Yang-Mills measure \cite{SenguptaAMS} differed from that given by the author in \cite{LevySMF}. 

Let us mention that the large $N$ limit of the two-dimensional Yang-Mills theory was specifically studied by Gross, Taylor and Matytsin \cite{GrossMatytsin,GrossTaylor2,GrossTaylor}, but in relation with string theory rather than with the master field. We studied some of the formulas displayed in these papers in our previous work \cite{LevyAIM}, but we do not pursue this investigation in the present paper.

\subsection{The master field} The master field is a non-commutative stochastic process indexed by $\Loop_{0}(\R^{2})$, that is, a collection $(h_{l})_{l\in \Loop_{0}(\R^{2})}$ of elements of a complex involutive unital algebra $\A$ endowed with a tracial state $\tau$. The distribution of this process, which is by definition the value of $\tau(h_{l_{1}}^{\epsilon_{1}}\ldots h_{l_{r}}^{\epsilon_{r}})$ for all $r\geq 1$, all $l_{1},\ldots,l_{r}\in \Loop_{0}(\R^{2})$ and all $\epsilon_{1},\ldots,\epsilon_{r}\in \{1,*\}$, is uniquely characterised by the following properties.\\

\indent ${\rm MF}_{1}$. It is unitary and anti-multiplicative in the sense that the equalities $h_{l_{1}^{-1}}=h_{l_{1}}^{*}=h_{l_{1}}^{-1}$ and $h_{l_{1}l_{2}}=h_{l_{2}}h_{l_{1}}$ hold for any two loops $l_{1}$ and $l_{2}$.\\
\indent ${\rm MF}_{2}$. It is continuous, in the sense that if $(l_{n})_{n\geq 0}$ is a sequence of loops which converges in $1$-variation towards a loop $l$, then the sequence $(h_{l_{n}})_{n\geq 0}$ converges in $L^{2}(\A,\tau)$ towards $h_{l}$.\\
\indent ${\rm MF}_{3}$. For each graph $\G$ traced on $\R^{2}$ such that $0$ is one of the vertices of $\G$ and each lasso basis $\{\lambda_{F} : F\in \F^{b}\}$ of $\Loop_{0}(\G)$, the finite collection $(h_{\lambda_{F}})_{F\in \F^{b}}$ is a collection of mutually free non-commutative random variables such that for each bounded face $F$, $h_{\lambda_{F}}$ has the distribution of the free unitary Brownian motion taken at the time equal to the Euclidean area of $F$. This means that if $F$ has area $t$, then for all $n\in \N$, one has the equalities
\begin{equation}\label{moment lambdaF}
\tau(h_{\lambda_{F}}^{n})=\tau(h_{\lambda_{F}}^{-n})=e^{-\frac{nt}{2}}\sum_{k=0}^{n-1}\frac{(-t)^{k}}{k!}\binom{n}{k+1}n^{k-1}.
\end{equation}

This description of the master field is of course meant to emphasise its close relationship to the Yang-Mills process. Let us emphasise a specificity of this anti-multiplicative non-commutative process indexed by a group. The properties of anti-multiplicativity and unitarity imply that the distribution of the master field is completely described by the complex-valued function 
\begin{align*} \Phi: \Loop_{0}(\R^{2}) & \longrightarrow \C\\
l&\longmapsto \tau(h_{l}).
\end{align*}
Indeed, any quantity of the form $\tau(h_{l_{1}}^{\epsilon_{1}}\ldots h_{l_{r}}^{\epsilon_{r}})$ can be computed from $\Phi$ alone because it is equal to $\Phi(l_{r}^{\epsilon_{r}}\ldots l_{1}^{\epsilon_{1}})$, with the convention that $l^{*}=l^{-1}$ for any loop $l$. From a more abstract point of view, the function $\Phi$ extends by linearity to a state on the complex algebra of the group $\Loop_{0}(\R^{2})^{op}$ and on the pair $(\A,\tau)=(\C[\Loop_{0}(\R^{2})^{op}],\Phi)$, the non-commutative process $(h_{l}=l)_{l\in \Loop_{0}(\R^{2})}$ is a realisation of the master field.

The main results of the present work are the construction of the master field as defined by the properties above, the convergence of the Yang-Mills process with structure group $\SO(N)$, $\U(N)$, or $\Sp(N)$ to the master field as $N$ tends to infinity, and the computation of the function $\Phi$. As we shall explain now, this involves a study a the large $N$ limit of the Brownian motions on the unitary, orthogonal and symplectic groups.

\subsection{Brownian motions}\label{sec intro BM}
Among the three properties which we used to characterise the Yang-Mills measure, the most specific is the third, which involves the Brownian motion on the compact connected Lie group $G$. Taking the large $N$ limit of this measure means setting $G=\U(N)$, the unitary group of order $N$, and letting $N$ tend to infinity. It is thus not surprising that the first step in the study of the master field is the study of the large $N$ limit of the Brownian motion on $\U(N)$.

The description of the limit and the proof of the convergence were achieved by P. Biane in 1995, in \cite{Biane}. Let us recall his result. For each $N\geq 1$, endow the Lie algebra $\u(N)$ of the unitary group $\U(N)$ with the scalar product $\langle X,Y\rangle_{\u(N)}=-N\Tr(XY)$. Denote by $(U_{N,t})_{t\geq 0}$ the associated Brownian motion on $\U(N)$ issued from the identity matrix $I_{N}$. The random matrices $\{U_{N,t} : t\geq 0\}$ form a collection of elements of the non-commutative probability space $(L^{\infty}(\Omega,\mathcal F,\P)\otimes \Mat_{N}(\C),\E\otimes \frac{1}{N}\Tr)$, where $(\Omega,\mathcal F,\P)$ denotes the underlying probability space. Biane proved that the non-commutative distribution of the collection $\{U_{N,t} : t\geq 0\}$ converges, as $N$ tends to infinity, to the distribution of a free multiplicative Brownian motion, which is by definition a collection $\{u_{t} : t\geq 0\}$ of unitary elements of a non-commutative probability space $(\A,\tau)$ such that the process $(u_{t})_{t\geq 0}$ has free and stationary increments, and such that these increments have the distribution whose moments are given by \eqref{def mu n} and \eqref{moments nu}, and which are also those of \eqref{moment lambdaF}.

To say that there is convergence of the non-commutative distributions means that for each integer $r\geq 1$, all $t_{1},\ldots,t_{r}\geq 0$ and all $\epsilon_{1},\ldots,\epsilon_{r}\in \{1,*\}$, one has the convergence
\[\lim_{N\to\infty}\E\left[\frac{1}{N}\Tr \left(U_{N,t_{1}}^{\epsilon_{1}} \ldots U_{N,t_{r}}^{\epsilon_{r}}\right)\right]=\tau(u_{t_{1}}^{\epsilon_{1}}\ldots u_{t_{r}}^{\epsilon_{r}}).\]

The first result of the present work extends Biane's convergence result to Brownian motions on the special orthogonal and symplectic groups (this is Theorem \ref{limite brown}). With the appropriate normalisation of these Brownian motions, the limiting non-commutative process is the same as in the unitary case.

This first result, combined with previously known results of asymptotic freeness, suffices to imply the existence of a large $N$ limit to the collection of random matrices $(H_{l})_{l\in \Loop_{0}(\G)}$ when $\G$ is any graph on $\R^{2}$ containing the origin as one of its vertices. In fact, it even allows to prove the existence of the large $N$ limit for all piecewise affine loops at once.

At this point, there is no information about the speed of convergence of $\E\left[\frac{1}{N}\Tr(H_{l})\right]$ to $\tau(h_{l})$ when $l$ is piecewise affine, and it has yet to be proved that a similar convergence holds for an arbitrary loop $l$.

Given an arbitrary loop $l$, we can approximate it by a sequence $(l_{n})_{n\geq 0}$ of piecewise affine loops, and we have the following diagram:
\[\xymatrix{ \E\left[\frac{1}{N}\Tr(H_{l_{n}})\right] \ar[d]_{N\to \infty} \ar[r]^{n\to\infty} &  \E\left[\frac{1}{N}\Tr(H_{l})\right] \ar[d]^{N\to\infty} \\ \tau(h_{l_{n}}) \ar[r]^{n\to\infty} & ?}\]
The top horizontal convergence is granted by the stochastic continuity of the Yang-Mills measure and the left vertical convergence by our study of Brownian motions and our understanding of the structure of the group of loops in a graph. 

In order to complete the diagram, we prove, and this is the second result of the present work, that the left vertical convergence occurs at a speed which is controlled by the length of the loop $l_{n}$. More precisely, we prove (see Theorem \ref{conv unif long}) that for each piecewise affine loop $l_{n}$, whose length is denoted by $\ell(l_{n})$, and for all $N\geq 1$, we have the inequality
\begin{equation}\label{unif intro}
\left|\E\left[\frac{1}{N}\Tr(H_{l_{n}})\right]-\tau(h_{l_{n}})\right|\leq \frac{1}{N} \ell(l_{n})^{2}e^{\ell(l_{n})^{2}}.
\end{equation}
This is strong enough to allow us to conclude that the master field exists, is the large $N$ limit of the Yang-Mills field, and has the properties by which we characterised it above. In fact, we prove not only that the expectation of $\frac{1}{N}\Tr(H_{l_{n}})$ converges, but in fact that the random variable itself converges in probability, with an explicit bound on its variance.

In order to prove \eqref{unif intro}, one expresses the loop $l_{n}$ as a word in the elements of a lasso basis of a graph, and then applies a quantitative version of Theorem \ref{limite brown}, our first result of convergence. It is thus on one hand necessary to control the speed at which the expectation of the trace of a word of independent Brownian motions converges to its limit, and this involves a certain measure of the complexity of the word. This is done in Theorem \ref{main estim}. It is on the other hand necessary to prove that, by an appropriate choice of the lasso basis of the group of loops in a graph, the complexity of the word which expresses a loop in this graph can be controlled by its length. This is done in Proposition \ref{NCA w l}. For a definition of the measure of the complexity of a word which we use, see \eqref{def NCA w}.

The final product of this study is in a sense nothing more than a complex-valued function on the set of rectifiable loops on $\R^{2}$, the function $\Phi$ defined by
\begin{equation}\label{limit intro}
\forall l\in \Loop_{0}(\R^{2}), \; \Phi(l)=\mathop{P{\rm \mbox{-}lim}}_{N\to\infty} \frac{1}{N}\Tr(H_{l}).
\end{equation}
We prove that the function $\Phi$ is in fact real-valued, satisfies $\Phi(l^{-1})=\Phi(l)$, is bounded by $1$, and continuous in the topology of $1$-variation (see Theorem \ref{main mf}).

In the course of our study, we devote some attention to the structure of the set of loops $\Loop_{0}(\R^{2})$. We explained at the beginning of this introduction that it is naturally a monoid, and that it can be turned into a group by taking its quotient by an appropriate sub-monoid. The crucial technical ingredient of this construction is due to B. Hambly and T. Lyons \cite{HamblyLyons}, and it is in their work a consequence of a difficult theorem which states that a rectifiable path is uniquely characterised by its signature. We offer a more elementary proof of the property which is needed for the construction of the group $\Loop_{0}(\R^{2})$ (see Proposition \ref{tree like p 2}). Our proof is based on classical topological arguments, in particular on a result of Fort \cite{Fort} in the proof of which we fill a small gap. The method applies in a slightly more general setting than what we strictly need, and applies to loops which are not necessarily rectifiable, but whose range has Hausdorff dimension strictly smaller than $2$.

At a somewhat heuristic level, it seems in the end that the function $\Phi$, which is thus defined on the genuine group $\Loop_{0}(\R^{2})$, can be thought of as the character of an infinite-dimensional unitary self-dual representation of this group. We feel that a geometrically natural realisation of this representation has yet to be given, and that the work of Anshelevitch and Sengupta \cite{AnshelevitchSengupta} might contain promising leads in this direction.

\subsection{The Makeenko-Migdal equations}\label{introsec mm} Since all the information about the master field is contained in the function $\Phi : \Loop_{0}(\R^{2})\to [-1,1]$, it is natural to seek efficient ways of actually computing it. The last section of the present work (Section \ref{variation area}) is entirely devoted to this question, and is inspired by the work of several physicists on this problem, in particular Kazakov, Kazakov and Kostov, and Makeenko and Migdal \cite{Kazakov, KazakovKostov, MakeenkoMigdal}. 

Taking advantage of the continuity of $\Phi$, we restrict our attention to the class of loops, which we call elementary, which have transverse and finite self-intersection. For such a loop $l$, the strategy is to see $\Phi(l)$ as a function of the areas of the bounded connected components of the complement of the range of $l$. This is consistent with the approach which we used to derive quantitative estimates for Brownian motions in Section \ref{section:speed}. 

Our first result is that for all elementary loop $l$, the number $\Phi(l)$ and each of its approximations $\E\left[\frac{1}{N}\Tr(H_{l})\right]$ can be computed by solving a finite first-order linear differential system with constant coefficients. This is Theorem \ref{matrice M}. The system is however very large in general, and this theorem is far from providing us with an efficient algorithm for the computation of $\Phi$.

The key to the improvement of this result is the discovery by Makeenko and Migdal that the alternated sum of the derivatives of $\Phi(l)$ with respect to the areas of the four faces surrounding a point of self-intersection of $l$ is equal to $\Phi(l_{1})\Phi(l_{2})$, where $l_{1}$ and $l_{2}$ are the two loops which are formed by changing the way in which the two strands of $l$ which are incoming at the self-intersection point are connected to the two outgoing strands (see Figure \ref{mmintro}). 

\begin{figure}[h!]
\begin{center}
\includegraphics{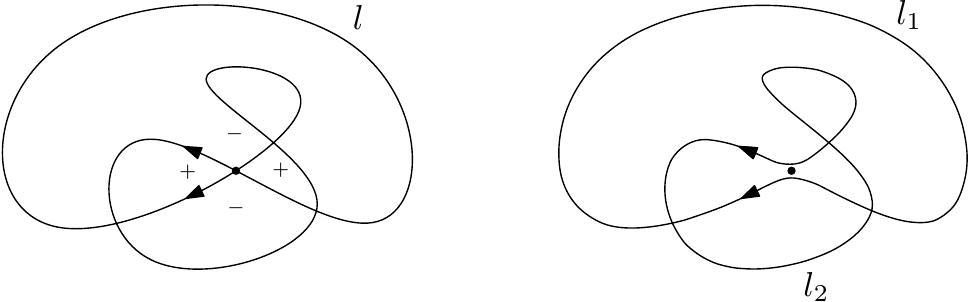}
\caption{\label{mmintro} A graphical representation of the Makeenko-Migdal equations in the large $N$ limit. The signs indicate with respect to the areas of which faces the derivatives must be taken, and with which signs.}
\end{center}
\end{figure}

The derivation of this formula by Makeenko and Migdal was based on an integration by parts with respect to the ill-defined measure $\mu_{\YM}$. Our second result is a proof that equations of which the Makeenko-Migdal equations are a particular case hold. This is the content of Propositions \ref{derive loc inv}, \ref{prop: mmeq} and Theorem \ref{statement mmmf}. This allows us to simplify greatly the algorithm of computation of $\Phi$ (see Theorem \ref{recursion}). Finally, we prove that a system of coordinates proposed by Kazakov allows one to simplify even further the formulation of the algorithm (see Proposition \ref{algo kazakov}). This allows us for example to prove that, the combinatorial structure of a loop $l$ being fixed, and the areas of the faces which it delimits being allowed to vary, $\Phi(l)$ is a polynomial function of these areas and the exponential of $-\frac{1}{2}$ times these areas. This is Proposition \ref{mf poly}.

It is possible a posteriori to give the following axiomatic description of the function $\Phi$. It is the unique real-valued function on $\Loop_{0}(\R^{2})$ with the following properties.\\

\indent $\Phi_{1}.$ It is continuous in the topology of $1$-variation.\\
\indent $\Phi_{2}.$ It is invariant under the group of diffeomorphisms of $\R^{2}$ which preserve the area.\\
\indent $\Phi_{3}.$ If $l$ is the constant loop at the origin, then $\Phi(l)=1$.\\
\indent $\Phi_{4}.$ For all elementary loop $l$, the derivative of $\Phi(l)$ with respect to the area of a face adjacent to the unbounded face is equal $-\frac{1}{2}\Phi(l)$. Pictorially, we have
\[\frac{d}{dt} \Phi\left(\raisebox{-2mm}{\includegraphics[width=1.5cm]{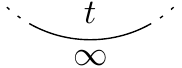}}\right)=-\frac{1}{2}\Phi\left(\raisebox{-2mm}{\includegraphics[width=1.5cm]{tinfini}}\right).\]
\indent $\Phi_{5}.$ It satisfies the Makeenko-Migdal equations. Pictorially, these write
\[\left(\frac{d}{dt_{1}}-\frac{d}{dt_{2}}+\frac{d}{dt_{3}}-\frac{d}{dt_{4}}\right)\Phi\left(\raisebox{-6.5mm}{\includegraphics[width=1.5cm]{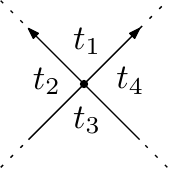}}\right)=\Phi\left(\raisebox{-6.5mm}{\includegraphics[height=1.5cm]{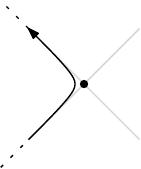}}\right)\Phi\left(\raisebox{-6.5mm}{\includegraphics[height=1.5cm]{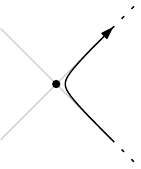}}\right).\]
In this left-hand side of this equation, we agree to replace the derivative with respect to the area of the unbounded face, should it occur, by $0$.\\

Let us emphasise that the idea, on which the proof of the Makeenko-Migdal equation is based, that certain combinatorial features of the unitary Brownian motion can be translated into combinatorial operations on loops, in relation with the computation of expectations of Wilson loops, can be traced back to the work of L. Gross, C. King and A. Sengupta \cite{GrossKingSengupta}. A related idea was present in our previous work \cite{LevyJGP}.

\subsection{The original derivation of the Makeenko-Migdal equations}
Before concluding this introduction, we would like to describe the way in which Makeenko and Migdal originally formulated and proved the equation which now bear their names. The striking contrast between the mathematically unorthodox character - to say the less - of the derivation of the equation, and the beauty and simplicity of the equation itself was one of the motivations of the author for undertaking the present study. 

Makeenko and Migdal derived their equation (see Theorem \ref{statement mmmf}) as a particular instance of the Schwinger-Dyson equations, which are the equations which one obtains by formally extending the integration by parts formula to the framework of functional integrals. The finite-dimensional prototype of these equations is the fact that for all smooth function $f:\R^{n}\to \R$ with bounded differential, and for all $h\in \R^{n}$, the equality
\[\int_{\R^{n}} d_{x}f(h) e^{-\frac{1}{2}\|x\|^{2}}\; dx=\int_{\R^{n}}\langle x,h\rangle f(x)e^{-\frac{1}{2}\|x\|^{2}}\; dx\]
holds. This equality ultimately relies on the invariance by translation of the Lebesgue measure on $\R^{n}$ and it can be proved by writing
\[0=\frac{d}{dt}_{|t=0}\int_{\R^{n}}f(x+th)e^{-\frac{1}{2}\|x+th\|^{2}}\; dx.\]

In our description of the Yang-Mills measure $\mu_{\YM}$ (see \eqref{def mu intro}), we mentioned that the measure $D\omega$ on $\A(P)$ was meant to be invariant by translations. This is the key to the derivation of the Schwinger-Dyson equations, as we will now explain.

Let $\psi:\A(P)\to \R$ be an observable, that is, a function. In general, we are interested in the integral of $\psi$ with respect to the measure $\mu_{\YM}$. The tangent space to the affine space $\A(P)$ is the linear space $\Omega^{1}(\Sigma)\otimes \Ad(P)$. The invariance of the measure $D\omega$ yields
\[0=\frac{d}{dt}_{|t=0}\int_{\A(P)}\psi(\omega+t\eta)e^{-\frac{1}{2}S_{\YM}(\omega+t\eta)}\; D\omega,\]
and the Schwinger-Dyson equations follow in their abstract form
\begin{equation}\label{SD intro}
\int_{\A(P)}d_{\omega}\psi(\eta) \; \mu_{\YM}(d\omega)=\frac{1}{2}\int_{\A(P)}\psi(\omega)d_{\omega}S_{\YM}(\eta)\; \mu_{\YM}(d\omega).
\end{equation}
The computation of the directional differential of the Yang-Mills action is standard and rigorously grounded, and its main difficulty lies in the careful unfolding of the definitions of the objects involved. It is detailed for example in \cite{Bleecker}. The expression of $d_{\omega}S_{\YM}(\eta)$ is most easily written using the covariant exterior differential $d^{\omega}:\Omega^{0}(\Sigma)\otimes \Ad(P)\to \Omega^{1}(\Sigma)\otimes \Ad(P)$ defined by $d^{\omega}\alpha=d\alpha+[\omega,\alpha]$, and it reads
\[d_{\omega}S_{\YM}(\eta)=2\int_{\Sigma}\langle \eta \wedge d^{\omega} *\!\Omega\rangle.\]
Substituting in \eqref{SD intro}, this yields the Schwinger-Dyson equation for the Yang-Mills measure. In order to extract information from it, one must choose an appropriate observable and a direction of derivation.

Assuming that $G$ is a matrix group, and with the unitary group $\U(N)$ in mind, Makeenko and Migdal applied \eqref{SD intro} to the observable defined by choosing a loop $l$ on $\Sigma$, an element $X\in \g$ and setting, for all $\omega\in \A(P)$,
\[\psi_{l,X}(\omega)=\Tr(X\hol(\omega,l)).\]
To make this definition perfectly meaningful, one should rather think of $X$ as an element of the fibre of $\Ad(P)$ over the base point of $l$. Alternatively, one can choose of a reference point in the fibre of $P$ over the base point of $l$. We will assume that such a point has been chosen, and that holonomies are computed with respect to this point. 

If we choose a parametrisation $l:[0,1]\to \Sigma$ of $l$, then the directional derivative of the observable $\psi_{l,X}$ in the direction $\eta\in \Omega^{1}(\Sigma)\otimes \Ad(P)$ is given by
\begin{equation}\label{dpsiomega}
d_{\omega}\psi_{l,X}(\eta)=-\int_{0}^{1}\Tr\left(X\hol(\omega,l_{[t,1]})\eta(\dot l(t)) \hol(\omega,l_{[0,t]})\right) \; dt,
\end{equation}
where we denote by $l_{[a,b]}$ the restriction of $l$ to the interval $[a,b]$. At first glance, this expression may seem to require the choice of a point in $P_{l(t)}$ for each $t$, but in fact it does not, for the way in which the two holonomies and the term $\eta(\dot l(t))$ would depend on the choice of this point cancel exactly.

The final ingredient of the application of the Schwinger-Dyson equation is the choice of the direction $\eta$ in which one differentiates. Let us consider the case where $l:[0,1]\to \Sigma$ is an elementary loop, that is, a loop with transverse and finite self-intersection. Let us assume that for some $t_{0}\in (0,1)$, the equality $l(t_{0})=l(0)$ holds. The assumption that $l$ is elementary prevents the vectors $\dot l(0)$ and $\dot l(t_{0})$ from being collinear. Let us assume that $\det(\dot l(0),\dot l(t_{0}))=1$. Makeenko and Migdal choose for $\eta$ a distributional $1$-form, which one could write as
\[\forall m\in \Sigma, \forall v\in T_{m}\Sigma, \; \eta_{m}(v)=\delta_{m,l(0)} \det(\dot l(0),v) X.\]
Since $\eta$ is non-zero only at the base point of $l$, the choice of a reference point in $P_{l(0)}$ allows us to see $\eta$ as a $\g$-valued form rather than an $\Ad(P)$-valued one. On the other hand, since $\eta$ is a distribution, the equation \eqref{dpsiomega} is not invariant by change of parametrisation of $l$ and this explains why we normalised the speed at $t_{0}$ by the condition $\det(\dot l(0),\dot l(t_{0}))=1$.

With this choice of $\eta$, the directional derivative of $\psi_{l,X}$ is given by
\[d_{\omega}\psi_{l,X}(\eta)=-\Tr\left(X\hol(\omega,l_{[t_{0},1]}) X \hol(\omega,l_{[0,t_{0}]})\right).\]
Let us now specify on the Lie algebra $\u(N)$ the invariant scalar product $\langle X,Y\rangle=-N\Tr(XY)$. The directional derivative of the Yang-Mills action is given by 
\[d_{\omega}S_{\YM}(\eta)=-2\langle X, (d^{\omega}\!*\!\Omega)(\dot l(0))\rangle=-2N\Tr\left(X d^{\omega}\!*\!\Omega(\dot l(0))\right),\]
or so it seems from a naive computation. We shall soon see that this expression needs to be reconsidered.

The Schwinger-Dyson equation for the observable $\psi_{l,X}$ and the derivation in the direction $\eta$ as we have obtained it reads
\begin{align*}
\int_{\A(P)}\Tr\left(X\hol(\omega,l_{[t_{0},1]}) X \hol(\omega,l_{[0,t_{0}]})\right)\;\mu_{\YM}(d\omega)&=\\
&\hspace{-2cm} N\int_{\A(P)}\Tr(X\hol(\omega,l)) \Tr(X d^{\omega}\! * \! \Omega(\dot l(0))) \;\mu_{\YM}(d\omega).
\end{align*}
Let us take the sum of these equations when $X$ takes all the values $X_{1},\ldots,X_{N^{2}}$ of an orthonormal basis of $\u(N)$. With the scalar product which we chose, the relations 
\[\sum_{k=1}^{N^{2}}\Tr(X_{k}AX_{k}B)=-\frac{1}{N}\Tr(A)\Tr(B) \mbox{ and } \sum_{k=1}^{N^{2}}\Tr(X_{k}A)\Tr(X_{k}B)=-\frac{1}{N}\Tr(AB)\]
hold for any two matrices $A$ and $B$, so that we find
\begin{align*}
\int_{\A(P)}\frac{1}{N}\Tr(\hol(\omega,l_{[0,t_{0}]}))\frac{1}{N}\Tr(\hol(\omega,l_{[t_{0},1]}))\; \mu_{\YM}(d\omega)&=\\
&\hspace{-2cm}\int_{\A(P)}\frac{1}{N}\Tr\left(\hol(\omega,l) d^{\omega}\! * \! \Omega(\dot l(0))\right)\; \mu_{\YM}(d\omega).
\end{align*}

There remains to interpret both sides of this equation. For the left-hand side, this is easily done thanks to \eqref{limit intro}, at least in the limit where $N$ tends to infinity. Indeed, the integrand converges to the constant $\Phi(l_{[0,t_{0}]})\Phi(l_{[t_{0},1]})$ and the integral converges towards the same limit, which is the right-hand side of the Makeenko-Migdal equation as written at the end of Section \ref{introsec mm}, and labelled $\Phi_{5}$.

In order to understand the right-hand side, we must interpret the term $d^{\omega}*\!\Omega(\dot l(0))$. This interpretation relies on two facts. The first is that $d^{\omega}$ acts by differentiation in the horizontal direction. More precisely, if $s$ is a section of $\Ad(P)$, then
\[\hol(\omega,l)d^{\omega} s(\dot l(0))=\frac{d}{dt}_{|t=0} \hol(\omega,l_{[t,1]})s(l(t))\hol(\omega,l_{[0,t]}).\]
The second fact is that $*\Omega$ computes the holonomy along infinitesimal rectangles. More precisely, for all $m\in \Sigma$ and all vectors $v,w\in T_{m}\Sigma$ such that $\det(v,w)=1$, one has
\[*\Omega(m)=\lim_{\epsilon\to 0} \frac{\hol(\omega,R_{\epsilon}(v,w))-I_{N}}{\epsilon},\]
where $R_{\epsilon}$ is the rectangle of which $m$ is a corner and with sides $\sqrt{\epsilon} v$ and $\sqrt{\epsilon} v$, so as to have area $\epsilon$. We will choose, in order to build our infinitesimal rectangles, $v=-\dot l(0)$ and $w$ its image by a rotation of $\frac{\pi}{2}$.

We will now combine these two facts. However, before applying blindly the formula which computes $d^{\omega}s(\dot l(0))$, we must remember where this term comes from, namely the computation of the exterior product of the distributional form $\eta$ with the form $d^{\omega}*\!\Omega$. A more prudent analysis of what we could mean by the distributional form $\eta$ makes it plausible that, instead of a derivative with respect to $t$ at $t=0$, we should have the difference between the values at $0^{+}$ and at $0^{-}$, which we denote by $\Delta_{t=0}$.

With all this preparation, the right-hand side of the Schwinger-Dyson equation can finally be drawn as follows. In the following picture, $\dot l(0)$ points north-eastwards.
\[\Delta_{|t=0} \frac{d}{d\epsilon}_{|\epsilon=0} \Phi\left(\;\raisebox{-10mm}{\scalebox{1}{\includegraphics{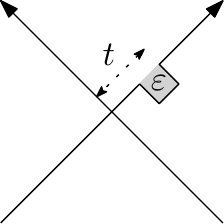}}}\;\right)=\frac{d}{d\epsilon}_{|\epsilon=0} \Phi\left(\;\raisebox{-10mm}{\scalebox{1}{\includegraphics{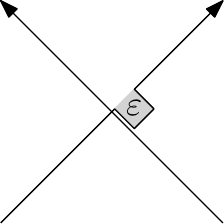}}}\;\right)- \frac{d}{d\epsilon}_{|\epsilon=0}\Phi\left(\;\raisebox{-10mm}{\scalebox{1}{\includegraphics{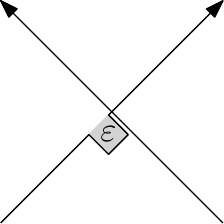}}}\;\right)\]
This is indeed the left-hand side of the Makeenko-Migdal equation $\Phi_{5}$.

\subsection{Structure of the paper} The present paper is organised in two parts of three sections each. The first part is purely devoted to the study of the Brownian motion on the orthogonal, unitary and symplectic groups, and contains no mention of the Yang-Mills measure. 

After establishing some notation and collecting some preliminary information in Section \ref{sec brown classical}, we give in Section \ref{section one BM} a short proof of Biane's convergence result in the unitary case and adapt our argument to prove that, with the correct normalisation of the invariant scalar products on $\so(N)$ and $\sp(N)$, the orthogonal and symplectic Brownian motions have the same large $N$ limit as the unitary Brownian motion (Theorem \ref{limite brown}). Our approach has a combinatorial flavour and aims at proving that the differential systems satisfied by the moments of the limiting distributions are the same as in the unitary case. The main novelty in the orthogonal and symplectic cases is the appearance of Brauer diagrams in the combinatorics, instead of permutations alone. This is ultimately due to the form taken in these cases by the first fundamental theorem of invariant theory. The set of Brauer diagrams includes in a natural sense the symmetric group and our analysis shows that the diagrams which are not permutations do not contribute to the large $N$ limit. 

In this first part, we consistently try to emphasise the similarities between the orthogonal, unitary and symplectic groups, in particular by treating them respectively as the real, complex and quaternionic unitary groups. 

The second part of the paper is devoted to the application of the results of the first part to the study of the master field. In Section \ref{section YM}, we recall the parts of the construction of the Yang-Mills measure which are useful for our purposes. Section \ref{mf} contains the heart of this paper, namely the construction of the master field as the large $N$ limit of the Yang-Mills field. The key technical tool for this, which we have not mentioned in this introduction, is the maximal Amperean area of a loop, which allows us to compare the length of an elementary loop to the complexity of the word which expresses it in a lasso basis of the graph on which it is traced. In the course of Section \ref{mf}, we study the group of rectifiable loops, along the lines mentioned in Section \ref{sec intro BM} above.
Finally, in \ref{variation area}, we discuss the actual computation of the function $\Phi$, the Makeenko-Migdal equations and their analysis by Kazakov. 

\part{Large $N$ limit of Brownian motions.}\label{Large N BM}

In the first part of this paper, we study the {\em large $N$ limit} of the Brownian motion on a compact matrix group and prove two main convergence results. In the first result, we consider the distribution of the eigenvalues of a matrix taken in a compact matrix group under the heat kernel measure at a fixed time, and prove the convergence of this distribution as the size of the group tends to infinity. By letting the size tend to infinity, we mean that we consider the three series of special orthogonal, unitary and symplectic groups $\SO(N)$, $\U(N)$ and $\Sp(N)$, and let $N$ tend to infinity. From the point of view of the asymptotic distribution of the eigenvalues, there is no difference between odd and even orthogonal groups. 

In the unitary case, the result was proved by P. Biane \cite{Biane} using harmonic analysis and, with a more combinatorial approach relying on the Schur-Weyl duality, by the author in \cite{LevyAIM}.  We recall and slightly improve the latter proof, and extend it to the orthogonal and symplectic cases by showing that the polynomial differential system which characterises the limiting moments of the distribution of the eigenvalues is the same as in the unitary case. In our treatment of this problem, we try to emphasise the similarities between the three series of classical groups by viewing each of them as the series of unitary groups over one of the three associative real division algebras. We also pay special attention to the symplectic case and to the signs involved in the multiplication of the elements of the Brauer algebra, especially to the one which is hidden behind one of the very last sentences\footnote{\label{citation Brauer}{\em One has, however, to add a factor $\phi(S_{1},S_{2})$ on the right side, whose value is $+1$, $-1$ or $0$.} Brauer unfortunately does not give the definition of $\phi(S_{1},S_{2})$.} of Brauer's original article \cite{Brauer}, on which a substantial part of the literature seems ultimately to rely.

Our first main result, combined with a general property of asymptotic freeness for large independent and rotationally invariant matrices, proved by D. Voiculescu in the unitary case (see \cite{VDN}) and by B. Collins and P. \'Sniady in the orthogonal and symplectic case (see \cite{CollinsSniady}), implies a convergence result for expected traces of words of independent matrices taken under the heat kernel measures at various times. Our second main result is an explicit estimate of the speed of this convergence in terms of a certain measure of the complexity of the word under consideration and which we call its non-commutative Amperean area. This notion turns out to be very well suited to the study which we develop in the second part of this work of the large $N$ limit of the Yang-Mills theory on the Euclidean plane.

This first part is divided in three sections. In the first section, we define the Brownian motions which we consider, with the appropriate normalisations, and compute explicitly the Casimir elements of the various Lie algebras involved. Then, the second section is devoted to the proof of our first main theorem and the third and last section to the proof of our second main theorem.

\section{Brownian motions on classical groups}\label{sec brown classical}

In this section, we define the Brownian motion on the orthogonal, unitary, and symplectic groups and establish a concise formula for the expected value of any polynomial function of the entries of a sample of this Brownian motion at a given time. To the extent possible, we treat the three cases on the same footing, by seeing them as the unitary group over the reals, complex numbers, and quaternions. In particular, we avoid as much as possible considering the symplectic group $\Sp(N)$ as a subgroup of $\U(2N)$. 

\subsection{Classical groups}

Let $\K$ be one of the three associative real division algebras $\R$, $\C$ and $\H$. If $x\in \K$, we denote by $x^{*}$ the conjugate of $x$. If $M\in \Mat_{N}(\K)$, the adjoint of $M$ is the matrix $M^{*}$ defined by $(M^{*})_{ab}=(M_{ba})^{*}$. We consider the following compact real Lie group, which depend on an integer $N\geq 1$:
\[\U(N,\K)=\{M\in \Mat_{N}(\K) : M^{*}M=I_{N}\}^{0},\]
where the exponent $0$ indicates, for the needs of the real case, that we take the connected component of the unit element. 
The Lie algebra of this Lie group is the real vector space
\[\u(N,\K)=\{X\in \Mat_{N}(\K) : X^{*}+X=0\}.\]
We thus have the following table, in which we include the value of classical parameter $\beta=\dim_{\R}\K$.
\begin{equation}\label{rch}
\begin{array}{|c|ccc|}
\hline  & \U(N,\K) & \u(N,\K) & \beta \\
\hline \R & \SO(N) & \so(N) & 1 \\
\C & \U(N) & \u(N) & 2 \\
 \H & \Sp(N) & \sp(N) & 4 \\ \hline
\end{array}
\end{equation}

Let ${\mathfrak a}_{N}$ and ${\mathfrak s}_{N}$ denote respectively the linear spaces of skew-symmetric and symmetric real matrices of size $N$. Denoting by $\{1,\ii,\j,\k\}$ the standard $\R$-basis of $\H$, we have the equalities
\begin{equation}\label{algebres de Lie}
\so(N)= {\mathfrak a}_{N}, \; \u(N)= {\mathfrak a}_{N}\oplus \ii {\mathfrak s}_{N}, \mbox{ and } \sp(N)={\mathfrak a}_{N}\oplus \ii {\mathfrak s}_{N}\oplus \j {\mathfrak s}_{N}\oplus \k {\mathfrak s}_{N},
\end{equation}
from which it follows that
\begin{equation}\label{dimension}
\dim \U(N,\K)=\frac{N(N-1)}{2}+(\beta-1)\frac{N(N+1)}{2}=\frac{\beta}{2}N^{2}+\left(\frac{\beta}{2}-1\right)N.
\end{equation}

Let us add to our list the special unitary group $\SU(N)=\{U\in \U(N), \det U=1\}$ whose Lie algebra is  $\su(N)=\{X\in \u(N), \Tr(X)=0\}$, and which has dimension $N^{2}-1$.

\subsection{Invariant scalar products} The first step in defining a Brownian motion on a compact Lie group is the choice of a scalar product on its Lie algebra invariant under the adjoint action. Excepted the $1$-dimensional centre of $\U(N)$, the Lie groups which we consider are simple, so that their Lie algebras carry, up to a scalar multiplication, a unique invariant scalar product. As long as $N$ is fixed, a rescaling of the scalar product corresponds merely to a linear time-change for the Brownian motion. However, since we are going to let $N$ tend to infinity, the way in which we normalise the scalar products matters. 

Let $\Tr:\Mat_{N}(\K)\to \K$ denote the usual trace, so that $\Tr(I_{N})=N$. We endow our Lie algebras with the following scalar products: 
\begin{equation}\label{normalization}
\forall X,Y \in \u(N,\K), \; \langle X,Y\rangle = \frac{\beta N}{2}\Re\Tr(X^{*}Y)=-\frac{\beta N}{2}\Re\Tr(XY), 
\end{equation}
and the scalar product on $\su(N)$ is the restriction of that on $\u(N)$. The real part is needed only for the quaternionic case, for $\Tr(X^{*}Y)$ is real whenever $X$ and $Y$ are complex anti-Hermitian.

\subsection{Casimir elements}\label{section casimir} Let $\g\subset \Mat_{N}(\K)$ be one of our Lie algebras, of dimension $d$. Let $\{X_{1},\ldots,X_{d}\}$ be an orthonormal $\R$-basis of $\g$. The tensor
\[C_{\g}=\sum_{k=1}^{d}X_{k}\otimes X_{k},\]
seen abstractly as an element of $\g \otimes \g$ or more concretely as an element of $\Mat_{N}(\K)\otimes_{\R}\Mat_{N}(\K)$, does not depend on the choice of the orthonormal basis. It is called the Casimir element of $\g$. 

Let $\{E_{ab} :a,b=1\ldots N\}$ denote the set of elementary matrices in $\Mat_{N}(\R)$, defined by $(E_{ab})_{ij}=\delta_{i,a}\delta_{j,b}$. 
Let us define two elements $T$ and $P$ of $\Mat_{N}(\R)^{\otimes 2}$ by
\begin{equation}\label{tau kappa}
T=\sum_{a,b=1}^{N}E_{ab}\otimes E_{ba} \mbox{ and } P=\sum_{a,b=1}^{N} E_{ab}\otimes E_{ab}.
\end{equation}
The letters $T$ and $P$ stand respectively for {\em transposition} and {\em projection}. The operators $T$ and $P$ can conveniently be depicted as in Figure \ref{TetW} below.

\begin{figure}[h!]
\begin{center}
\scalebox{0.7}{\includegraphics{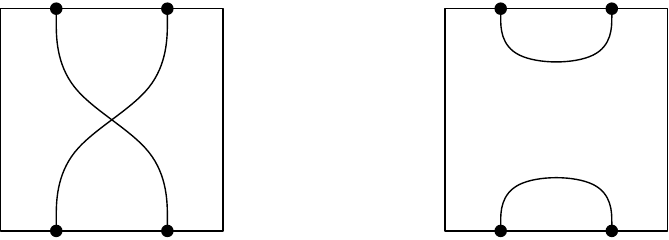}}
\caption{\label{TetW}The operators $T$ and $P$.}
\end{center}
\end{figure}

On the other hand, set $\I(\K)=\{1,\ii,\j,\k\}\cap \K$ and let us define two elements ${\rm Re}^{\K}$ and ${\rm Co}^{\K}$ of $\K\otimes_{\R} \K$ by
\begin{equation}\label{re co}
{\rm Re}^{\K}=\sum_{\gamma\in \I(\K)} \gamma \otimes \gamma ^{-1} \mbox{ and } {\rm Co}^{\K}=\sum_{\gamma\in \I(\K)} \gamma \otimes \gamma .
\end{equation}
The names $\rm Re$ and $\rm Co$ stand for {\em real part} and {\em conjugation}, with the quaternionic case in mind. Indeed, the following two relations hold, which will prove very useful: for all quaternion $q$,
\begin{equation}\label{re co 2}
q-\ii q \ii -\j q \j - \k q\k =4\Re(q) \mbox{ and } q+\ii q \ii +\j q \j + \k q\k =-2q^{*}.
\end{equation}

In the next lemma, and later in this work, we will use the natural identifications $\Mat_{N}(\K)\simeq \Mat_{N}(\R)\otimes \K$ and $\Mat_{N}(\K)^{\otimes n}\simeq \Mat_{N}(\R)^{\otimes n}\otimes \K^{\otimes n}$.

\begin{lemma}\label{lemme casimir} The Casimir element of $\u(N,\K)$ is given by
\begin{equation}\label{casimir general}
C_{\u(N,\K)}=\frac{1}{\beta N} \left(-T \otimes {\rm Re}^{\K} +P \otimes {\rm Co}^{\K} \right).
\end{equation}

Moreover, $C_{\su(N)}=C_{\u(N)}-\frac{1}{N^{2}}\ii I_{N}\otimes \ii I_{N}$.
\end{lemma}

\begin{proof} The spaces ${\mathfrak a}_{N}$ and ${\mathfrak s}_{N}$, each endowed with the scalar product $\langle X,Y \rangle=\frac{1}{2}\Tr(X^{*}Y)$ are Euclidean spaces in which we can compute the sum of the tensor squares of the elements of an orthonormal basis. We find $C_{{\mathfrak a}_{N}}=-T+P$ and $C_{{\mathfrak s}_{N}}=T+P$. The result then follows from \eqref{algebres de Lie} and \eqref{normalization}.
\end{proof}

Because tensor products in \eqref{casimir general} are over $\R$, the expression in the case of $\U(N)$ is not the most natural one. From now on, let us make the convention that tensor products are on $\R$ when we deal with orthogonal or symplectic matrices, and over $\C$ when we deal with unitary ones. Then in particular ${\rm Re}^{\C}=2$ and ${\rm Co}^{\C}=0$.
Thus, we have
\begin{equation}\label{casimir o u}
C_{\so(N)}=-\frac{1}{N}(T-P) \mbox{ and } C_{\u(N)}=-\frac{1}{N}T.
\end{equation}

The explicit expression \eqref{casimir general} of the Casimir operators allows us to compute any expression of the form $\sum_{k=1}^{d} B(X_{k},X_{k})$ where $B$ is an $\R$-bilinear map. For example, we can compute the sum of the squares of the elements of an orthonormal basis.

\begin{lemma}\label{sum squares} Let $\g\subset \Mat_{N}(\K)$ be one of our Lie algebras, of dimension $d$. Let $\{X_{1},\ldots,X_{d}\}$ be an orthonormal basis of $\g$. Then $\sum_{k=1}^{d} X_{k}^{2}=c_{\g}I_{N}$, where the real constant $c_{\g}$ is given by
\begin{equation}\label{sum square unif}
c_{\u(N,\K)}=-1+\frac{2-\beta}{\beta N},
\end{equation}
and $c_{\su(N)}=-1+\frac{1}{N^{2}}$.
\end{lemma}

\begin{proof} This equality follows from Lemma \ref{lemme casimir} and the following facts: the images of $T$ and $P$ by the mapping $X\otimes Y \mapsto XY$ are respectively $NI_{N}$ and $I_{N}$ (see Figure \ref{TetW2} below for a graphical proof), and the sums $\sum_{\gamma\in \I(\K)} \gamma\gamma^{-1}$ and $\sum_{\gamma\in \I(\K)} \gamma\gamma$ are respectively equal to $\beta$ and $2-\beta$.
\end{proof}

\begin{figure}[h!]
\begin{center}
\scalebox{0.7}{\includegraphics{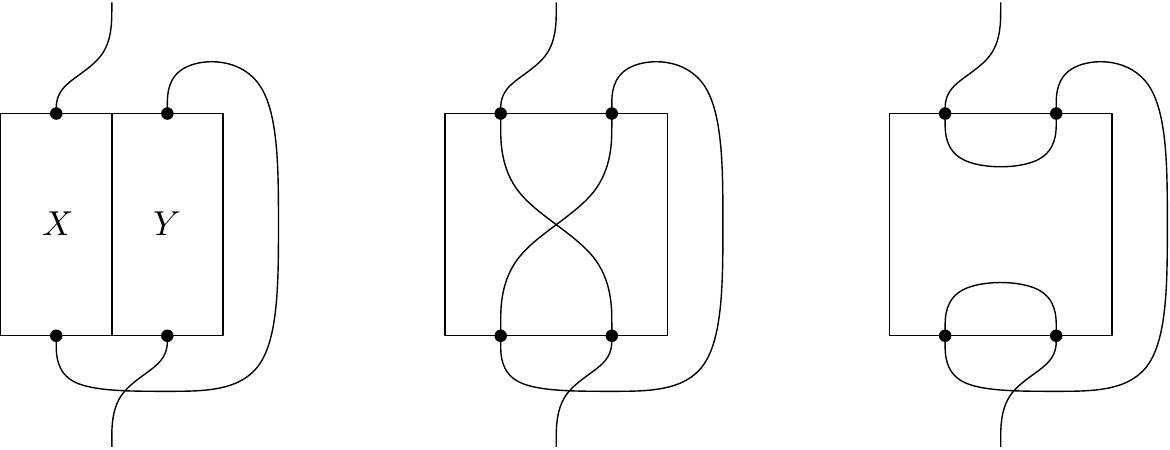}}
\caption{\label{TetW2} The images of the operators $T$ and $P$ by the mapping $X\otimes Y\mapsto XY$ can be computed graphically by joining the top right dot to the bottom left dot of the box. A loop carries a free index and produces a factor $N$.}
\end{center}
\end{figure}

\subsection{Brownian motions}\label{def BM}

Let $\g\subset \Mat_{N}(\K)$ be one of our Lie algebras and let $G$ be the corresponding group. Let $(K_{t})_{t\geq 0}$ be the linear Brownian motion in the Euclidean space $(\g,\langle \cdot, \cdot\rangle)$, that is, the continuous $\g$-valued Gaussian process such that for all $X,Y\in \g$ and all $s,t\geq 0$, one has
\[\E\left[\langle X,K_{t} \rangle \langle Y,K_{s}\rangle\right]=\min(s,t) \langle X,Y\rangle.\]
Alternatively, $K$ can be constructed by picking an orthonormal basis $(X_{k})_{k=1\ldots d}$ of $\g$, a collection $(B^{(k)})_{k=1\ldots d}$ of independent standard real Brownian motions, and by setting
\begin{equation}\label{def K}
K_{t}=\sum_{k=1}^{d} B_{t}^{(k)}X_{k}.
\end{equation}
The quadratic variation of $K$ is easily expressed in terms of the Casimir operator of $\g$: we have 
\begin{equation}\label{quadratic tens}
dK_{t}\otimes  dK_{t}=C_{\g}  dt,
\end{equation}
from which one deduces, in the same way as Lemma \ref{sum squares} was deduced from Lemma \ref{lemme casimir},
\begin{equation}\label{quadratic variation}
(dK dK)_{t}=c_{\g}I_{N}  dt.
\end{equation}

The Brownian motion on $G$ is defined as the solution $(V_{t})_{t\geq 0}$ of the following linear It\^{o} stochastic differential equation in $\Mat_{N}(\K)$:
\begin{align}\label{def brown}
\left\{\!\!\begin{array} {l}\displaystyle  dV_{t}=dK_{t}V_{t}  +\frac{c_{\g}}{2}  V_{t} dt,\\
\displaystyle  V_{0}=I_{N}.\end{array}\right.
\end{align}

\begin{lemma} With probability $1$, the matrix $V_{t}$ belongs to $G$ for all $t\geq 0$.
\end{lemma}

\begin{proof} One has $dV_{t}^{*}=-V_{t}^{*}dK_{t}+\frac{1}{2} c_\g V_{t}^{*} dt$.
Hence, It\^{o}'s formula and the expression \eqref{quadratic variation} of the quadratic variation of $K$ imply that $d(V_{t}^{*}V_{t})=0$. This proves the assertion, except for $\SU(N)$. In order to treat this case, write the stochastic differential equation satisfied by the columns of $V_{t}$ and deduce an expression of $d(\det V_{t})$. Using the fact that $\Tr(dK_{t})=0$ and the fact that  $C_{\su(N)}=-\frac{1}{N}T+\frac{1}{N^{2}}I_{N}\otimes I_{N}$ in $\Mat_{N}(\C)\otimes_{\C}\Mat_{N}(\C)$, this yields $d(\det V_{t})=0$, as expected.
\end{proof}

Let us recall some fundamental properties of this process. The reader may consult the book of M. Liao \cite{Liao} for more details.

\begin{lemma}\label{inverses} 1. The processes $(V_{t})_{t\geq 0}$ and $(V_{t}^{*})_{t\geq 0}$ have the same distribution.\\
2. The process $(V_{t})_{t\geq 0}$ has independent left increments. In other words, for all $0\leq t_{1}\leq \ldots \leq t_{n}$, the random variables $V_{t_{1}},V_{t_{2}}V_{t_{1}}^{-1},\ldots,V_{t_{n}}V_{t_{n-1}}^{-1}$ are independent. Moreover, for all $s\leq t$, the increment $V_{t}V_{s}^{-1}$ has the same distribution as $V_{t-s}$.
\\
3. The distribution of the process $(V_{t})_{t\geq 0}$ is invariant by conjugation: For all $U\in G$, the processes $(V_{t})_{t\geq 0}$ and $(UV_{t}U^{-1})_{t\geq 0}$ have the same distribution. \end{lemma}

\begin{proof} 1. Let $(L_{t})_{t\geq 0}$ be the solution of the stochastic differential equation $dL_{t}=-V_{t}^{*}dK_{t}V_{t}$, with initial condition $L_{0}=0$. The process $(L_{t})_{t\geq 0}$ is a continuous martingale issued from $0$ in $\g$. Let us show that it is a Brownian motion. This will prove the lemma, for $V^{*}$ is the solution of the equation $dV^{*}_{t}=dL_{t}V^{*}_{t}+\frac{1}{2} c_\g V_{t}^{*} dt$.

Let $(X_{1},\ldots,X_{d})$ be an orthonormal basis of $\g$. Let us write $L_{t}=\sum_{k=1}^{d} L^{(k)}_{t} X_{k}$ and $K_{t}=\sum_{k=1}^{d} B^{(k)}_{t} X_{k}$. We know that $B^{(1)},\ldots,B^{(d)}$ are independent standard real Brownian motions. For each $t\geq 0$, let $R=(R_{kl})_{k,l=1\ldots d}$ be the orthogonal matrix representing the isometric transformation $X\mapsto -V_{t}^{*}XV_{t}$ of $\g$ in the basis $(X_{1},\ldots,X_{d})$. Then for all $k\in \{1,\ldots,d\}$ we have $dL_{t}^{(k)}=\sum_{l=1}^{d} R_{kl}dB^{(l)}_{t}$, from which it follows that $L^{(1)},\ldots,L^{(d)}$ are also independent standard real Brownian motions. Hence, $L$ is a Brownian motion on $\g$.

2. The process $(W_{t})_{t\geq s}=(V_{t}V_{s}^{-1})_{t\geq s}$ is the solution of the stochastic differential equation
$dW_{t}=dK_{t}W_{t}+ \frac{1}{2} c_\g W_{t} dt$, with initial condition $W_{s}=I_{N}$. Hence, $W_{t}$ is measurable with respect to $\sigma(K_{u} : u\in [s,t])$ and has the same distribution as $V_{t-s}$. The result follows immediately.

3. This assertion follows from the fact that for all $U\in G$, the process $(UK_{t}U^{-1})_{t\geq 0}$ is a Brownian motion in $\g$.
\end{proof}

We will adopt the following notational convention: the Brownian motions on $\SO(N)$, $\U(N)$ and $\Sp(N)$ will respectively be denoted by $(R_{t})_{t\geq 0}$, $(U_{t})_{t\geq 0}$, and $(S_{t})_{t\geq 0}$.

\subsection{Expected values of polynomials of the entries}

Let $n\geq 1$ be an integer and $t\geq 0$ be a real. We give a formula for the expected value of all homogeneous polynomial functions of degree $n$ in the entries of the Brownian motion on one of our groups at time $t$. 

For all integers $i,j$ such that $1\leq i<j\leq n$, let us denote by $\iota_{i,j}:\Mat_{N}(\K)^{\otimes 2}\to \Mat_{N}(\K)^{\otimes n}$ the linear mapping defined by 
\begin{equation}\label{def iota}
\iota_{i,j}(X\otimes Y)=I_{N}^{\otimes (i-1)} \otimes X \otimes I_{N}^{\otimes (j-i-1)} \otimes Y \otimes I_{N}^{\otimes (n-j)}.
\end{equation}
We will also often write $(X\otimes Y)_{ij}$ instead of $\iota_{i,j}(X\otimes Y)$

\begin{proposition}\label{espérances groupes} Let $(V_{t})_{t\geq 0}$ be the Brownian motion on one of the groups which we consider, with Lie algebra $\g$. Let $n\geq 1$ be an integer. Let $t\geq 0$ be a real. We have
\begin{equation}\label{expectation polynomial}
\E\left[V_{t}^{\otimes n}\right]=\exp\left(\frac{nc_{\g}t}{2} + t\sum_{1\leq i<j\leq n} \iota_{i,j}(C_{\g})\right).
\end{equation}
In particular, if $(R_{t})_{t\geq 0}$ denotes the Brownian motion on $\SO(N)$, then
\begin{equation}\label{expectation O}
\E\left[R_{t}^{\otimes n}\right]=\exp\left(-\frac{N-1}{N}\frac{nt}{2} - \frac{t}{N}\sum_{1\leq i<j\leq n} T_{ij}-P_{ij}\right).
\end{equation}
If $(U_{t})_{t\geq 0}$ denotes the Brownian motion on $\U(N)$, then
\begin{equation}\label{expectation U}
\E\left[U_{t}^{\otimes n}\right]=\exp\left(-\frac{nt}{2} - \frac{t}{N}\sum_{1\leq i<j\leq n} T_{ij} \right).
\end{equation}
Finally, if $(S_{t})_{t\geq 0}$ denotes the Brownian motion on $\Sp(N)$, then
\begin{equation}\label{expectation S}
\E\left[S_{t}^{\otimes n}\right]=\exp\left(-\frac{2N+1}{2N}\frac{nt}{2} - \frac{t}{4N}\sum_{1\leq i<j\leq n} \left((T\otimes {\rm Re}^{\H})_{ij}-(P\otimes {\rm Co}^{\H})_{ij}\right) \right).
\end{equation}
Also, if $(\widetilde{U}_{t})_{t\geq 0}$ denotes the Brownian motion on $\SU(N)$, then $\E[\widetilde{U}_{t}^{\otimes n}]=\exp\left(\frac{n^{2}t}{2N^{2}}\right)\E[U_{t}^{\otimes n}]$.
\end{proposition}

\begin{proof} Both sides of \eqref{expectation polynomial} are equal to $I_{N}^{\otimes n}$ for $t=0$. Moreover, It\^o's formula for $V_{t}^{\otimes n}$ seen as an element of $\Mat_{N}(\K)^{\otimes n}$ writes
\[d\left(V_{t}^{\otimes n}\right)=\left(\sum_{i=1}^{n} I_{N}^{\otimes(i-1)}\otimes dK_{t} \otimes I_{N}^{\otimes (n-i)} + \frac{nc_{\g}}{2} + \sum_{1\leq i<j \leq 2} \iota_{i,j}(dK_{t}\otimes dK_{t})\right)V_{t}^{\otimes n}.\] 
Using \eqref{quadratic tens}, this implies that the time derivatives of both sides of \eqref{expectation polynomial} are equal.

The special unitary case follows from the unitary case and the relations $c_{\su(N)}=c_{\u(N)}+\frac{1}{N^{2}}$, $C_{\su(N)}=C_{\u(N)}+\frac{1}{N^{2}}$.
\end{proof}

\section{Convergence results for one Brownian motion}\label{section one BM}

In this section, we analyse the asymptotic behaviour of the repartition of the eigenvalues of the Brownian motion at time $t$ on $\U(N,\K)$ as $N$ tends to infinity, the time $t$ being fixed. We start by briefly discussing the issue of eigenvalues in the symplectic case.

\subsection{Moments of the empirical spectral measure}\label{moments empirique}
Let $M$ be a real or complex matrix of size $N$ with complex eigenvalues $\lambda_{1},\ldots,\lambda_{N}$. We define the empirical spectral measure of $M$ by
\[\hat\mu_{M}=\frac{1}{N}\sum_{k=1}^{N} \delta_{\lambda_{k}}.\]
The moments of this measure can be expressed as traces of powers of $M$. Indeed, for all integer $n\geq 0$, $\int_{\C} z^{n}\; \hat\mu_{M}(dz)=\frac{1}{N}\Tr(M^{n})=\tr(M^{n})$, where $\tr$ denotes the normalised trace, so that $\tr(I_{N})=1$. If $M$ is invertible, then these equalities hold for all $n\in \Z$.

For a matrix with quaternionic entries, the very notion of eigenvalue must be handled with care. A matrix $M\in \Mat_{N}(\H)$ is said to admit the right eigenvalue $q\in \H$ if there exists a non-zero vector $X\in \H^{N}$ such that $MX=Xq$. If $q$ is a right eigenvalue of $M$, then any quaternion conjugated to $q$ is also a right eigenvalue of $M$, because for all non-zero quaternion $u$, one has $M(Xu^{-1})=M(Xu^{-1})uqu^{-1}$. 

It is an elementary property of $\H$ that two quaternions are conjugated if and only if they have the same real part and the same norm. In particular, each conjugacy class of $\H$ either consists of a single real element, or meets $\C$ at exactly two conjugated non-real elements. Thus, a matrix with quaternionic entries determines real eigenvalues, which are to be counted twice, and conjugate pairs of complex eigenvalues.

It is convenient to momentarily see $\H$ as $\C\oplus \j\C$, to write any vector $X\in \H^{N}$ as $X=Z+\j W$ with $Z,W\in \C^{N}$, and to write any matrix $M\in \Mat_{N}(\H)$ as $M=A+\j B$ with $A,B\in \Mat_{N}(\C)$. The mappings $X\mapsto \upsilon(X)=\begin{pmatrix} Z \\W \end{pmatrix}$ and $M\mapsto \iota( M)=\begin{pmatrix} A & -\overline{B} \\ B & \overline{A}\end{pmatrix}$ are respectively an isomorphism of right complex vector spaces $\H^{N}\to \C^{2N}$ and an injective homomorphism of involutive algebras $\Mat_{N}(\H)\to \Mat_{2N}(\C)$. These morphisms are compatible in the sense that $\upsilon(MX)=\iota(M) \upsilon(X)$ for all $M\in \Mat_{N}(\H)$ and $X\in \H^{N}$.

It turns out that the complex eigenvalues of $\iota(M)$ are exactly the complex right eigenvalues of $M$, counted twice if they are real. Thus, $M$ admits exactly $2N$ complex right eigenvalues  $\{\lambda_{1},\lambda_{1}^{*},\ldots,\lambda_{N},\lambda_{N}^{*}\}$. We define the empirical spectral measure of $M$ as the spectral empirical measure of $\tilde M$:
\[\hat\mu_{M}=\frac{1}{2N}\sum_{k=1}^{N}\delta_{\lambda_{k}}+\delta_{\lambda_{k}^{*}}.\]
Observe that the mapping $M\mapsto \iota( M)$ does not preserve the trace, but rather verifies $\Tr(\iota( M))=2\Re\Tr(M)$. Hence, the moments of $\hat\mu_{M}$ are given by $\int_{\C} z^{n}\; \hat\mu_{M}(dz)=\frac{1}{2N}\Tr(\iota(M)^{n})=\Re\tr(M^{n})$ for all $n\geq 0$, and also for all $n\in \Z$ if $M$ is invertible. The situation is thus almost the same as in the real and complex case, the only difference being that the trace is replaced by its real part. One should however keep in mind that, from the point of view of eigenvalues, the natural non-normalised trace on $\Mat_{N}(\H)$ is {\em twice} the real part of the usual trace. Indeed, for instance, with our way of counting, the eigenvalue $1$ of $I_{N}\in\Mat_{N}(\H)$ has multiplicity $2N$.

Note finally that orthogonal and unitary matrices have eigenvalues of modulus $1$. Similarly, symplectic matrices have quaternionic right eigenvalues of norm $1$, and in all cases, the empirical spectral measures which we consider are supported by the unit circle of the complex plane, which we denote by $\UC=\{z\in \C : |z|=1\}$.

\subsection{First main result: convergence of empirical spectral measures}
Let us introduce the limiting measure which appears in our first main result and was first described by P. Biane in the unitary case. It is a one-parameter family of probability measures on $\UC$ which plays for compact matrix groups the role played for Hermitian matrices by the Wigner semi-circle law. The simplest description of this family is through its moments.

For all real $t\geq 0$ and all integer $n\geq 0$, set
\begin{equation}\label{def mu n}
\mu_{n}(t)=e^{-\frac{nt}{2}}\sum_{k=0}^{n-1} \frac{(-t)^{k}}{k!} n^{k-1}\binom{n}{k+1}.
\end{equation}
It follows from Biane's result (Theorem \ref{limit brown U} below) that there exists a probability measure $\nu_{t}$ on $\UC$ such that for all integer $n\geq 0$, one has
\begin{equation}\label{moments nu}
\int_{\UC}z^{n} \; \nu_{t}(dz)=\int_{\UC}z^{-n} \; \nu_{t}(dz)=\mu_{n}(t).
\end{equation}
There is no simple expression for the density of this measure. Nevertheless, some information about this measure can be found in \cite{Biane,LevyAIM}. The result in the unitary case is the following.

\begin{theorem}[Biane, \cite{Biane}]\label{limit brown U} Let $(U_{N,t})_{t\geq 0}$ be the Brownian motion on the unitary group $\U(N)$, or on the special unitary group $\SU(N)$. Let $r\geq 1$ be an integer and $m_{1},\ldots,m_{r}\geq 0$ be integers. Let $t\geq 0$ be a real. Then
\[\lim_{N\to\infty} \E\left[\tr(U_{N,t}^{m_{1}})\ldots \tr(U_{N,t}^{m_{r}})\right]=\mu_{m_{1}}(t)\ldots \mu_{m_{r}}(t).\]
Moreover, for all $n\in \Z$, 
\[\lim_{N\to\infty} \E\left[\tr(U_{N,t}^{n})\right] = \mu_{|n|}(t).\]
\end{theorem}

Our first main result is the following.

\begin{theorem}\label{limite brown} Let $(R_{N,t})_{t\geq 0}$ be the Brownian motion on the special orthogonal group $\SO(N)$, and $(S_{N,t})_{t\geq 0}$ be the Brownian motion on the symplectic group $\Sp(N)$. Let $r\geq 1$ be an integer and $m_{1},\ldots,m_{r}\geq 0$ be integers. Let $t\geq 0$ be a real. Then 
\[\lim_{N\to\infty} \E\left[\tr(R_{N,t}^{m_{1}})\ldots \tr(R_{N,t}^{m_{r}})\right]=\lim_{N\to\infty} \E\left[\Re\tr(S_{N,t}^{m_{1}})\ldots \Re\tr(S_{N,t}^{m_{r}})\right]=\mu_{m_{1}}(t)\ldots \mu_{m_{r}}(t).\]
Moreover, for all $n\in \Z$, 
\[\lim_{N\to\infty} \E\left[\tr(R_{N,t}^{n})\right] = \lim_{N\to\infty} \E\left[\Re\tr(S_{N,t}^{n})\right] =\mu_{|n|}(t).\]
\end{theorem}

The rest of this section is devoted to the proof of Theorem \ref{limite brown}.

\subsection{Characterisation of the moments of the limiting distribution}
Before we jump into the computation of the limiting distribution of the eigenvalues of our Brownian motions, let us say a few words about the disguise under which the moments $(\mu_{n})_{n\geq 0}$ of the limiting distribution will appear.

These moments are defined by \eqref{def mu n} and this is the form under which they appear in the original proof of Theorem \ref{limit brown U} by P. Biane. There are at least two other ways in which they are amenable to appear. The first is purely combinatorial and related to minimal factorisations of an $n$-cycle in the symmetric group $\S_{n}$. Recall the elementary fact that the $n$-cycle $(1\ldots n)$ cannot be written as a product of less than $n-1$ transpositions, and the classical fact that the number of ways of writing it as a product of exactly $n-1$ transpositions is $n^{n-2}$. More generally, the product of $(1\ldots n)$ and $k$ transpositions cannot have more than $k+1$ cycles. The following result is proved in \cite{LevySym} in a bijective way.
\begin{proposition}\label{count paths} Let $\T_{n}$ be the set of transpositions in the symmetric group $\S_{n}$. Let $k\geq 0$ be an integer. The set
\[\left\{(\tau_{1},\ldots,\tau_{k})\in (\T_{n})^{k} : (1\ldots n)\tau_{1}\ldots \tau_{k} \mbox{ has exactly } k+1 \mbox{ cycles}\right\}\]
is empty if $k\geq n$ and has otherwise $n^{k-1}\binom{n}{k+1}$ elements.
\end{proposition}
This result, combined with the equality \eqref{expectation U}, allows one to give a quick proof of Theorem \ref{limit brown U}. It is however a proof which is not easily generalised to the orthogonal and symplectic cases, because it is more difficult to count paths in the set of standard generators of the Brauer algebra than in the symmetric group.

The second way in which the moments $(\mu_{n})_{n\geq 0}$ may and in fact will appear is the following. Define a sequence of polynomials $(L_{n})_{n\geq 0}$ by setting $L_{0}(t)=1$ and, for all $n\geq 1$,
 \begin{equation}\label{def L}
 L_{n}(t)=e^{\frac{nt}{2}}\mu_{n}(t)=\sum_{k=0}^{n-1} \frac{(-t)^{k}}{k!} n^{k-1}\binom{n}{k+1}.
 \end{equation}

\begin{lemma} The sequence $(L_{n})_{n\geq 0}$ is the unique sequence of functions of one real variable such that $L_{0}=1$ and
\begin{equation}\label{rec R}
\forall n\geq 1, \; L_{n}(0)=1 \mbox{ and }\dot L_{n}=-\frac{n}{2}\sum_{k=1}^{n-1}L_{k}L_{n-k}.
\end{equation}
\end{lemma}
 
Despite the relatively simple explicit form of $L_{n}$, this statement seems to resist a direct verification. One way to prove it is to recognise the link between the recurrence relation \eqref{rec R} and the problem of enumeration of paths in the symmetric group solved by Proposition \ref{count paths}, but this could hardly be called a simple proof.

\begin{proof}
The shortest proof seems to consists in recognising that \eqref{rec R} is equivalent to an easily solved equation in the reciprocal of the generating function of the sequence $(L_{n})_{n\geq 0}$.
Indeed, consider the formal series $g(t,z)=\sum_{n\geq 1}L_{n}(t)z^{n}$. The recurrence relation \eqref{rec R} is equivalent to the differential equation $\partial_{t}g(t,z)=-zg(t,z)\partial_{z}g(t,z)$ with initial condition $g(0,z)=\frac{z}{1-z}$. This differential equation is in turn equivalent, for the reciprocal formal series $f(t,z)$, defined by $f(t,g(t,z))=z$, to the differential equation $\partial_{t}f(t,z)=zf(t,z)$, with the initial condition $f(0,z)=\frac{z}{1+z}$. This last equation is solved by $f(t,z)=\frac{z}{1+z}e^{tz}$ and Lagrange's inversion formula yields the value of the polynomials $(L_{n})_{n\geq 0}$.
\end{proof}

The reason why reciprocals of generating functions on one hand and paths of shortest length in the symmetric group on the other hand, although apparently rather remote from each other, allow one to prove Theorem \ref{limit brown U}, is that both are governed by the combinatorics of the lattice of non-crossing partitions of a cycle (see \cite{Speicher,Biane2}).

\subsection{The unitary case revisited}\label{unitary revisited}
The basis of our proof in the orthogonal and symplectic cases is the proof in the unitary case, which we review in this section. We take this opportunity to introduce useful notation, and also to offer what we believe to be a simpler and clearer proof than what can be found in the literature.

Before we start, let us make a short comment on our strategy of exposition. Rather than spending a lot of time introducing from the beginning, and with little motivation, all the tools which will be needed for the three series of groups, we have chosen to introduce the various objects progressively. The drawback of this approach is that several tools will have to be redefined, some more than once, each new definition containing and superseding the previous ones.

\begin{proof}[Proof of Theorem \ref{limit brown U}]
Let $n\geq 1$ be an integer. We denote by $\S_n$ the symmetric group of order $n$. Let  $\rho:\S_n\to \GL((\C^N)^{\otimes n})$ denote the action given by 
\begin{equation}\label{def rho C}
\rho(\sigma)(x_1\otimes \ldots \otimes x_n)=x_{\sigma^{-1}(1)} \otimes \ldots \otimes x_{\sigma^{-1}(n)}.
\end{equation}
For all $\sigma\in \S_n$, let us denote by $\ell(\sigma)$ the number of cycles of $\sigma$. To each $\sigma\in \S_n$ we associate two complex-valued functions $P_\sigma$ and $p_\sigma$ on $M_N(\C)$ by setting 
\[ P_\sigma(M)=\Tr^{\otimes n}\left( \rho(\sigma) \circ M^{\otimes n}\right) \mbox{ and } p_\sigma(M)=N^{-\ell(\sigma)} P_\sigma(M),\]
where by $\Tr^{\otimes n}(M_{1}\otimes \ldots \otimes M_{n})$ we mean $\Tr(M_{1})\ldots\Tr(M_{n})$.
If the lengths of the cycles of the permutation $\sigma$ are $m_1,\ldots,m_{\ell(\sigma)}$, then these functions can be written in more elementary terms as
\begin{equation}\label{P cycles}
P_\sigma(M)=\prod_{i=1}^{\ell(\sigma)} \Tr(M^{m_i}) \mbox{ and } p_\sigma(M)=\prod_{i=1}^{\ell(\sigma)} \tr(M^{m_i}).
\end{equation}
The use of the letter $P$ is motivated here by the fact that the functions $P_{\sigma}$ and $p_{\sigma}$ are {\em power sums} of the eigenvalues. We hope that no confusion will arise from our using the same letter $P$ for the projection defined in \eqref{tau kappa}.

Let $(U_{N,t})_{t\geq 0}$ be a Brownian motion on the unitary group $\U(N)$. We are going to study the complex-valued functions $F_N$ and $f_N$ defined on $\R_+\times \S_n$ by
\[F_N(t,\sigma)=\E\left[P_\sigma(U_N(t))\right] \mbox{ and } f_N(t,\sigma)=\E\left[p_\sigma(U_N(t))\right].\]
Let $\T_n\subset \S_n$ denote the set of transpositions. An application of It\^{o}'s formula and the fact that the Casimir operator of $\u(N,\C)$ is equal to $-\frac{1}{N}T$, where $T$ is the flip operator on $\C^{N}\otimes \C^{N}$  (see \eqref{casimir o u}), allow us to prove the following fundamental relation: for all $t\geq 0$ and all $\sigma\in \S_n$, 
\begin{align}
\nonumber \frac{d}{d t} F_N(t,\sigma) &= \E\left[\Tr^{\otimes n}\left(\rho(\sigma) \circ \left(-\frac{n}{2}-\frac{1}{N}\sum_{1\leq i<j\leq n} \rho((i\, j))\right)  \circ U_{t}^{\otimes n} \right)\right]\\
&=-\frac{n}{2}F_N(t,\sigma)-\frac{1}{N}\sum_{\tau \in \T_n}F_N(t, \sigma \tau). \label{edp F_N u}
\end{align}
With the large $N$ limit in view, it is preferable to work with the function $f_N$ rather than the function $F_N$: for example, one has $F_N(0,\sigma)=N^{\ell(\sigma)}$ but $f_N(0,\sigma)=1$. When we divide \eqref{edp F_N u} by $N^{\ell(\sigma)}$, we must take care about the number of cycles of the permutations $\sigma\tau$, which is not the same as that of $\sigma$. More precisely, for each $\tau$, we have $\ell(\sigma\tau)\in \{\ell(\sigma)+1,\ell(\sigma)-1\}$. Let us define
\[\T_n^\pm(\sigma)=\{\tau \in \T_n : \ell(\sigma\tau)=\ell(\sigma)\pm 1\}.\]
With this notation, we have
\begin{equation}\label{edp f_N u}
\frac{d}{d t} f_N(t,\sigma)=-\frac{n}{2}f_N(t,\sigma)-\sum_{\tau \in \T_n^+(\sigma)} f_N(t,\sigma\tau)-\frac{1}{N^2}\sum_{\tau \in \T_n^-(\sigma)} f_N(t,\sigma\tau).
\end{equation}
Let us denote by $\Lop_{\U(N)}$ the linear operator on the space $\Fun(\S_{n})$ of complex-valued functions on $\S_n$ defined by
\[(\Lop_{\U(N)} f)(\sigma)=-\frac{n}{2}f(\sigma)-\sum_{\tau \in \T_n^+(\sigma)} f(\sigma\tau)-\frac{1}{N^2}\sum_{\tau \in \T_n^-(\sigma)} f(\sigma\tau),\]
and by $\1 \in \Fun(\S_n)$ the function identically equal to $1$. We have the equality
\[\forall t\geq 0, \; f_N(t,\cdot)=e^{t\Lop_{\U(N)}} \1.\]
This expression allows us to let $N$ tend to infinity very easily. Indeed, if $\Lop$ denotes the limit of $\Lop_{\U(N)}$ as $N$ tends to infinity (with $n$ staying fixed), that is, the operator defined by
\begin{equation}\label{def op L}
(\Lop f)(\sigma)=-\frac{n}{2}f(\sigma)-\sum_{\tau \in \T_n^+(\sigma)} f(\sigma\tau),
\end{equation}
then it is readily checked that the sequence of functions $f_N$, seen as a sequence of functions from $\R_+$ to $\Fun(\S_n)$, converges uniformly on every compact subset of $\R_+$ towards the function $f(t,\cdot)$ defined by
\begin{equation}\label{edp f u}
\forall t\geq 0, \; f(t,\cdot)=e^{t\Lop} \1.
\end{equation}
In order to compute this exponential, let us make the Ansatz that $f(t,\sigma)$ factorises with respect to the lengths of the cycles of $\sigma$, that is, that there exists a sequence $(\tilde L_{n})_{n\geq 1}$ of functions such that for all $t\geq 0$ and all permutation $\sigma$ with cycles of lengths $m_{1},\ldots,m_{r}$, we have $f(t,\sigma)=e^{-\frac{nt}{2}}\tilde L_{m_{1}}(t)\ldots \tilde L_{m_{r}}(t)$. A little computation shows that \eqref{edp f u} is equivalent to the recurrence relation \eqref{rec R} for the sequence $(\tilde L_{n})_{n\geq 1}$, of which we know that the sequence $(L_{n})_{n\geq 1}$ defined by \eqref{def L} is the unique solution. This finishes the proof of the first assertion of the theorem on the unitary group. 

The second assertion follows from the first and the fact that, by Lemma \ref{inverses}, $U_{N,t}$ has the same distribution as $U_{N,t}^{-1}$.

Let us finally consider the case of the special unitary group. It follows from the last assertion of Proposition \ref{espérances groupes} that the functions $F_{N}$ and $f_{N}$ get simply multiplied by the factor $\exp\frac{n^{2}t}{2N^{2}}$ if we replace $(U_{N,t})_{t\geq 0}$ by a Brownian motion on $\SU(N)$ in their definition. Thus, the operator which replaces $\Lop_{\U(N)}$ in this case is $\Lop_{\SU(N)}=\Lop_{\U(N)}+\frac{n^{2}}{2N^{2}}$ and the conclusion of the proof is the same.
\end{proof}

\subsection{The Brauer algebra I}\label{section:brauer I}

In the orthogonal and symplectic cases, the role played by the symmetric group will be held by an algebra known as the Brauer algebra, which we now describe.

The integer $n\geq 1$ being fixed, let $\B_{n}$ be the set of partitions of the set $\{1,\ldots,2n\}$ by pairs. Let $\lambda$ be a real number. The real Brauer algebra $\Br_{n,\lambda}$ admits, as a real vector space, a basis which is in one-to-one correspondence with $\B_{n}$ and which we identify with it. For example, $\Br_{2,\lambda}$ has dimension $3$ and the basis $\B_2$ consists in the three pairings $\{\{1,2\},\{3,4\}\}$, $\{\{1,3\},\{2,4\}\}$ and $\{\{1,4\},\{2,3\}\}$.

An element of $\B_{n}$ can be represented by a horizontal box with $n$ dots on its bottom edge labelled from $1$ to $n$ and $n$ dots on its top edge labelled from $n+1$ to $2n$, both from left to right, the appropriate pairs of dots being joined by lines inside the box. The product $\pi_{1}\pi_{2}$ of two elements $\pi_{1}$ and $\pi_{2}$ of $\B_{n}$ is computed by putting the box representing $\pi_{1}$ on the top of the box representing $\pi_{2}$. This produces a new pairing $\pi$
between the points on the bottom of the box representing $\pi_{2}$ and those on the top of the box representing $\pi_{1}$. The superposition of two boxes may moreover lead to the formation of loops inside the box. If $r$ loops appear in the process, then we set $\pi_{1}\pi_{2}=\lambda^{r}\pi$ (see Figure \ref{brauer} for an example). 

\begin{figure}[h!]
\begin{center}
\scalebox{0.8}{\includegraphics{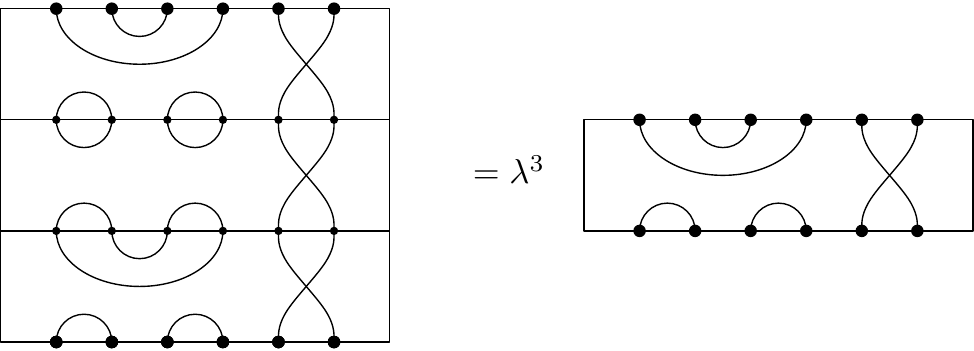}}
\caption{\label{brauer}\small With $\pi_{1}=\{\{1,2\},\{3,4\},\{5,12\},\{6,11\},\{7,10\},\{8,9\}\}$ and $\pi_{2}=\{\{1,2\},\{3,4\},\{5,12\},\{6,11\},\{7,8\},\{9,10\}\}$, we have $\pi_{1}\pi_{2}\pi_{1}=\lambda^{3}\pi_{1}$.}
\end{center}
\end{figure}

Let $\S_{n}$ denote the symmetric group of order $n$. There is a natural inclusion $\S_{n}\subset \B_{n}$ which to a permutation $\sigma\in \S_n$ associates the pairing $\{\{i,\sigma(i)+n\} : i\in \{1,\ldots,n\}\}$. Since the multiplication of pairings associated with permutations does never make loops appear, this correspondence determines an injective homomorphism of algebras $\R[\S_n] \hookrightarrow \Br_{n,\lambda}$, regardless of the value of $\lambda$. 

Let us call horizontal edge of a pairing $\pi$ a pair of $\pi$ which is contained either in $\{1,\ldots,n\}$ or in $\{n+1,\ldots,2n\}$. The pairings of $\S_{n}$ are characterised in $\B_{n}$ by the fact that they have no horizontal edge. On the other hand, a pairing which has one horizontal edge must have at least one in $\{1,\ldots,n\}$ and one in $\{n+1,\ldots,2n\}$, because it has the same number of horizontal edges in both sets. It follows that multiplying this pairing on either side by another pairing cannot make all horizontal edges disappear. Hence, the linear subspace of $\Br_{n,\lambda}$ spanned by $\B_{n}\setminus \S_{n}$ is an ideal of $\Br_{n,\lambda}$.

For all integers $r,s$ such that $1\leq r<s\leq n$, we denote by $(r\, s)$ the element of $\B_{n}$ corresponding to the transposition which exchanges $r$ and $s$. We also denote by $\langle r \, s\rangle$ the partition of $\{1,\ldots,2n\}$ which consists of the pairs $\{k,k+n\}$ for $k\in \{1,\ldots,n\}\setminus\{r,s\}$, and the two pairs $\{r,s\}$ and $\{r+n,s+n\}$. We call this pairing an elementary projection. We denote by $\T_{n}$ the set of all transpositions and by $\Co_n$ the subset of $\mathfrak B_n$ which consists of all contractions. Note that the algebra $\Br_{n,\lambda}$ is generated by $\T_{n}\cup \Co_{n}$.

For the needs of the orthogonal case, let us define an action of the Brauer algebra $\Br_{n,N}$ on $(\R^N)^{\otimes n}$, that is, a morphism of algebras $\rho:\Br_{n,N}\to \Mat_{N}(\R)^{\otimes n}$. Let $(e_1,\ldots,e_N)$ denote the canonical basis of $\R^N$. Let $\pi\in \B_n$ be a basis vector of $\Br_{n,N}$, which we identify with the partition in pairs of $\{1,\ldots,2n\}$ which labels it. We set 
\begin{equation}\label{def rho}
\rho(\pi) = \sum_{i_{1},\ldots,i_{2n}=1}^{N} \left(\prod_{\{k,l\}\in \pi} \delta_{i_k,i_l}\right) E_{i_{n+1}i_{1}} \otimes \ldots \otimes E_{i_{2n}i_{n}}.
\end{equation}
Consider two elements $\pi_{1},\pi_{2}\in \B_{n}$. In the product $\rho(\pi_{1})\rho(\pi_{2})$, the only non-zero contributions come from the terms in which the $n$ bottom indices of $\pi_{1}$ are equal to the $n$ top indices of $\pi_{2}$. Moreover, any loop carries a free index which runs from $1$ to $N$ and thus produces a factor $N$. Hence, if $r$ loops are formed in the product of $\pi_{1}$ and $\pi_{2}$, then $\rho(\pi_{1}\pi_{2})=N^{r}\rho(\pi_{1})\rho(\pi_{2})$. This shows that the unique linear extension of $\rho$ to $\Br_{n,N}$ is a homomorphism of algebras $\rho : \Br_{n,N} \to \End\left((\R^N)^{\otimes n}\right)$.

The restriction of $\rho$ to the subalgebra $\C[\S_{n}]$ coincides with the action of the symmetric group which we considered in the unitary case.

\subsection{The orthogonal case}\label{sec : orth 1M} On the orthogonal group $SO(N)$, the Casimir operator is equal to 
\begin{equation}\label{casimir so}
C_{\so(N)}=-\frac{1}{N} (T-P)
\end{equation}
so that for all $i,j$ such that $1\leq i<j\leq n$, we have
\begin{equation}\label{iotacasimir so}
\iota_{i,j}(C_{\so(N)})=-\frac{1}{N} \left(\rho((i\, j)) - \rho(\langle i\, j\rangle)\right).
\end{equation}

Because of the presence of $P$, the orthogonal analogues of the functions $\left(t\mapsto F_N(t,\sigma)\right)_{\sigma\in \S_{n}}$ do not satisfy a closed differential system anymore. We must therefore introduce new functions, which are naturally indexed by the elements of the Brauer algebra.

\begin{proof}[Proof of Theorem \ref{limite brown}  in the orthogonal case]
Let $n\geq 1$ be an integer. To each element $\pi\in \B_n$ we associate the function $P_\pi$ on $M_N(\R)$ by setting 
\[ P_\pi(M)=\Tr^{\otimes n}\left(\rho(\pi)\circ M^{\otimes n} \right).\]
For example, if $\pi$ is the element of $\B_6$ depicted on the right-hand side of Figure \ref{brauer}, then $P_\pi(M)=\Tr(M\t{M}M\t{M})\Tr(M^{2})$. Note that when it is restricted to the orthogonal group, the function $M\mapsto P_{\pi}(M)$ can be a polynomial in the entries of $M$ of degree strictly smaller than $n$. It is possible, but unnecessary at this stage, to give for the function $P_{\pi}$ an expression similar to \eqref{P cycles}. Our treatment of the symplectic case will however require such a formula, and it may be instructive to look briefly at \eqref{P cycles o}.

The correct definition of the normalised function $p_\pi$ requires an appropriate definition of the number of cycles of $\pi$. The simplest way to define this number is through the equality $P_{\pi}(I_{N})=N^{\ell(\pi)}$. Alternatively, it is the number of loops formed after completing the diagram of $\pi$ by the $n$ vertical lines which join $k$ to $n+k$ for all $k$ between $1$ and $n$. We set, as in the unitary case,
\begin{equation}\label{def p o}
p_\pi(M)=N^{-\ell(\pi)} P_\pi(M).
\end{equation}
We extend the definitions of $P_{\pi}$ and $p_{\pi}$ by linearity to any $b\in \Br_{n,N}$. Note however that the function $\ell$ is only defined on the elements of $\B_{n}$. We extend it to multiple of basis elements by setting $\ell(c \pi)=\ell(\pi)$ for all complex number $c\neq 0$.

Let $(R_{N,t})_{t\geq 0}$ be a Brownian motion on the orthogonal group $\SO(N)$. As in the unitary case, we are going to study the functions $F_N$ and $f_N$ defined on $\R_+\times \Br_{n,N}$ by
\[F_N(t,b)=\E\left[P_b(R_{N,t})\right] \mbox{ and } f_N(t,b)=\E\left[p_b(R_{N,t})\right].\]
The normalisation has been chosen such that $f_{N}(0,\pi)=1$ for all $\pi\in \B_{n}$.
With these definitions and considering the stochastic differential equation which defines the Brownian motion on $\SO(N)$, an application of It\^{o}'s formula yields the following fundamental relation: for all $t\geq 0$ and all $b\in \Br_{n,N}$, one has
\begin{equation}\label{edp F_N o}
\frac{d}{d t} F_N(t,b)=-\frac{n(N-1)}{2N}F_N(t,b)-\frac{1}{N}\sum_{\tau \in \T_n} F_N(t, b \tau)+\frac{1}{N}\sum_{\kappa \in \Co_n} F_N(t, b \kappa),\end{equation}
from which it follows immediately that for all $\pi\in \B_n$,
\begin{equation}\label{edp f_N o}
\frac{d}{d t} f_N(t,\pi)=-\frac{n(N-1)}{2N}f_N(t,\pi)-\sum_{\tau \in \T_n} N^{\ell(\pi \tau)-\ell(\pi)-1}f_N(t,\pi \tau)+\sum_{\kappa \in \Co_n} N^{\ell(\pi \kappa)-\ell(\pi)-1}f_N(t,\pi \kappa).\end{equation}
Note that in this equation, $\pi \tau $ and  $\pi \kappa $ might be non-trivial scalar multiples of basis elements, thus possibly introducing extra powers of $N$ in the expression. Note also that, for the same reason, we are using the extended definition of the function $\ell$. 

In fact, the only case where a loop is formed is for the product $\pi \kappa $ when $\kappa=\langle i\, j\rangle$ and the pair $\{i,j\}$ belongs to $\pi$. Moreover, in this case, $\pi \kappa =N\pi$.

Let us denote by $\Lop_{\SO(N)}$ the linear operator on the dual space $\Br_{n,N}^{*}$ of linear forms on $\Br_{n,N}$ characterised by the fact that for all $\pi\in \B_{n}$,
\begin{align}\label{def lso}
(\Lop_{\SO(N)} f)(\pi)=-\frac{n(N-1)}{2N}f(\pi)&-\sum_{\tau \in \T_n} N^{\ell(\pi \tau)-\ell(\pi)-1}f(\pi \tau)+ \sum_{\kappa\in \Co_n} N^{\ell(\pi \kappa)-\ell(\pi)-1} f(\pi \kappa).
\end{align}
We also denote by $\1 \in \Br_{n,N}^{*}$ the linear form equal to $1$ on each basis vector. Then we have the equality
\[\forall t\geq 0, \; f_N(t,\cdot)=e^{t\Lop_{\SO(N)}} \1.\]

Let us now determine which powers of $N$ appear in $\Lop_{\SO(N)}$. First of all, the observation which we made just after  \eqref{edp f_N o} and an elementary verification show that $\Lop_{\SO(N)}$ is a polynomial of degree at most $2$ in $N^{-1}$. We also know that if $\pi$ is a permutation, then $\ell(\pi \tau)$ belongs to $\{\ell(\pi)-1,\ell(\pi)+1\}$, so that the second term of $\Lop_{\SO(N)}$ involves terms of order $N^{0}$ and $N^{-2}$, but not of order $N^{-1}$.

Now comes the crucial argument, namely the observation that multiplying a permutation by an elementary projection does never create a loop nor increase the number of cycles. The first assertion is a consequence of the fact that for all $\pi\in \B_{n}$ and all $i,j\in \{1,\ldots,n\}$ with $i<j$, the product $\pi\langle i\, j\rangle$ involves a loop if and only if the pair $\{i,j\}$ belongs to $\pi$. If $\pi$ is a permutation, this never happens. Moreover, one checks, depending on whether $i$ and $j$ belong to the same cycle of $\sigma$ or not, that $\ell(\pi \langle i \, j\rangle)$ belongs to $\{\ell(\pi)-1,\ell(\pi)\}$. These observations imply that when $\pi$ is a permutation, the third term of $\Lop_{\SO(N)}$ contains no term of order $N^{0}$ and is hence dominated by $N^{-1}$. 

A less important but simpler observation is that when $\pi$ is not a permutation, none of the elements which appear in \ref{def lso} is a permutation, according to our discussion of horizontal edges in \ref{section:brauer I}.

Recall the definition of the operator $\Lop$ from the unitary case (see \eqref{def op L}). The previous discussion shows that, in the basis of $\Br_{n,N}^{*}$ dual to $\B_{n}$, split into dual permutations on one hand and the other dual basis elements on the other hand, the matrix of $\Lop_{\SO(N)}$ is
\[\Lop_{\SO(N)}=\left(\begin{array}{c|c}  & \\ \Lop_{\U(N)} + \frac{n}{2N} I_{n!} & O(N^{-1}) \\ & \\ \hline  0 & * \end{array}\right),\]
where the bottom right block of this matrix is a polynomial of degree $2$ in $N^{-1}$. In particular, $\Lop_{\SO(N)}$ admits a limit as $N$ tends to infinity and this limit is of the form
\[\lim_{N\to\infty}\Lop_{\SO(N)}=\left(\begin{array}{c|c}  \Lop & 0 \\  \hline  0 & * \end{array}\right).\]

Ignoring the second column of this matrix, we conclude that the sequence of functions $(f_N)_{n\geq 1}$, restricted to $\R_{+}\times \S_{n}$, converges uniformly on every compact subset of $\R_+$ towards the function $f(t,\cdot)$ defined by
\begin{equation}\label{edp f o}
\forall t\geq 0, \; f(t,\cdot)=e^{t\Lop} \1.
\end{equation}
We recognise here the equation \eqref{edp f u}.
\end{proof}

\subsection{The Brauer algebra II}\label{Brauer II} In the treatment of the symplectic case, we will consider a homomorphism of algebras $\rho_{\H}:\Br_{n,-2N}\to\Mat_{N}(\H)^{\otimes n}$. This homomorphism will be constructed as the tensor product of the homomorphism $\rho$ considered in the orthogonal case and another homomorphism $\gamma:\Br_{n,-2}\to \H^{\otimes n}$, which we define and study in this section.

In order to define $\gamma$, we need to discuss a cyclic structure on $\{1,\ldots,2n\}$ associated to each element of $\B_{n}$. We have already implicitly considered this cyclic structure in the definition of $\ell(\pi)$ just before \eqref{def p o}. 

Let us consider a pairing $\pi$ of $\{1,\ldots,2n\}$. Let us consider the usual graph associated with $\pi$, with vertices $\{1,\ldots,2n\}$ and $n$ edges, one joining $i$ and $j$ for each pair $\{i,j\}\in \pi$. We call these $n$ edges the primary edges. Let us add to this graph $n$ other edges, one joining $i$ to $i+n$ for each $i\in \{1,\ldots,n\}$. We call these edges the secondary edges. We get a graph in which each vertex has degree 2, being adjacent to one primary and one secondary edge. This graph is thus a union of disjoint cycles of even length, for which $\pi$ provides no canonical orientation. We decide to orient each of these cycles by declaring that the primary edge adjacent to the smallest element of each cycle is outgoing at this vertex. In this way, we get a partition of $\{1,\ldots, 2n\}$ by oriented cycles, that is, a permutation of $\{1,\ldots,2n\}$, which we denote by $\Sigma_{\pi}$. For an example of this construction, see Figure \ref{ptp}.

We are now going to use the permutation $\Sigma_{\pi}\in \S_{2n}$ to define a permutation $\sigma_{\pi}\in \S_{n}$ and to attach a sign to each integer $\{1,\ldots,n\}$. Let us start with the signs. For each $i\in \{1,\ldots,n\}$, we set $\epsilon_{\pi}(i)=1$ if $\{i,\Sigma_{\pi}(i)\}$ is a primary edge and $\epsilon_{\pi}(i)=-1$ otherwise. If $(i\, n+i)$ is a cycle of $\Sigma_{\pi}$, then $\epsilon_{\pi}(i)=1$. Then, we define $\sigma_{\pi}$ as the permutation of $\{1,\ldots,n\}$ obtained by removing the integers $\{n+1,\ldots,2n\}$ from their cycles in $\Sigma_{\pi}$. Note that $\Sigma_{\pi}$, and hence $\sigma_{\pi}$, have exactly $\ell(\pi)$ cycles. For example, if $\pi$ is a permutation, then $\epsilon_{\pi}(i)=1$ for all $i\in \{1,\ldots,n\}$ and $\sigma_{\pi}=\pi$.

\begin{figure}[h!]
\begin{center}
\scalebox{0.9}{\includegraphics{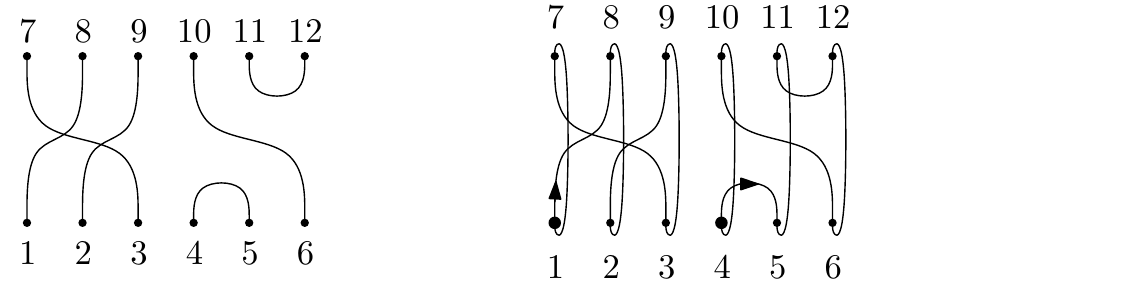}}\\
\caption{\small \label{ptp} Consider $\pi=\{\{1,8\},\{2,9\},\{3,7\},\{4,5\},\{6,10\},\{11,12\}\} \in \B_{6}$. The primary edges are represented on the left and the full graph on the right. There are two cycles with respective smallest element $1$ and $4$. We thus have $\Sigma_{\pi}=(1 \; 8 \; 2\; 9\; 3\; 7)(4\; 5 \; 11 \; 12 \; 6 \; 10)$ and $\sigma_{\pi}=(1\; 2\; 3)(4\; 5\; 6)$. For each $i\in \{1,\ldots,6\}$, $\epsilon_{\pi}(i)$ equals $1$ if $i$ is traversed upwards and $-1$ if $i$ is traversed downwards. Here, $\epsilon_{\pi}(5)=-1$ and the other signs are $1$.}
\end{center}
\end{figure}

The signification of the permutation $\sigma_{\pi}$ and the signs $\epsilon_{\pi}(1),\ldots,\epsilon_{\pi}(n)$ is given by the following formula. Recall the definition of $\rho$ from \eqref{def rho}. 

\begin{proposition} Let $\pi$ be an element of $\B_{n}$. Let $R_{1},\ldots,R_{n}$ be elements of $\SO(N)$. Let us write $(i_{1}\ldots i_{s})\preccurlyeq \sigma_{\pi}$ if $(i_{1}\ldots i_{s})$ is a cycle of $\sigma_{\pi}$. Then
\begin{equation}\label{P cycles o}
\Tr^{\otimes n}(\rho(\pi) \circ R_{1}\otimes \ldots \otimes R_{n})=\prod_{(i_{1}\ldots i_{s}) \preccurlyeq \sigma_{\pi}} \Tr(R_{i_{s}}^{\epsilon_{\pi}(i_{s})}\ldots R_{i_{1}}^{\epsilon_{\pi}(i_{1})}).
\end{equation}
The same identity holds with arbitrary matrices provided inverse matrices are replaced by transposed ones.
\end{proposition}

\begin{proof} If $\pi$ is a permutation, then a direct computation shows that the formula holds. Now, let us pick an arbitrary pairing $\pi\in \B_{n}$, an integer $i\in \{1,\ldots,n\}$ and let us consider the pairing $\pi'$ obtained by exchanging $i$ and $n+i$ in the pairs to which they belong in $\pi$. We have $\sigma_{\pi'}=\sigma_{\pi}$, $\epsilon_{\pi'}(i)=-\epsilon_{\pi}(i)$ and $\epsilon_{\pi'}(j)=\epsilon_{\pi}(j)$ for all $j\neq i$. Moreover, 
\[\Tr^{\otimes n}(\rho(\pi') \circ R_{1}\otimes \ldots \otimes R_{n})=\Tr^{\otimes n}(\rho(\pi) \circ R_{1}\otimes \ldots \otimes {}^{t}R_{i} \otimes \ldots \otimes R_{n}).\]
Hence, if \eqref{P cycles o} holds for $\pi$, it also holds for $\pi'$. It only remains to convince oneself that any pairing can be turned into a permutation by a finite succession of  exchanges of the sort which we have just considered.
\end{proof}

Through the mapping $\pi \mapsto (\sigma_{\pi},\epsilon_{\pi})$, we associate to each element of $\B_{n}$ an element of $\S_{n}$ and an element of $(\Z/2\Z)^{n}$, that is, an element of the hyperoctahedral group $H_{n}=\S_{n}\ltimes (\Z/2\Z)^{n}\subset \S_{2n}$. Since $\B_{n}$, seen as the set of fixed point free involutions of $\{1,\ldots,2n\}$, is isomorphic to the quotient $\S_{2n}/H_{n}$, it would be natural to expect a neater definition of the pair $(\sigma_{\pi},\epsilon_{\pi})$, but I was not able to find it.

Let us now turn to the definition of the mapping $\gamma$. Recall that $\I(\H)$ denotes the subset $\{1,\ii,\j,\k\}$ of $\H$.
For each pairing $\pi\in \B_{n}$, set 
\begin{equation}\label{def Srep}
\gamma(\pi)=\frac{1}{(-2)^{n}}\sum_{\gamma_{1},\ldots,\gamma_{n}\in\I(\H)} \left(\prod_{(i_{1}\ldots i_{s}) \preccurlyeq \sigma_{\pi}}(-2\Re)(\gamma_{i_{s}}\ldots \gamma_{i_{1}})\right) \gamma_{1}^{-\epsilon_{\pi}(1)}I_{N}\otimes \ldots \otimes \gamma_{n}^{-\epsilon_{\pi}(n)}I_{N}.
\end{equation}

If $\pi$ is the pairing corresponding to the identity permutation, then $\gamma(\pi)=I_{N}^{\otimes n}$. We set 
\begin{equation}\label{def rho H}
\rho_{\H}(\pi)=\rho(\pi)\otimes \gamma(\pi)
\end{equation}
and will sometimes use the lighter notation $\rho_{\H}\pi$.

Recall \eqref{re co} and observe that $\gamma((1\, 2))=-\frac{1}{2}{\rm Re}^{\H}$, $\gamma(\langle 1\, 2\rangle)=-\frac{1}{2}{\rm Im}^{\H}$, 
so that
\[\rho_{\H}(1\, 2)=-\frac{1}{2}T\otimes {\rm Re}^{\H} \mbox{ and } \rho_{\H}\langle 1\, 2 \rangle=-\frac{1}{2} P\otimes {\rm Im}^{\H},\]
and by comparing with \eqref{casimir general}, we have for all $i,j$ with $1\leq i < j \leq n$ the equality
\begin{equation}\label{iotacasimir}
\iota_{i,j}(C_{\sp(N)})=-\frac{1}{-2N} \left(\rho_{\H}(i\, j) - \rho_{\H}\langle i\, j\rangle\right).
\end{equation}
This is a first piece of a justification for our arguably strange definition of $\gamma$. A second piece of justification is given by the following lemma. By analogy with the real and complex cases, we denote by $\circ$ the product in the algebra $\Mat_{N}(\H)^{\otimes n}$, but we would like to emphasise that the natural action of this algebra on $(\H^{N})^{\otimes n}$ which is implicit in this notation is the action of a real algebra on the tensor product over $\R$ of real linear spaces. The trace denote by $\Tr$ on the other hand is still the usual trace on $\Mat_{n}(\H)$.

\begin{lemma}\label{tr tr sp} For all $n\geq 1$, all $\pi\in \B_{n}$ and all $S_{1},\ldots,S_{n}\in \Sp(N)$, we have
\[(-2\Re\Tr)^{\otimes n}(\rho_{\H}(\pi)\circ S_{1}\otimes \ldots \otimes S_{n})=\prod_{(i_{1}\ldots i_{s}) \preccurlyeq \sigma_{\pi}} (-2\Re\Tr)(S_{i_{s}}^{\epsilon_{\pi}(i_{s})}\ldots S_{i_{1}}^{\epsilon_{\pi}(i_{1})}).\]
The same identity holds with arbitrary matrices provided inverse matrices are replaced by adjoint ones.
\end{lemma}

\begin{proof} Relabelling the matrices $S_{1},\ldots,S_{n}$ if necessary and using the fact that both $\rho(\pi)$ and $\gamma(\pi)$ factorise according to the cycles of $\sigma_{\pi}$, we may reduce the problem to the case where $\sigma_{\pi}$ has a single cycle, and we may choose the cycle $(n\ldots 1)$. In this case, after developing the traces, the equality results from the following identity, valid for all quaternions $q_{1},\ldots,q_{n}$:
\[\sum_{\gamma_{1},\ldots,\gamma_{n}\in \I(\H)} \gamma_{1}\ldots \gamma_{n} \Re(\gamma_{1}^{-\epsilon_{1}} q_{1}) \ldots \Re(\gamma_{n}^{-\epsilon_{n}} q_{n})=q_{1}^{*_{1}}\ldots q_{n}^{*_{n}},\]
where we set $q_{i}^{*_{i}}=q_{i}$ if $\epsilon_{i}=1$ and $q_{i}^{*_{i}}=q_{i}^{*}$ if $\epsilon_{i}=-1$. 
\end{proof}

The main property of $\gamma$ is the following, which determined its definition.

\begin{proposition}\label{mult sp} The unique extension of $\gamma$ to a linear mapping $\Br_{n,-2}\to \H^{\otimes n}$ is a homomorphism of algebras.
\end{proposition}

\begin{proof} Since the algebra $\Br_{n,-2}$ is generated by $\T_{n}\cup\Co_{n}$, it suffices to prove that for all pairing $\pi\in \B_{n}$ and all $i,j$ with $1\leq i< j \leq n$, we have $\gamma(\pi (i\, j))=\gamma(\pi) \gamma((i\, j))$ and $\gamma(\pi \langle i\, j\rangle)=\gamma(\pi) \gamma (\langle i\, j\rangle)$. For each equality, there are three cases to consider: the case where $i$ and $j$ do not belong to the same cycle of $\sigma_{\pi}$, then the case where they do, which itself is subdivided into the sub-cases $\epsilon_{\pi}(i)=\epsilon_{\pi}(j)$ and $\epsilon_{\pi}(i)=-\epsilon_{\pi}(j)$. In each of the six cases, the key of the result is one of the following elementary identities, valid for all $q_{1},q_{2}\in \H$:
\begin{align}
\label{quat 1}\tag{I}&\frac{1}{4}\sum_{\gamma_{1},\gamma_{2}\in \I(\H)}(-2\Re)(\gamma_{1}\gamma_{2}) (-2\Re)(\gamma_{1}^{-1}q_{1})(-2\Re)(\gamma_{2}^{-1}q_{2})= (-2\Re)(q_{1}q_{2}),\\
\label{quat 2}\tag{II}&\frac{1}{4}\sum_{\gamma_{1},\gamma_{2}\in \I(\H)}(-2\Re)(\gamma_{1}\gamma_{2})(-2\Re)(\gamma_{1}^{-1}q_{1})(-2\Re)(\gamma_{2}q_{2})=(-2\Re)(q_{1}q_{2}^{*}),\\
\label{quat 3}\tag{III}&\frac{1}{4}\sum_{\gamma_{1},\gamma_{2}\in \I(\H)}(-2\Re)(\gamma_{1}\gamma_{2})(-2\Re)(\gamma_{1}^{-1}q_{1}\gamma_{2}^{-1}q_{2})=(-2\Re)(q_{1})(-2\Re)(q_{2}),\\
\label{quat 4}\tag{IV}&\frac{1}{4}\sum_{\gamma_{1},\gamma_{2}\in \I(\H)}(-2\Re)(\gamma_{1}\gamma_{2})(-2\Re)(\gamma_{1}^{-1}q_{1}\gamma_{2}q_{2})=(-2\Re)(q_{1}q_{2}^{*}).
\end{align}
The first equality is the multiplication rule in $\H$ and the second follows by replacing $q_{2}$ by $q_{2}^{*}$. The third and fourth equality follow from the identities \eqref{re co 2}.

Let us give the details of the proof of the equality $\gamma(\pi)\gamma(\langle i\, j\rangle)=\gamma(\pi \langle i\, j\rangle)$ in the the case where $i$ and $j$ belong to the same cycle of $\sigma_{\pi}$ and $\epsilon_{\pi}(i)=-\epsilon_{\pi}(j)$. Recall the notation $\iota_{i,j}$ from \eqref{def iota}. To start with, we have
\[\gamma(\langle i\, j\rangle)=\frac{1}{4}\sum_{\alpha_{1},\alpha_{2}\in \I(\H)} (-2\Re)(\alpha_{1}\alpha_{2}) \iota_{i,j}(\alpha_{1}^{-1} \otimes \alpha_{2} ).\]
Let us write $(i\, i_{1}\ldots i_{s}\, j \, j_{1}\ldots j_{t})$ the cycle of $\sigma_{\pi}$ which contains $i$ and $j$. Reversing the orientation of this cycle if necessary, we may assume that $\epsilon_{\pi}(i)=1$ and $\epsilon_{\pi}(j)=-1$. In the expression of $\gamma(\pi)\gamma(\langle i\, j\rangle)$, we have the sum over all possible values of $\gamma_{1},\ldots,\gamma_{n},\alpha_{1},\alpha_{2}$ in $\I(\H)$ of the product of a term
\[ \frac{1}{4}(-2\Re)(\alpha_{1}\alpha_{2})(-2\Re)(\gamma_{j_{t}}\ldots\gamma_{j_{1}}\gamma_{j}\gamma_{i_{s}}\ldots\gamma_{i_{1}}\gamma_{i}) \ldots\]
and a term
\[ \ldots \otimes \gamma_{i}^{-1} \alpha_{1}^{-1}\otimes \ldots \otimes \gamma_{j}\alpha_{2} \otimes \ldots.\]
In this sum, we would like to perform a change of variables and to replace $\gamma_{i}$ by $\alpha_{1}^{-1}\gamma_{i}$ and $\gamma_{j}$ by $\gamma_{j}\alpha_{2}^{-1}$. This would however introduce troublesome minus signs. The neatest way to do this is to allow temporarily our variables to vary in the set  $\I(\H)\cup -\I(\H)$ instead of $\I(\H)$, to the price of a factor $\frac{1}{2}$ for each variable. This does not affect the sum otherwise, because each variable appears exactly twice. The advantage is that $\I(\H)\cup -\I(\H)$ is a subgroup of $\H$, so that the change of variables is justified. After this change of variables, the two terms which we are considering are replaced respectively by
\[ \frac{1}{4}(-2\Re)(\alpha_{1}\alpha_{2})(-2\Re)(\gamma_{j_{t}}\ldots\gamma_{j_{1}}\gamma_{j}\alpha_{2}^{-1}\gamma_{i_{s}}\ldots\gamma_{i_{1}}\alpha_{1}^{-1}\gamma_{i}) \ldots\]
and
\[ \ldots \otimes \gamma_{i}^{-1} \otimes \ldots \otimes \gamma_{j}  \otimes \ldots.\]
Thanks to the third of the four elementary identities mentioned above, summing over $\alpha_{1}$ and $\alpha_{2}$ transforms the first term into 
\[(-2\Re)(\gamma_{j_{t}}\ldots\gamma_{j_{1}}\gamma_{j}\gamma_{i})(-2\Re)(\gamma_{i_{s}}\ldots\gamma_{i_{1}})\ldots.\]
On the other hand, the cycles of $\sigma_{\pi\langle i\, j\rangle}$ are the same as those of $\sigma_{\pi}$, except for $(i\, i_{1}\ldots i_{s}\, j \, j_{1}\ldots j_{t})$ which is replaced by $(i\, j \, j_{1}\ldots j_{s})(i_{1}\ldots i_{s})$. Moreover, for all $k\in \{1,\ldots,n\}$, we have $\epsilon_{\pi\langle i\, j\rangle}(k)=\epsilon_{\pi}(k)$.

Finally, it may happen that $s=0$, in which case the cycle $(i_{1}\ldots i_{s})$ is absent in $\pi\langle i\, j\rangle$. In this case, the fact that $\sigma_{\pi}(i)=j$, $\epsilon_{\pi}(i)=1$ and $\epsilon_{\pi}(j)=-1$ imposes that $\{i,j\}$ is a pair of $\pi$. Since we are working in $B_{n,-2}$, the appearance of a loop in the multiplication of $\pi$ and $\langle i\, j \rangle$ brings the missing factor $-2$. In fact, this is the only case in the whole proof where a loop is formed and where the parameter of the Brauer algebra plays a role.

Let us indicate what differs in the proof of $\gamma(\pi)\gamma((i\, j))=\gamma(\pi(i\, j))$ in the same case, when $i$ and $j$ belong to the same cycle of $\sigma_{\pi}$ and $\epsilon_{\pi}(i)=-\epsilon_{\pi}(j)$. With the same notation, using
\[\gamma(( i\, j))=\frac{1}{4}\sum_{\alpha_{1},\alpha_{2}\in \I(\H)} (-2\Re)(\alpha_{1}\alpha_{2}) \iota_{i,j}(\alpha_{1}^{-1}\otimes \alpha_{2}^{-1} )\]
and performing exactly the same steps, only applying the fourth elementary equality instead of the third, we end up with a term
\[(-2\Re)(\gamma_{j_{t}}\ldots\gamma_{j_{1}}\gamma_{j}\gamma_{i_{1}}^{-1}\ldots\gamma_{i_{s}}^{-1}\gamma_{i})\ldots.\]
A second change of variables is needed at this point, and justified as the first, by which we replace $\gamma_{i_{1}},\ldots,\gamma_{i_{s}}$ by their inverses. This comes in agreement with the fact that not only $\sigma_{\pi(i\, j)}$ has $(i\, i_{s}\ldots i_{1}\, j\, j_{1}\ldots j_{s})$ as a cycle, but $\epsilon_{\pi(i\, j)}(i_{k})=-\epsilon_{\pi}(i_{k})$ for all $k\in \{1,\ldots,s\}$, the other signs being unchanged.

Nothing new is needed to check the four other cases and we spare the reader a detailed account of them. 
\end{proof}

It follows from this result and from our earlier study of $\rho$ that the linear extension $\rho_{\H} : \Br_{n,-2N}\to \Mat_{N}(\H)^{\otimes n}$ is a homomorphism of algebras.

At this point, we can uniformise our definitions of the representations $\rho$ and $\rho_{\H}$. Indeed, we have defined, for each $\K\in \{\R,\C,\H\}$, with the corresponding value of $\beta=\dim_{\R}\K$, a representation
\begin{equation}\label{def rho K}
\rho_{\K} : \Br_{n,(2-\beta)N} \to \Mat_{N}(\K)^{\otimes n}.
\end{equation}
In the case $\K=\C$, we set $\rho_{\C}(\pi)=0$ whenever $\pi\in \B_{n}$ is not a permutation. We shall henceforward use the notation $\rho_{\K}$, that is, in particular, $\rho_{\R}$ instead of $\rho$.

We can now proceed to the proof of our first main theorem in the symplectic case.

\subsection{The symplectic case}\label{section symp} The symplectic case is similar to the orthogonal case, but more complicated, since there is no expression of the Casimir operator which is really simpler than \eqref{casimir general}. One possibility would be to work through the embedding $\Sp(N)\to \U(2N)$, but this is not the approach which we choose.

\begin{proof}[Proof of Theorem \ref{limite brown} in the symplectic case] To each element $\pi\in \B_n$, we associate the function $P_\pi$ on $M_N(\H)$ by setting 
\[ P_\pi(M)=(-2\Re\Tr)^{\otimes n}\left(\rho_{\H} (\pi) \circ M^{\otimes n}\right),\]
and the function $p_{\pi}(M)=(-2N)^{-\ell(\pi)}P_{\pi}(M)$. By Lemma \ref{tr tr sp}, we have $p_{\pi}(I_{N})=1$.

Let $(S_{N,t})_{t\geq 0}$ be a Brownian motion on the symplectic group $\Sp(N)$. We define the functions $F_N$ and $f_N$ on $\R_+\times \B_{n}$ by
\[F_N(t,\pi)=\E\left[P_\pi(S_{N,t})\right] \mbox{ and } f_N(t,\pi)=\E\left[p_\pi(S_{N,t})\right],\]
and extend them by linearity to $\R_{+} \times \Br_{n,-2N}$.
The normalisation has been chosen such that $f_{N}(0,\pi)=1$ for all $\pi\in \B_{n}$. 

Let us apply It\^{o}'s formula in this new context. Thanks to \eqref{iotacasimir} and Proposition \eqref{mult sp}, we have, for all $t\geq 0$ and all $b\in \Br_{n,-2N}$, 
\begin{align}
\nonumber \frac{d}{d t} F_N(t,b) &= \E\left[\rho_{\H}(b) \circ \left(-\frac{2N+1}{2N}\frac{n}{2}+\frac{1}{2N}\sum_{1\leq i<j\leq n} (\rho_{\H}(i\, j)-\rho_{H}\langle i\, j \rangle)\right) \circ S_{t}^{\otimes n} \right]\\
&=-\frac{n(2N+1)}{4N}F_N(t,b)-\frac{1}{-2N}\sum_{\tau \in \T_n} F_N(t, b \tau) + \frac{1}{-2N}\sum_{\kappa \in \Co_n} F_N(t, b \kappa), \label{edp F_N sp}
\end{align}
which is the symplectic version of \eqref{edp F_N o}. From this equality, we deduce that for all $\pi\in \B_n$,
\begin{align}\label{edp f_N sp}
\frac{d}{d t} f_N(t,\pi)=-\frac{n(2N+1)}{4N}f_N(t,\pi)&-\sum_{\tau \in \T_n} (-2N)^{\ell(\pi \tau)-\ell(\pi)-1}  f_N(t, \pi \tau)\\ &+\sum_{\kappa \in \Co_n} (-2N)^{\ell(\pi \kappa)-\ell(\pi)-1} f_N(t, \pi \kappa).\nonumber
\end{align}
Recall that in \eqref{edp f_N sp}, $\pi \tau$ and  $\pi\kappa$ can be scalar multiples of basis elements. Just as in the orthogonal case, a loop is formed in the product $\pi\kappa$ only when $\kappa=\langle i\, j\rangle$ and the pair $\{i,j\}$ belongs to $\pi$, and in this case, we have $\pi\kappa=N\pi$.

Let us denote by $\Lop_{\Sp(N)}$ the linear operator on $\Br_{n,-2N}^{*}$ defined by the following equality, valid for all $\pi\in \B_{n}$:
\begin{align*}
(\Lop_{\Sp(N)} f)(\pi)=-\frac{n(2N+1)}{4N}f(\pi)&-\sum_{\tau \in \T_n} (-2N)^{\ell(\pi \tau)-\ell(\pi)-1}f(\pi \tau) \\ &+ \sum_{\kappa\in \Co_n} (-2N)^{\ell(\pi \kappa)-\ell(\pi)-1}f(\pi \kappa).
\end{align*}
We also denote by $\1 \in \Br_{n,-2N}^{*}$ the linear form equal to $1$ on each element of $\B_{n}$. Then we have the equality
\[\forall t\geq 0, \; f_N(t,\cdot)=e^{t\Lop_{\Sp(N)}} \1.\]
Our discussion of the powers of $N$ involved in the operator $\Lop_{\SO(N)}$ did not depend on the signs of the coefficients, or of factors independent of $N$. It remains thus entirely valid for the operator $\Lop_{\Sp(N)}$. Thus, in the basis of $\Br_{n,-2N}^{*}$ dual to $\B_{n}$, split as in the orthogonal case, the matrix of $\Lop_{\Sp(N)}$ is again
\[\Lop_{\Sp(N)}=\left(\begin{array}{c|c}  & \\ \Lop_{\U(N)} + \frac{n}{2(-2N)} I_{n!} & O(N^{-1})\\ & \\ \hline  0 & * \end{array}\right),\]
where as in the orthogonal case, the second column is a polynomial of degree $2$ in $N^{-1}$. In fact, we have, formally, the equality $\Lop_{\Sp(N)}=\Lop_{\SO(-2N)}$.

In particular, $\Lop_{\Sp(N)}$ admits a limit as $N$ tends to infinity and this limit is of the form
\[\lim_{N\to\infty}\Lop_{\Sp(N)}=\left(\begin{array}{c|c}  \Lop & 0 \\  \hline  0 & * \end{array}\right).\]
We can conclude the proof as in the orthogonal case.
\end{proof}

\subsection{Uniform matrices} In the second part of this work, we will make use of a result about uniform random matrices on $\U(N,\K)$ which can be proved using the techniques which we introduced in our study of the Brownian motion.

\begin{proposition}\label{haar unitary} Choose $\K\in \{\R,\C,\H\}$. Let $(W_{N})_{N\geq 1}$ be a sequence of random matrices such that for all $N\geq 1$, $W_{N}$ is distributed according to the Haar measure on $\U(N,\K)$. Then for all $n\in \Z\setminus\{0\}$,
\[\E\left[\tr(W_{N}^{n})\right]=O(N^{-1}),\]
where $\tr$ must be replaced by $\Re\tr$ if $\K=\H$.
\end{proposition}

This result is very elementary, and of course not optimal, in the unitary case, since these expectations are actually equal to $0$. It is also known at least in the orthogonal case (see for example \cite[Section 4.2]{HiaiPetz}). 

\begin{proof} Choose $n\in \Z\setminus\{0\}$. Replacing $n$ by $-n$ leaves the expectation unchanged, or conjugates it if $\K=\C$. In all cases, we may and will assume that $n\geq 1$.

If $\K=\C$, the invariance of the Haar measure by translation by scalar matrices implies immediately that 
$\E\left[\tr(W_{N}^{n})\right]=0$.

Let us focus on the case where $\K\in \{\R,\H\}$. We will write down the proof in the symplectic case and indicate the very small modifications which should be made to treat the orthogonal case.

For each $N\geq 1$, let $(S_{N,t})_{t\geq 0}$ be a Brownian motion on $\U(N,\H)$ issued from $I_{N}$ and independent of $W_{N}$. Then the process $(S_{N,t}W_{N})_{t\geq 0}$ satisfies the first equation of the system \eqref{def brown}. Moreover, its distribution is stationary, equal to the Haar measure on $\U(N,\H)$. It is a stationary Brownian motion. Let us define $f_{N}\in \Br_{n,-2N}^{*}$ by setting, for each $\pi\in \B_{n}$,
\[f_{N}(\pi)=\E\left[p_{\pi}(V_{N,t}W_{N})\right]=\E\left[p_{\pi}(W_{N})\right],\]
the definition of $p_{\pi}$ being the same as in the proof of Theorem \ref{limite brown}. In the orthogonal case, $f_{N}$ would belong to $\Br_{n,N}^{*}$, and the definition of $p_{\pi}$ would be different, but still the same as in the proof of Theorem \ref{limite brown} for the orthogonal case.

The main difference with the cases which we studied previously is that $f_{N}$, as a function of $t$, is constant, thanks to the stationarity of the process $(S_{N,t}W_{N})_{t\geq 0}$.

In the proof of Theorem \ref{limite brown}, we made use of the fact that the Brownian motion $(S_{N,t})_{t\geq 0}$ was issued from $I_{N}$ only at the very last step, in order to specify the value at $t=0$ of the function $f_{N}(t,\cdot)$. Before that point, we only made use of the first equation of \eqref{def brown}. Hence, the arguments which we applied to $(S_{N,t})_{t\geq 0}$ hold equally for the stationary Brownian motion $(S_{N,t}W_{N})_{t\geq 0}$ and the function $f_{N}$ satisfies the differential equation
\[0=\frac{d}{dt}f_{N}=\Lop_{\Sp(N)}f_{N}.\]
Let us use the basis of $\Br_{n,-2N}^{*}$ dual to $\B_{n}$ and see $f_{N}$ as a column vector accordingly. Splitting this column as $f_{N}=f_{N}^{1}+f_{N}^{2}$, according to the decomposition of the dual basis into dual permutations and the other dual elements, and using the form of $\Lop_{\Sp(N)}$, we find
\[\Lop f_{N}^{1}=\left(\frac{1}{N}A+\frac{1}{N^{2}}B\right)f_{N}\]
for some rectangular matrices $A$ and $B$ which do not depend on $N$. Since all the components of $f_{N}$ are bounded by $1$, it follows that for any norm on $\Br_{n,-2N}^{*}$, we have $\|\Lop f_{N}^{1}\|=O(N^{-1})$.

From the definition of $\Lop$ given by \eqref{def op L}, and from the fact that a sequence of elements of $\S_{n}$ with increasing number of cycles has length at most $n$, we deduce that $(\Lop+\frac{n}{2})^{n}=0$. This implies that the spectrum of $\Lop$ is reduced to $\{\frac{n}{2}\}$, so that $\Lop$ is injective.

Hence, we have in fact $\|f_{N}^{1}\|=O(N^{-1})$, and the expectation $\E[\Re\tr(W_{N}^{n})]$, which is one of the components of $f^{1}_{N}$, is also dominated by $N^{-1}$.
\end{proof}

\section{Speed of convergence for words of independent Brownian motions}\label{section:speed}

Theorems \ref{limit brown U} and \ref{limite brown}, together with a classical result of D. Voiculescu and its extension to the orthogonal and symplectic cases by B. Collins and P. \'Sniady, allow one to determine the limit of expected traces of arbitrary words in independent Brownian motions on $\U(N,\K)$ as $N$ tends to infinity. Our second main result provides a quantitative estimate of the rate of convergence of such expected traces, in terms of a certain measure of the complexity of the word considered. 

Let us start by recalling how the results of Voiculescu and Collins-\'Sniady apply in the present context.

\subsection{Free limits}\label{free proba} We shall not review here the basic definitions of free probability theory. We recommend \cite{NicaSpeicher} as a general reference.

Recall from \eqref{moments nu} the definition of the measures $(\nu_{t})_{t\geq 0}$. A {\em free multiplicative Brownian motion} is a family $(u_t)_{t\geq 0}$ of unitary elements of a non-commutative probability space $(\A,\tau)$ such that for all $0\leq t_1\leq \ldots \leq t_n$, the increments $u_{t_2}u_{t_1}^*,\ldots,u_{t_n}u_{t_{n-1}}^*$ are free and have respectively the distributions $\nu_{t_2-t_1}, \ldots,\nu_{t_n-t_{n-1}}$. 
Free multiplicative Brownian motions exist and can be realised as the large $N$ limit of the Brownian motion on the unitary group.

\begin{theorem}[Biane, \cite{Biane}]\label{limite libre u} For each $N\geq 1$, let $(U_{N,t})_{t\geq 0}$ be a Brownian motion on $\U(N)$ issued from $I_N$, associated with the scalar product $\langle X,Y\rangle=N\Tr(X^{*}Y)$ on $\u(N)$, defined on a probability space $(\Omega_N,\A_N,\P_N)$. Then the collection $\{ U_{N,t} : t\geq 0\}$ of elements of the non-commutative probability space $(L^\infty(\Omega_N,\A_N,\P_N)\otimes \Mat_N(\C), \E\otimes \tr)$ converges in non-commutative distribution as $N$ tends to infinity to a free unitary Brownian motion. Moreover, independent Brownian motions converge to free unitary Brownian motions which are mutually free.
\end{theorem}

It follows from our study of the orthogonal and symplectic case, and from a result of Collins and \'Sniady \cite[Thm. 5.2]{CollinsSniady} that a similar result holds for orthogonal and symplectic Brownian motions.

There is a small complication due to the fact that we do not regard symplectic matrices as complex matrices. Indeed, the algebra $L^\infty(\Omega_N,\A_N,\P_N)\otimes \Mat_{N}(\H)$ is a real algebra and not a complex one, and we are slightly outside the usual framework of non-commutative probability theory. Here is the short argument which we need to find ourselves back into it.

Consider a real involutive unital algebra $\A$ endowed with a linear form $\tau$ such that $\tau(1)$ and for all $a\in \A$, one has $\tau(aa^{*})\geq 0$. We shall call such a pair $(\A,\tau)$ a real non-commutative probability space. It is straightforward to check that the complexified algebra $\A\otimes \C$ endowed with the involution $(a\otimes z)^{*}=a^{*}\otimes \bar z$ and the linear form $\tau\otimes \id_{\C}$ is a non-commutative probability space in the usual sense. Moreover, for all $a\in \A$, the moments of $a\otimes 1$ in $(A\otimes \C,\tau\otimes \id_{\C})$ are the same as those of $a$ in $(\A,\tau)$.

This being said, we take the liberty of using the language of free probability in a real non-commutative probability space.

\begin{theorem}\label{limite libre} For each $N\geq 1$, let $(R_{N,t})_{t\geq 0}$ be a Brownian motion on $\SO(N)$  issued from $I_N$, associated with the scalar product $\langle X,Y\rangle=\frac{N}{2}\Tr(\t{X}Y)$ on $\so(N)$, defined on a probability space $(\Omega_N,\A_N,\P_N)$. Then the collection $\{ R_{N,t} : t\geq 0\}$ of elements of the non-commutative probability space $(L^\infty(\Omega_N,\A_N,\P_N)\otimes \Mat_N(\R), \E\otimes \tr)$ converges in non-commutative distribution as $N$ tends to infinity to a free unitary Brownian motion. Moreover, independent Brownian motions converge to free unitary Brownian motions which are mutually free.

For each $N\geq 1$, let $(S_{N,t})_{t\geq 0}$ be a Brownian motion on $\Sp(N)$ issued from $I_N$, associated with the scalar product $\langle X,Y\rangle=2N\Re\Tr(X^{*}Y)$ on $\sp(N)$, defined on a probability space $(\Omega_N,\A_N,\P_N)$. Then the collection $\{ S_{N,t} : t\geq 0\}$ of elements of the non-commutative probability space $(L^\infty(\Omega_N,\A_N,\P_N)\otimes \Mat_{N}(\H), \E\otimes \Re\tr)$ converges in non-commutative distribution as $N$ tends to infinity to a free unitary Brownian motion. Moreover, independent Brownian motions converge to free unitary Brownian motions which are mutually free.
\end{theorem}

The main result on which this theorem relies is the following.

\begin{theorem}[Voiculescu ; Collins, \'Sniady]\label{inv asymp free} Choose $\K\in \{\R,\C,\H\}$. Let $(A_{N,1},\ldots,A_{N,n})_{N\geq 1}$ and $(B_{N,1},\ldots,B_{N,n})_{N\geq 1}$ be two sequences of families of random matrices with coefficients in $\K$. Let $a_{1},\ldots,a_{n}$ and $b_{1},\ldots,b_{n}$ be two families of elements of a non-commutative probability space $(\mathcal A,\tau)$. Assume that the convergences in non-commutative distribution
\[(A_{N,1},\ldots,A_{N,n}) \build{\longrightarrow}_{N\to\infty}^{\rm{n.c.d.}} (a_{1},\ldots,a_{n}) \mbox{ and } (B_{N,1},\ldots,B_{N,n}) \build{\longrightarrow}_{N\to\infty}^{\rm{n.c.d.}} (b_{1},\ldots,b_{n})\]
hold. Assume also that for all $N$, given a random matrix $U$ distributed according to the Haar measure on $\U(N,\K)$ and independent of $(A_{N,1},\ldots,A_{N,n},B_{N,1},\ldots,B_{N,n})$, the two families $(A_{N,1},\ldots,A_{N,n},B_{N,1},\ldots,B_{N,n})$ and $(UA_{N,1}U^{-1},\ldots,UA_{N,n}U^{-1},B_{N,1},\ldots,B_{N,n})$ have the same distribution. Then the families $\{a_{1},\ldots,a_{n}\}$ and $\{b_{1},\ldots,b_{n}\}$ are free.
\end{theorem}

For the sake of completeness, and also because our treatment of the symplectic case differs from the most frequent one, we give a proof of this theorem in Appendix \ref{appendice}. 

\begin{proof}[Proof of Theorem \ref{limite libre}] We shall treat the symplectic case, and say at the end how the proof must be adapted to suit the orthogonal case. Let us choose a free multiplicative Brownian motion $(u_{t})_{t\geq 0}$ on a non-commutative probability space $(\A,\tau)$.
We prove by induction on $n$ that for all ordered $n$-tuple of reals $t_{1}<\ldots <t_{n}$, the increments $S_{N,t_{1}},S_{N,t_{2}} S_{N,t_{1}}^{-1},\ldots,S_{N,t_{n}}S_{N,t_{n-1}}^{-1}$ converge in non-commutative distribution towards mutually free unitaries with respective distributions $\nu_{t_{1}},\nu_{t_{2}-t_{1}},\ldots,\nu_{t_{n}-t_{n-1}}$, that is, towards  $(u_{t_{1}},u_{t_{2}}u_{t_{1}}^{-1},\ldots,u_{t_{n}}u_{t_{n-1}}^{-1})$.

For $n=1$, since $S_{N,t_{1}}$ is unitary, the convergence of the moments of arbitrary integer order, including negative integer order, granted by Theorem \ref{limite brown}, implies the convergence of $\E[\Re\tr p(S_{N,t_{1}},S_{N,t_{1}}^{*})]$ to $\tau(p(u_{t_{1}},u_{t_{1}}^{*}))$ for all polynomial $p$, which is the definition of the convergence in non-commutative distribution. Since for $n=1$, the assertion of freeness is empty, the property is proved. 

Let us consider $n>1$ and assume that the property has already been proved for $n-1$ increments. Let us consider $t_{1}<\ldots < t_{n}$. By Lemma \ref{inverses}, the increment $S_{N,t_{n}}S_{N,t_{n-1}}^{-1}$ has the distribution of $S_{N,t_{n}-t_{n-1}}$, so that the property for  $n=1$ implies that it converges in non-commutative distribution to $u_{t_{n}}u_{t_{n-1}}^{-1}$. On the other hand, the property at rank $n-1$ implies that  the increments $S_{N,t_{1}},S_{N,t_{2}}S_{N,t_{1}}^{-1},\ldots,S_{N,t_{n-1}}S_{N,t_{n-2}}^{-1}$ converge in non-commutative distribution towards $u_{t_{1}},u_{t_{2}}u_{t_{1}}^{-1},\ldots,u_{t_{n-1}}u_{t_{n-2}}^{-1}$. Moreover, by Lemma \ref{inverses} again, the increment $S_{N,t_{n}}S_{N,t_{n-1}}^{-1}$ is independent of the other increments which we are considering, and its distribution is invariant by conjugation by any deterministic element of $\Sp(N)$, hence by conjugation by an independent uniform random matrix.

Theorem \ref{inv asymp free} implies that the limits in non-commutative distribution of $S_{N,t_{n}}S_{N,t_{n-1}}^{-1}$ and $(S_{N,t_{1}},S_{N,t_{2}}S_{N,t_{1}}^{-1},\ldots,S_{N,t_{n-1}}S_{N,t_{n-2}}^{-1})$ are free. Hence,  
$(S_{N,t_{1}},S_{N,t_{2}}S_{N,t_{1}}^{-1},\ldots,S_{N,t_{n}}S_{N,t_{n-1}}^{-1})$ converges in non-commutative distribution to  $(u_{t_{1}},u_{t_{2}}u_{t_{1}}^{-1},\ldots,u_{t_{n}}u_{t_{n-1}}^{-1})$.
\end{proof}

\subsection{Second main result: speed of convergence}\label{subsec soc}

In this section, we state our second main result, firstly in its most natural form and then in the form under which we will prove it.

Let $\K$ be one of our three division algebras. We denote generically by $(V_{N,s})_{s\geq 0}$ a Brownian motion on $\U(N,\K)$ as defined in Section \ref{def BM}. We are going to consider several independent copies of this Brownian motion, with which we are going to form a word, of which in turn we will estimate the expected trace. The number of independent copies which we use to form our word will not appear in our final estimates, and this is one of their main strengths. We will nevertheless fix this number and denote it by $q$. Let us thus choose an integer $q\geq 1$, which will stay fixed until the end of Section \ref{section:speed}.

We shall denote by $\Mr_q$ the free monoid generated by the $2q$ letters $x_1,\ldots,x_q,x_{1}^{-1},\ldots,x_{q}^{-1}$. As a set, $\Mr_q$ consists in all finite words in these $2q$ letters, which are to be treated as $2q$ unrelated symbols. Two words can be concatenated, without any cancellation, and this endows $\Mr_{q}$ with an associative operation for which the empty word is the unit element. Let $w$ be an element of $\Mr_{q}$. It is thus a word of the form $x_{i_{1}}^{\epsilon_{1}}\ldots x_{i_{r}}^{\epsilon_{r}}$, where $r\geq 0$ is the length of $w$, and $\epsilon_{1},\ldots,\epsilon_{r}$ belong to $\{-1,1\}$. If $u_{1},\ldots,u_{q}$ are invertible elements of an algebra, we denote by $w(u_{1},\ldots,u_{q})$ the element $u_{i_{1}}^{\epsilon_{1}}\ldots u_{i_{r}}^{\epsilon_{r}}$ of this algebra. We shall use this notation for matrices and for elements of non-commutative probability spaces. The following notation will also be useful later: if $U_{1},\ldots,U_{q}$ belong to $\U(N,\K)$, we shall denote by $w_{\otimes}(U_{1},\ldots,U_{q})$ the element $U_{i_{1}}^{\epsilon_{1}}\otimes \ldots \otimes U_{i_{r}}^{\epsilon_{r}}$ of $\Mat_{N}(\K)^{\otimes r}$. 

The notation suggests a natural homomorphism of monoids from $\Mr_{q}$ to the free group $\Fr_{q}$ on $q$ letters, which sends $x_{i} \in \Mr_{q}$ to $x_{i}\in \Fr_{q}$, and  $x_{i}^{-1} \in \Mr_{q}$ to the inverse of $x_{i}\in \Fr_{q}$. The definition of $w(u_{1},\ldots,u_{q})$ depends on $w$ only through its image by this homomorphism and we shall also use it when $w$ is an element of $\Fr_{q}$. Observe however that this is not true for $w_{\otimes}(U_{1},\ldots,U_{q})$.

We use the free monoid $\Mr_{q}$ to produce a non-commutative probability space in the usual way. Let $\C[\Mr_q]$ be the complex algebra of the monoid $\Mr_q$. It is isomorphic to the algebra of complex polynomials in $2q$ non-commuting indeterminates. It carries an involution characterised by the equality $(\lambda x_{i})^*=\overline{\lambda} x_{i}^{-1}$, valid for all $i\in \{1,\ldots,q\}$ and all $\lambda\in \C$. This involution is anti-multiplicative, so that for all words $w_{1}$ and $w_{2}$, one has $(w_{1}w_{2})^{*}=w_{2}^{*}w_{1}^{*}$.

Let us fix an integer $N\geq 1$. Let $(V_{N,1,s})_{s\geq 0},\ldots,(V_{N,q,s})_{s\geq 0}$ be $q$ independent Brownian motions on the group $\U(N,\K)$. Let also $(u_{1,s})_{s\geq 0},\ldots,(u_{q,s})_{s\geq 0}$ be $q$ free unitary Brownian motions which are mutually free, carried by a non-commutative probability space $(\mathcal A,\phi)$.

In the words which we shall consider, each of our $q$ Brownian motions will possibly appear several times, but always evaluated at the same time. Since the increments of a Brownian motion are independent, and since the number of independent Brownian motions which we consider does not affect our estimates, this does not entail any loss of generality. The times at which we evaluate our Brownian motions are of course important, and we put them into a vector $t=(t_{1},\ldots,t_{q})\in \R_{+}^{q}$.

Let us define a state $\tauknt$ on $\C[\Mr_q]$ by setting, for all $w\in \Mr_{q}$, 
\begin{equation}\label{def tauknt}
\tauknt(w)=\left\{\begin{array}{ll} \E\left[\tr\left(w(V_{N,1,t_{1}},\ldots,V_{N,q,t_{q}}\right)\right] &\mbox{if } \K=\R \mbox{ or }\C, \\[2pt]
\E\left[\Re\tr\left(w(V_{N,1,t_{1}},\ldots,V_{N,q,t_{q}}\right)\right] &\mbox{if } \K=\H.\end{array}\right.
\end{equation}
Theorems \ref{limite libre u} and \ref{limite libre} assert that, as $N$ tends to infinity, $\tauknt$ converges pointwise to the state $\tau_t$ defined by
\[\tau_t(w)=\phi(w(u_{1,t_{1}},\ldots,u_{q,t_{q}})).\]

The main result of this section gives an explicit bound on $|\tauknt(w)-\tau_{t}(w)|$. This bound must of course depend on the word $w$. It does so through a certain non-negative real which we assign to each pair $(w,t)\in \Mr_{q}\times \R_{+}^{q}$, and which we call its Amperean area, for a reason which shall become clear in the second part of this work.

Let us define $q$ functions $\n_{1},\ldots,\n_{q}:\Mr_{q}\to \N$, which could be called partial lengths, as follows. For all element $w$ of $\Mr_q$, written as $w=x_{i_{1}}^{\epsilon_{1}}\ldots x_{i_{r}}^{\epsilon_{r}}$, and all $k\in\{1,\ldots,q\}$, we define $\n_{k}(w)$ as the total number of occurrences of the letter $x_{k}$ in $w$, that is, 
\begin{equation}\label{nk}
\n_{k}(w)=\#\{j\in \{1,\ldots,r\} : i_{j}=k\}.
\end{equation}
We do not make any distinction between $x_{k}$ or $x_{k}^{-1}$. For example, $\n_{3}(x_{3}x_{1}x_{3}^{-1})=2$. Using these partial lengths, we define the Amperean area of the word $w$ relative to $t$ as the real number
\begin{equation}\label{def NCA w}
\Aa_t(w)=\sum_{k=1}^q t_{k} \n_{k}(w)^{2}.
\end{equation}
Let us emphasise that this number does not really depend on $q$. We could see the word $w$ as a word in infinitely many letters, and $t$ as a infinite vector with only finitely many non-zero components. The main estimate is the following.

\begin{theorem}\label{main estim} For all word $w\in \Mr_{q}$ and all $N\geq 1$, the following inequality holds:
\begin{equation}
\left| \tauknt(w)-\tau_t(w)\right| \leq \left\{
\begin{aligned} &\frac{1}{N^{2}}\;  \Aa_t(w)e^{\Aa_t(w)} && \mbox{if } \K=\C,\\
&\frac{1}{N} \; \Aa_t(w)e^{\Aa_t(w)} && \mbox{if } \K=\R \mbox{ or }\H. 
\end{aligned}
\right.
\end{equation}
\end{theorem}

We will in fact prove a more general result, which asserts that the same bounds hold for quantities which are built from the word $w$ but which are more general than $\tauknt(w)$ and $\tau_{t}(w)$. Just as in the proofs of Theorems \ref{limit brown U} and \ref{limite brown}, this generalisation is meant to provide us with a finite set of functions of $t=(t_{1},\ldots,t_{q})$ which satisfies an autonomous differential system. The quantities which we will consider are very similar to the functions $f_{N}(t,\pi)$ considered in these proofs. In particular, we will need a larger set of quantities in the orthogonal and symplectic cases as in the unitary case.

Let us start by the unitary case. For this, let us consider again an element $w$ of $\Mr_{q}$, written as $w=x_{i_{1}}^{\epsilon_{1}}\ldots x_{i_{r}}^{\epsilon_{r}}$. Let us consider a permutation $\sigma\in \S_{r}$. We write $(j_{1}\ldots j_{s})\preccurlyeq \sigma$ to indicate that $(j_{1}\ldots j_{s})$ is a cycle of $\sigma$. Recall from the beginning of Section \ref{unitary revisited} that we defined $\rho_{\C}(\sigma)\in \Mat_{N}(\C)^{\otimes r}$. Recall also, from the beginning of the current section, the notation $w_{\otimes}(U_{N,1,t_{1}},\ldots,U_{N,q,t_{q}})$. In agreement with the convention made at the end of Section \ref{def BM}, we denote respectively by $R$, $U$ and $S$ the orthogonal, unitary and symplectic Brownian motions. With all this preparation, we set
\begin{align*}
\ptcn(w,\sigma)&=N^{-\ell(\sigma)}\E\left[\Tr^{\otimes r}(\rho_{\C}(\sigma) \circ w_{\otimes}(U_{N,1,t_{1}},\ldots,U_{N,q,t_{q}}))\right]\\
&=\E\left[\prod_{(j_{1}\ldots j_{s})\preccurlyeq \sigma}  \tr\left(U_{N,i_{j_{s}},t_{i_{j_{s}}}}^{\epsilon_{j_{s}}}\ldots U_{N,i_{j_{1}},t_{i_{j_{1}}}}^{\epsilon_{j_{1}}}\right)\right],
\end{align*}
and
\[p_{t}(w,\sigma)=\prod_{(j_{1}\ldots j_{s})\preccurlyeq \sigma} \phi\left(u_{i_{j_{s}},t_{j_{s}}}^{\epsilon_{j_{s}}}\ldots u_{i_{j_{1}},t_{j_{1}}}^{\epsilon_{j_{1}}}\right).\]
As usual, we extend these definitions by linearity with respect to $\sigma$, so as to allow an arbitrary element of $\C[\S_{n}]$ to replace $\sigma$.

In the orthogonal and symplectic cases, we introduce the analogous functions indexed by pairings. Let $r\geq 1$ be an integer. Let $\pi\in \B_{r}$ be a pairing of $\{1,\ldots,2r\}$. Recall the construction of the permutation $\sigma_{\pi}\in \S_{r}$ and the signs $\epsilon_{\pi}(1),\ldots,\epsilon_{\pi}(r)$ made at the beginning of Section \ref{section symp}. The following definitions imitate the equation \eqref{P cycles o}.
We define, in the orthogonal case,
\begin{align*}
\ptrn(w,\pi)&=N^{-\ell(\pi)}\E\left[\Tr^{\otimes r}(\rho_{\R}(\pi) \circ w_{\otimes}(R_{N,1,t_{1}},\ldots,R_{N,q,t_{q}}))\right]\\&=\E\left[\prod_{(j_{1}\ldots j_{s})\preccurlyeq\sigma_{\pi}} \tr\left(\left(R_{N,i_{j_{s}},t_{i_{j_{s}}}}^{\epsilon_{j_{s}}}\right)^{\epsilon_{\pi}(s)}\ldots \left(R_{N,i_{j_{1}},t_{i_{j_{1}}}}^{\epsilon_{j_{1}}}\right)^{\epsilon_{\pi}(1)}\right)\right],
\end{align*}
and, in the symplectic case,
\begin{align*}
\pthn(w,\pi)&=(-2N)^{-\ell(\pi)}\E\left[(-2\Re\Tr)^{\otimes r}(\rho_{\H}(\pi) \circ w_{\otimes}(S_{N,1,t_{1}},\ldots,S_{N,q,t_{q}}))\right]\\&=\E\left[\prod_{(j_{1}\ldots j_{s})\preccurlyeq\sigma_{\pi}} \Re\tr\left(\left(S_{N,i_{j_{s}},t_{i_{j_{s}}}}^{\epsilon_{j_{s}}}\right)^{\epsilon_{\pi}(s)}\ldots \left(S_{N,i_{j_{1}},t_{i_{j_{1}}}}^{\epsilon_{j_{1}}}\right)^{\epsilon_{\pi}(1)}\right)\right].
\end{align*}
We have left the case $r=0$ aside. In this case, $w$ is the empty word, the unit element of $\Mr_{q}$, and $\pi$ the empty pairing of the empty set. For the sake of this case, we define $\ptkn(1,\varnothing)=1$.

Let us also define, for both the orthogonal and symplectic cases, $p_{t}(1,\varnothing)=1$ and
\[p_{t}(w,\pi)=\prod_{(j_{1}\ldots j_{s})\preccurlyeq \sigma} \phi\left(\left(u_{i_{j_{s}},t_{j_{s}}}^{\epsilon_{j_{s}}}\right)^{\epsilon_{\pi}(s)} \ldots \left(u_{i_{j_{1}},t_{j_{1}}}^{\epsilon_{j_{1}}}\right)^{\epsilon_{\pi}(1)}\right).\]
We extend these definitions by linearity with respect to $\pi$, in order to be able to replace $\pi$ by an arbitrary element of $\Br_{r,N}$ in the orthogonal case, or $\Br_{r,-2N}$ in the symplectic case.

If we apply these new definitions with $\sigma=(1\ldots r)$, we find $\ptkn(w,(1\ldots r))=\tauknt(w)$ and $p_{t}(w,(1\ldots n))=\tau_t(w)$, so that the following proposition implies Theorem \ref{main estim}.

\begin{proposition}\label{unif estim}  Let $w\in \Mr_{q}$ be a word of length $r\geq 0$. Let $N\geq 1$ be an integer. The following inequalities hold:
\begin{equation*}
\max_{\sigma\in \S_{r}} \left| \ptcn(w,\sigma)-p_{t}(w,\sigma)\right| \leq \frac{1}{2N^{2}} \Aa_t(w)e^{\frac{1}{2}\Aa_t(w)},
\end{equation*}
and, for $\K=\R$ or $\K=\H$,
\begin{equation*}
\max_{\pi\in \B_{r}} \left| \ptkn(w,\pi)-p_{t}(w,\pi)\right| \leq \frac{1}{N} \Aa_t(w)e^{\Aa_t(w)}.
\end{equation*}

Finally, if we replace the Brownian motion on $\U(N)$ by the Brownian motion on $\SU(N)$ in the definition of $p_{t}^{\C,N}$, then 
\begin{equation*}
\max_{\sigma\in \S_{r}} \left| \ptcn(w,\sigma)-p_{t}(w,\sigma)\right| \leq \frac{1}{2N^{2}} \Aa_t(w)e^{\frac{1}{2}\Aa_t(w)}+e^{\frac{1}{2N^{2}}\Aa_{t}(w)}-1.
\end{equation*}
\end{proposition}

We will moreover get the following information from the proof of this proposition.

\begin{proposition}\label{real} For $\K=\R$ or $\K=\C$, the expected trace of any word in independent Brownian motions on $\U(N,\K)$ is real.
\end{proposition}

\subsection{It\^o's equation for words} With our present notation, Section \ref{section one BM} was devoted to the study of quantities of the form $\ptkn(w,\pi)$ when $w$ is a non-negative power of a single letter. In the present setting, we need to extend this study in two respects: firstly, we need to allow more than one letter to appear in $w$ and secondly, we need to allow negative powers of letters to appear. The treatment of the latter issue requires the introduction of some new notation, which is forced upon us by It\^o's formula. Let us see how.

Let $w\in \Mr_{q}$ be a word of length $r\geq 0$. In this paragraph, we will write It\^o's formula for $w_{\otimes}(V_{N,1,t_{1}},\ldots,V_{N,q,t_{q}})$ when among the times $t_{1},\ldots,t_{q}$, all but one are fixed. The integer $N$ is fixed in this section and we will omit it in the notation. The first fundamental relation is the stochastic differential equation satisfied by $V_{t}^{*}$, namely
\[dV_{t}^{*}=-V_{t}^{*}dK_{t}+\frac{c_{\u(N,\K)}}{2} V_{t}^{*} dt.\]
The algebra $\Mat_{N}(\K)^{\otimes r}$ is both a left and right $\Mat_{N}(\K)$-module in $r$ different ways: for each $i\in \{1,\ldots,r\}$ and all $X,M_{1},\ldots,M_{r}\in \Mat_{N}(\K)$, we define
\begin{align}
\nonumber \theta_{i}^{+}(X) \cdot M_{1}\otimes \ldots \otimes M_{r}&=M_{1}\otimes \ldots \otimes XM_{i} \otimes \ldots \otimes M_{r},\\
\label{def theta}\theta_{i}^{-}(X) \cdot M_{1}\otimes \ldots \otimes M_{r}&=M_{1}\otimes \ldots \otimes M_{i}X \otimes \ldots \otimes M_{r}.
\end{align}
With this notation, we unify It\^o's formulas for $V_{t}$ and $V_{t}^{-1}=V_{t}^{*}$. Indeed, for all $\epsilon\in \{-1,1\}$, 
\begin{equation}\label{ito +-}
dV_{t}^{\epsilon}=\epsilon \theta_{1}^{\epsilon}(dK_{t}) \cdot V_{t}^{\epsilon} + \frac{c_{\u(N,\K)}}{2}V_{t}^{\epsilon}\; dt.
\end{equation}
Here and thereafter, we identify the sets $\{-1,1\}$ and $\{-,+\}$ in the obvious way without further comment.

Note that $\theta_{i}^{+}$ and $\theta_{i}^{-}$ satisfy the following relation of adjunction: for all $\xi_{1},\xi_{2}\in \Mat_{N}(\K)^{\otimes r}$, 
\begin{equation}\label{adjunction}
\Tr^{\otimes r}\left((\theta_{i}^{\pm}(X)\cdot \xi_{1}) \xi_{2}\right)=\Tr^{\otimes r}\left( \xi_{1}( \theta_{i}^{\mp}(X)\cdot \xi_{2})\right).
\end{equation}

Let us write $w=x_{i_{1}}^{\epsilon_{1}}\ldots x_{i_{r}}^{\epsilon_{r}}$. For each $k\in \{1,\ldots,q\}$, let us record the positions of $x_{k}$ and $x_{k}^{-1}$ in $w$ by defining
\begin{equation}\label{Sk+-}
X_{k}(w)=\{ j \in \{1,\ldots,r\} : i_{j}=k\}.
\end{equation}
For example, if $w=x_{2}x_{1}^{-1}x_{3}x_{1}^{2}x_{2}$, then $X_{1}(w)=\{2,4,5\}$. Recall that $\n_{k}(w)$ is the cardinal of $X_{k}(w)$.

\begin{lemma}\label{Ito for w} Choose $k\in \{1,\ldots,q\}$. Choose $q-1$ reals $t_{1},\ldots,t_{k-1},t_{k+1},\ldots,t_{q}\geq 0$. It\^o's formula for the process $\left(w_{\otimes}(V_{1,t_{1}},\ldots,V_{q,t_{q}})\right)_{t_{k}\geq 0}$ reads
\begin{align}
\nonumber d_{t_{k}} w_{\otimes}(V_{1,t_{1}},\ldots,V_{q,t_{q}}) =& \sum_{l\in X_{k}(w)} \epsilon_{l} \theta^{\epsilon_{l}}_{l}(dK_{t_{k}}) \cdot w_{\otimes}(V_{1,t_{1}},\ldots,V_{q,t_{q}})\\
\nonumber & +\n_{k}(w)\frac{c_{\u(N,\K)}}{2} w_{\otimes}(V_{1,t_{1}},\ldots,V_{q,t_{q}}) dt_{k}\\
\label{ito V}& + \sum_{\substack{l,m\in X_{k}(w)\\ l<m}} \epsilon_{l}\epsilon_{m} \left(\theta^{\epsilon_{l}}_{l}\otimes \theta^{\epsilon_{m}}_{m}\right)(C_{\g})\cdot w_{\otimes}(V_{1,t_{1}},\ldots,V_{q,t_{q}}) dt_{k}.
\end{align}
In particular, for all $\pi\in \B_{r}$, or all $\pi\in \S_{r}$ if $\K=\C$,
\begin{align}
\nonumber \frac{\partial}{\partial t_{k}} \ptkn(w,\pi)=&\frac{\n_{k}(w)c_{\u(N,\K)}}{2} \ptkn(w,\pi) \\
\label{ito general}&\hspace{-2cm}+ \sum_{\substack{l,m\in X_{k}(w)\\ l<m}}\epsilon_{l}\epsilon_{m} N^{-\ell(\pi)}\E\left[\Tr^{\otimes r}\left[\Big(\left(\theta^{-\epsilon_{l}}_{l}\otimes \theta^{-\epsilon_{m}}_{m}\right)(C_{\u(N,\K)})\cdot \rho_{\K}(\pi)\Big)   w_{\otimes}(V_{1,t_{1}},\ldots,V_{q,t_{q}})\right]\right],
\end{align}
or the same equality with $N$ replaced by $(-2N)$ and $\Tr$ replaced by $(-2\Re \Tr)$ if $\K=\H$.
\end{lemma}

\begin{proof} The equality \eqref{ito V} is only a matter of notation. We apply It\^o's formula in its most usual form to $w_{\otimes}(V_{1,t_{1}},\ldots,V_{q,t_{q}})$, using It\^o's formula for a single Brownian motion as written in \eqref{ito +-} and with the help of the operators $\theta_{i}^{\pm}$ defined by \eqref{def theta}. The Casimir operator appears thanks to the expression \eqref{quadratic tens} of the quadratic variation of $(K_{t})_{t\geq 0}$.

Equation \eqref{ito general} follows from \eqref{ito V}, the definition of $\ptkn(w,\pi)$ given earlier in this section, and the adjunction relation \eqref{adjunction}.
\end{proof}

\subsection{The Brauer algebra III}\label{Brauer III}

It appears in \eqref{ito general} that we need to compute the quantity $\left(\theta_{l}^{-\epsilon_{l}} \otimes \theta_{m}^{-\epsilon_{m}}\right)(C_{\u(N,\K)}) \cdot \rho_{\K}(\pi)$. We are already familiar with this quantity when $\epsilon_{l}=\epsilon_{m}=1$, since in this case it is simply $\rho_{\K}(\pi) (C_{\u(N,\K)})_{lm}$. Similarly, when $\epsilon_{l}=\epsilon_{m}=-1$, it is $(C_{\u(N,\K)})_{lm}\rho_{\K}(\pi)$.

Our aim in this last section devoted to the Brauer algebra is to describe these quantities for all values of $\epsilon_{l}$ and $\epsilon_{m}$. For this, we will introduce six linear operations on the Brauer algebra $\Br_{n,\lambda}$ which generalise the operations which we have already encountered of left and right multiplication by transpositions and contractions. Note that we consider the Brauer algebra of order $n$, although in the context of Section \ref{section:speed}, we take $n$ to be the length of our word $w$, which we denote by $r$.

Let us choose $n\geq 1$ and two distinct integers $a,b$ in $\{1,\ldots,2n\}$. Let us start by describing two simple linear operations associated to $a$ and $b$ on the Brauer algebra $\Br_{n,\lambda}$. For this, let us choose a pairing $\pi\in \B_{n}$. Let $\{a,a'\}$ and $\{b,b'\}$ be the pairs of $\pi$ which contain $a$ and $b$. These pairs must not be distinct. The first operation which we define is the swap of $a$ and $b$: we set
\[S_{a,b}(\pi)=\left(\pi\setminus\{\{a,a'\},\{b,b'\}\}\right) \cup \{\{b,a'\},\{a,b'\}\}.\]
The second operation is the forcing of the pair $\{a,b\}$: we set
\[F_{a,b}(\pi)=\left\{\begin{array}{ll} \left(\pi\setminus\{\{a,a'\},\{b,b'\}\}\right) \cup \{\{a,b\},\{a',b'\}\}. & \mbox{if } \{a,b\}\notin \pi, \\
\lambda \pi & \mbox{if } \{a,b\}\in \pi.
\end{array}\right.
\]
The factor $\lambda$ in this definition can be understood as follows: applying $F_{a,b}$ consists in adding twice the pair $\{a,b\}$, once to form the pair itself, and once to form, by contiguity with the pairs $\{a,a'\}$ and $\{b,b'\}$, the pair $\{a',b'\}$. If the pair $\{a,b\}$ is already present in $\pi$, then this procedure forms a loop, hence the factor $\lambda$.

We can now define the six operations which we are interested in. Let us choose two distinct integers $l,m$ in $\{1,\ldots,n\}$. We define six linear endomorphisms of $\Br_{n,\lambda}$, which we denote by $T_{lm}^{++}$, $T_{lm}^{--}$, $T_{lm}^{+-}$,$P_{lm}^{++}$, $P_{lm}^{--}$ and $P_{lm}^{+-}$, according to the following table, where the second row defines the first.

\begin{center}
\begin{tabular}{|c|c|c|c|c|c|}
\hline
$T_{lm}^{++}$ & $T_{lm}^{--}$ & $T_{lm}^{+-}$ &$P_{lm}^{++}$ & $P_{lm}^{--}$& $P_{lm}^{+-}$ \\[2.5pt]
$S_{l,m}$ & $S_{n+l,n+m}$ & $F_{l,n+m}$ & $F_{l,m}$ & $F_{n+l,n+m}$ & $S_{l,n+m}$\\[1.5pt]
\hline
\end{tabular}
\end{center}
We complete these definitions by setting $T_{lm}^{-+}=T_{ml}^{+-}$ and $P_{lm}^{-+}=P_{ml}^{+-}$. One checks that if $\{l,m\}\cap\{i,j\}=\varnothing$, then with all possible choices of signs, the following commutation relations hold:
\begin{equation}\label{commute}
[T_{lm}^{**},T_{ij}^{**}]=[T_{lm}^{**},P_{ij}^{**}]=[P_{lm}^{**},P_{ij}^{**}]=0.
\end{equation}

It follows immediately from the definitions that the following equalities hold:
\[T_{lm}^{++}(\pi)=\pi (l\, m), \; T_{lm}^{--}(\pi)=(l\, m)\pi, \; P_{lm}^{++}(\pi)= \pi \langle l\, m\rangle, \; P_{lm}^{--}(\pi)= \langle l\, m\rangle \pi.\]
The definitions of $T_{lm}^{+-}$ and $P_{lm}^{+-}$ may look inconsistent with the previous ones, but the following lemma explains why we chose them in this way.

\begin{figure}[h!]
\begin{center}
\scalebox{0.8}{\includegraphics{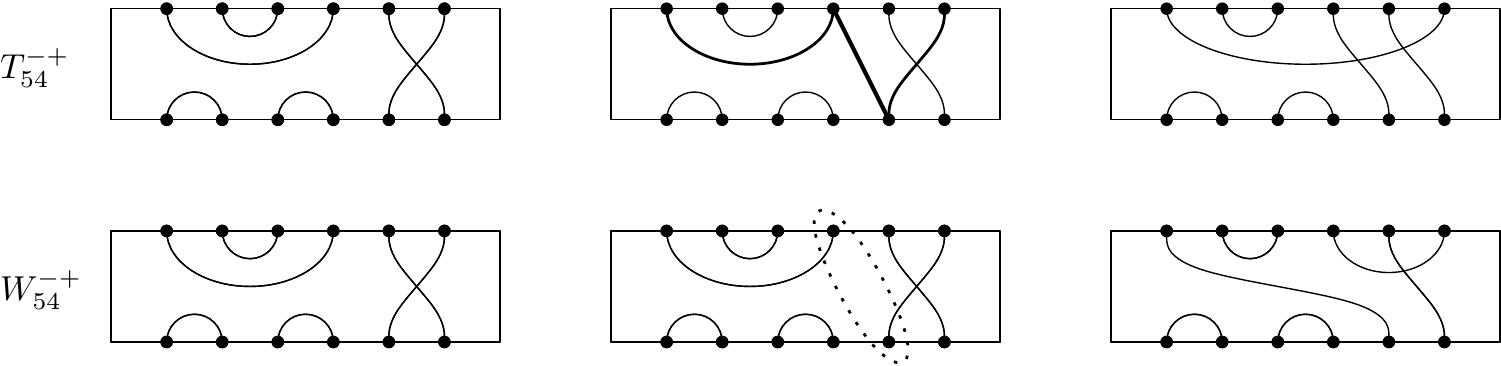}}\\
\caption{\small \label{TW} In the first line, the operation $T_{54}^{-+}$ is applied to the pairing represented on the left. The result is represented on the right. The second line is a similar representation of the operation $P_{54}^{-+}$.}
\end{center}
\end{figure}

\begin{lemma}\label{calcul TW} Let $\pi\in \B_{n}$ be a pairing. Let $l,m$ be distinct integers between $1$ and $n$. Let $\epsilon_{l},\epsilon_{m}$ be two elements of $\{-1,1\}$. The following equalities hold in $\Br_{n,N}$: 
\[\left(\theta_{l}^{-\epsilon_{l}} \otimes \theta_{m}^{-\epsilon_{m}}\right)(T) \cdot \rho(\pi)=\rho(T_{lm}^{\epsilon_{l}\epsilon_{m}}(\pi)) \mbox{ and }  \left(\theta_{l}^{-\epsilon_{l}} \otimes \theta_{m}^{-\epsilon_{m}}\right)(P) \cdot \rho(\pi)=\rho(P_{lm}^{\epsilon_{l}\epsilon_{m}}(\pi)).\]
\end{lemma}

\begin{proof} In the case where $\epsilon_{l}=\epsilon_{m}=-1$, the first equality follows from the identity $\left(\theta_{l}^{+} \otimes \theta_{m}^{+}\right)(T) \cdot \rho(\pi)=\rho((l\, m))\rho(\pi)$ and the fact that $\rho:\Br_{n,N}\to \Mat_{N}(\R)^{\otimes r}$ is a homomorphism of algebras (see Section \ref{section:brauer I}). The same arguments apply to the second equality, as well as to both equalities in the case where $\epsilon_{l}=\epsilon_{m}=1$.

Let us compute $\left(\theta_{l}^{-} \otimes \theta_{m}^{+}\right)(T) \cdot \rho(\pi)$. We will make the assumption that $l<m$ but this plays no role in the computation. We find
\begin{align*}
\left(\theta_{l}^{-} \otimes \theta_{m}^{+}\right)(T) \cdot \rho(\pi) &= \sum_{i_{1},\ldots,i_{2n},a,b=1}^{N}\left( \prod_{\{u,v\}\in\pi }\delta_{i_{u},i_{v}}\right) \ldots \otimes E_{i_{n+l},i_{l}}E_{ab} \otimes \ldots \otimes E_{ba}E_{i_{n+m},i_{m}}\otimes \ldots\\
& =\sum_{i_{1},\ldots,i_{2n},b=1}^{N}\left(\delta_{i_{l},i_{n+m}} \prod_{\{u,v\}\in\pi }\delta_{i_{u},i_{v}}\right) \ldots \otimes E_{i_{n+l},b} \otimes \ldots \otimes E_{b,i_{m}}\otimes \ldots.
\end{align*}
If $\{l,n+m\}$ is a pair of $\pi$, then the factor $\delta_{i_{l},i_{n+m}}$ is already present in the product over the pairs of $\pi$, the matrices $E_{i_{n+l},b}$ and $E_{b,i_{m}}$ can be respectively replaced by $E_{i_{n+l},i_{l}}$ and $E_{i_{n+m},i_{m}}$, and we recover $\rho(\pi)$, multiplied by the factor $N$ due to the now superfluous index $b$.
If $\{l,n+m\}$ is not a pair of $\pi$, then we perform the summation over $i_{l}$ and $i_{n+m}$ which do not appear in the tensor product anymore. We have the partial sum $\sum_{i_{l},i_{n+m}} \delta_{i_{l'},i_{l}}\delta_{i_{l},i_{n+m}}\delta_{i_{n+m},i_{m'}}=\delta_{i_{l'},i_{m'}}$. We finally use the index $b$ to reintroduce $i_{l}$ and $i_{n+m}$, according to the relation
\[\sum_{b} \ldots \otimes E_{i_{n+l},b} \otimes \ldots \otimes E_{b,i_{m}}\otimes \ldots = \sum_{i_{l},i_{n+m}} \delta_{i_{l},i_{n+m}} 
\ldots \otimes E_{i_{n+l},i_{l}} \otimes \ldots \otimes E_{i_{n+m},i_{m}}\otimes \ldots,\]
and find ourselves left with the very definition of $\rho(T_{lm}^{+-}(\pi))$.

The computation of $\left(\theta_{l}^{-} \otimes \theta_{m}^{+}\right)(P) \cdot \rho(\pi)$ is similar, but the difference is
significant enough for us to deem it necessary to give some details. We have
\begin{align*}
\left(\theta_{l}^{-} \otimes \theta_{m}^{+}\right)(P) \cdot \rho(\pi) &= \sum_{i_{1},\ldots,i_{2n},a,b=1}^{N}\left( \prod_{\{u,v\}\in\pi }\delta_{i_{u},i_{v}}\right) \ldots \otimes E_{i_{n+l},i_{l}}E_{ab} \otimes \ldots \otimes E_{ab}E_{i_{n+m},i_{m}}\otimes \ldots\\
& =\sum_{i_{1},\ldots,i_{2n},a,b=1}^{N}\left(\delta_{i_{l},a} \delta_{i_{n+m},b} \prod_{\{u,v\}\in\pi }\delta_{i_{u},i_{v}}\right) \ldots \otimes E_{i_{n+l},b} \otimes \ldots \otimes E_{a,i_{m}}\otimes \ldots.
\end{align*}
If $\{l,n+m\}$ is a pair of $\pi$, then the only non-zero contributions come from the terms where $a=b=i_{l}=i_{n+m}$ and the last expression is equal to $\rho(\pi)$. Otherwise, we can sum over $i_{n+m}$ and $i_{l}$ thanks to $\sum_{i_{n+m},i_{l}} \delta_{a,i_{l}}\delta_{i_{l},i_{l'}}\delta_{b,i_{n+m}}\delta_{i_{n+m},i_{m'}}=\delta_{a,i_{l'}}\delta_{b,i_{m'}}$ and use the same formula in the reverse direction, only exchanging $i_{l}$ and $i_{n+m}$, thus replacing $\delta_{a,i_{l'}}\delta_{b,i_{m'}}$ by $\sum_{i_{n+m},i_{l}} \delta_{a,i_{n+m}}\delta_{i_{n+m},i_{l'}}\delta_{b,i_{l}}\delta_{i_{l},i_{m'}}$. If we finally replace $a$ by $i_{n+m}$ and $b$ by $i_{l}$, we find $\rho(P_{lm}^{+-}(\pi))$.
\end{proof}

In the unitary case, we are going to apply Lemma \ref{calcul TW} only when $\pi$ is a permutation and considering only the actions derived from $T$. It follows from the observation made just after their definition that $T_{lm}^{++}$ and $T_{lm}^{--}$ leave the subspace $\C[\S_{n}]$ of $\Br_{n,N}$ invariant. The next lemma asserts the same of $T_{lm}^{+-}$.

\begin{lemma}\label{T sigma} Let $\lambda$ be a complex number. The linear subspace $\C[\S_{n}]$ of $\Br_{n,\lambda}$ is stable by $T_{lm}^{+-}$. More precisely, let $\sigma$ be an element of $\S_{n}$. For all distinct integers $l,m$ between $1$ and $n$, we have the following equality in $\C[\S_{n}]$:
\[T_{lm}^{+-}(\sigma)=\lambda^{\delta_{\sigma(l),m}} (\sigma(l)\, m) \sigma.\]
\end{lemma}

\begin{proof} The pair $\{l,n+m\}$ belongs to the pairing associated to $\sigma$ if and only if $\sigma(l)=m$. The formula is thus true in this case. Let us assume that $\sigma(l)\neq m$. Then $T_{lm}^{+-}(\sigma)$ is the pairing associated to $\sigma$ in which the pairs $\{l,n+\sigma(l)\}$ and $\{\sigma^{-1}(m),n+m\}$ have been replaced by $\{\sigma^{-1}(m),n+\sigma(l)\}$ and $\{l,n+m\}$. It is the pairing associated to a permutation $\tilde \sigma$, which satisfies $\tilde\sigma(i)=\sigma(i)$ for all $i\in \{1,\ldots,n\}\setminus\{\sigma^{-1}(m),l\}$, $\tilde\sigma(\sigma^{-1}(m))=\sigma(l)$ and $\tilde\sigma(l)=m$. Thus, $\tilde\sigma=(\sigma(l)\, m)\sigma$, as expected.
\end{proof}

In the symplectic case, we will need, in addition to the tools developed for the unitary and orthogonal cases, a description similar to that given by Lemma \ref{calcul TW} of the behaviour of the homomorphism $\gamma$ defined by \eqref{def Srep} with respect to the operations $T_{lm}^{+-}$ and $P_{lm}^{+-}$. Recall from \eqref{re co} the definition of ${\rm Re}^{\H}$ and ${\rm Co}^{\H}$.

\begin{proposition}\label{calcul TW sp}Let $\pi\in \B_{n}$ be a pairing. Choose two distinct integers $l,m$ in $\{1,\ldots,n\}$. Then
\begin{align*}
&\left(\theta_{l}^{-\epsilon_{l}}\otimes \theta_{m}^{-\epsilon_{m}}\right)\left({\rm Re^{\H}} \right) \cdot \gamma(\pi)=\gamma(T_{lm}^{\epsilon_{l}\epsilon_{m}}(\pi)),\\
&\left(\theta_{l}^{-\epsilon_{l}}\otimes \theta_{m}^{-\epsilon_{m}}\right)\left({\rm Co^{\H}}\right) \cdot \gamma(\pi)=\gamma(P_{lm}^{\epsilon_{l}\epsilon_{m}}(\pi)).
\end{align*}
\end{proposition}

\begin{proof} When $\epsilon_{l}=\epsilon_{m}=-1$, the two assertions are a consequence of Proposition \ref{mult sp}. The other cases are treated exactly in the same way as we proved Proposition \ref{mult sp}. We summarise in Figure \ref{nonper} the information which is needed to the prove each equality on the model of the computation given extensively in the proof of Proposition \ref{mult sp}. This table contains in fact all cases, including those of Proposition \ref{mult sp} itself.
\end{proof}

\begin{figure}[h!]
\begin{center}
\begin{tabular}{|c|c||ccc|ccc|} 
\cline{3-8} 
\multicolumn{2}{c|}{} & $T_{lm}^{++}$ & $T_{lm}^{--}$ & $T_{lm}^{+-}$ & $P_{lm}^{++}$ & $P_{lm}^{--}$ & $P_{lm}^{+-}$ \\[2pt]
\hhline{--|======|}
\multirow{2}{*}{$\begin{array}{cc} \mbox{Same} \\ \mbox{cycle} \\ \epsilon_{l}=1\end{array}$} & 
$\epsilon_{m}=1$ & 
$\begin{array}{c} 1 \\ {\hspace{2pt}} \end{array}$  &
$\begin{array}{c} 1 \\ \scriptsize\eqref{quat 3}\end{array}$  & 
$\begin{array}{c}0^{*} \mbox{ or } 1 \\ {\hspace{2pt}}\end{array}$  & 
$\begin{array}{c} 0 \\  {\hspace{2pt}}\end{array}$ & 
$\begin{array}{c} 0 \\ \scriptsize\eqref{quat 4}\end{array}$ & 
$\begin{array}{c} 0 \\ {\hspace{2pt}}\end{array}$ \\
\hhline{||~-||------}
& $\epsilon_{m}=-1$ & 
$\begin{array}{c} 0 \\ {\hspace{2pt}}\end{array}$ & 
$\begin{array}{c} 0 \\ \scriptsize\eqref{quat 4}\end{array}$ & 
$\begin{array}{c} 0 \\ {\hspace{2pt}}\end{array}$ & 
$\begin{array}{c}0^{*} \mbox{ or } 1 \\ {\hspace{2pt}}\end{array}$ & 
$\begin{array}{c}0^{\dagger} \mbox{ or } 1 \\ \scriptsize\eqref{quat 3}\end{array}$ & 
$\begin{array}{c} 1 \\ {\hspace{2pt}}\end{array}$ \\
\hhline{--||------}
\multicolumn{2}{|c||}{\multirow{1}{*}{Different cycles}} & 
$\begin{array}{c} -1 \\ {\hspace{2pt}}\end{array}$&
$\begin{array}{c} -1 \\ \scriptsize\eqref{quat 1}\end{array}$&
$\begin{array}{c} -1 \\ {\hspace{2pt}}\end{array}$&
$\begin{array}{c} -1 \\ {\hspace{2pt}}\end{array}$&
$\begin{array}{c} -1 \\ \scriptsize\eqref{quat 2}\end{array}$&
$\begin{array}{c} -1 \\ {\hspace{2pt}}\end{array}$\\
\hline
\multicolumn{4}{c}{}& \multicolumn{3}{|c}{$\scriptstyle\mbox{\scriptsize If and only if }\; {}^{*}\sigma_{\pi}(l)=m$} \mbox{\scriptsize or } ${}^{\dagger}\scriptstyle\sigma_{\pi}(m)=l.$ &  \multicolumn{1}{|c}{}\\
\multicolumn{4}{c}{}& \multicolumn{3}{|c|}{${}^{*,\dagger}\scriptstyle\mbox{\scriptsize A factor }  \lambda \mbox{ \scriptsize  is produced}.$} &  \multicolumn{1}{c}{}\\
\cline{5-7}
\end{tabular}
\caption{\label{nonper}\small The table is read as follows. Consider a paring $\pi\in \B_{n}$. Choose $l,m$ distinct integers between $1$ and $n$. Whether $l$ and $m$ are in the same cycle of $\sigma_{\pi}$ or not, and if they are, whether $\epsilon_{\pi}(l)\epsilon_{\pi}(m)=1$ or $-1$ determines which row of the table we must look at. When $l$ and $m$ are in the same cycle, we orient this cycle in such a way that $\epsilon_{l}=1$. The entry of the table corresponding to the operation we are interested in tells us how the number of cycles of $\pi$ will be affected by this operation, if it will produce a factor $\lambda$ (the parameter of the Brauer algebra), and which of the four identities \eqref{quat 1} - \eqref{quat 4} is used in the proof of the corresponding part of Proposition \ref{calcul TW sp}. 
} 
\end{center}
\end{figure}

\subsection{The unitary case}
We now turn to the proof of Proposition  \ref{unif estim} in the unitary case. Just as in the proof of the first main result, the strategy is to differentiate with respect to $t=(t_{1},\ldots,t_q)$, to show that $\ptcn(w,\sigma)$ and $p_{t}(w,\sigma)$ satisfy differential relations which are not very different. The difference with the first main result is that we will quantify the difference between the differential systems and draw quantitative conclusions on the difference between $\ptcn(w,\sigma)$ and $p_{t}(w,\sigma)$. The following elementary and well-known fact will be instrumental. 

\begin{lemma} Let $d\geq 1$ be an integer. Let $\|\cdot \|$ be a norm of algebra on $\Mat_{d}(\C)$. Let $A,B$ be two elements of $\Mat_d(\C)$. Then
\begin{equation}\label{diff exp}
\left\|e^{A+B}-e^{A}\right\|\leq \| B \| e^{\max(\|A+B\|, \|A\|)}.
\end{equation}
\end{lemma}

\begin{proof} We simply write
\begin{align*}
\left\|e^{A+B}-e^{A}\right\| &= \left\|\int_{0}^{1} \frac{d}{dt} \left[e^{t(A+B)}e^{(1-t)A}\right] \; dt\right\|
\leq \int_{0}^{1} \left\|e^{t(A+B)} B e^{(1-t)A}\right\| \; dt\\
&\leq \|B\| \int_{0}^{1} e^{t\|A+B\|+(1-t)\|A\|} \; dt
\leq \|B\| e^{\max(\|A+B\|,\|A\|)},
\end{align*}
and find the expected inequality.
\end{proof}

We will apply this result with the norm on $\Mat_{d}(\C)$ associated to the $\ell^{\infty}$ norm on $\C^{d}$. It matters for us that this norm is given explicitly, for a matrix $A=(A_{ij})_{i,j=1\ldots d}$, by
\begin{equation}\label{matrix norm}
\|A\|=\max_{i=1\ldots d} \sum_{j=1}^{d} |A_{ij}|.
\end{equation}

\begin{proof}[Proof of Proposition \ref{unif estim} in the unitary case] Let $w=x_{i_{1}}^{\epsilon_{1}}\ldots x_{i_{r}}^{\epsilon_{r}}$ be a an element of $\Mr_{q}$. Let $\sigma\in \S_{r}$ be a permutation. We start from the result of Lemma \ref{Ito for w} and more specifically from \eqref{ito general}, applied to the word $w$, the pairing $\pi=\sigma$, and an integer $k\in \{1,\ldots,q\}$.

Let us apply Lemmas \ref{calcul TW} and \ref{T sigma}.  Thanks to the expression \eqref{casimir o u} of $C_{\u(N)}$, we find that $\frac{\partial}{\partial t_{k}} \ptcn(w,\sigma)+\frac{\n_{k}(w)}{2} \ptcn(w,\sigma)$ is equal to
\[-\sum_{\substack{l,m\in X_{k}(w)\\ l<m}} \epsilon_{l}\epsilon_{m} N^{-\ell(\sigma)-1}\E\left[\Tr^{\otimes r}\left(\rho_{\C}(T_{lm}^{\epsilon_{l}\epsilon_{m}} (\sigma)) \circ   w_{\otimes}(U_{N,1,t_{1}},\ldots,U_{N,q,t_{q}})\right)\right].\]

Let us write $l\stackrel{\sigma}{\sim} m$ if $l$ and $m$ are in the same cycle of $\sigma$, and $l\not\stackrel{\sigma}{\sim} m
$ otherwise. Using the left half of the first third rows of the table \ref{nonper}, we find
\begin{equation}
\frac{\partial}{\partial t_{k}} \ptcn(w,\sigma)= -\frac{\n_{k}(w)}{2} \ptcn(w,\sigma)- \sum_{\substack{l,m\in X_{k}(w)\\ l<m}}\left({\1}_{l\stackrel{\sigma}{\sim} m }+\frac{1}{N^{2}} {\1}_{l \not\stackrel{\sigma}{\sim} m}\right) \epsilon_{l}\epsilon_{m} \ptcn(w,T_{lm}^{\epsilon_{l}\epsilon_{m}}(\sigma) ).\label{edo ptcn}
\end{equation}
If a term occurs in this sum with $\epsilon_{l}=1$, $\epsilon_{m}=-1$ and $\sigma(l)=m$, then we have $T_{lm}^{+-}(\sigma)=N \sigma$ and this term produces a contribution of order $N^{0}$.

Let us write \eqref{edo ptcn} in its integral form 
\[\ptcn(w,\sigma)=1+\int_{0}^{t} \left(\mbox{r.h.s. of \eqref{edo ptcn} at }t=s\right) \; ds.\]
As $N$ tends to infinity, the pointwise convergence of $\ptcn(w,\sigma)$ towards $p_{t}(w,\sigma)$, the fact that $\left|\ptcn(w,\sigma)\right|\leq 1$ and the dominated convergence theorem imply that
\[p_{t}(w,\sigma)=1+\int_{0}^{t} \left(\mbox{r.h.s. of \eqref{edo ptcn} at }\frac{1}{N}=0 \mbox{ and }t=s\right) \; ds.\]
Hence, the family of functions $\{p_{t}(w,\sigma) : \sigma\in \S_{r}\}$ satisfies the following differential system: for all $\sigma\in \S_{r}$,
\begin{equation}
\nonumber\frac{\partial}{\partial t_{k}} p_{t}(w,\sigma)= -\frac{\n_{k}(w)}{2} p_{t}(w,\sigma)-\sum_{\substack{l,m\in X_{k}(w)\\ l<m, l\stackrel{\sigma}{\sim} m}} \epsilon_{l} \epsilon_{m} p_{t}(w,T_{lm}^{\epsilon_{l}\epsilon_{m}}(\sigma)). 
\label{edo ptc}
\end{equation}

To the word $w$, and for each $k\in \{1,\ldots,q\}$, we may thus associate two real $r! \times r!$ matrices $A_{k}$ and $C_{k}$, as follows. We define, for all $\sigma,\sigma'\in \S_{r}$,
\begin{align*}
(A_{k})_{\sigma,\sigma'}=-\frac{\n_{k}(w)}{2}\delta_{\sigma,\sigma'}-\sum_{\substack{l,m\in X_{k}(w)\\ l<m, l\stackrel{\sigma}{\sim} m}} \epsilon_{l}\epsilon_{m} \delta_{T_{lm}^{\epsilon_{l}\epsilon_{m}}(\sigma),\sigma'}
\end{align*}
and
\begin{align*}
(C_{k})_{\sigma,\sigma'}=&-\sum_{\substack{l,m\in X_{k}(w)\\ l<m, l\stackrel{\sigma}{\not\sim} m}} \epsilon_{l}\epsilon_{m} \delta_{T_{lm}^{\epsilon_{l}\epsilon_{m}}(\sigma),\sigma'}.
\end{align*}

For all distinct $k_{1},k_{2}\in\{1,\ldots,q\}$, the sets $X_{k_{1}}(w)$ and $X_{k_{2}}(w)$ are disjoint, so that \eqref{commute} implies the commutation relations
\begin{equation}\label{commute AC}
[A_{k_{1}},A_{k_{2}}]=[A_{k_{1}},C_{k_{2}}]=[C_{k_{1}},C_{k_{2}}]=0.
\end{equation}

Let us define the vector $\ptcn(w)=(\ptcn(w,\sigma))_{\sigma\in \S_{r}}$. Let us write explicitly the dependence of $\ptcn(w)$ on $t_{1},\ldots,t_q$. We have
\[\pcn_{(t_1,\ldots,t_q)}(w)=e^{t_{1}(A_{1}+\frac{1}{N^{2}}C_{1})}\pcn_{(0,t_2,\ldots,t_q)}(w).\]
Moreover, $\pcn_{(0,\ldots,0)}(w)=\1$, the vector of $\C^{r!}$ whose components are all equal to $1$. Thus, we have 
\begin{equation}
\pcn_{(t_1,\ldots,t_q)}(w)=\left(\prod_{k=1}^{q} e^{t_{k}\left(A_{k}+\frac{1}{N^{2}}C_{k}\right)}\right){\1},
\end{equation}
where the order in this product is irrelevant, thanks to \eqref{commute AC}.
Similarly, defining $p_{t}(w)=(p_{t}(w,\sigma))_{\sigma\in \S_{r}}$, we have
\begin{equation}\label{inflate pt}
p_{(t_1,\ldots,t_q)}(w)=\left(\prod_{k=1}^{q} e^{t_{k}A_{k}}\right){\1}.
\end{equation}
We can express the $\ell^{\infty}$ norm of the difference:
\begin{align*}
\|\pcn_{(t_1,\ldots,t_q)}(w)-p_{(t_1,\ldots,t_q)}(w)\|_{\infty} &= \left\| \left(\prod_{k=1}^{q} e^{t_{k}\left(A_{k}+\frac{1}{N^{2}}C_{k}\right)} - \prod_{k=1}^{q} e^{t_{k} A_{k}(w)}\right) {\1}\right\|_{\infty} \\
&\hskip -0cm \leq \sum_{l=1}^{q} \left\| \prod_{k=l+1}^{q} e^{t_{k}\left(A_{k}+\frac{1}{N^{2}}C_{k}\right)} \left(e^{t_{l}\left(A_{l}+\frac{1}{N^{2}}C_{l}\right)} - e^{t_{l} A_{l}}\right)\prod_{k=1}^{l-1} e^{t_{k}A_{k}} {\1} \right\|_{\infty}\\
&\hskip -0cm \leq \sum_{l=1}^{q} \prod_{k=l+1}^{q} e^{t_{k} \|A_{k}+\frac{1}{N^{2}}C_{k}\|} \left\|
e^{t_{l}\left(A_{l}+\frac{1}{N^{2}}C_{l}\right)} - e^{t_{l} A_{l}} \right\| \prod_{k=1}^{l-1} e^{t_{k}\|A_{k}\|}.
\end{align*}
Recall from \eqref{matrix norm} the definition of the norm which we are using on $\Mat_d(\C)$. It is easy to check that the following inequalities hold for all $N\geq 1$:
\begin{equation}\label{normesAB}
\|A_{k}\| \leq \frac{\n_{k}(w)^{2}}{2} \; , \;\; \|C_{k}\| \leq \frac{\n_{k}(w)^{2}}{2}\; , \;\; \|A_{k}+ \frac{1}{N^{2}}C_{k}\| \leq \frac{\n_{k}(w)^{2}}{2}.
\end{equation}
Now, applying (\ref{diff exp}) and thanks to (\ref{normesAB}), we find
\begin{align*}
\|\pcn_{(t_1,\ldots,t_q)}(w)-p_{(t_1,\ldots,t_q)}(w)\|_{\infty} &\leq 
\sum_{l=1}^{q} \prod_{k=l+1}^{q} e^{\frac{1}{2}t_{k} \n_{k}(w)^2} \frac{t_{l}\n_{l}(w)^{2}}{2N^{2}} e^{\frac{1}{2}t_{l}\n_{l}(w)^{2}} \prod_{k=1}^{l-1} e^{\frac{1}{2}t_{k} \n_{k}(w)^{2}}\\
&= \frac{1}{2N^{2}}\Aa(w)e^{\frac{1}{2}\Aa(w)},
\end{align*}
which is the expected inequality.

In the case of the special unitary group, the last assertion of Proposition \ref{espérances groupes} implies that each function $\ptcn(w)$ is multiplied, regardless of $\sigma$, by the factor
\[\exp\left(\frac{1}{2N^{2}} \sum_{k=1}^{q} \n_{k}(w)^{2} t_{k}\right)=e^{\frac{1}{2N^{2}} \Aa_{t}(w)}.\]
The inequality to prove in this case follows now from the fact that $|\ptcn(w)|\leq 1$.
\end{proof}

\subsection{The orthogonal case} The proof in the orthogonal case follows the same pattern as in the unitary case.

\begin{proof}[Proof of Proposition \ref{unif estim} in the orthogonal case] Let $w=x_{i_{1}}^{\epsilon_{1}}\ldots x_{i_{r}}^{\epsilon_{r}}$ be an element of $\Mr_{q}$ of length $r$. Let $\pi\in \B_{r}$ be a pairing of $\{1,\ldots,2r\}$. 
We start again from \eqref{ito general}, applied to $w$, $\pi$ and an integer $k\in \{1,\ldots,q\}$. We find, thanks to Lemma \ref{calcul TW} and \eqref{iotacasimir so}, that

$\frac{\partial}{\partial t_{k}} \ptrn(w,\pi)+\n_{k}(w)\frac{(N-1)}{2N} \ptrn(w,\pi)$ is equal to
\begin{align*}
&-\sum_{\substack{l,m\in X_{k}(w)\\ l<m}} \epsilon_{l} \epsilon_{m} N^{-\ell(\pi)-1}\E\left[\Tr^{\otimes r}\left(\rho(T_{lm}^{\epsilon_{l}\epsilon_{m}}(\pi)- P_{lm}^{\epsilon_{l} \epsilon_{m}}(\pi)) \circ   w_{\otimes}(R_{N,1,t_{1}},\ldots,R_{N,q,t_{q}})\right)\right].
\end{align*}

From this expression, we deduce
\begin{align}
\nonumber \frac{\partial}{\partial t_{k}} \ptrn(w,\pi)=& -\frac{\n_{k}(w)(N-1)}{2N} \ptrn(w,\pi)\\
 \nonumber &-\sum_{\substack{l,m\in X_{k}(w)\\ l<m}} \epsilon_{l} \epsilon_{m} N^{\ell(T_{lm}^{\epsilon_{l} \epsilon_{m}}(\pi))-\ell(\pi)-1} \ptrn(w,T_{lm}^{\epsilon_{l} \epsilon_{m}}(\pi) )\\
&+ \sum_{\substack{l,m\in X_{k}(w)\\ l<m}} \epsilon_{l} \epsilon_{m} N^{\ell(P_{lm}^{\epsilon_{l} \epsilon_{m}}(\pi) )-\ell(\pi)-1} \ptrn(w,P_{lm}^{\epsilon_{l} \epsilon_{m}}(\pi) ). \label{edo ptrn}
\end{align}
We claim that the only exponents of $N$ which can appear in this sum are $0$, $-1$ and $-2$. For the terms where $\epsilon_{l}=\epsilon_{m}$, this is something which we already discussed in the one-matrix case. Recall in particular from the proof of Theorem \ref{limite brown} in the orthogonal case, which we gave in Section \ref{sec : orth 1M}, that in the case where $\langle l\, m\rangle\pi =N\pi$ or $\pi\langle l\, m\rangle=N\pi$, we get a term of order $N^{0}$. The situation is the same for $T_{lm}^{+-}(\pi)$ and $P_{lm}^{+-}(\pi)$: both $\ell(T_{lm}^{+-}(\pi))-\ell(\pi)$ and $\ell(P_{lm}^{+-}(\pi))-\ell(\pi)$ belong to $\{-1,0,1\}$. Moreover, in the case where $T_{lm}^{+-}(\pi)=N\pi$, we get a term of order $N^{0}$.

As in the unitary case, the integral form $\ptrn(w,\pi)=1+\int_{0}^{t} \left(\mbox{r.h.s. of \eqref{edo ptrn} at }t=s\right) ds$ of \eqref{edo ptrn} converges, as $N$ tends to infinity, to $p_{t}(w,\pi)=1+\int_{0}^{t} \left(\mbox{r.h.s. of \eqref{edo ptrn} at }\frac{1}{N}=0 \mbox{ and }t=s\right) \; ds$.

Hence, the family of functions $\{p_{t}(w,\pi) : \pi\in \B_{r}\}$ satisfies the following differential system: for all $\pi\in \B_{r}$,
\begin{align}
\nonumber\frac{\partial}{\partial t_{k}} p_{t}(w,\pi)= -\frac{\n_{k}(w)}{2} p_{t}(w,\pi)&- \sum_{\substack{l,m\in X_{k}(w),l< m\\ \ell(T_{lm}^{\epsilon_{l}\epsilon_{m}}(\pi))=\ell(\pi)+1}}\hspace{-5mm} \epsilon_{l} \epsilon_{m}  p_{t}(w,T_{lm}^{\epsilon_{l} \epsilon_{m}}(\pi))\\
&+ \sum_{\substack{l,m\in X_{k}(w),l< m\\ \ell(P_{lm}^{\epsilon_{l} \epsilon_{m}}(\pi))=\ell(\pi)+1}}\hspace{-5mm} \epsilon_{l} \epsilon_{m}  p_{t}(w,P_{lm}^{\epsilon_{l}\epsilon_{m}}(\pi)).
\label{edo ptr}
\end{align}

Let $(2r)!!=\prod_{k=1}^{r} (2k-1)$ denote the cardinal of $\B_{r}$. To the word $w$, and for each $k\in \{1,\ldots,q\}$, we associate three matrices $A_{k}$, $B_{k}$ and $C_{k}$ in $\Mat_{(2r)!!}(\R)$, as follows. We define, for all $\pi,\pi'\in \B_{r}$,
\begin{align}
\nonumber (A_{k})_{\pi,\pi'}&=-\frac{\n_{k}(w)}{2}\delta_{\pi,\pi'}- \sum_{\substack{l,m\in X_{k}(w),l< m\\ \ell(T_{lm}^{\epsilon_{l} \epsilon_{m}}(\pi))=\ell(\pi)+1}}\hspace{-5mm}\epsilon_{l}\epsilon_{m} \delta_{T_{lm}^{\epsilon_{l} \epsilon_{m}}(\pi),\pi'} + \sum_{\substack{l,m\in X_{k}(w),l<m\\ \ell(P_{lm}^{\epsilon_{l} \epsilon_{m}}(\pi))=\ell(\pi)+1}}  \hspace{-5mm}\epsilon_{l}\epsilon_{m}\delta_{P_{lm}^{\epsilon_{l} \epsilon_{m}}(\pi),\pi'},\\
\nonumber (B_{k})_{\pi,\pi'}&=\frac{\n_{k}(w)}{2}\delta_{\pi,\pi'}- \sum_{\substack{l,m\in X_{k}(w),l< m\\ \ell(T_{lm}^{\epsilon_{l} \epsilon_{m}}(\pi))=\ell(\pi)}}\hspace{-5mm}\epsilon_{l}\epsilon_{m} \delta_{T_{lm}^{\epsilon_{l} \epsilon_{m}}(\pi),\pi'} +\sum_{\substack{l,m\in X_{k}(w),l< m\\ \ell(P_{lm}^{\epsilon_{l} \epsilon_{m}}(\pi))=\ell(\pi)}} \hspace{-5mm}\epsilon_{l}\epsilon_{m}\delta_{P_{lm}^{\epsilon_{l} \epsilon_{m}}(\pi),\pi'},\\
(C_{k})_{\pi,\pi'}&=-\sum_{\substack{l,m\in X_{k}(w),l< m\\ \ell(T_{lm}^{\epsilon_{l} \epsilon_{m}}(\pi))=\ell(\pi)-1}}\hspace{-5mm}\epsilon_{l}\epsilon_{m} \delta_{T_{lm}^{\epsilon_{l} \epsilon_{m}}(\pi),\pi'} +\sum_{\substack{l,m\in X_{k}(w),l< m\\ \ell(P_{lm}^{\epsilon_{l} \epsilon_{m}}(\pi))=\ell(\pi)-1}} \hspace{-5mm}\epsilon_{l}\epsilon_{m}\delta_{P_{lm}^{\epsilon_{l} \epsilon_{m}}(\pi),\pi'},\label{ABC}
\end{align}
which satisfy commutation relations analogous to \eqref{commute AC}: for all distinct $k_{1}$ and $k_{2}$ in $\{1,\ldots,q\}$, each of the matrices $A_{k_{1}},B_{k_{1}},C_{k_{1}}$ commutes with each of the matrices $A_{k_{2}},B_{k_{2}},C_{k_{2}}$. Setting $\ptrn(w)=(\ptrn(w,\pi))_{\pi\in \B_{r}}$, we have $\prn_{(0,\ldots,0)}(w)=\1$, the vector of $\R^{(2r)!!}$ whose components are all equal to $1$, and 
\begin{equation}
\prn_{(t_1,\ldots,t_q)}(w)=\left(\prod_{k=1}^{q} e^{t_{k}\left(A_{k}+\frac{1}{N}B_{k}+\frac{1}{N^{2}}C_{k}\right)}\right){\1}.
\end{equation}
Similarly, if we define $p_{t}(w)=(p_{t}(w,\pi))_{\pi\in \B_{r}}$, we have
\begin{equation}\label{sys p o}
p_{(t_1,\ldots,t_q)}(w)=\left(\prod_{k=1}^{q} e^{t_{k}A_{k}}\right){\1}.
\end{equation}
The same computation as in the unitary case shows that $\|\ptrn(w)-p_{t}(w)\|_{\infty}$ is smaller than
\[ \sum_{l=1}^{q} \prod_{k=l+1}^{q} e^{t_{k} \|A_{k}+\frac{1}{N}B_{k}+\frac{1}{N^{2}}C_{k}\|} \left\|
e^{t_{l}\left(A_{l}+\frac{1}{N}B_{l}+\frac{1}{N^{2}}C_{l}\right)} - e^{t_{l} A_{l}} \right\| \prod_{k=1}^{l-1} e^{t_{k}\|A_{k}\|}.\]
It is easy to check that the following inequalities hold for all $N\geq 1$:
\begin{equation}\label{normesABC}
\|A_{k}\|  \leq \n_{k}(w)^{2} \; , \;\; \left\|A_{k}+ \frac{1}{N}B_{k} + \frac{1}{N^{2}}C_{k}\right\| \leq \n_{k}(w)^{2} \; , \;\; \left\|\frac{1}{N}B_{k} + \frac{1}{N^{2}}C_{k}\right\|\leq \frac{\n_{k}(w)^{2}}{N}.
\end{equation}

Now, applying (\ref{diff exp}) and thanks to (\ref{normesABC}), we find
\begin{align*}
\|\prn_{(t_1,\ldots,t_q)}(w)-p_{(t_1,\ldots,t_q)}(w)\|_{\infty} &\leq 
\sum_{l=1}^{q} \prod_{k=l+1}^{q} e^{t_{k} \n_{k}(w)^2} \frac{t_{l}\n_{l}(w)^{2}}{N} e^{t_{l}\n_{l}(w)^{2}} \prod_{k=1}^{l-1} e^{t_{k} \n_{k}(w)^{2}}\\
&= \frac{1}{N}\Aa(w)e^{\Aa(w)},
\end{align*}
which is the expected inequality.
\end{proof}

\subsection{The symplectic case} 
\begin{proof}[Proof of Proposition \ref{unif estim} in the symplectic case] Let $w=x_{i_{1}}^{\epsilon_{1}}\ldots x_{i_{r}}^{\epsilon_{r}}$ be an element of $\Mr_{q}$ of length $r$. Let $\pi\in \B_{r}$ be a pairing of $\{1,\ldots,2r\}$. 
By \eqref{ito general}, Lemma \ref{calcul TW}, Proposition \ref{calcul TW sp} and \eqref{iotacasimir}, 
$\frac{\partial}{\partial t_{k}} \pthn(w,\pi)+\n_{k}(w)\frac{(2N+1)}{4N} \pthn(w,\pi)$ is equal to
\begin{align*}
&- \sum_{\substack{l,m \in  X_{k}(w)\\ l< m}} \epsilon_{l}\epsilon_{m} (-2N)^{-\ell(\pi)-1}\E\left[(-2\Re\Tr)^{\otimes r}\left(T_{lm}^{\epsilon_{l}\epsilon_{m}}(\pi)) \circ   w_{\otimes}(S_{N,1,t_{1}},\ldots,S_{N,q,t_{q}})\right)\right] \\
&+ \sum_{\substack{l,m \in  X_{k}(w)\\ l< m}} \epsilon_{l}\epsilon_{m} (-2N)^{-\ell(\pi)-1}\E\left[(-2\Re\Tr)^{\otimes r}\left(P_{lm}^{\epsilon_{l}\epsilon_{m}}(\pi)) \circ   w_{\otimes}(S_{N,1,t_{1}},\ldots,S_{N,q,t_{q}})\right)\right].
\end{align*}
From this expression, we deduce
\begin{align}
\nonumber \frac{\n_{k}(w)}{\n_{k}(w) t_{k}} \pthn(w,\pi)=& -\frac{\n_{k}(w)(2N+1)}{4N} \pthn(w,\pi)\\
\nonumber &- \sum_{\substack{l,m \in  X_{k}(w)\\ l< m}} \epsilon_{l}\epsilon_{m}  (-2N)^{\ell(T_{lm}^{\epsilon_{l}\epsilon_{m}}(\pi))-\ell(\pi)-1}\pthn(w,T_{lm}^{\epsilon_{l}\epsilon_{m}}(\pi)) \\
&+ \sum_{\substack{l,m \in  X_{k}(w)\\ l< m}} \epsilon_{l}\epsilon_{m}  (-2N)^{\ell(P_{lm}^{\epsilon_{l}\epsilon_{m}}(\pi))-\ell(\pi)-1}\pthn(w,P_{lm}^{\epsilon_{l}\epsilon_{m}}(\pi)).
\label{edo pthn}
\end{align}

For the same reason as in the orthogonal case, the only exponents of $N$ which can appear in this sum are $0$, $-1$ and $-2$. Still as in the unitary and orthogonal cases, the integral form $\ptrn(w,\pi)=1+\int_{0}^{t} \left(\mbox{r.h.s. of \eqref{edo pthn} at }t=s\right) ds$ converges, as $N$ tends to infinity, to $p_{t}(w,\pi)=1+\int_{0}^{t} \left(\mbox{r.h.s. of \eqref{edo pthn} at }\frac{1}{N}=0 \mbox{ and }t=s\right) \; ds$. We recover, in the limit, the differential system \eqref{edo ptr}.

To the word $w$, and for each $k\in \{1,\ldots,q\}$, we associate the same matrices $A_{k}$ and $C_{k}$ in $\Mat_{(2r)!!}(\R)$ defined by \eqref{ABC}, and a matrix $B'_{k}$, which differs from $B_{k}$ only by its diagonal terms, to compensate the difference between $c_{\so(N)}$ and $c_{\sp(N)}$: we define, for all $\pi,\pi'\in \B_{r}$,
\begin{equation*}
(B'_{k})_{\pi,\pi'}=(B_{k})_{\pi,\pi'}-\frac{3\n_{k}(w)}{4}\delta_{\pi,\pi'}.
\end{equation*}
Setting $\pthn(w)=(\pthn(w,\pi))_{\pi\in \B_{r}}$, we have $\phn_{(0,\ldots,0)}(w)=\1$, the vector of $\R^{(2r)!!}$ whose components are all equal to $1$, and 
\begin{equation}
\phn_{(t_1,\ldots,t_q)}(w)=\left(\prod_{k=1}^{q} e^{t_{k}\left(A_{k}+\frac{1}{(-2N)}B'_{k}+\frac{1}{(-2N)^{2}}C_{k}\right)}\right){\1}.
\end{equation}
By \eqref{sys p o} and the same computation as in the other cases, $\|\pthn(w)-p_{t}(w)\|_{\infty}$ is smaller than
\[ \sum_{l=1}^{q} \prod_{k=l+1}^{q} e^{t_{k} \|A_{k}-\frac{1}{2N}B'_{k}+\frac{1}{4N^{2}}C_{k}\|} \left\|
e^{t_{l}\left(A_{l}-\frac{1}{2N}B'_{l}+\frac{1}{4N^{2}}C_{l}\right)} - e^{t_{l} A_{l}} \right\| \prod_{k=1}^{l-1} e^{t_{k}\|A_{k}\|}.\]
It is easy to check that the following inequalities hold for all $N\geq 1$:
\begin{equation}\label{normesABC2}
\|A_{k}\| \leq \n_{k}(w)^{2} \; , \;\;  \left\|A_{k}-\frac{1}{2N}B_{k}+ \frac{1}{4N^{2}}C_{k}\right\| \leq \n_{k}(w)^{2} \; , \;\; \left\|-\frac{1}{2N}B_{k} + \frac{1}{4N^{2}}C_{k}\right\|\leq \frac{\n_{k}(w)^{2}}{N}.
\end{equation}

Now, applying (\ref{diff exp}) and thanks to (\ref{normesABC2}), we find
\begin{align*}
\|\phn_{(t_1,\ldots,t_q)}(w)-p_{(t_1,\ldots,t_q)}(w)\|_{\infty} &\leq 
\sum_{l=1}^{q} \prod_{k=l+1}^{q} e^{t_{k} \n_{k}(w)^2} \frac{t_{l}\n_{l}(w)^{2}}{N} e^{t_{l}\n_{l}(w)^{2}} \prod_{k=1}^{l-1} e^{t_{k} \n_{k}(w)^{2}}\\
&= \frac{1}{N}\Aa(w)e^{\Aa(w)},
\end{align*}
which is the expected inequality.
\end{proof}

\part{The master field on the plane}

In the second part of this work, we apply the results of the first part to the Yang-Mills measure on the plane and, specifically, to its large $N$ limit. 

\section{The Yang-Mills measure on the plane}\label{section YM}

Let us start by recalling the definition of the Yang-Mills measure on the plane. For a more detailed presentation, we refer the reader to \cite{LevyAMS}, although strictly speaking the case of the plane was not treated there.

We consider the plane $\R^2$ endowed with the usual Euclidean distance and the Lebesgue measure.

Let us choose a connected compact Lie group $G$ which will stay fixed throughout this section. The examples which we have in mind are of course the special orthogonal, unitary, special unitary and symplectic groups which we studied in the first part of this work, but for the purposes of the definition of the Yang-Mills measure, we do not need to specify $G$. We denote the Lie algebra of $G$ by $\g$ and we endow it with a scalar product invariant by the adjoint action of $G$, which we denote by $\langle \cdot , \cdot \rangle$. For example, one can think of $G$ being $\U(N)$ for some $N\geq 1$, so that $\g=\u(N)$, and the scalar product on $\g$ being given by $\langle X,Y \rangle=N\Tr(X^{*}Y)$. 

The Yang-Mills measure, or rather, the Yang-Mills process, is a collection of random variables with values in the group $G$, one for each loop with finite length on $\R^2$. In order to construct this collection, one proceeds by discrete approximation, considering at first only loops which are traced in a fixed  graph. We start by recalling the main aspects of this discrete theory.

\subsection{Discrete Yang-Mills field}\label{dymf}

Let us start by giving precise definitions of the sets of paths which we will consider. A {\em parametrised path} on $\R^2$ is a Lipschitz continuous mapping $c:[0,1] \to \R^2$ which is either constant or such that its speed is bounded below by a positive constant. A {\em path} is a parametrised path taken up to bi-Lipschitz increasing reparametrisation. The set of paths on $\R^2$ is denoted by $\Path(\R^2)$. 

The endpoints of a path are denoted respectively by $\src{c}=c(0)$ and $\tgt{c}=c(1)$. Two paths $c_1$ and $c_2$ such that $\tgt{c_1}=\src{c_2}$ can be concatenated to form a new path denoted by $c_1c_2$. This partially defined operation on $\Path(\R^2)$ is associative. For each path $c$ we define the path $c^{-1}$ which is the class of $t\mapsto c(1-t)$, the path $c$ traced backwards.

A path whose endpoints coincide is called a {\em loop}. The set of loops on $\R^2$ is denoted by $\Loop(\R^2)$. A loop whose restriction to $[0,1)$ is injective is called a {\em simple loop}. The set of loops starting, and hence finishing, at a point $m\in \R^2$ is denoted by $\Loop_m(\R^2)$. For all $m\in \R^2$, the set $\Loop_m(\R^2)$ endowed with the operation of concatenation is a monoid. We shall explain later (see Section \ref{group of loops}) that there is a natural, though not easy to define, equivalence relation on this monoid such that the quotient is actually a group.

Let us turn to graphs. An {\em edge} is a path which is either injective or a simple loop. Note that an edge traced backwards is still an edge, though distinct from the original one. A {\em graph} is a triple $\G=(\V,\E,\F)$ such that the following properties are satisfied.\\
\indent 1. The set $\E$ is a finite subset of $\Path(\R^2)$ consisting of edges. For all edge $e\in\E$, the edge $e^{-1}$ belongs to $\E$. Any two edges of $\E$ which are distinct and not each other's inverse meet, if at all, only at some of their endpoints.\\
\indent 2. The set $\V$ is the set of endpoints of the elements of $\E$.\\
\indent 3. The set $\F$ is the set of connected components of the complement in $\R^2$ of the {\em skeleton} of $\G$, which is the subset $\Sk(\G)=\bigcup_{e\in \E} e([0,1])$.\\
\indent 4. Each element of $\F$ is either a bounded subset of $\R^2$ homeomorphic to the open unit disk of $\R^2$, or an unbounded subset of $\R^2$ homeomorphic to the complement of the origin in $\R^2$.

The elements of $\V,\E,\F$ are called respectively the vertices, edges, and faces of $\G$. The fourth condition is equivalent to the fact that the skeleton of the graph is connected (see \cite[Prop. 1.3.10]{LevyAMS}). All faces of a graph are bounded but one, which we naturally call the unbounded face and which we usually denote by $F_\infty$. We shall use the notation $\F^{b}=\F\setminus\{F_\infty\}$ for the set of bounded faces. For each bounded face $F$, we denote by $|F|$ the area of $F$.

Let $\G$ be a graph. The set of paths which can be formed by concatenating edges of $\G$ is denoted by $\Path(\G)$. The subset of $\Path(\G)$ consisting of loops is denoted by $\Loop(\G)$. Each bounded face of $\G$ is positively bounded by a loop which we call its boundary and which is ill-defined because it has no preferred base point. Nevertheless, we denote by $\partial F$ the boundary of the face $F$, keeping in mind that this is not properly speaking a loop, but rather a collection of loops which differ only by their starting point.

The discrete Yang-Mills measure associated with the graph $\G$ and the group $G$ is a probability measure on a space which can be described in several equivalent and equally useful ways. Let $P$ be a subset of $\Path(\R^2)$. A function $h:P\to G$ is said to be {\em multiplicative} if for any $c\in P$ such that $c^{-1} \in P$ one has $h(c^{-1})=h(c)^{-1}$, and for any two paths $c_1$ and $c_2$ in $P$ such that $\tgt{c_1}=\src{c_2}$ and $c_1c_2 \in P$ one has 
\begin{equation}\label{def mult h}
h(c_1c_2)=h(c_2)h(c_1).
\end{equation}
We denote the set of multiplicative functions from $P$ to $G$ by $\M(P,G)$. The discrete Yang-Mills measure on $\G$ shall be defined as a probability measure on $\M(\Path(\G),G)$. 

Since any path traced in $\G$ is a concatenation of edges, a multiplicative function on $\Path(\G)$ is completely determined by its restriction to the set of edges. Actually, one needs only to know its value on one element of each pair $\{e,e^{-1}\}$, where $e$ spans the set of edges. We call {\em orientation} of the edges of $\G$ a subset $\E^{+}$ of $\E$ which contains exactly one element in each pair $\{e,e^{-1}\}$, $e\in \E$. An orientation of the edges of $\G$ being chosen, we have the following identifications
\begin{equation}\label{def config}
\M(\Path(\G),G) \simeq \M(\E,G) \simeq \M(\E^+,G) \simeq G^{\E^+}.
\end{equation}
The last identification expresses the fact that any function from $\E^+$ to $G$ is multiplicative, since 
the concatenation of two edges is never an edge. We call any of these spaces the {\em configuration space} of the discrete theory and denote it by $\Conf^{\G}_{G}$, or simply $\Conf^{\G}$ if there is no ambiguity on the group $G$. The reader who feels uncomfortable with such a row of identifications can take $\Conf^{\G}=G^{\E^{+}}$ as an efficient definition.

As announced, the discrete Yang-Mills measure is a Borel probability measure on $\Conf^{\G}$, which is naturally a compact topological space. The normalised Haar measure on the compact group $G$ determines, through the identifications above, a reference probability measure on $\Conf^{\G}$ which we denote by $dh=\bigotimes_{e\in \E^+} dg_{e^+}$. The Yang-Mills measure has a density with respect to this uniform measure and in order to define it, we must introduce the heat kernel on $G$, which is a one-parameter family of smooth positive functions on $G$, namely the fundamental solution of the heat equation. If $G$ is one of the groups which we studied in the first part, then this function is also the density of the distribution of the Brownian motion on the group, seen as a function of time and an element of the group.

The Lie algebra $\g$ of $G$ is the space of left-invariant first-order differential operators on $G$: to each element $X\in \g$, one associates the differential operator $\L_X$ defined by the equality, valid for all differentiable function $f:G\to \R$ and all $g\in G$, $(\L_X f)(g)=\frac{d}{dt}_{|t=0} f(ge^{tX})$. 

Let $d$ denote the dimension of $G$. Given an orthonormal basis $(X_1,\ldots,X_{d})$ of $\g$ with respect to the invariant scalar product which we have chosen on $\g$, we can form the second-order differential operator $\sum_{k=1}^{d} \L_{X_k}^2$. This operator does not depend on the choice of the orthonormal basis, it is called the Laplace operator on $G$ and we denote it by $\Delta$. 

The heat kernel on $G$ is the unique positive function $Q:\R^*_+ \times G \to \R^*_+$ such that $(\partial_t-\frac{1}{2}\Delta)Q=0$ and the measure $Q(t,g) \, dg$ converges weakly, as $t$ tends to $0$, to the Dirac measure at the unit element  of $G$. The measure $Q(t,g)\, dg$ is simply the distribution of the Brownian motion on $G$ at time $t$.

We will denote the number $Q(t,g)$ by $Q_t(g)$, thus seeing $Q$ as a one-parameter family of functions on $G$. A crucial property of these functions is that they are invariant by conjugation: they satisfy, for all $t>0$ and all $x,y\in G$, the equality $Q_t(yxy^{-1})=Q_t(x)$. This is a consequence of the fact that the Laplace operator belongs to the centre of the algebra of left-invariant differential operators on $G$. If $G$ is one of the groups which we studied in the first part of this paper, this is also a consequence of Lemma \ref{inverses}.

In order to define the Yang-Mills measure $\YM^{\G}_{G}$, or simply $\YM^{\G}$, on $\Conf^{\G}$, we only need to make a last observation: if $F$ is a face of a graph $\G$, and if $h$ is a multiplicative function on $\Path(\G)$, then for all $t>0$ the number $Q_t(h(\partial F))$ does not depend on the particular choice of the origin of the loop $\partial F$. Indeed, changing the origin of $\partial F$ alters $h(\partial F)$ by conjugating it in $G$, and this does not change the value of $Q_{t}$. The following expression is thus well defined: 
\begin{equation}\label{def YM}
\YM^\G(dh)=\prod_{F\in \F^{b}} Q_{|F|}(h(\partial F)) \; dh.
\end{equation}
This is indeed a probability measure, as one verifies by successively integrating over all edges using the convolution property of the heat kernel, according to which $\int_{G} Q_t(xy^{-1})Q_s(y)\;dy=Q_{t+s}(x)$, and finally the fact that $\int_{G} Q_t(x)\; dx=1$. Note that the product in this definition is over the set of bounded faces of $\G$. In fact, $Q_t$ converges on $G$ uniformly and exponentially fast to $1$ as $t$ tends to infinity, and we could just as well include the unbounded face in the product, provided we make the very natural convention $Q_{\infty}=1$.

With this definition, the Borel probability space $(\Conf^{\G},\YM^\G)$ is essentially the canonical space of the stochastic process $(H_{c})_{c\in\Path(\G)}$, which is defined simply by $H_{c}(h)=h(c)$ for all $c\in \Path(\G)$. The fact that we are working with multiplicative functions implies that the stochastic process $H$ is trajectorially multiplicative. This means that if $c_{1}$ and $c_{2}$ can be concatenated, then the functions $H_{c_{1}}H_{c_{2}}$ and $H_{c_{2}c_{1}}:\Conf^{\G}\to G$ are the same.

To conclude this section, let us observe that \eqref{def YM} would still make good sense and define a probability measure on $\Conf^{\G}$ if for each bounded face $F$ we replaced the area $|F|$ by any positive real number. We shall exploit this possibility in Section \ref{variation area}. In the mean time, in Sections \ref{section YM} and \ref{mf}, we shall use no definition of $\YM^{\G}$ other than \eqref{def YM}.

\subsection{Continuous Yang-Mills field}\label{CYMF}

The single most important property of the discrete Yang-Mills field is that it is consistent with respect to the subdivision, or refinement, of the underlying graph. The precise meaning of this assertion is the following. If $\G_{1}$ and $\G_{2}$ are two graphs, we say that $\G_{2}$ is finer that $\G_{1}$ if $\Path(\G_{1})$ is a subset of $\Path(\G_{2})$. In this case, there is a natural mapping of restriction $\M(\Path(\G_{2}),G)\to \M(\Path(\G_{1}),G)$ and the invariance of the Yang-Mills measure under refinement of the graph is the fact that the image of the measure $\YM^{\G_{2}}$ under this mapping is $\YM^{\G_{1}}$. 

The practical consequence of this invariance is that if a certain set $P$ of paths belongs to $\Path(\G_1)$ and $\Path(\G_2)$ for two graphs $\G_1$ and $\G_2$ such that one is finer than the other, then the distribution of the family of random variables $(H_{c})_{c\in P}$ is the same when it is computed under $\YM^{\G_1}$ or under $\YM^{\G_2}$. The same conclusion holds if there exists a graph $\G_{3}$ which finer than both $\G_{1}$ and $\G_{2}$.

Pushing this line of reasoning one step further, we would expect the invariance under subdivision of the discrete Yang-Mills measure to allow us to take the inverse limit of the probability spaces $(\Conf^{\G},\YM^{\G})$ along the partial order defined by the relation of fineness. Unfortunately, this partial order is not good enough for this, in that it is not directed: there does not always exist a graph which is finer than two given graphs. We are thus forced to consider in a first step a subset of the set of all graphs, for instance piecewise affine graphs, on which the partial order is directed, and in a second step to use a procedure of approximation to include all graphs in our definition. It is thus in particular necessary to consider an appropriate topology on the set of paths, which we now describe.

Let $c_1$ and $c_2$ be two paths. We denote by $\ell(c_1)$ and $\ell(c_2)$ respectively the lengths of $c_1$ and $c_2$. The uniform distance between $c_1$ and $c_2$ is $d_\infty(c_1,c_2)=\inf_{\phi_1,\phi_2} \sup\{|c_1(\phi_1(t))-c_2(\phi_2(t))| : t\in [0,1]\}$, where the infimum is taken over all pairs of increasing bi-Lipschitz homeomorphisms of $[0,1]$. We define two distances between $c_1$ and $c_2$ by setting 
\[d_1(c_1,c_2)=|c_{1}(0)-c_{2}(0)|+\int_0^1 |\dot c_1(t)-\dot c_2(t)| \; dt\] and 
\[d_\ell(c_1,c_2)=d_\infty(c_1,c_2) + |\ell(c_1)-\ell(c_2)|.\] 
The first distance is the distance in $1$-variation and the second we call the length distance. Although the first makes $\Path(\R^{2})$ a complete metric space and the second does not, it can be shown that these distances determine the same topology (see \cite[Prop. 1.2.14]{LevySMF}). We thus simply speak of convergence of paths, without mentioning a distance. We shall also frequently use the notion of convergence with fixed endpoints of a sequence of paths, where all the paths of the sequence are required to have the same endpoints as the limiting path. 

The result of the construction which we have just sketched is summarised in the following theorem, which defines the Yang-Mills measure. It is a consequence of \cite[Thm. 4.3.1]{LevySMF} where the L\'evy process on $G$ must be chosen to be the Brownian motion. It is proved on a compact Riemannian surface rather than on the plane $\R^{2}$, but the proof is valid without any modification.

\begin{theorem}\label{continuous YM} There exists on the space $\M(\Path(\R^{2}),G)$ endowed with the cylinder $\sigma$-algebra a unique probability measure $\YM_{G}$ such that the following two properties are satisfied.\\
\indent 1. For all graph $\G=(\V,\E,\F)$, the family of random variables $(H_c)_{c\in \Path(\G)}$ has the same distribution under $\YM_{G}$ as under $\YM^{\G}_{G}$. \\
\indent 2. For all path $c\in \Path(\R^2)$ and all sequence $(c_n)_{n\geq 1}$ of paths converging with fixed endpoints to $c$, the sequence $(H_{c_n})_{n\geq 1}$ converges in probability to $H_{c}$.
\end{theorem}

\subsection{The group of loops in a graph}\label{grouploops}

In Section \ref{mf}, we shall prove the central result of the second part of this work, which asserts the existence of a limit as $N$ tends to infinity for the Yang-Mills process on the plane when the group $G$ is one of the groups $\U(N,\K)$ which we considered in the first part of this work. In a first step, we shall study the large $N$ limit of the discrete Yang-Mills measure associated with a graph on the plane, using a very explicit description of this measure in terms of a collection of independent random variables with values in $G$, some uniform and some distributed according to the heat kernel measure. In preparation for this, we need to understand very concretely the structure of the set of paths and loops on the graph $\G$, and this is what we explain now. What we have to say in the present section is still valid for any compact connected Lie group $G$. More details about what we explain can be found in \cite[Sec. 2.4]{LevySMF}.

Let $\G$ be a graph on $\R^2$. There is a very natural  equivalence relation on the set $\Path(\G)$ for which two paths are equivalent if it is possible to transform one into the other by a finite sequence of insertions or erasures of sub-paths of the form $ee^{-1}$, where $e$ is an edge. For example, the paths $e_0 e_1 e_2 e_3 e_3^{-1} e_2^{-1}$ and $e_{0}e_{2}^{-1} e_{2}  e_1 e_3 e_4^{-1} e_4 e_3^{-1}$ are equivalent. One proves that in each equivalence class for this relation there is a unique path of shortest combinatorial length, that is, a unique path which traverses a minimal number of edges. It is characterised by the fact that it is {\em reduced}, which means that it does not contain any sub-path of the form $ee^{-1}$. In the example above, none of the two paths are reduced, and the unique reduced path to which they are equivalent is $e_0e_1$. The equivalence relation thus defined preserves the endpoints and is compatible with concatenation. For all vertex $v\in \V$, the quotient of the set of loops $\Loop_v(\G)$ based at $v$ by this equivalence relation becomes a group for the operation of concatenation. The unit element is the class of the constant loop at $v$. Instead of a quotient of $\Loop_{v}(\G)$, one can think of this group as the subset $\RL_v(\G)$ of $\Loop_{v}(\G)$ which consists of reduced loops based at $v$, endowed with the operation of concatenation-reduction. 

If $v$ and $w$ are two vertices of $\G$, and if $c$ is a path in $\G$ which joins $v$ to $w$, then the mapping $l\mapsto clc^{-1}$ induces an isomorphism of groups between $\RL_w(\G)$ and $\RL_v(\G)$. It is thus enough to understand the structure of $\Loop_{v}(\G)$ for some vertex $v$. The first crucial fact is that for all $v\in \V$, the group $\RL_v(\G)$ is a free group of rank equal to the number of bounded faces of $\G$. The second crucial fact, which is very useful for the purposes of the discrete Yang-Mills theory, is that this free group possesses bases, indeed many bases, which are naturally indexed by the set $\F^{b}$ of bounded faces of $\G$. We will spend the next paragraphs describing a particular way of associating such a basis to each choice of a spanning tree of $\G$, or equivalently to each choice of a spanning tree of the dual graph of $\G$.

The {\em dual graph} of $\G$ is a graph which is not exactly of the same nature as $\G$ insofar its edges are not concretely embedded in the plane. We define it as the quadruple $\wG=(\wV,\wE,s,t)$, where $\wV=\F$ is the set of faces of $\G$, $\wE$ is the set of triples $(F_0,e,F_1) \in \F\times \E\times \F$ such that the edge $e$ bounds $F_0$ positively and $F_1$ negatively, and $s,t:\wE\to \wV$ are the two mappings defined by $s(F_0,e,F_1)=F_0$ and $t(F_0,e,F_1)=F_1$. We call respectively the elements of $\wV$ and $\wE$ {\em dual vertices} and {\em dual edges}. Each edge $e$ appears in a unique dual edge $(F_0,e,F_1)$ which we denote by $\hat e$. We define the inverse of the dual edge $\hat e=(F_0,e,F_1)$ by $\hat e^{-1}=(F_{1},e^{-1},F_{0})$, so that $\hat e^{-1}=\widehat{e^{-1}}$. Observe that the equality $F_{0}=F_{1}$ is not excluded in these definitions. Finally, the unbounded face of $\G$ determines a particular dual vertex which we denote by $\hat F_\infty$ and call the dual root. 

\begin{figure}[h!]
\begin{center}
\scalebox{1}{\includegraphics{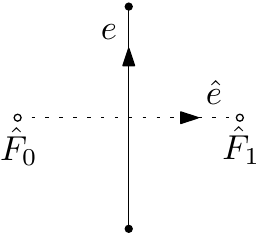}}
\caption{\label{gdual} \small An edge and the corresponding dual edge.}
\end{center}
\end{figure}

Recall that a spanning tree of $\G$ (resp. $\widehat \G$) is a subset $\T\subset \E$ (resp. $\wT \subset \wE$) which is the set of edges (resp. dual edges) of a connected sub-graph of $\G$ (resp. $\widehat\G$) without cycles and which contains every vertex (resp. dual vertex). We take as a part of the definition that a spanning tree contains its edges with both orientations. Given a spanning tree $\T$ of $\G$ and two vertices $v_{1},v_{2}\in \V$, we define $[v_{1},v_{2}]_{T}$ as the unique reduced path in $\G$ which goes from $v_{1}$ to $v_{2}$ using only edges of $\T$. We define similarly the path $[\hat F_{1},\hat F_{2}]_{\wT}$ in $\widehat \G$.

Since we are working on the plane, rather than on a multiply connected surface, spanning trees of $\G$ are in bijection with spanning trees of $\widehat \G$, through the dual and reciprocal formulas  $\T\mapsto \wT=\{\hat e \in \wE: e \notin \T\}$ and $\wT \mapsto \T=\{e \in \E: \hat e \notin \wT\}$.

Let us choose a spanning tree $\wT$ of $\widehat \G$. Let $\T$ be the corresponding spanning tree of $\G$. Let $v_{0}$ be a vertex of $\G$.
We are going to use $\wT$ to produce a basis of the free group $\Loop_{v_{0}}(\G)$ indexed by $\F^{b}$. Let $F$ be a bounded face of $\G$. Let $\hat e$ be the dual edge of $\wT$ issued from $\hat F$ in the direction of the dual root $\hat F_{\infty}$, that is, the first edge traversed by the path $[\hat F,\hat F_{\infty}]_{\wT}$. Let $\partial_{e} F$ be the loop which goes once around the boundary of $F$, starting with the edge $e$. We define the loop $\lambda_{F}\in \RL_{v_{0}}(\G)$ by
\[\lambda_{F}=[v_{0},\underline{e}]_{T}\partial_{e} F [\underline{e},v_{0}]_{T},\]
being understood that $\lambda_{F}$ is the reduced loop equivalent to the loop on the right-hand side. 

Let us emphasise that the family  $\{\lambda_{F}: F\in \F^{b}\}$ depends on the choice of the spanning tree $\T$ of $\G$ and of the vertex $v_{0}$, and that these choices can be made independently. The first result is the following.

\begin{proposition}\label{basis RL} The family of loops $\{\lambda_{F}: F\in \F^{b}\}$ is a basis of the group $\RL_{v_{0}}(\G)$.
\end{proposition}

This result is proved in \cite[Sec. 2.4]{LevySMF} in the more general situation of orientable or non-orientable compact surfaces with or without boundary. The present case of the plane, which corresponds to the case of the disk in \cite{LevyAMS}, is fortunately much simpler. For the convenience of the reader, and because a familiarity with the ideas used in this proof will be helpful for the understanding of Section \ref{maa}, we recall its main arguments.

\begin{proof} The choice of the spanning tree $\T$ and the vertex $v_{0}$ determines a subset of $\RL_{v_{0}}(\G)$ which is obviously a basis, but not the one we are interested in. The proof consists in proving that our basis is essentially deduced from this obvious basis by a triangular array of multiplications.

For each edge $e\in \E\setminus \T$, define $\beta_{e}=[v_{0},\underline{e}]_{T} e [\overline{e},v_{0}]_{T}$. It is not difficult to check that for all orientation $\E^{+}$ of $\G$, $\RL_{v_{0}}(\G)$ is freely generated by $\{\beta_{e} : e\in \E^{+}\setminus\T\}$. It is equivalent and more convenient to say that $\RL_{v_{0}}(\G)$ is generated by the family $\{\beta_{e} : e\in \E\setminus\T\}$, which is subject to the relations $\beta_{e}\beta_{e^{-1}}=1$.

The triangular array which allows one to pass from the family $\{\beta_{e}: e\in \E\setminus \T\}$ to the family $\{\lambda_{F} : F\in \F^{b}\}$ is dictated by the geometry of the  spanning tree $\wT$. This geometry can be encoded as follows. 

The orientation of the plane determines a cyclic order on the set of dual edges issued from each dual vertex in $\widehat \G$, hence in $\wT$. Our choice of $v_{0}$ breaks the cyclic symmetry of the dual edges issued from $\hat F_{\infty}$ and allows us to order them totally. Moreover, for each dual vertex $\hat F$ which is not the dual root, there is one distinguished dual vertex adjacent to $\hat F$, namely the dual vertex visited by the path $[\hat F,\hat F_{\infty}]_{\wT}$ immediately after leaving $\hat F$. We call this dual vertex the predecessor of $\hat F$ and denote it by $\pi(\hat F)$. Having chosen $\pi(\hat F)$ determines a total order on the set of the other dual vertices adjacent to $\hat F$.

These orders determine a way of labelling each dual vertex by a word of integers. We start by labelling the dual root $\hat F_{\infty}$ by the empty word $\varnothing$. The dual vertices which are adjacent in $\wT$ to the dual root are labelled by the one-letter words $1,2,\ldots,k(\varnothing)$ in the total order which we have just considered. Then, a dual vertex $\hat F$ being labelled by the word $u$, we label its neighbours in $\wT$ other than $\pi(\hat F)$ in their total order by the words $u1,\ldots,uk(u)$, where by $ul$ we mean the word $u$ to which the letter $l$ has been added at the end. The dual vertex $\pi(\hat F)$ has already been labelled, by the word $\pi(u)$ obtained from $u$ by removing its last letter. 

We will now designate the dual vertices by their labels. For example, each pair $(u,v)$ of dual vertices adjacent in $\wT$ determines a dual edge $\hat e$, hence an edge $e$, and we use the notation $\beta_{u,v}=\beta_{e}$. The main triangular relation is now the following: for all dual vertex $\hat F$ other than the dual root, labelled by the word $u\neq \varnothing$, we have
\begin{equation}\label{beta to lambda}
\lambda_{u}=\lambda_{F}=\beta_{u,\pi(u)}\beta_{u1,u}^{-1}\ldots \beta_{uk(u),u}^{-1}.
\end{equation}

\begin{figure}[h!]
\begin{center}
\scalebox{0.8}{\includegraphics{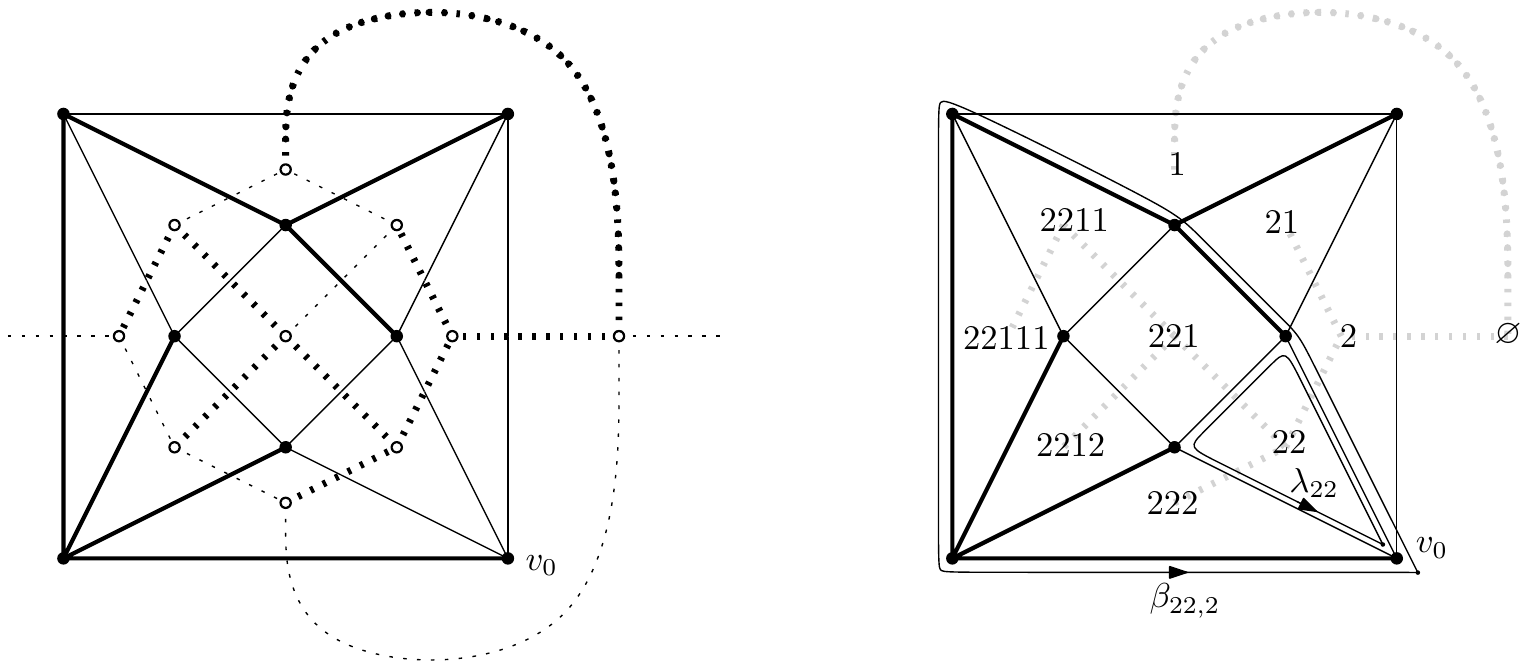}}
\caption{\label{ex base}\small In this example, one checks that $\lambda_{22}=\beta_{22,2}\beta_{221,22}^{-1}\beta_{222,22}^{-1}$ and $\beta_{22,2}=\lambda_{22}\lambda_{222}\lambda_{221}\lambda_{2212}\lambda_{2211}\lambda_{22111}$.}
\end{center}
\end{figure}

The reason why \eqref{beta to lambda} is invertible is that each $\lambda_{u}$ is the product of $\beta_{u,\pi(u)}$ and a word in the loops $\beta_{u',\pi(u')}$ where $u'$ stays in the sub-tree of $\wT$ above $u$, that is, the set of dual vertices $u'\neq u$ such that the path $[u',\varnothing]_{\wT}$ visits $u$. In order to invert \eqref{beta to lambda}, one must then start by the loops $\beta_{u,\pi(u)}$ where $u$ is a leaf of $\wT$, that is, a dual vertex which is not the dual root and which is of degree $1$ in $\wT$, and proceed inwards, towards the dual root. One has in fact for all $u\neq \varnothing$ the explicit relation
\begin{equation}\label{lambda to beta}
\beta_{u,\pi(u)}=\lambda_{u}\lambda_{u_{1}}\ldots \lambda_{u_{p}},
\end{equation}
where $(u_{1},\ldots,u_{p})$ is the list of the dual vertices located in the sub-tree above $u$, ordered in the lexicographic order corresponding to the reversed natural order on $\N$.

The explicit relations \eqref{beta to lambda} and \eqref{lambda to beta} imply that $\{\lambda_{u} : u\neq \varnothing\}$ is a basis of $\RL_{v_{0}}(\G)$.
\end{proof}

We shall call {\em lassos} the loops of the form $\lambda_{F}$, and {\em lasso basis} associated to $\T$ or to $\wT$ the basis $\{\lambda_{F} : F\in \F^{b}\}$, which we shall denote by $\Lambda_{\T}$ or $\Lambda_{\wT}$.

From Proposition \ref{basis RL} we can deduce a normalised way of writing not only loops, but paths in $\G$. To formulate this, observe that the quotient of $\Path(\G)$ by the relation of equivalence endowed with the partial operation of concatenation is a groupoid. This means that, although concatenation is only partially defined, it is associative and each element has an inverse, in the sense that for each path $c$ the path $cc^{-1}$ is equivalent to a constant path. We denote this groupoid by $\RP(\G)$. Let us choose an orientation $\E^{+}$ of the edges of $\G$ and set $\T^+=\T\cap \E^+$. 

\begin{corollary}\label{ecrire chemins} The groupoid $\RP(\G)$ is freely generated by the elements $\{\lambda_{F}: F\in \F^{b}\}$ and $\{e : e\in \T^{+}\}$: each path on $\G$ is equivalent to a unique reduced word in these paths.

In fact, for each path $c$ in $\G$, there exists a unique sequence of faces $F_1,\ldots,F_n\in \F^{b}$ and a unique sequence of signs $\epsilon_1,\ldots\epsilon_n \in \{-1,1\}$, of the same length, possibly empty and such that for all $k\in \{1,\ldots,n-1\}$ one has $F_k\neq F_{k+1}$ or $\epsilon_k= \epsilon_{k+1}$, such that $c$ is equivalent to the path 
\[[\underline{c},v_0]_{\T} \lambda_{F_1}^{\epsilon_1} \ldots \lambda_{F_n}^{\epsilon_n} [v_0,\overline{c}]_{\T}.\] 
\end{corollary}

We are thus able to write any path in $\G$ as a word in a certain alphabet of elementary paths. The number of these elementary paths is the number of edges of a spanning tree plus the number of bounded faces. Let us denote by $\v,\e,\f$ the numbers of vertices, unoriented edges and bounded faces of $\G$. Here, by the number of unoriented edges, we mean the half of the number of elements of $\E$. There are $\v-1$ unoriented edges in $\T$, so that the number of elementary paths is $\v+\f-1$. On the other hand, Euler's relation for $\G$ reads $\v-\e+\f=1$, hence $\v+\f-1$ is the number of edges of $\G$. Let us choose an orientation $\E^{+}$ of the edges of $\G$ and set $\T^+=\T\cap \E^+$. We can thus add a new identification 
\begin{equation}\label{config lambda}
\Conf^{\G}_{G}=\M(\Path(\G),G) \simeq G^{\E^+}\simeq G^{\F^{b}}\times G^{\T^+}
\end{equation}
to the row \eqref{def config}, the last isomorphism being given by $h\mapsto ((h(\lambda_{F}) : F\in \F^{b}),(h(e) : e\in \T^{+}))$. This mapping encodes a lot of the geometry of the graph (see Figure \ref{ex base} for an example).

\begin{figure}[h!]
\begin{center}
\scalebox{0.8}{\includegraphics{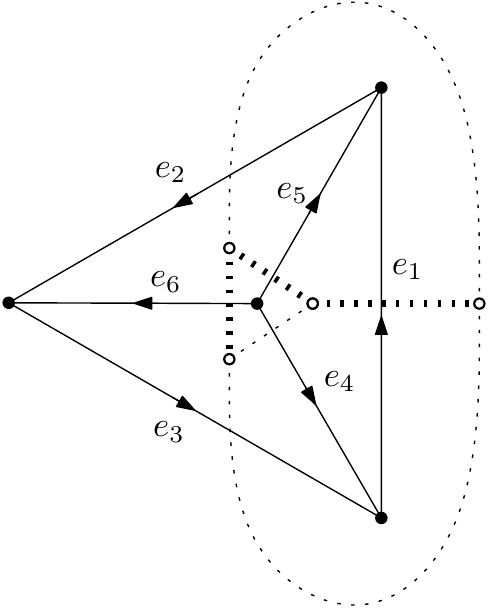}}
\caption{\label{ex base 2}\small In this example, the last identification of \eqref{config lambda} is the following:
\[(g_{1},g_{2},g_{3},g_{4},g_{5},g_{6}) \mapsto (g_{4}g_{5}^{-1}g_{1},g_{4}g_{6}^{-1}g_{2}g_{5}g_{4}^{-1},g_{3}g_{6}g_{4}^{-1},g_{2},g_{3},g_{4}).\]}
\end{center}
\end{figure}

The interest of the last description of the configuration space of the discrete Yang-Mills theory is that it allows a very pleasant description of the probability measure $\YM^\G_{G}$. The following result is a consequence of \cite[Cor. 2.4.9]{LevySMF}.

\begin{proposition}\label{YM basis} Through the identification $\Conf^{\G}_{G} \simeq  G^{\F^{b}}\times G^{\T^+}$, the discrete Yang-Mills measure $\YM^\G_{G}$ corresponds to the measure
\[ \bigotimes_{F\in \F^{b}} Q_{|F|}(g)\; dg \otimes \bigotimes_{e\in \T^+} dg.\]
In other words, under  $\YM^\G_{G}$, the random variables $\{H_{\lambda_F}: F\in \F^{b}\}\cup\{H_{e} : e\in \T^+\}$ are independent, each $H_{\lambda_F}$ distributed according to the heat kernel measure at time $|F|$ on $G$ and each $H_{e}$ distributed according to the Haar measure on $G$.
\end{proposition}

With this description in hand, we can safely turn to the study of the large $N$ limit of the Yang-Mills field.

Before we do so however, and because this will be useful at a later stage of this work, let us review the definition of the gauge group and its action on the configuration space, and give an invariant version of Proposition \ref{YM basis}.

We are given a graph $\G=(\V,\E,\F)$ and a compact connected Lie group $G$. The gauge group is by definition the group $G^\V$ equipped with pointwise multiplication. It acts on $\M(\Path(\G),G)$ according to the following rule: given $j=(j(v))_{v\in\V} \in G^{\V}$ and a multiplicative function $h$, we have for all path $c$
\[(j\cdot h)(c)=j(\overline{c})^{-1}h(c)j(\underline{c}).\]
One checks easily that this is a right action in the sense that if $j$ and $k$ belong to the gauge group and $h$ is a multiplicative function, then $(jk)\cdot h=k\cdot (j\cdot h)$.

The gauge group acts naturally on the space of smooth functions on the configuration space $\Conf^{\G}_{G}$: if $f$ is such a smooth function, $j$ a gauge transformation and $h$ a configuration, then we have, by definition,
\[(j\cdot f)(h)=f(j^{-1}\cdot h),\]
so that again, if $k$ is another gauge transformation, $(jk)\cdot f=k\cdot (j\cdot f)$. We say that a function on the configuration space $G^{\E^+}$ is {\em invariant} if it is invariant under the action of the gauge group.

Let now $\T$ be a spanning tree of $\G$. Let $v_{0}$ be a vertex. For each configuration $h\in \Conf^{\G}_{G}$, let us consider the element $j_{h,\T}$ of the gauge group defined by
\begin{equation}\label{def jhT}
j_{h,\T}(v)=h([v_{0},v]_{\T}),
\end{equation}
where $[v_{0},v]_{\T}$ denotes the unique reduced path in $\T$ from $v_{0}$ to $v$. Then $j_{h,\T}\cdot h$ is identically equal to $1$, the unit element of $G$, on each edge of $\T$. Moreover, for each loop based at $v_{0}$, one has $(j_{h,\T}\cdot h)(l)=h(l)$. The next result follows immediately from this observation and Proposition \ref{YM basis}.

\begin{proposition}\label{YM basis invariant} Through the identification $\Conf^{\G}_{G} \simeq  G^{\F^{b}}\times G^{\T^+}$, the image of the discrete Yang-Mills measure $\YM^\G_{G}$ by the mapping $h\mapsto j_{h,\T}\cdot h$ 
corresponds to the measure
\[ \bigotimes_{F\in \F^{b}} Q_{|F|}(g)\, dg \otimes \bigotimes_{e\in \T^+} \delta_{1},\]
where $\delta_{1}$ is the Dirac mass at $1$. In other words, for all smooth function $f:\Conf^{\G}_{G}\to \C$ seen as a smooth function on $G^{\F^{b}}\times G^{\T^+},$ the following equality holds:
\[\int_{\Conf^{\G}_{G}} f(j_{h,\T}\cdot h) \; \YM^{\G}_{G}(dh)= \int_{G^{\F^{b}}} f(\{g_{F}: F\in \F^{b}\},\{1 : e\in \T^{+}\})\; \prod_{F\in \F^{b}}Q_{|F|}(g_{F})\, dg_{F}.\]
\end{proposition}

\section{The master field on the plane}\label{mf}

In this section, we turn to the proof of the main result of the second part of this work, indeed the main motivation for this whole work. Our goal is to describe the large $N$ limit of the Yang-Mills field with structure group $\U(N,\K)$ for $\K\in \{\R,\C,\H\}$.

The study of this limit follows the construction of the field itself. We start by applying the results of Section \ref{section one BM} to the discrete theory, on a graph ; then take an easy step and assemble the results for a large family of graphs in order to be able to treat all piecewise affine loops at once ; and finally, apply the results of Section \ref{section:speed} in order to tackle the approximation procedure involved in the construction of the Yang-Mills field, and succeed in obtaining the limit for all rectifiable loops.

\subsection{Large $N$ limit of the Yang-Mills field on a graph}\label{large N graph}
 Let us choose one of the three division algebras $\R,\C,\H$ and denote it by $\K$, as we did in the first part. For each $N\geq 1$, let us consider the Yang-Mills field on $\R^{2}$ with structure group $\U(N,\K)$, associated with the scalar product given by \eqref{normalization}. We denote by $(H_{N,c}^{\K})_{c\in\Path(\R^{2})}$ the corresponding process.

For each $N\geq 1$, the random variables $(H_{N,c}^{\K})_{c\in \Path(\G)}$ form a family of non-commutative random variables in the non-commutative probability space $(L^{\infty}(\Conf^{\G}_{\U(N,\K)},\YM^{\G}_{\U(N,\K)})\otimes \Mat_{N}(\K),\E\otimes \tr)$, where $\tr$ must be replaced by $\Re\tr$ when $\K=\H$. When $\K=\R$ or $\K=\H$, this is a real non-commutative probability space, in the sense described just before the statement of Theorem \ref{limite libre}. Let us describe the convergence result in the discrete setting. 

Let $\G$ be a graph. Let $\E^{+}$ be an orientation of $\G$. Let $\T$ be a spanning tree of $\G$. Let $v_{0}$ be a vertex of $\G$. 

Let $(\A,\tau)$ be a non-commutative probability space. Recall from \eqref{moments nu} the definition of the measures $\nu_{t}$ on $\UC$. Let $((u_{F}:F\in \F^{b}),(u_{e} : e\in \T^{+}))$ be a family of unitary elements of $\A$ which are mutually free and such that for all $F\in \F^{b}$, $u_{F}$ has the distribution $\nu_{|F|}$, and for all $e\in \T^{+}$, $u_{e}$ is a Haar unitary. Recall that a Haar unitary is a unitary element $u$ such that $\tau(u^{n})=\delta_{n,0}$ for all $n\in \Z$, that is, a unitary element whose non-commutative distribution is the uniform probability measure on $\UC$. Finally, for all $e\in \T^{+}$, set $u_{e^{-1}}=u_{e}^{-1}$.

For each path $c\in \Path(\G)$, let $c=e_{1}\ldots e_{r} \lambda_{F_{1}}^{\epsilon_{1}}\ldots \lambda_{F_{n}}^{\epsilon_{n}}e_{r+1}\ldots e_{r+s}$ be the decomposition of $c$ given by Corollary \ref{ecrire chemins} as a product of loops of $\Lambda_{\wT}$ and edges of $\T$. Set
\[u_{c}=u_{e_{r+s}}\ldots u_{e_{r+1}} u_{F_{n}}^{\epsilon_{n}} \ldots u_{F_{1}}^{\epsilon_{1}}u_{e_{r}}\ldots u_{e_{1}}.\]

\begin{theorem} \label{Large N DT} 
The family of random matrices $(H_{N,c}^{\K} : c\in \Path(\G))$ converges in non-commutative distribution as $N$ tends to infinity to the family $(u_{c} : c\in \Path(\G))$.
\end{theorem}

\begin{proof} Thanks to Corollary \ref{ecrire chemins}, it suffices to prove the result for the family of random matrices $\{H^{\K}_{N,\lambda_{F}} : F\in \F^{b}\}\cup \{H^{K}_{N,e}: e\in \T^{+}\}$. Proposition \ref{YM basis} describes for each $N\geq 1$ the distribution of these random matrices. They are independent, distributed respectively according to a heat kernel measure and to the uniform measure. In particular their distributions are conjugation invariant and, granted the fact that a uniformly random matrix on $\U(N,\K)$ converges towards a Haar unitary, the result is a consequence of Theorems \ref{limite libre} and \ref{inv asymp free}. 

The fact that a uniform matrix on $\U(N,\K)$ converges to a Haar unitary is in turn a direct consequence of Proposition \ref{haar unitary}.
\end{proof}

Theorem \ref{Large N DT} provides us with the distribution of a family of non-commutative random variables indexed by all paths on $\G$. By construction, the random variable $u_{c}$ depends only on the equivalence class of $c$ in the set of reduced paths $\RP(\G)$. Moreover, the family $\{u_{c}:c\in \Path(\G)\}$ is multiplicative, in that $u_{c_{1}c_{2}}=u_{c_{2}}u_{c_{1}}$ whenever $c_{1}$ can be concatenated with $c_{2}$. In order to take this multiplicativity more explicitly into account, it is tempting to incorporate the structure of groupoid of $\RP(\G)$ into the construction of the family of non-commutative variables. Unfortunately, it seems delicate to define a non-commutative probability space on the algebra of a groupoid, in particular regarding the definition of the unit of this algebra.

We are thus led to lower our ambitions and to consider, instead of all paths, only the loops based at some vertex. From a physical point of view, considering loops instead of paths is natural, for they contain all the gauge-invariant information. Moreover, we already know that the groups of reduced loops based at any two vertices $v_{0}$ and $v_{1}$ are isomorphic by an explicit isomorphism, and that for any path $c$ joining $v_{1}$ to $v_{0}$, the collection $\{u_{l} : l\in \RL_{v_{1}}(\G)\}$ is conjugated to the collection $\{u_{l} : l\in \RL_{v_{0}}(\G)\}$ by $u_{c}$. One can check that $u_{c}$ is a Haar unitary which is free with $\{u_{l} : l\in \RL_{v_{0}}(\G)\}$, so that the distribution of one family is easily deduced from the distribution of the other. 

Let us choose a vertex $v_{0}$. Since the multiplicativity of the family $(u_{l}:l\in \Loop_{v_{0}}(\G))$ writes $u_{l_{1}l_{2}}=u_{l_{2}}u_{l_{1}}$ rather than $u_{l_{1}l_{2}}=u_{l_{1}}u_{l_{2}}$, the group which it is appropriate to consider is not exactly $\RL_{v_{0}}(\G)$, but the opposite group $\RL_{v_{0}}(\G)^{op}$, which is the same set endowed with the reversed group operation $l_{1}\cdot l_{2}=l_{2}l_{1}$. We thus consider the complex unital algebra $\C[\RL_{v_{0}}(\G)^{op}]$,
endowed with the usual involution for the algebra of a group, namely the involution given by 
\[\bigg(\sum_{l}  \alpha_{l} l \bigg)^{*}=\sum_{l} \overline{\alpha_{l}} l^{-1}.\]
For each $N\geq 1$ and each $\K\in \{\R,\C,\H\}$, let us define a state $\Phi^{\G,\K}_{N}$ on $\C[\RL_{v_{0}}(\G)^{op}]$ by setting, for all $l\in \RL_{v_{0}}(\G)^{op},$
\[ \Phi^{\G,\K}_{N}(l)=\E\left[\tr\left(H_{N,l}^{\K}\right)\right],\]
with $\tr$ replaced by $\Re\tr$ if $\K=\H$. It is indeed a state because $\Phi^{\G,\K}_{N}(1)=1$ and, for all matrix $M\in \Mat_{N}(\K)$, the number $\tr(MM^{*})$ is a non-negative real. 

Finally, let us define the collection $\{h_{l} : l\in \Loop_{v_{0}}(\G)\}$ of elements of $\C[\RL_{v_{0}}(\G)^{op}]$ by letting $h_{l}$ be equal to the unique reduced loop equivalent to $l$. We can reformulate the convergence expressed by Theorem \ref{Large N DT} as follows. 

\begin{proposition} \label{convergence graph group} Let $\G$ be a graph. Let $v_{0}$ be a vertex of $\G$. On the complex unital involutive algebra $\C[\RL_{v_{0}}(\G)^{op}]$, the sequence of states $(\Phi^{\G,\K}_{N})_{N\geq 1}$ converges pointwise to a state $\Phi^{\G}$ which does not depend on $\K$. 

The state $\Phi^{\G}$ can be described as follows. Let $\T$ be a spanning tree of $\G$. Let $\E^{+}$ be an orientation of $\G$ and set $\T^{+}=\T\cap \E^{+}$. Let $v_{0}$ be a vertex of $\G$. Let $\{\lambda_{F} : F\in \F^{b}\}$ be the corresponding basis of $\RL_{v_{0}}(\G)$.

The family $((h_{F}:F\in \F^{b}),(h_{e} : e\in \T^{+}))$ is free with respect to $\Phi^{G}$ and such that for all $F\in \F^{b}$, $h_{F}$ has the distribution $\nu_{|F|}$, and for all $e\in \T^{+}$, $h_{e}$ is a Haar unitary. 

Moreover, as $N$ tends to infinity, and regardless of the value of $\K$, the collection of random matrices $(H^{\K}_{N,l})_{l\in \Loop_{v_{0}}(\G)}$ converges in non-commutative distribution to the distribution of the family $(h_{l})_{l\in \Loop_{v_{0}}(\G)}$ with respect to $\Phi^{\G}$.
\end{proposition}

\begin{proof} These assertions are straightforward consequences of Theorem \ref{Large N DT}. The second depends also on the fact that for all $\K$ and all $N\geq 1$, and all loops $l_{1},l_{2}\in \Loop_{v_{0}}(\G)$ which are equivalent, one has the equality of random variables $H_{N,l_{1}}^{\K}=H_{N,l_{2}}^{\K}$.
\end{proof}

The state $\Phi^{\G}$ on the involutive algebra $\C[\RL_{v_{0}}(\G)^{op}]$ is the discrete counterpart of what we shall call the master field on the plane.

\subsection{Large $N$ limit for piecewise affine loops}\label{large N aff} Having understood the large $N$ limit of the theory on a graph, it is easy to go one step beyond and to consider several graphs simultaneously. As in the construction of the Yang-Mills field itself, we can however not consider all graphs at once but we must restrict ourselves to a class of graphs where any two graphs are dominated in the partial order of fineness by a third one. Graphs with piecewise affine edges are an example of such a class. 

Consider two graphs $\G_{1}$ and $\G_{2}$ such that $\G_{2}$ is finer than $\G_{1}$. Let $v_{0}$ be a vertex of $\G_{1}$, hence of $\G_{2}$. The inclusion $\Loop(\G_{1})\subset \Loop(\G_{2})$ is of course compatible with the equivalence of paths, for two loops in $\G_{1}$ which are equivalent in $\G_{1}$ are also equivalent in $\G_{2}$. There is thus a quotient mapping $\RL_{v_{0}}(\G_{1})\to \RL_{v_{0}}(\G_{2})$ which is a group homomorphism.

\begin{lemma}\label{reduce} The homomorphism $\RL_{v_{0}}(\G_{1})\hookrightarrow \RL_{v_{0}}(\G_{2})$ is injective.
\end{lemma}

\begin{proof} Assume that the kernel of this homomorphism contains a non-constant reduced loop $l\in \RL_{v_{0}}(\G_{1})$. Its image, which is $l$ itself but seen as a loop in $\G_{2}$, is then equivalent to the constant loop. Since $l$ is not the constant loop, it is not reduced in $\G_{2}$. Thus, $l$ contains a sub-path of the form $ee^{-1}$ for some edge $e$ of $\G_{2}$. In particular, $l$ backtracks at the final point of $e$, which must then be a vertex of $\G_{1}$. It follows that $e$ is the last segment in $\G_{2}$ of an edge $e'$ of $\G_{1}$, and $l$ contains the sub-path $e'e'^{-1}$. We arrive to the contradiction that $l$ is not reduced in $\G_{1}$. 
\end{proof}

We have an inclusion of groups $\RL_{v_{0}}(\G_{1})\subset \RL_{v_{0}}(\G_{2})$, hence also of the opposite groups $\RL_{v_{0}}(\G_{1})^{op}\subset \RL_{v_{0}}(\G_{2})^{op}$, and of the corresponding algebras $\C[\RL_{v_{0}}(\G_{1})^{op}]\subset \C[\RL_{v_{0}}(\G_{2})^{op}]$. 

The invariance under refinement of the Yang-Mills measure can be expressed as follows (see \cite[Prop. 4.3.4]{LevySMF}).

\begin{proposition} \label{inv ref YM}The inclusion of algebras $\C[\RL_{v_{0}}(\G_{1})^{op}]\subset \C[\RL_{v_{0}}(\G_{2})^{op}]$ is compatible with the states $\Phi^{\G_{1},\K}_{N}$ and $\Phi^{\G_{2},\K}_{N}$, in the sense that
\[\left.\Phi^{\G_{2},\K}_{N}\right|_{\C[\RL_{v_{0}}(\G_{1})^{op}]}=\Phi^{\G_{1},\K}_{N}.\]
\end{proposition}

A consequence of Lemma \ref{reduce} is that any loop which can be traced in a graph has a unique reduced representative, which is defined independently of the choice of a graph in which the loop can be traced. This is in particular the case for piecewise affine loops. We may thus speak of reduced piecewise affine loops without specifying a graph, and we denote by $\RAff_{0}(\R^{2})$ the set of reduced piecewise affine loops based at the origin in $\R^{2}$. As a set, it is the direct limit of the sets $\RL_{0}(\G)$ along the set of graphs with piecewise affine edges which have the origin of $\R^{2}$ as a vertex:
\[\RAff_{0}(\R^{2})=\build{\lim}_{\longrightarrow}^{} \RL_{0}(\G).\]
Since the inclusions described in Lemma \ref{reduce} are group homomorphisms, $\RAff_{0}(\R^{2})$ is also a group. This simply means that piecewise affine loops based at the origin can be concatenated and reduced without making explicit reference to any graph.

The direct limit can in fact be taken at the level of the group algebras and even, thanks to Proposition \ref{inv ref YM},
at the level of non-commutative probability spaces. We thus define, for each $\K$ and each $N\geq 1$,
\[\left(\C[\RAff_{0}(\R^{2})^{op}],\Phi^{\Aff,\K}_{N}\right)=\build{\lim}_{\longrightarrow}^{} \left(\C[\RL_{0}(\G)^{op}],\Phi^{\G,\K}_{N}\right).\]
Concretely, $\C[\RAff_{0}(\R^{2})^{op}]$ is the involutive algebra of formal complex linear combinations of piecewise affine loops based at $0$ and the state $\Phi^{\Aff,\K}_{N}$ is defined by the equality 
\[\Phi^{\Aff,\K}_{N}(l)=\Phi^{\G,\K}_{N}(l)=\E\left[\tr\left(H_{N,l}^{\K}\right)\right],\]
where $\G$ is any graph with piecewise affine edges such that $l$ belongs to $\Loop_{0}(\G)$.

Let us denote by $\Aff_{0}(\R^{2})$ the set of piecewise affine loops on $\R^{2}$ based at the origin. Consider an element $l\in \Aff_{0}(\R^{2})$. For each graph $\G$ in which $l$ can be traced, we defined just before Proposition 
 \ref{convergence graph group} a non-commutative random variable $h_{l}$ in $(\C[\RL_{0}(\G)],\Phi^{\G,\K}_{N})$. These definitions for all possible graphs $\G$ are compatible and define $h_{l}$ as an element of the direct limit $\C[\RAff_{0}(\R^{2})]$, which is simply the reduced loop equivalent to $l$.

\begin{proposition} \label{large N Aff} On the algebra $\C[\RAff_{0}(\R^{2})^{op}]$, for each choice of $\K$, the sequence of states $(\Phi^{\Aff,\K}_{N})_{N\geq 1}$  converges pointwise as $N$ tends to infinity to a state $\Phi^{\Aff}$.

As $N$ tends to infinity, and regardless of the value of $\K$, the collection of random matrices $(H^{\K}_{N,l})_{l\in \Aff_{0}(\R^{2})}$ converges in non-commutative distribution to the distribution of the family $(h_{l})_{l\in \Aff_{0}(\R^{2})}$ with respect to $\Phi^{\Aff}$.
\end{proposition}

\begin{proof} Both statements are obtained by taking the direct limit of the assertions of Proposition \ref{convergence graph group} along the set of graphs with piecewise affine edges which have the origin as a vertex, directed by the relation of fineness.
\end{proof}

The function $\Phi^{\Aff}$ can immediately be extended to the set $\Aff(\R^{2})$ of all piecewise affine loops on $\R^{2}$, either by replacing the origin of $\R^{2}$ by any other point of $\R^{2}$ in the statements above, or by setting, for all piecewise affine loop $l$,  $\Phi^{\Aff}(l)=\Phi^{\Aff}(clc^{-1})$ where $c$ is the line segment which joins the origin of $\R^{2}$ to the base point of $l$. The two points of view yield of course the same extended function $\Phi^{\Aff} : \Aff(\R^{2}) \to \C$.\\

In their recent work \cite{AnshelevitchSengupta}, M. Anshelevitch and A. Sengupta prove a result similar to Proposition \ref{large N Aff}, and provide a model for the limiting distribution which is in a sense more natural than ours. 
The authors consider a slightly different class of paths, which they call {\em basic loops}, and which are finite concatenations of radial segments and paths which can be parametrised in polar coordinates under the form $\theta\mapsto (r(\theta),\theta)$. In the context of axial gauge fixing in which they work, this class of paths plays essentially the role of our class of piecewise affine edges. Using free stochastic calculus, the authors achieve the construction of a free process indexed by the set of basic loops on the algebra of bounded operators on the full Fock space on $L^{2}(\R^{2})\otimes \u(N)$. This is in very suggestive agreement with the informal description of the Yang-Mills measure by means of a functional integral, which through an appropriate choice of gauge, becomes a Gaussian measure on the Hilbert space $L^{2}(\R^{2})\otimes \u(N)$. The transition from a commutative Gaussian setting to a non-commutative semi-circular setting is thus naturally reflected in the transition from the symmetric Fock space to the full Fock space, although the former usually stays hidden behind the probabilistically more familiar white noise.
M. Anshelevitch and A. Sengupta worked with the unitary group, but given the results of Section \ref{section one BM} of the present work, it should be possible to extended their results to the orthogonal and symplectic cases. 

\subsection{Uniformity of the convergence towards the master field (statement)}\label{sec: uniformity}
Theorem \ref{Large N DT}, of which we have now exhausted the algebraic consequences, was proved by blending the notion of freeness with a combinatorial description of the set of paths in a graph. We are now going to enter more deeply into the convergence that it expresses, in order to prove that this convergence has a property of uniformity on sets of paths with bounded length. This is the crucial result which will allow us to take the last step in the construction of the master field and to extend the state $\Phi^{\Aff}$ to an algebra constructed from all rectifiable loops.

We shall prove the uniformity of the convergence for a class of loops which is slightly more restricted than the class of piecewise affine loops. Let us call {\em elementary loop} a loop which can be traced in a graph with piecewise affine edges, with the additional constraint that it traverses at most one edge of each pair $\{e,e^{-1}\}$, and that it traverses it at most once. Thus, an elementary loop is a product of edges of a graph which are pairwise distinct and not equal to each other's inverse. Elementary loops are in particular piecewise affine and reduced. We denote by $\EL(\R^2)$ the set of elementary loops. Note however that $\EL_{0}(\R^{2})$, the set of elementary loops based at the origin, is not a subgroup of $\RAff_{0}(\R^{2})$.

Recall that we denote the Euclidean length of a loop $l$ by $\ell(l)$. For any complex-valued random variable $Z$, we call variance of $Z$ the number $\Var(Z)=\E\left[|Z|^{2}\right]-\left|\E[Z]\right|^{2}$. Our main result of uniformity is the following.

\begin{theorem}\label{conv unif long} Let $l$ be an elementary loop. Then, for all $\K$ and all $N\geq 1$, one has the inequalities
\begin{equation}\label{main esperance}
\left| \E\left[\tr\left(H_{N,l}^{\K}\right)\right] - \lim_{N\to\infty} \E\left[\tr\left(H_{N,l}^{\K}\right)\right] \right| =\left| \Phi^{\Aff,\K}_{N}(l) - \Phi^{\Aff}(l) \right| \leq \frac{1}{N} \ell(l)^2 e^{\ell(l)^2}
\end{equation}
and
\begin{equation}\label{main variance}
\Var\left(\tr\left(H_{N,l}^{\K}\right)\right)\leq \frac{1}{N} \ell(l)^{2} e^{\ell(l)^{2}},
\end{equation}
where $\tr$ must be replaced by $\Re\tr$ if $\K=\H$. Moreover, when $\K=\C$, the inequalities hold with the factor $\frac{1}{N}$ replaced by $\frac{1}{N^{2}}$.

In particular, for all real $L\geq 0$, the convergence of the sequence of functions $(\Phi^{\Aff,\K}_{N})_{N\geq 1}$ towards $\Phi^{\Aff}$ is uniform on the set of elementary loops with length smaller than $L$.
\end{theorem}

This result will be deduced from the main result of the first part of the present work, which is Theorem \ref{main estim}, along the following lines. Let $l$ be an elementary loop based at some point $v_{0}$. Let $\G$ be a graph in which $l$ can be traced. We know from Proposition \ref{basis RL} that $l$ can be expressed as a word in the elements of a basis of $\RL_{v_{0}}(\G)$:
\[l=w(\lambda_{F_{1}},\ldots,\lambda_{F_{q}}),\]
where $\F^{b}=\{F_{1},\ldots,F_{q}\}$ is the set of bounded faces of $\G$. From this equality follows
\[H_{N,l}^{\K}=w^{op} (H_{N,\lambda_{F_{1}}}^{\K},\ldots,H_{N,\lambda_{F_{q}}}^{\K}),\]
where $w^{op}$ is the word $w$ read backwards. By Proposition \ref{YM basis}, the random variables on the right-hand side of this equality have the distribution of independent Brownian motions on $\U(N,\K)$ at times $|F_{1}|,\ldots,|F_{q}|$. We are thus in a situation where Theorem \ref{main estim} provides us with an explicit estimate (compare with \eqref{def tauknt}). This estimate involves the Amperean area of $w^{op}$ (see \eqref{def NCA w}), which is the same as the Amperean area of $w$. The main step of the proof of Theorem \ref{conv unif long} consists in proving that provided the basis of $\RL_{v_{0}}(\G)$ has been chosen in a certain appropriate way, the Amperean area of the word $w$ can be explicitly controlled by the length of the loop $l$ alone. This is stated below as Proposition \ref{quant basis RL}.

\subsection{Maximal Amperean area}\label{maa} In this paragraph, as explained immediately above, we prove a quantitative version of Proposition \ref{basis RL} by relating the length of a loop in a graph and the Amperean area of the word which expresses this loop in terms of a lasso basis of the group of reduced loops in this graph. We introduce a third quantity which we call the maximal Amperean area of the loop, which allows us to relate the two quantities which we want to compare. Let us start by defining the maximal Amperean area of a loop and comparing it to its length. 

In preparation for this, we associate a graph to each elementary loop. We call degree of a vertex of a graph the number of edges which start from this vertex.

\begin{lemma} \label{graph gl} Let $l\in \Aff(\R^{2})$ be a piecewise affine loop. There exists a unique graph $\G_{l}$ such that each graph in which $l$ can be traced is finer than $\G_{l}$. Moreover, if $l$ is an elementary loop, then the degree of every vertex of $\G_{l}$ is even, and at least equal to $4$ except for the origin of $l$ which may have degree $2$. 
\end{lemma}

\begin{proof} Let $\G$ be a graph in which $l$ can be traced. Define $\E_{0}$ as the set of edges of $\G$ which are contained in the range of $l$. Since the union of the ranges of the edges of $\E_{0}$ is the range of $l$ which is connected,
$\E_{0}$ is the set of edges of a graph $\G_{0}$ (see Section \ref{dymf}).
It is not difficult to check that any vertex of $\G_{0}$ other than the origin of $l$ and which has degree $2$ can be removed from $\G_{0}$ by replacing two edges  by their concatenation, without altering the fact that $l$ can be traced in $\G_{0}$.
Removing in this way all vertices of degree $2$ other than the origin of $l$, we arrive at a graph $\G_{l}$ of which we claim that it is the least fine graph on which $l$ can be traced. Indeed, the local structure of the range of $l$ around each vertex of $\G_{l}$ is that of a point from which are issued either one or at least three half-lines, except perhaps if the vertex is the origin of $l$. All the vertices of $\G_{l}$ must then be vertices of any graph whose skeleton contains the range of $l$.
\end{proof}

Let $l$ be an elementary loop. In what follows, we will identify several times the loop $l$ with its range. Recall that the winding number of $l$ is the function $\n_{l}:\R^{2}\setminus l \to \Z$ defined on the complement of $l$ and which to each point $x\in \R^{2}$ associates the index of $l$ with respect to $x$. It is integer-valued, locally constant, and it has compact support. The Banchoff-Pohl inequality (see \cite{BanchoffPohl}), which generalises the isoperimetric inequality in this context, compares the Amperean area of the loop $l$, defined by
\begin{equation}\label{def AmpA}
\Amp(l)=\int_{\R^{2}} \n_{l}(x)^{2}\; dx,
\end{equation}
to its length, and reads
$$\int_{\R^{2}} \n_{l}(x)^{2}\; dx \leq \frac{1}{4 \pi}\ell(l)^{2}.$$
The Amperean area of $l$ owes its name to the fact that it can be understood as the energy of the magnetic field induced by a unit current flowing along $l$.

Let us introduce another integer-valued function $\na_{l}:\R^{2}\setminus l \to \N$, this time with non-negative values. In words, for all $x\in \R^{2}\setminus l$, $\na_{l}(x)$ is the minimal number of crossings between a path which joins $x$ to infinity and the loop $l$. Formally, let us consider the graph $\G_{l}$ and its dual graph $\widehat{\G}_{l}$. Recall that $\widehat{\G}_{l}$ has a distinguished vertex $\hat F_{\infty}$ which corresponds to the unbounded face of $\R^{2}\setminus l$. We define, for all $x\in \R^{2}\setminus l$, the dual vertex $\hat F_{x}$ as the vertex of $\widehat{\G}_{l}$ corresponding to the face of $\G$ which contains $x$. Finally, we denote by $\hat d$ the graph distance in $\widehat{\G}_{l}$. We define the function $\na_{l}$ by setting, for all $x\in \R^{2}\setminus l $,
\[ \na_{l}(x)=\hat d(\hat F_{x},\hat F_\infty).\]
We call the function $\bar \n_{l}$ the maximal winding number of $l$. Note that it depends on $l$ only through $\G_{l}$, and in particular not on the direction in which $l$ traverses the edges of $\G_{l}$. The definition of the maximal Amperean area of $l$, denoted by $\Aa(l)$, is obtained by replacing in the definition of the Amperean area the winding number $\n_{l}$ by the maximal winding number $\bar \n_{l}$: 
\begin{equation}\label{def NCA}
\Aa(l)=\int_{\R^{2}} \na_{l}(x)^{2}\; dx.
\end{equation}

For our purposes, the first main property of the maximal Amperean area is the following.

\begin{proposition}\label{BP ineq NCA} The maximal Amperean area of an elementary loop satisfies the Banchoff-Pohl inequality. By this we mean that for all $l\in \EL(\R^{2})$,  
$$\Aa(l)\leq \frac{1}{4\pi} \ell(l)^{2}.$$
\end{proposition}

This proposition follows at once from the following result, which also justifies the name of the maximal Amperean area.

\begin{lemma} \label{llbar} Let $l\in \EL(\R^{2})$ be an elementary loop.\\
1. The inequality $|\n_{l}|\leq\bar \n_{l}$ holds on $\R^{2}\setminus l$. In particular, $\Amp(l)\leq \Aa(l)$.\\
2. There exists $\bar l \in \EL(\R^{2})$ with the same range and length as $l$ such that $\n_{\bar l}=\na_{l}$.
\end{lemma}

That this lemma implies Proposition \ref{BP ineq NCA} is straightforward. Indeed, if $l$ is an elementary loop and  $\bar l$ is given by the second assertion of this lemma, then $\Aa(l)=\Amp(\bar l)\leq \frac{1}{4\pi}\ell(\bar l)^{2}=\frac{1}{4\pi}\ell(l)^{2}$.

\begin{proof} 1. Consider an edge of $\G_{l}$. Since the faces located on either side of this edge correspond to two vertices of the dual graph $\widehat{\G}_{l}$ which are equal or nearest neighbours, the values of $\na_{l}=\hat d(\cdot,\hat F_{\infty})$ on both sides of this edge are equal or differ by $1$. Let us start by proving that they cannot be equal. 

Since the loop $l$ is elementary, it traverses each edge exactly once. Hence, the value of the winding number $\n_{l}$ changes by $1$ or $-1$ when one crosses an edge. The set of vertices of the graph $\widehat{\G}_{l}$ can be partitioned according to the parity of the value of $\n_{l}$ and we shall speak of even and odd vertices. This partition is a bipartition in the sense that any two neighbours have different parities. The dual root $\hat F_\infty$ is an even vertex. Hence, the parity of any vertex $\hat F$ is that of $\hat d(\hat F, \hat F_\infty)$. If two neighbours were to have the same distance to $\hat F_\infty$, they would also have the same parity and this is impossible. Hence, $\na_{l}$ cannot take the same value on two faces which share a bounding edge.

Let us use this observation to prove the first assertion. Consider $x\in \R^{2}\setminus l$. Choose a shortest path from $\hat F_{\infty}$ to $\hat F_{x}$ in $\widehat{\G}_{l}$. At each step of this path, $\n_{l}$ varies by $\pm 1$, and $\bar \n_{l}$ increases by $1$. The conclusion follows immediately, as well as the second part of the first assertion.

2. We look for $\bar l$ as a Eulerian circuit in $\G_{l}$, that is, a cycle which traverses exactly once each edge. Let us start by proving that the direction in which each edge should be traversed by $\bar l$ is determined by $\na_{l}$. Consider an edge of $\G_{l}$. We have proved that the values of $\na_{l}$ on both sides of this edge differ by $1$. If $\n_{\bar l}$ is to be equal to $\na_{l}$, then $\bar l$ must traverse this edge in such a way that the largest value of $\na_{l}$ is on its left-hand side. Thus, each edge of $\G_{l}$ carries an orientation which is the direction in which $\bar l$ must traverse it in order for the equality $\n_{\bar l}=\na_{l}$ to hold. 

We claim that there are, at each vertex of $\G_{l}$, as many incoming edges as there are outcoming ones. Indeed, the values of $\na_{l}$ around each vertex read in cyclic order form a sequence of integers which jumps by $1$ or $-1$ and comes back to its initial point. Thus, there must be an equal number of rises and falls, which correspond respectively to incoming and outgoing edges.

We use now the classical fact that an oriented graph in which each vertex has equal incoming and outcoming degrees carries a Eulerian circuit, that is, a loop which traverses each edge exactly once, and does so in the direction given by the orientation of the edge. We choose one of these circuits and call it $\bar l$. It is a loop with the same length as $l$. The functions $\n_{\bar l}$ and $\na_{l}$ are both integer-valued, locally constant on the complement of  $l$, equal to $0$ at infinity, and both vary by $1$ or $-1$ in the same way across each edge of $\G_{l}$. Hence, they are equal. \end{proof}

We now turn to the main result of this section, which is the following. Recall the notation $w(\cdot,\ldots,\cdot)$ which we defined at the beginning of Section \ref{subsec soc}.

\begin{proposition}\label{quant basis RL}
Let $l$ be an elementary loop. Let $\F^{b}_{l}=\{F_{1},\ldots,F_{q}\}$ be the set of bounded faces of $\G_{l}$. Let $t=(|F_{1}|,\ldots,|F_{q}|)$ be the vector of the areas of these faces. 

It is possible to choose a spanning tree $\T$ of $\G_{l}$ in such a way that, denoting by $\{\lambda_{F}: F\in \F^{b}_{l}\}$ the lasso basis of $\RL_{v_{0}}(\G_{l})$ determined by the choice of $\T$, and by $w$ the unique element of the free group on $q$ letters such that $l=w(\lambda_{F_{1}},\ldots,\lambda_{F_{q}})$, one has the inequality
\[\Aa_{t}(w)\leq \frac{1}{4\pi} \ell(l)^{2}.\]
\end{proposition}

The crucial step in the proof of this result is given by the next proposition. Recall from \eqref{def NCA w} the definition of the Amperean area of $w$ relative to $t$, denoted by $\Aa_{t}(w)$. 

\begin{proposition}\label{NCA w l} With the notation of Proposition \ref{quant basis RL}, it is possible to choose the spanning tree $\T$ of $\G_{l}$ in such a way that the inequality $\Aa_{t}(w)\leq \Aa(l)$ holds. 
\end{proposition}

Let us for one minute take this assertion for granted and see how it implies Proposition \ref{quant basis RL}.

\begin{proof}[Proof of Proposition \ref{quant basis RL}] By Proposition \ref{NCA w l}, one can choose the basis of $\RL_{v_{0}}(\G_{l})$ in such a way that $\Aa_{t}(w)\leq \Aa(l)$. On the other hand, by Proposition \ref{BP ineq NCA}, $\Aa(l)\leq \frac{1}{4\pi} \ell(l)^{2}$. 
\end{proof}

It remains to prove Proposition \ref{NCA w l}. Rather than choosing a spanning tree of $\G_{l}$, we will in fact choose a spanning tree $\wT$ of the dual graph $\widehat{\G}_{l}$, but we know from Section \ref{grouploops} that this is equivalent. Given such a spanning tree $\wT$, define, for each dual vertex $\hat F$, the integer $\hat d_{\widehat \T}(\hat F, \hat F_\infty)$ as the graph distance between $\hat F$ and $\hat F_\infty$ in $\widehat \T$. This number is also called the height of $\hat F$ in $\widehat \T$ and it is the length of the word of integers which labels $F$ in the labelling which we described in the course of the proof of Proposition \ref{basis RL}. The inequality $\hat d(\hat F,\hat F_\infty)\leq \hat d_{\widehat \T}(\hat F, \hat F_\infty)$ holds for all $\hat F$. We claim that $\widehat \T$ can be chosen in such a way that it is an equality. This is in fact a perfectly general property of any finite graph.

\begin{lemma}\label{Td} There exists a spanning tree $\widehat \T$ such that, for all dual vertex $\hat F$ of $\widehat{\G}_{l}$, the equality $\hat d(\hat F, \hat F_\infty)=\hat d_{\widehat \T}(\hat F, \hat F_\infty)$ holds.
\end{lemma}

\begin{proof} Construct $\widehat \T$ by choosing, for each dual vertex different from $\hat F_\infty$, one edge which joins this vertex to a vertex which is strictly closer from $\hat F_\infty$. The subgraph thus obtained is spanning and connected, for each vertex is joined inside it to the dual root. It has one vertex more than it has edges, it is thus a tree. It is a spanning tree.
\end{proof}

We can now finish the proof of Proposition \ref{NCA w l}.

\begin{proof}[Proof of Proposition \ref{NCA w l}]
Let $\widehat \T$ be a spanning tree of $\widehat{\G}_{l}$ such that $\hat d=\hat d_{\widehat\T}$. Such a spanning tree exists by Lemma \ref{Td}. Let $v_{0}$ be the base point of $l$. Let $\{\lambda_{F}: F\in \F^{b}_{l}\}$ be the basis of $\RL_{v_{0}}(\G_{l})$ determined by $\widehat\T$, according to Proposition \ref{basis RL}. Set $\F^{b}_{l}=\{F_{1},\ldots,F_{q}\}$. Let $w$ be the element of the free group $\Fr_{q}$ such that $l= w(\lambda_{F_{1}},\ldots,\lambda_{F_{q}})$.

In order to bound the Amperean area of $w$, we need to understand how the loop $l$ is decomposed as a word in the lassos $\lambda_{F_{1}},\ldots,\lambda_{F_{q}}$. Fortunately, we already did the work in the proof of Proposition \ref{basis RL}. Indeed, let us first write $l$ in the basis $\{\beta_{e} : e\in \E^{+}\setminus \T\}$. It suffices for this to record the ordered list $e_{1},\ldots,e_{n}$ of edges which are traversed by $l$ and which do not belong to $\T$. Then $l=\beta_{e_{1}}\ldots \beta_{e_{n}}$. We now use the triangular relation \eqref{beta to lambda} to convert this into an expression of $l$ as a word in the lassos $\lambda_{F_{1}},\ldots,\lambda_{F_{q}}$.

It remains to count how many times a given lasso or its inverse appear. According to \eqref{beta to lambda}, for each $i\in \{1,\ldots,q\}$ and all $j\in \{1,\ldots,n\}$, the lasso $\lambda_{F_{i}}$ or its inverse appears in $\beta_{e_{j}}$ exactly once if $e_{j}$ crosses $[\hat F_{i},\hat F_{\infty}]_{\wT}$, and not at all otherwise. There are $\hat d_{\T}(\hat F,\hat F_{\infty})$ unoriented edges in $\G_{l}$ which cross $[\hat F_{i},\hat F_{\infty}]_{\wT}$ and none of them belong to $\T$. Moreover, $l$, which is an elementary loop, traverses exactly once each unoriented edge of $\G_{l}$. Finally, for each $i\in \{1,\ldots,q\}$, the lasso $\lambda_{F_{i}}$ or its inverse appears exactly $\hat d_{\T}(\hat F,\hat F_{\infty})$ times in the decomposition of $l$.
We chose the spanning tree $\wT$ in such a way that this number is equal to $\hat d(\hat F_{i},\hat F_\infty)$, that is, by definition, the value of $\bar \n_{l}$ on $F_{i}$.

The expression of $l$ as a word $w$ in $\lambda_{F_{1}},\ldots,\lambda_{F_{q}}$ which we obtain by applying \eqref{beta to lambda} to the equality $l=\beta_{e_{1}}\ldots \beta_{e_{n}}$ may not be reduced. Simplifying it can however only decrease the Amperean area of $w$, which finally satisfies $\Aa_{t}(w)\leq \sum_{i=1}^{q} |F_{i}| \bar \n_{l}(F_{i})^{2}=\Aa(l)$.
\end{proof}

\subsection{Uniformity of the convergence towards the master field (proof)}\label{sec: uniformity proof}

We can finally prove the result of uniform convergence of the expected trace of the Yang-Mills field towards its limit on sets of elementary loops of bounded length.

\begin{proof}[Proof of Theorem \ref{conv unif long}] 
Let $l$ be an elementary loop. We use the notation of Proposition \ref{quant basis RL}. The law of the random variable $H_{N,l}^{\K}$ does not depend on the graph in which it is computed. We choose to consider the graph $\G_{l}$ (see Lemma \ref{graph gl}). Let us choose a basis of $\RL_{v_{0}}(\G_{l})$ in which the conclusion of Proposition \ref{quant basis RL} holds.

On one hand, thanks to Proposition \ref{YM basis} for $\E [\tr(H_{N,l}^{\K})]$ and to Proposition \ref{convergence graph group} for $\Phi^{\Aff}(l)$, we have
\[\left| \E [\tr(H_{N,l}^{\K})] - \Phi^{\Aff}(l) \right|=|\tauknt(w^{op})-\tau_t(w^{op}) |,\]
where $w^{op}$ denotes the word $w$ read backwards. It is understood, as usual, that in the quaternionic case, $\tr$ has to be replaced by $\Re\tr$.
On the other hand, by Proposition \ref{main estim},
\[\left| \tauknt(w^{op})-\tau_t(w^{op})\right| \leq \frac{1}{N} \Aa_t(w^{op})e^{\Aa_t(w^{op})},\]
where $\frac{1}{N}$ can be replaced by $\frac{1}{N^{2}}$ if $\K=\C$.

But Proposition \ref{quant basis RL} ensures that $\Aa_t(w)\leq \frac{1}{4\pi}\ell(l)^2$. Since $\Aa_{t}(w^{op})=\Aa_{t}(w)$, we find
\[\left| \E [\tr(H_{N,l}^{\K})] - \Phi^{\Aff}(l) \right| \leq \frac{1}{N} \frac{1}{4\pi}\ell(l)^2 e^{\frac{1}{4\pi}\ell(l)^2},\]
which is even slightly better than the expected inequality.

Let us turn to the second inequality. We are going to apply Proposition \ref{unif estim} with a word which is not $w^{op}$ and with a permutation which is not a single cycle. Let $r$ denote the length of the word $w$. Recall that if $w=x_{i_{1}}^{\epsilon_{1}}\ldots x_{i_{r}}^{\epsilon_{r}}$, then we denote by $w^{*}$ the word $x_{i_{r}}^{-\epsilon_{r}}\ldots x_{i_{1}}^{-\epsilon_{1}}$. Let us apply Proposition \ref{unif estim} to the word $w^{op}(w^{op})^{*}$ and to the permutation $\sigma=(1\ldots r)(r+1\ldots 2r)$. This is the place where we benefit from having allowed the word which we consider in Proposition \ref{unif estim} not to be reduced. Indeed, the image of $w^{op}(w^{op})^{*}$ in the free group $\F_{q}$ is the unit element.

With this notation, we have $\ptkn(w^{op}(w^{op})^{*},\sigma)=\E[|\tr(H_{N}(l))|^{2}]$, or $\E[\Re\tr(H_{N}(l))^{2}]$ if $\K=\H$. We also have $p_{t}(w^{op}(w^{op})^{*},\sigma)=|\tau(h_{l})|^{2}=|\Phi^{\Aff}(l)|^{2}$. The Amperean area of $w^{op}(w^{op})^{*}$ satisfies $\Aa_{t}(w^{op}(w^{op})^{*})=\Aa_{t}(w^{2})=4\Aa_{t}(w)$. Hence, and with the usual replacement of $\tr$ by $\Re \tr$ if $\K=\H$, we have
\[\left|\E\left[|\tr(H_{N,l}^{\K})|^{2}\right]-|\Phi^{\Aff}(l)|^{2}\right|\leq \frac{4}{N}\Aa_{t}(w)e^{4\Aa_{t}(w)}\leq \frac{1}{\pi N}\ell(l)^{2}e^{\ell(l)^{2}}.\]
Since $\E[|\tr(H_{N,l})|]\leq 1$ and $\left|\Phi^{\Aff}(l)\right|\leq 1$, we deduce from this inequality that
\begin{align*}
\left|\E\left[|\tr(H_{N,l}^{\K})|^{2}\right]-|\E[\tr(H_{N,l}^{\K})]|^{2}\right| &\leq \left|\E\left[|\tr(H_{N,l}^{\K})|^{2}\right]-|\Phi^{\Aff}(l)|^{2}\right| + 2 \left|\E[|\tr(H_{N,l}^{\K})|]-|\Phi^{\Aff}(l)|\right|\\
& \leq \frac{1}{N}\ell(l)^{2}e^{\ell(l)^{2}}\left(\frac{1}{\pi}+\frac{1}{2\pi}\right)\\
&\leq\frac{1}{2}\frac{1}{N}\ell(l)^{2}e^{\ell(l)^{2}},
\end{align*} 
as desired. Here as in the first inequality, the factor $\frac{1}{N}$ can be replaced by $\frac{1}{N^{2}}$ when $\K=\C$.
\end{proof}

\subsection{The distribution of the master field}

Let us summarise what we have done since the beginning of Section \ref{mf}. We considered three classes of loops, each contained in  the next: elementary loops, piecewise affine loops, and loops, so that $\EL(\R^{2})\subset \Aff(\R^{2})\subset \Loop(\R^{2})$. A division algebra $\K\in \{\R,\C,\H\}$ being fixed, we have for all $N\geq 1$ a function $\Phi^{\K}_{N}:\Loop(\R^{2})\to \C$ defined by
\[\forall l\in \Loop(\R^{2}), \; \Phi^{\K}_{N}(l)=\E[\tr(H^{\K}_{N,l})],\]
where, as always, $\tr$ must be replaced by $\Re\tr$ if $\K=\H$. Note that these functions are continuous by the second assertion of Theorem \ref{continuous YM}, and bounded by $1$ by construction. They are real-valued on $\Aff^{\R^{2}}$ by Proposition \ref{real}, hence on $\Loop(\R^{2})$.

We proved that the restriction on $\Aff(\R^{2})$ of the sequence $(\Phi^{\K}_{N})_{N\geq 1}$ converges pointwise towards a function $\Phi^{\Aff}$, which does not depend on $\K$. This is the convergence expressed by Proposition \ref{large N Aff}. Moreover, we proved in Theorem \ref{conv unif long} that for each positive $L\geq 0$, this convergence is uniform on the set of elementary loops whose length is smaller than or equal to $L$. From this and a straightforward observation, we will now deduce that the convergence holds and is uniform on the whole space $\Loop(\R^{2})$.  

For each $L\geq 0$, set $\Loop_{L}(\R^{2})=\{l\in \Loop(\R^{2}) : \ell(l)\leq L\}$ and $\EL_{L}(\R^{2})=\EL(\R^{2})\cap \Loop_{L}(\R^{2})$. Set also $\Loop_{L^{-}}(\R^{2})=\{l\in \Loop(\R^{2}) : \ell(l)< L\}$. In the following lemma, we consider, as always, $\Loop(\R^{2})$ endowed with the topology of the convergence in $1$-variation with fixed endpoints. The length is by definition a continuous function on $\Loop(\R^{2})$ equipped with this topology, so that $\Loop_{L}(\R^{2})$ is a closed subset, and $\Loop_{L^{-}}(\R^{2})$ an open subset of $\Loop(\R^{2})$.

\begin{lemma}\label{topo loops} 1. For all $L\geq0$, the closure of $\EL_{L}(\R^{2})$ is $\Loop_{L}(\R^{2})$. \\
\indent 2. For all $L\geq 0$, the interior of $\Loop_{L}(\R^{2})$ is $\Loop_{L^{-}}(\R^{2})$.\\
\indent 3. The union of the sets $\Loop_{L^{-}}(\R^{2})$ for $L\geq 0$ is $\Loop(\R^{2})$.
\end{lemma}

\begin{proof} 1. The piecewise affine interpolations of any given loop parametrised at constant speed converge with fixed endpoints, as the mesh of the interpolation tends to $0$, to the loop itself. These piecewise affine interpolations have moreover a length which is not greater than that of the original loop. Contracting them slightly by an affine homothecy centred at the base point of the loop allows us to be sure that their length is strictly smaller than that of the loop which we are approximating. This gives us freedom to deal with the last possible issue, which is the fact that our scaled interpolations may not be elementary loops. Fortunately, any piecewise affine loop can be turned into an elementary loop by an arbitrarily small modification of the endpoints of its affine pieces, simply by making sure that no two of them are equal and no three of them are collinear.\\
2. For $L=0$, the statement is true. Choose $L>0$. The interior of $\Loop_{L}(\R^{2})$ contains the open subset $\Loop_{L^{-}}(\R^{2})$ of $\Loop(\R^{2})$. To prove the other inclusion, consider a loop of length greater than or equal to $L$. Any neighbourhood of this loop contains its images by small affine dilations around its basepoint, and these images have length strictly larger than $L$. Hence, no loop of length $L$ or more belongs to the interior of $\Loop_{L}(\R^{2})$.\\
3. This assertion barely deserves a proof.
\end{proof}

We can now extend our result of convergence to the class of all rectifiable loops. 

\begin{theorem}\label{existence Phi}  The function $\Phi^{\Aff} : \Aff(\R^2)\to \R$ can be extended in a unique way to a continuous function $\Phi:\Loop(\R^2)\to \R$, which for all $\K\in \{\R,\C,\H\}$ and all $L\geq 0$ is the uniform limit of the sequence of functions $(\Phi^{\K}_{N})_{N\geq 1}$ on $\Loop_{L}(\R^{2})$. 

More precisely, for all loop $l\in \Loop(\R^2)$ and all $N\geq 1$, the following inequalities hold:
\begin{align*}
\left|\E\left[\tr (H_{N,l}^{\K})\right]-\Phi(l)\right|&\leq\frac{1}{N} \ell(l)^2 e^{\ell(l)^2}\\
\Var\left(\tr(H_{N,l}^{\K})\right)&\leq \frac{1}{N} \ell(l)^{2}e^{\ell(l)^{2}},
\end{align*}
where $\tr$ must be replaced by $\Re\tr$ if $\K=\H$ and the factor $\frac{1}{N}$ can be replaced by $\frac{1}{N^{2}}$ if $\K=\C$. \\

In particular the following convergence in probability holds:
\[\tr(H_{N,l}^{\K}) \build{\xrightarrow{\hspace{8mm}}}_{N\to \infty}^{P} \Phi(l),\]
and in the case where $\K=\C$, this convergence is fast in the sense that the series
\[\sum_{N\geq 1} \P(|\tr(H_{N,l}^{\C})-\Phi(l)|>\epsilon)\]
converges for all $\epsilon>0$.
\end{theorem}

\begin{proof} For each $L\geq 0$, it follows from Theorem \ref{conv unif long} that the sequence $(\Phi^{\K}_{N})_{N\geq 1}$ of continuous functions on $\Loop_{L}(\R^{2})$ converges uniformly to $\Phi^{\Aff}$ on the subset $\EL_{L}(\R^{2})$ of $\Loop_{L}(\R^{2})$, which is dense by Lemma \ref{topo loops}. Hence, the sequence $(\Phi^{\K}_{N})_{N\geq 1}$ converges uniformly on $\Loop_{L}(\R^{2})$ to the unique continuous extension of $\Phi^{\Aff}$. Since, by Lemma \ref{topo loops} again, the interiors of the subspaces $\Loop_{L}(\R^{2})$ with $L\geq 0$ cover $\Loop(\R^{2})$, the convergence holds pointwise on the whole space $\Loop(\R^{2})$ and the limiting function is continuous.

Let us turn to the second part of the theorem. Consider a loop $l\in \Loop(\R^{2})$. Let $(l_{n})_{n\geq 1}$ be a sequence of elementary loops converging to $l$ with fixed endpoints. By the second assertion of Theorem \ref{continuous YM}, and for all $N\geq 1$, the sequence $(H^{\K}_{N,l_{n}})_{n\geq 1}$ converges in probability to $H^{\K}_{N,l}$. Hence, the same convergence holds for the traces and, since those are bounded, the convergence holds in $L^{2}$. Thus, we have
\[\Var(\tr(H^{\K}_{N,l}))=\lim_{n\to\infty} \Var(\tr(H^{\K}_{N,l_{n}})).\]
By the second assertion of Theorem \ref{conv unif long}, we have thus the inequality
\[\Var(\tr(H^{\K}_{N,l})) \leq \limsup_{n\to \infty}\frac{1}{N} \ell(l_{n})^{2}e^{\ell(l_{n})^{2}}=\frac{1}{N}\ell(l)^{2}e^{\ell(l)^{2}},\]
with $\frac{1}{N}$ replaced by $\frac{1}{N^{2}}$ if $\K=\C$.
The rest of the theorem follows immediately. 
\end{proof}

We would like our final result to be of the same nature as Propositions \ref{convergence graph group} and \ref{large N Aff}, describing the master field $\Phi$ as a trace on the algebra of a group of loops rather than simply a function on the set of loops. We need in particular to define the continuous analogue of the group $\RL_{v_{0}}(\G)$ of reduced loops traced in a graph. This is a deep problem, and a fascinating one in its own right, which, fortunately for us, has already been solved, precisely in the case of rectifiable paths, by B. Hambly and T. Lyons, in a way which we briefly review, and marginally extend, in the next section.

\subsection{The group of rectifiable loops}\label{group of loops} A beautiful result of B. Hambly and T. Lyons on rectifiable paths \cite{HamblyLyons} allows one, among other things, to make sense on the set $\Loop_{0}(\R^{2})$ of all rectifiable loops based at the origin of an equivalence relation analogous to the one which we used on the set of loops traced in a graph (see Section \ref{grouploops}). The central notion in their approach is that of tree-like loop, which turns out to be the appropriate continuous analogue of a loop in a graph equivalent to the constant loop. In order to to define a tree-like loop, one needs to use a certain notion of continuous tree of which we start by recalling the definition.

A {\em compact $\R$-tree} is an arcwise connected compact metric space in which any two distinct points are the endpoints of a unique subset homeomorphic to a segment, and such that the unique such subset which joins two distinct points is not only homeomorphic, but isometric to a segment\footnote{It is not explicitly contained in the definition of an arcwise connected space, at least not the one which we use nowadays, that any two distinct points of such a space are the endpoints of a subspace homeomorphic to a segment. We merely insist that they be joined by a curve, which may have self-intersection. The fact that in an arcwise connected metric space any two distinct points are indeed joined by an injective curve seems to be interestingly non-trivial, and is in any case a consequence of various substantial theorems due to Hahn, Mazurkiewicz, Moore, Menger, Sierpinski and which are summarised in the treatise of Kuratowski \cite{Kuratowski}, \textsection 45.
}. 

The next theorem gives five equivalent properties of a Lipschitz continuous loop, which all characterise tree-like loops. In this theorem, $(E,d)$ denotes a complete metric space and we think of the circle $S^{1}$ as $\R/\Z$.

\begin{theorem}\label {thin tree like} Let $l:S^{1}\to E$ be a Lipschitz continuous loop. The following assertions are equivalent.\\
1. There exists a compact $\R$-tree $T$ and two Lipschitz continuous mappings $f:S^{1}\to T$ and $g: T \to E$ such that $l= g \circ f$.\\
1'. There exists a compact $\R$-tree $T$ and two continuous mappings $f:S^{1}\to T$ and $g: T \to E$ such that $l= g \circ f$.\\
2. There exists a Lipschitz continuous function $h:[0,1]\to\R_{+}$ such that $h(0)=h(1)=0$ and, for all $s,t\in [0,1]$, the following inequality holds:
\[d(l(s),l(t))\leq h(s)+h(t)-2 \inf\{h(u) : u \in [s,t]\}.\]
3. The loop $l$ is homotopic to a constant loop within its own range, that is, the mapping $l:S^{1}\to l(S^{1})$ is inessential.\\
3'. The loop $l$ is homotopic to a constant loop within the union of the ranges of finitely many lipschitz continuous loops, that is, there exist some lipschitz continuous loops $l_{1},\ldots,l_{n}$ such that mapping $l:S^{1}\to l(S^{1})\cup l_{1}(S^{1})\cup \ldots \cup l_{n}(S^{1})$ is inessential. 
\end{theorem}

If any of these equivalent properties is satisfied, the loop $l$ is said to be {\em tree-like}. A loop which satisfies property 3 is also called a {\em thin loop} by some authors (see for example \cite{GambiniPullin}).

Not all these characterisations appear in the work of Hambly and Lyons, in particular not the last two, which are slightly remote from their point of view. We thus offer a proof of their equivalence. 

\begin{proof}[First part of the proof of Theorem \ref{thin tree like}.] 1' $\Rightarrow$ 2. 
Set $\rho=f(0)$, of which we think as the root of the tree. For all $x,y\in T$, let us denote by $V_{g}(x,y)$ the total variation of $g$ along the unique segment which joins $x$ to $y$, that is, the total variation of the function $g\circ \gamma_{x,y}$, where $\gamma_{x,y}:[0,1]\to T$ is an injective continuous path from $x$ to $y$:
\[V_{g}(x,y)=\sup_{0\leq t_{0} \leq \ldots \leq t_{n} \leq 1} \sum_{k=0}^{n-1} d(g\circ \gamma_{x,y}(t_{k}),g \circ \gamma_{x,y}(t_{k+1})).\]

 We claim that the function $h:[0,1]\to \R_{+}$ defined by 
\[h(t)=V_{g}(\rho,f(t))\]
satisfies the second property.

Let us prove that $h$ is finite and Lipschitz continuous. Let $K$ denote the Lipschitz norm of $l$. For all $x,y\in T$, let us denote by $x\wedge y$ the midpoint of $\rho$, $x$ and $y$, that is, the unique point located simultaneously on the three geodesics from $\rho$ to $x$, from $x$ to $y$ and from $y$ to $\rho$. Firstly, we have, for all $s,t\in [0,1]$,
\[|h(t)-h(s)|=|V_{g}(f(s)\wedge f(t),f(t))-V_{g}(f(s)\wedge f(t),f(s))|\leq V_{g}(f(t),f(s)) \leq K |t-s|.\]

Now, for all $s,t\in [0,1]$, $l(s)$ is joined to $l(t)$ by the image by $g$ of the geodesic from $f(s)$ to $f(t)$. Hence, if $v\in[s,t]$ is such that $f(v)=f(s)\wedge f(t)$, then
\begin{align*}
d(l(s),l(t)) &\leq V_{g}(f(s),f(t))=V_{g}(f(s),f(v))+V_{g}(f(v),f(t))\\
&=  V_{g}(\rho,f(t)) -V_{g}(\rho,f(v)) + V_{g}(\rho,f(s))-V_{g}(\rho,f(v)) \\
&= h(t)+h(s)-2 h(v)\\
&\leq  h(t)+h(s)-2\inf \{h(u) : u\in [s,t]\}.
\end{align*}

\begin{figure}[h!]
\begin{center}
\scalebox{0.8}{\includegraphics{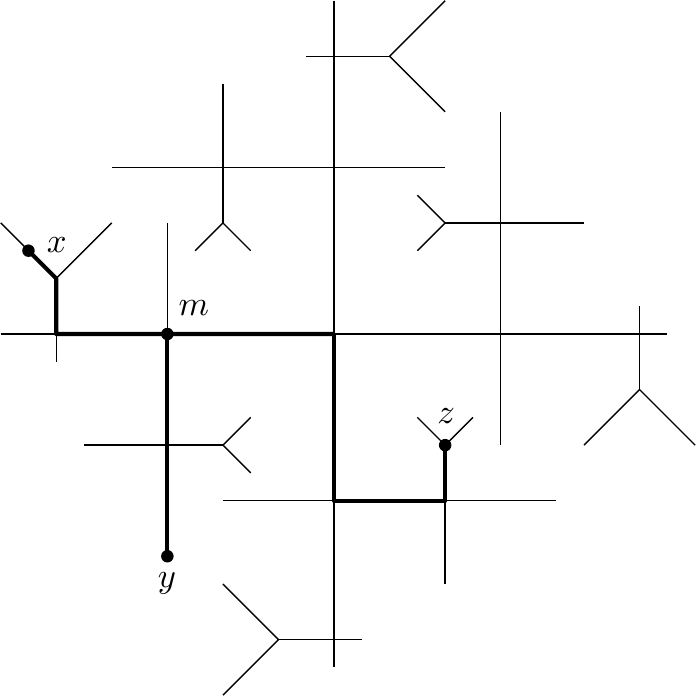}}
\caption{\label{realtree} A compact $\R$-tree on which three points $x,y,z$ have been chosen, and the geodesics which join them. The intersection of the three geodesics, denoted by $m$, is the midpoint of $x$, $y$ and $z$.}
\end{center}
\end{figure}

2 $\Rightarrow$ 1. It is a classical fact that the function $\delta(s,t)=h(s)+h(t)-2 \inf\{h(u) : u \in [s,t]\}$ is a pseudo-distance on $[0,1]$ and that the quotient by the relation which identifies $s$ and $t$ whenever $\delta(s,t)=0$ is a compact $\R$-tree, which we denote by $T$. Moreover, the canonical projection $p:[0,1]\to T$ is Lipschitz continuous with a Lipschitz constant bounded by that of $h$. By assumption, $l$ descends to a $1$-Lipschitz continuous function $\tilde l : T\to E$. With this notation, $l=\tilde l\circ p$ is the sought-after decomposition.  

Since 1 tautologically implies 1', this proves the equivalence of the first three assertions. 
\end{proof}

In order to prove the equivalence of 1' and 3, we use a result due to Fort \cite{Fort}, which we start by recalling, for the convenience of the reader and also because Fort's proof seem to contain a small gap, which we fill at the end of this section. Fort's statement and proof make use of a terminology which has gone slightly out of fashion, and we deem it possibly useful to recall a few definitions. Firstly, by a {\em dendrite}, Fort means a continuum (that is, a compact connected metric space) which is locally connected and contains no subspace homeomorphic to the circle $S^{1}$. Dendrites are studied for instance in Kuratowski's treatise \cite[Chap. VI,  \textsection 46, VI]{Kuratowski}. Secondly, a continuous mapping $f:K\to L$ between two continua is said to be {\em light} if for each $y\in L$ the set $f^{-1}(y)$ is empty or $0$-dimensional, that is, a subspace of $K$ whose topology admits a basis formed by sets which are both open and closed (see the book \cite{Hurewicz} of Hurewicz and Wallman for a comprehensive account of dimension theory). Finally, a non-empty separable metric space is said to be $1$-dimensional if it is not $0$-dimensional and if its topology admits a basis which consists of open sets whose topological boundary is empty or of dimension $0$. Fort's theorem is the following.

\begin{theorem}[Fort,\cite{Fort}]\label{fort} Let $f$ be a mapping on $S^{1}$ into a $1$-dimensional space $Y$. Then, $f$ is inessential if and only if there exists a dendrite $K$ and mappings $f_{1}$ and $f_{2}$ such that\\
\indent 1. $f=f_{2}f_{1}$,\\
\indent 2. $f_{1}$ maps $S^{1}$ onto $K$, and\\
\indent 3. $f_{2}$ is a light mapping on $K$ into $Y$.
\end{theorem}

Any compact $\R$-tree is a dendrite, and it turns out that any dendrite is homeomorphic to a compact $\R$-tree (see \cite{Bowditch}, Lemma 1.7). We can thus replace the word {\em dendrite} by {\em compact $\R$-tree} in Fort's statement.

\begin{proof}[Second part of the proof of Theorem \ref{thin tree like}.]
1' $\Rightarrow$ 3. As a compact $\R$-tree, $T$ is contractible. Let $\eta:[0,1]\times T\to T$ be a homotopy from the identity to the constant map equal to $\rho=f(0)$. Then $(s,t)\mapsto g(\eta(s,f(t)))$ is a homotopy between the loop $l(S^{1})$ and the constant loop equal to $l(0)$, within $l(S^{1})$.

3 $\Rightarrow$ 1'. This is the part where we use Fort's theorem. If $l$ is constant, we may take for $T$ a singleton. Let us now assume that $l$ is not constant. The crucial fact is that the range of $l$, being the image of an interval by a non-constant Lipschitz continuous mapping, has Hausdorff dimension $1$, hence topological dimension $1$ (see \cite[Theorem VII 2]{Hurewicz}). Hence, by Fort's theorem applied to $Y=l(S^{1})$, the fact that $l$ is homotopic to a constant loop within its own range implies that $l$ factorises through an $\R$-tree. Thus, property 1' holds. 

The assertion 3 certainly implies 3'. We finish by proving that 3' implies 1. This is the same argument as the proof that 3 implies 1, with the following modification. The sum theorem (see Theorem III 2 in \cite{Hurewicz}) asserts that a countable union of closed subspaces of dimension 1 of a topological space is still of dimension 1. Hence, the union of the compact ranges of finitely many Lipschitz continuous loops has dimension 1, and Fort's theorem applies also in this case.
\end{proof}

The result of B. Hambly and T. Lyons which matters most for our purposes, and which actually is a corollary of their main result, is the following.

\begin{theorem}[\cite{HamblyLyons}] The relation $\sim$ defined on $\Loop_{0}(\R^{2})$ by declaring $l_{1}\sim l_{2}$ if and only if $l_{1}l_{2}^{-1}$ is tree-like is an equivalence relation. Moreover, each equivalence class contains a unique loop of shortest length, which is characterised by the fact that no restriction of this loop is a tree-like loop. 
\end{theorem}

A loop which is the shortest in its equivalence class is said to be {\em reduced}. In \cite{HamblyLyons}, this theorem is inferred from considerations on an algebraic object associated to a path which the authors call its signature. It turns out that the fact that the relation $\sim$ is an equivalence relation can be deduced in a slightly more elementary way from the definition 3' of a tree-like loop and, since this definition was not considered in \cite{HamblyLyons}, we take a moment to give the argument. The point is of course that 3' allows one to see $\sim$ as a relation of homotopy, of which we know that it is an equivalence relation.

\begin{proof}[Alternative proof of the fact that $\sim$ is an equivalence relation] That the relation $\sim$ is reflexive and symmetric is straightforward. The problem is to prove that it is transitive. Let us assume that three loops $l_{1}$, $l_{2}$ and $l_{3}$ are such that $l_{1}\sim l_{2}$ and $l_{2}\sim l_{3}$. Then, in the union of the ranges of $l_{1}$, $l_{2}$ and $l_{3}$, the loops $l_{1}$ and $l_{2}$ are homotopic, as well as $l_{2}$ and $l_{3}$. Hence, $l_{1}$ and $l_{3}$ are homotopic, so that $l_{1}\sim l_{3}$.
\end{proof}

Note that the equivalence between the characterisations 1' and 3 of tree-like loops holds for all loops whose range has topological dimension 1. Since the topological dimension is smaller than the Hausdorff dimension (see \cite[Theorem VII 2]{Hurewicz}), this holds for loops whose range has Hausdorff dimension strictly smaller than 2, in particular for loops with finite $p$-variation for $p\in [1,2)$. It even holds for loops whose range has Hausdorff dimension 2, provided the measure of their range is zero. Thus, the following result holds.

\begin{proposition}\label{tree like p 2} Let $(E,d)$ be a metric space. On the space of continuous loops $l:S^{1}\to E$ such that the $2$-dimensional Hausdorff measure of $l(S^{1})$ is equal to zero, the following relation is an equivalence relation:
\[l_{1}\sim l_{2} \mbox{ if and only if } l_{1}l_{2}^{-1} \mbox{ is tree-like},\]
in the sense of the property 1' of Theorem \ref{thin tree like}.

In particular, $\sim$ is an equivalence relation on the set of all continuous loops which have finite $p$-variation for some $p<2$.
\end{proposition}

Let us go back to our initial setting where we consider Lipschitz continuous loops in $\R^{2}$. With the help of the very natural equivalence relation which we described, we may proceed in the same way as in the discrete setting and consider the quotient of $\Loop_{0}(\R^{2})$ equipped with the operation of concatenation. Equivalently, we may consider the group $\RL_{0}(\R^{2})$ of reduced loops with the operation of concatenation and reduction. Note that, contrary to what happens in the discrete case, and perhaps also to intuition, this group is not a free group. Indeed, it contains a subgroup isomorphic to the fundamental group of a topological space called the Hawaiian earring \cite{Hawaii}, which is known not to be free. Since, by a classical theorem of Nielsen and Schreier, any subgroup of a free group is free, the group $\RL_{0}(\R^{2})$ is not free.

It would be very desirable at this point to possess a nice structure of topological group on $\RL_{0}(\R^{2})$. Unfortunately, we do not know how to define such a structure. We shall therefore content ourselves with the bare algebraic structure.

As an appendix to this section, we discuss a particular point of Fort's proof of Theorem \ref{fort} which we found unsatisfactory, and give an alternative argument. Fort's proof rests on the following lemma, in which $S^{1}$ is seen as the boundary of the unit disk $D$ of the plane $\R^{2}$, which Fort denotes by $P$.

\begin{lemma}[\cite{Fort}] Let $Y$ be a metric space. If $f:S^{1}\to Y$ is inessential, then there exists a continuous extension $F:D\to Y$ of $f$ such that none of the components of the inverse sets $F^{-1}(y), y\in Y$, separates the plane $P$.
\end{lemma}

The condition on $F$ is thus that for each $y\in Y$, each connected component of $F^{-1}(y)$ has a connected complement in $\R^{2}$.

\begin{proof} Let $g:D\to Y$ be a continuous extension of $f$, which exists because $f$ is inessential. Let $A$ be the set of connected components of sets of the form $f^{-1}(y)$, $y\in Y$. Then $A$ is a partition of $D$ by connected compact subsets. It is moreover upper semi-continuous, in the sense that the union of the elements of $A$ which meet any given closed subset of $D$ is again a closed subset of $D$. Upper semi-continuous decomposition of continua are defined and studied by Kuratowski in \cite[IV, \textsection 39, V]{Kuratowski}. Fort is also using here the assertion V, \textsection 42, VI, 8 of the same treatise, which originally is due to Eilenberg \cite{Eilenberg}.

Fort defines a partial order on $A$ by setting $a<b$ if $b$ separates $a$ from infinity. His idea is to consider the set of maximal elements of $A$ and, for each such maximal element $m$, to replace the function $g$ by the function which is constant on the union of $m$ and all the bounded connected components of its complement, equal there to the unique value of $g$ on $m$.

The gap lies in the proof that every element of $A$ is dominated by a maximal element. Fort proves that the set $\{b\in A : a\leq b\}$ is totally ordered, but then appeals to Zorn's lemma to produce a maximal element of this set. This is unfortunately not right, since Zorn's lemma applied to a totally ordered set reduces to the tautological statement that the set admits a maximal element if it admits a maximal element. It is nevertheless true that $\{b\in A : a\leq b\}$ admits a maximal element for each $a\in A$, as we shall now prove. 

Let us introduce some notation and make a few observations. For each $b\in A$, let us follow Fort in denoting by $b^{*}$ the union of $b$ and the bounded components of $\R^{2}\setminus b$. The boundary of $b^{*}$ is a subset of $b$. Hence, if $a$ and $b$ are distinct elements of $A$, then the boundaries of $a^{*}$ and $b^{*}$ are disjoint. Since $a^{*}$ and $b^{*}$ are connected, this implies that they are either included one in the other, or disjoint. Moreover, $a^{*}\subset b^{*}$ if and only if $a\leq b$. Finally, the definition of the partial order on $A$ can usefully be reformulated as follows: we have $a\leq b$ if and only if every closed connected subset $\gamma$ of $D$ which meets both $a$ and $S^{1}$ also meets $b$.

Let us choose $a\in A$. The set $C=\bigcup_{a\leq b} b^{*}$ contains $a$ and is contained in $D$. It is thus neither empty nor equal to $\R^{2}$, and its boundary is not empty. Let us choose a point $x\in \partial C$. Let $c$ be the element of $A$ which contains $x$. We claim that $c$ is the greatest element of $\{b\in A : a\leq b\}$.

To start with, let $U$ be an open set containing $c$. Since $A$ is upper semi-continuous, there exists an open subset $V$ of $U$ which contains $c$ and which is a union of elements of $A$. The set $V$, being a neighbourhood of $x$, meets $C$ and hence contains an element $b$ of $A$ such that $a\leq b$.

Let $\gamma$ be a closed connected set which meets both $a$ and $S^{1}$. Since every neighbourhood of $c$ contains an element $b$ of $A$ such that $a\leq b$, every neighbourhood of $c$ meets $\gamma$. Since $c$ and $\gamma$ are closed, this implies that $c$ meets $\gamma$. Hence, $a\leq c$. 

If there existed $b\in A$ such that $c<b$, then $x$ would belong to the interior of $b^{*}$, hence to the interior of $C$, and this is not the case. Thus, $c$ is the greatest element of $\{b\in A : a\leq b\}$.

We proved that each point $x$ of $D$ is contained in $c_{x}^{*}$ for a unique maximal element $c_{x}$ of $A$. Fort defines $F:D\to Y$ by setting, for all $x\in D$, $F(x)$ equal to the unique value of $g$ on $c_{x}$. Fort claims essentially without proof that $F$ is continuous, and since we are reviewing his argument, we complete this point too. 

Let $x$ be a point of $D$. Let $c$ be the maximal element of $A$ such that $c^{*}$ contains $x$. If $x$ belongs to the interior of $c^{*}$, then $F$ is constant in a neighbourhood of $x$, hence continuous at $x$. Otherwise, $x$ belongs to the boundary of $c^{*}$, hence to $c$, so that $F(x)=g(x)$. Let then $W$ be a neighbourhood of $F(x)$. Since $W$ is a neighbourhood of $g(x)$ and $g$ is continuous, there exists an open neighbourhood $U$ of $x$ such that $g(U)\subset W$. Moreover, we may, and do, choose $U$ connected. We claim that $F(U)\subset g(U)$. 

Indeed, let $z$ be a point of $U$. Let $d$ be the maximal element of $A$ such that $d^{*}$ contains $z$. The set $U$, containing $z$, meets $d^{*}$. On the other hand, $x$ is not included in the interior of $d^{*}$, so that $U$ also meets the complement of $d^{*}$. Therefore $U$, being connected, meets the boundary of $d^{*}$, hence $d$ itself. It follows that $F(z)$, which is the value taken by $g$ on $d$, belongs to $g(U)$.

This proves our claim, and the fact that $F$ is continuous. The reason why $F$ is an extension of $f$, given by Fort, is that any element of $A$ which contains a point of $S^{1}$ is maximal.
\end{proof}

\subsection{The master field as a free process}

The discussion of the previous section provides us with a natural algebra on which the master field is defined, namely the algebra $\C[\RL_{0}(\R^{2})]$ of the group of reduced rectifiable loops based at the origin.

Not only can we restrict the function $\Phi$ defined by Theorem \ref{existence Phi} to the set of reduced loops, but it is in fact compatible with the equivalence of paths.

\begin{lemma} Let $l_{1},l_{2} \in \Loop_{0}(\R^{2})$ be two loops based at the origin. Assume that $l_{1}\sim l_{2}$. Then for all $\K\in \{\R,\C,\H\}$ and all $N\geq 1$, the equality $H^{\K}_{N,l_{1}}=H^{\K}_{N,l_{2}}$ holds almost surely. In particular, for all $\K$ and all $N\geq 1$, $\Phi^{\K}_{N}(l_{1})=\Phi^{\K}_{N}(l_{2})$, and $\Phi(l_{1})=\Phi(l_{2})$.
\end{lemma}

\begin{proof} The second assertion follows immediately from the first. By the multiplicativity of the Yang-Mills field, the first assertion will be proved if we show that for all tree-like loop $l$, the random matrix $H^{\K}_{N,l}$ is almost surely equal to the identity matrix for all $N\geq 1$. 

In a first step, let us consider a tree-like loop $l$ traced in a graph $\G$, based at some vertex $v_{0}$. We need to show that $l$ is combinatorially equivalent to a constant loop. We know that $l$ is homotopic to a constant loop within its own range, hence within the skeleton of $\G$. The description of the group $\RL_{v_{0}}(\G)$ as the free group over a set of facial lassos shows that this group is isomorphic to the fundamental group of $\Sk(\G)$. Hence, $l$ is equal to $1$ in this group, which means that it is combinatorially equivalent to a constant loop. Then, the multiplicativity of the Yang-Mills field entails that $H^{\K}_{N,l}=I_{N}$ almost surely for all $N\geq 1$. The conclusion of this paragraph applies in particular to any piecewise affine tree-like loop.

In a second step, let us consider a tree-like loop $l$, without any further assumption. We claim that $l$ is the limit of a sequence of piecewise affine tree-like loops. In order to prove this, let us consider a factorisation $l=g\circ f$ through an $\R$-tree $T$. For each $n\geq 1$, consider a finite subset of $T$ whose $2^{-n}$-neighbourhood covers $T$ and let $T_{n}$ be the convex hull of this subset, which is a finite sub-tree of $T$. Construct $\tilde g_{n}:T_{n}\to \R^{2}$ as the unique mapping which coincides with $g$ on the vertices of $T_{n}$ and is affine on each edge of $T_{n}$. Finally, let $p_{n}:T\to T_{n}$ denote the retraction which is the identity on $T_{n}$ and sends each connected component of $T\setminus T_{n}$ onto the unique point of $T_{n}$ which belongs to the closure of this component. Let us define, for all $x\in T$, $g_{n}(x)=\tilde g_{n}(p_{n}(x))$. Then it is not difficult to check that $g_{n}\circ f$ is piecewise affine and converges to $l$ as $n$ tends to infinity. Hence, by continuity of the Yang-Mills field for fixed $N\geq 1$, we have $H^{\K}_{N,l}=I_{N}$ almost surely for all $N\geq 1$. This concludes the proof.
\end{proof}

According to this lemma, the functions $\Phi^{\K}_{N}$ and $\Phi$ descend to functions on the quotient $\Loop_{0}(\R^{2})/\sim$, or on $\RL_{0}(\R^{2})$. We still denote these functions by $\Phi^{\K}_{N}$ and $\Phi$. It follows from Theorem \ref{existence Phi} that, on the complex involutive unital algebra $\C[\RL_{0}(\R^{2})^{op}]$, the sequence of states $(\Phi^{\K}_{N})_{N\geq 1}$ converges pointwise to $\Phi$, which is also a state.

We can thus define the master field as a free process.

\begin{definition} \label{def mf} Let $\C[\RL_{0}(\R^{2})^{op}]$ be the complex group algebra of the opposite group of reduced rectifiable loops on $\R^{2}$ endowed with the operation of concatenation-reduction. Let $\Phi$ be the linear form on this algebra characterised by the equality
\begin{equation}\label{def phi def}
\forall l \in \RL_{0}(\R^{2}), \;  \Phi(l)=\lim_{N\to\infty} \E_{\YM_{\U(N)}}\left[\tr (H_{N,l}^{\C})\right].
\end{equation}

On the non-commutative space $(\C[\RL_{0}(\R^{2})^{op}],\Phi)$, define the process $\{h_{l}:l\in \Loop_{0}(\R^{2})\}$ by letting, for all $l\in \Loop_{0}$, the non-commutative random variable $h_{l}$ be the image of $l$ by the composed mapping $\Loop_{0}(\R^{2})\to \RL_{0}(\R^{2}) \to \C[\RL_{0}(\R^{2})^{op}]$. In other words, $h_{l}$ is the unique reduced loop equivalent to $l$, seen as an element of the group algebra of the group of reduced loops.

The process $(h_{l})_{l\in \Loop_{0}(\R^{2})}$ is called the {\em master field} on the plane.
\end{definition}

We can state the main theorem of the present work in its final form. 

\begin{theorem}\label{main mf} Choose $\K\in \{\R,\C,\H\}$. For each $N\geq 1$, consider the Yang-Mills field on the plane $\R^{2}$ with structure group $\U(N,\K)$, associated to the Lebesgue measure on $\R^{2}$ and the scalar product on $\u(N,\K)$ given by \eqref{normalization}. This is a process $(H_{N,l}^{\K})_{l\in \Loop_{0}(\R^{2})}$ with values in $\U(N,\K)$. Consider this process as a non-commutative process with respect to the state $\E\otimes \tr$ if $\K\in \{\R,\C\}$ and $\E\otimes \Re\tr$ if $\K=\H$.

1. As $N$ tends to infinity, the Yang-Mills field converges in non-commutative distribution towards the master field on the plane. 

2a. If $\G$ is a graph and $\{\lambda_{F} : F\in \F^{b}\}$ is a lasso basis of the group of reduced loops in $\G$ (see Section \ref{grouploops}), then the non-commutative random variables $\{h_{\lambda_{F}} : F\in \F^{b}\}$ are free, and each random variable $h_{\lambda_{F}}$ has the distribution of a free unitary Brownian motion at time $|F|$.

2b. If $l_{1}$ and $l_{2}$ are two loops, then $h_{l_{1}l_{2}}=h_{l_{2}}h_{l_{1}}$.

2c. The process $\{h_{l} : l\in \Loop_{0}(\R^{2})\}$ is continuous in the $L^{2}$ topology. This means that if a sequence of loops $(l_{n})_{n\geq 0}$ converges to a loop $l$, then $\Phi((h_{l_{n}}-h_{l})(h_{l_{n}}-h_{l})^{*})$ tends to $0$ as $n$ tends to infinity.
More generally, if the sequence $(l_{n})_{n\geq 0}$ converges to $l$, then for all integer $q\geq 1$, all loops $m_{1},\ldots,m_{q}$ and all word $w\in \Fr_{q+2}$ in $q+2$ letters, the following convergence holds:
\[\lim_{n\to\infty} \Phi(w(l_{n},l_{n}^{-1},m_{1},\ldots,m_{q}))= \Phi(w(l,l^{-1},m_{1},\ldots,m_{q})).\]

3. The properties 2a, 2b and 2c characterise the distribution of the master field.

4. The function $\Phi:\Loop_{0}(\R^{2})\to \C$ determined by $\Phi(l)=\Phi(h_{l})$ satisfies $\Phi(l^{-1})=\Phi(l)$ for all loop $l$, takes its values in the real segment $[-1,1]$, and is continuous with respect to the convergence in $1$-variation.
\end{theorem}

\begin{proof} 1. This is part of the content of Theorem \ref{existence Phi}.

2a. By Proposition \ref{YM basis}, the random variables $H_{N,l_{1}}^{\C},\ldots, H_{N,l_{n}}^{\C}$ are independent. The claim is thus a consequence of the theorem of Voiculescu \cite{VDN} (see also Theorem \ref{inv asymp free}) which asserts asymptotic freeness for large independent random matrices invariant in distribution by unitary conjugation.

2b. This follows from the very definition of the process $(h_{l})_{l\in \Loop_{0}(\R^{2})}$ and the fact that we consider the group $\RL_{0}(\R^{2})$ with its opposite multiplication.

2c. Assume that $(l_{n})_{n\geq 1}$ converges to $l$. For each $N\geq 1$, $H^{\C}_{N,l_{n}}$ converges in probability to $H^{\C}_{N,l}$, so that
\begin{align*}
\lim_{n\to \infty} \Phi^{\C}_{N}((h_{l_{n}}-h_{l})(h_{l_{n}}-h_{l})^{*})&=2-2\lim_{n\to \infty} \Re \Phi^{\C}_{N}(h_{l_{n}l^{-1}})\\
&=2-2\lim_{n\to \infty} \Re\E[\tr((H^{\C}_{N,l})^{-1}H^{\C}_{N,l_{n}})]=0.
\end{align*}
Since the convergence of $\Phi^{\C}_{N}$ towards $\Phi$ is uniform on the set $\{l_{n}:n\geq 1\} \cup \{l\}$, for the length function is bounded on this set, the convergence holds at the limit when $N$ tends to infinity. 
The last assertion follows from the same argument applied to $w^{op}(H_{N,l_{n}}^{\C},H_{N,l_{n}^{-1}}^{\C},H_{N,m_{1}}^{\C},\ldots,H_{N,m_{q}}^{\C})$.

3.  Properties 2a and 2b characterise the distribution of the master field on the set of loops traced in a graph, hence on the set of piecewise affine loops. Property 2c guarantees that the distribution on $\Loop_{0}(\R^{2})$ is given by the unique extension by continuity of that on $\Aff_{0}(\R^{2})$.

4. By Theorem \ref{real}, the function $\Phi$ is real on $\Aff_{0}(\R^{2})$. It is continuous on $\Loop(\R^{2})$ by the third assertion of the present theorem. It is thus real-valued on $\Loop_{0}(\R^{2})$. Since for all loop $l$ one has $\Phi(l^{-1})=\Phi(l)^{*}$, $\Phi(l^{-1})$ is also equal to $\Phi(l)$. Definition \eqref{def phi def} shows that it is bounded by $1$. 
\end{proof}

\section{Computing the master field}\label{variation area}

From the study of the large $N$ limit of the Yang-Mills field on the Euclidean plane presented in this work, and the substance of which is summarised in Theorems \ref{existence Phi} and \ref{main mf}, it emerges that the master field, which is the limiting object, is completely described by a plain deterministic bounded, real-valued, continuous function $\Phi$ on the set $\Loop_{0}(\R^{2})$ of loops with finite length based at the origin. In this section, we address the following obvious question: given a loop $l$ on the plane, how can we actually compute the real number $\Phi(l)$ ?

We are going to provide several more or less explicit pieces of answer to this question. They all rely on the  fundamental principle that one should see $\Phi(l)$, and its approximations $\Phi^{\K}_{N}(l)$, as functions of the areas of the faces delimited by $l$, and that the information we are looking for can be obtained by differentiating these functions. 

It is clear from this general description that this approach will only work for loops which delimit a finite number of faces. Accordingly, the level at which we address the problem is that of the discrete theory. The answer which we are seeking is thus combinatorial in nature.

The content of the present section is in part guided by the desire to understand at a mathematical level of rigour and to elaborate on previous work of Makeenko and Migdal \cite{MakeenkoMigdal}, Kazakov \cite{Kazakov}, and Kostov \cite{KazakovKostov} on this question.

In a first step, we shall compute in a fairly general framework the derivative of the Yang-Mills measure on a graph with respect to the areas of the faces. Our expressions will involve differential operators on the configuration space of the discrete theory which we will, in a second time, interpret in a combinatorial language. This second step will be meaningful only for a special class of observables known as the Wilson loops, which are on one hand very natural, on the other hand general enough to generate the algebra of all invariant observables, and most importantly general enough to contain the functions which we are interested in, namely the functions $\Phi^{\K}_{N}$.

\subsection{Differential operators on the configuration space}\label{op diff config}

To start with, we introduce some differential operators on the configuration space of the discrete Yang-Mills theory. The computations which we are going to do in the first sections are valid for any structure group. We thus choose a connected compact Lie group $G$, with Lie algebra $\g$.

Let $\G=(\V,\E,\F)$ be a graph. The configuration space $\Conf^{\G}=\M(\Path(\G),G)$ is in a canonical way a smooth manifold through the identification $\Conf^{\G}\simeq G^{\E^{+}}$, regardless of the orientation $\E^{+}$ of $\G$ that we choose.
We are going to define certain vector fields on this manifold. It is tempting to this end to use the Lie group structure inherited from  $G^{\E^{+}}$, but this structure depends on the orientation. In a first time, it is more convenient to use the following description of the configuration space:
\[\Conf^{\G}=\M(\E,G)=\{(h(e))_{e\in \E} : \forall e\in \E, h(e^{-1})=h(e)^{-1}\},\]
as a submanifold of $G^{\E}$.

Let $e\in \E$ be an edge. Let $X$ be an element of the Lie algebra $\g$. We define the vector field $\D_{X}^{e}$ on $\Conf^{\G}$ by setting, for all $h\in \Conf^{\G}$,
\[\left(\D_{X}^{e}\right)(h)=\frac{d}{dt}_{|t=0} h_{t},\] 
where, for all $t\in \R$ and all $e'\in \E$,
\[h_{t}(e')=\left\{\!\!\begin{array}{ll} h(e') & \mbox{if } e'\notin\{e,e^{-1}\}, \\
h(e)e^{tX} & \mbox{if } e'=e, \\
e^{-tX}h(e^{-1}) & \mbox{if } e'=e^{-1}. \\
\end{array}\right.\]
Let us extend slightly this definition. Let $c\in \Path(\G)$ be a path which ends at the starting point of $e$. We define the vector field $\D_{X}^{c,e}$ by setting, for all $h\in \Conf^{\G}$,
\[\left(\D_{X}^{c,e}\right)(h)=(\D_{\Ad(h(c))X}^{e})(h).\]
In vague but perhaps more intuitive terms, the field $\D^{e}_{X}$ corresponds to the adjunction of an infinitesimal loop with holonomy $X$ at the starting point of the edge $e$. This starting point should however not be understood as the vertex $\underline{e}$, since there may be edges other than $e$ which are issued from $\underline{e}$, but 
the field $\D_{X}^{e}$ does not affect the configuration on these other edges. Let us rather imagine that an infinitesimal loop with holonomy $X$ is inserted {\em at the very beginning} of the edge $e$ (see Figure \ref{decx}). Similarly, the field $\D^{c,e}_{X}$ varies the current configuration by inserting, at the very beginning of $e$, a loop formed by the path $c^{-1}$ followed by an infinitesimal loop with holonomy $X$ and then the path $c$.

\begin{figure}[h!]
\begin{center}
\scalebox{0.9}{\includegraphics{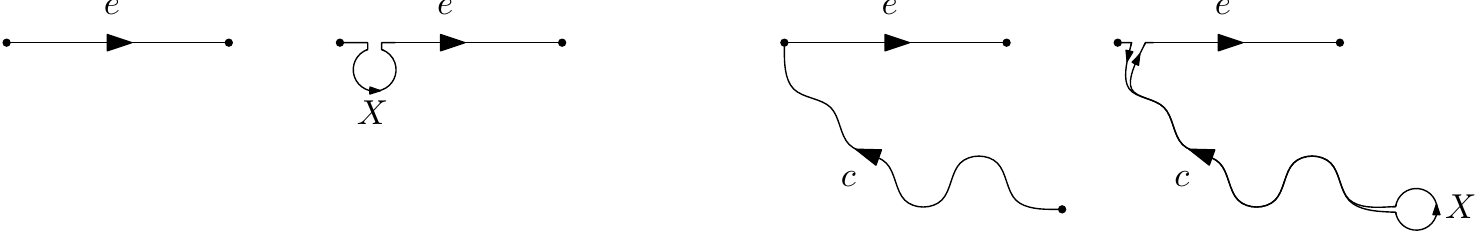}}
\caption{\label{decx} The action of the vector fields $\D^{e}_{X}$ (on the left) and $\D^{c,e}_{X}$ (on the right).}
\end{center}
\end{figure}

Note that for all $e\in \E$, the mapping $\g\to\mathcal X(\Conf^{\G})$ into the Lie algebra of smooth vector fields on $\Conf^{\G}$ which sends $X$ to $\D_{X}^{e}$ is linear and a homomorphism of Lie algebras.

Let us assume that $\g$ is endowed with an invariant scalar product, which we denote by $\langle \cdot, \cdot \rangle$. Let $f:\Conf^{\G}\to \C$ be a smooth function. We define the gradient at $e$ of $f$ by choosing an orthonormal basis $(X_{1},\ldots,X_{d})$ of $\g$ and setting, for all $h\in \Conf^{\G}$,
\[\left(\nabla^{e}f\right)(h)=\sum_{k=1}^{d} \left((\D_{X_{k}}^{e}f)(h)\right) X_{k}.\]
The gradient thus defined does not depend on the choice of the orthonormal basis of $\g$, for example because it is a linear function of the Casimir element of $\g$ (see Section \ref{section casimir}). Similarly, we define
\[\left(\nabla^{c,e}f\right)(h)=\sum_{k=1}^{d} \left((\D_{X_{k}}^{c,e}f)(h)\right) X_{k}.\]
Both $\nabla^{e}f$ and $\nabla^{c,e}f$ are smooth $\g$-valued functions on $\Conf^{\G}$. They are related by
\[\left(\nabla^{c,e}f\right)(h)=\Ad(h(c))^{-1} \left((\nabla^{e}f)(h)\right).\]
In particular, using the invariance of the scalar product on $\g$, we deduce from this equality that if $f_{1}$ and $f_{2}$ are smooth functions on $\Conf^{\G}$, if $e_{1}$ and $e_{2}$ are two edges of $\G$ and if $d_{1}$ and $d_{2}$ are two paths which join a same vertex to the starting points of $e_{1}$ and $e_{2}$ respectively, then
the equality
\begin{equation}\label{chemin grad grad}
\left\langle  \nabla^{d_{1},e_{1}}f_{1}, \nabla^{d_{2},e_{2}} f_{2}\right\rangle=\left\langle \nabla^{d_{2}^{-1}d_{1},e_{1}} f_{1}, \nabla^{e_{2}}f_{2}\right\rangle
\end{equation}
holds pointwise on $\Conf^{\G}$.

Let us now define second-order differential operators. Let $e_1,e_2$ be two edges of $\G$. Let $c\in\Path(\G)$ be a path which joins the starting point of $e_2$ to the starting point of $e_1$. Let $(X_1,\ldots,X_{d})$ be an orthonormal basis of $\g$. We define
\[\Delta^{e_2;c,e_1} = \sum_{k=1}^{d} \D_{X_k}^{e_2} \D_{X_k}^{c,e_1}.\]
If $e_{1}$ and $e_{2}$ are issued from the same vertex and $c$ is the constant path at this vertex, we write $\Delta^{e_{2};e_{1}}=\Delta^{e_{2};c,e_{1}}$. If moreover $e_1=e_2=e$, we write
\[\Delta^{e}=\Delta^{e_2;c,e_{1}}= \sum_{k=1}^{d} \left(\D_{X_k}^{e}\right)^2 .\]
As before, none of these definitions depend on the choice of the orthonormal basis of $\g$. Let us however emphasise that the order of the derivatives in the definition of $\Delta^{e_{2};c,e_{1}}$ matters, since in general,
\[\sum_{k=1}^{d} \D^{e_{2}}_{X_{k}}\D^{c,e_{1}}_{X_{k}} \neq \sum_{k=1}^{d} \D^{c,e_{1}}_{X_{k}} \D^{e_{2}}_{X_{k}},\]
unless $e_{1}\neq e_{2}$ and the path $c$ does not traverse the edge $e_{2}$. 

We have defined the differential operators $\D^{c,e}_{X}$, $\nabla^{c,e}$ and $\Delta^{e_{2};c,e_{1}}$ on the configuration space $\Conf^{\G}$ seen as a submanifold of $G^{\E}$. It is however usually simpler, when one is computing on the configuration space, to choose an orientation $\E^{+}$ of $\G$ and to use the identification $\Conf^{\G}\simeq G^{\E^{+}}$. Let us write down the definition of our differential operators in this language. It is enough to write the definition of $\D^{e}_{X}$, since all others are built from this one. 

Let $\E^{+}$ be an orientation of $\G$. For each $e\in \E^+$ and all $X\in \g$, let us denote by $X^{e}$ the element $(0,\ldots,0,X,0, \ldots,0)$ of $\g^{\E^+}$ whose only possibly non-zero component is that corresponding to the edge $e$ and is equal to $X$. Let $f:G^{\E^{+}}\to \C$ be a smooth observable. Let $h\in \G^{\E^{+}}$ be a configuration. Let $e\in \E$ be an edge, and $X$ an element of $\g$. If $e$ belongs to $\E^{+}$, we have
\[\left(\D_X^{e} f \right)(h)=\frac{d}{dt}_{|t=0} f\left(he^{tX^{e}}\right),\]
and if $e^{-1}$ belongs to $\E^{+}$, then
\begin{equation} \label{dex}
\left(\D_X^{e} f \right)(h)=\frac{d}{dt}_{|t=0} f\big(e^{-tX^{e^{-1}}}h\big).
\end{equation}

Let us collect some properties of the differential operators which we have just defined and which we will need in the proof of Proposition \ref{main deriv}. We denote, as in Section \ref{dymf}, by $\Delta$ the Laplace operator on $G$, and by $(Q_{t})_{t>0}$ the associated heat kernel.

\begin{lemma}\label{rules} 1. Let $f_1,f_2:G^{\E^+} \to \R$ be two smooth functions. Let $e$ be an edge of $\G$. We have, for all $X\in \g$,
\[\int_{G^{\E^+}} f_1(h) (\D_X^{e} f_2)(h) \; dh = - \int_{G^{\E^+}} (\D_X^{e} f_1)(h) f_2(h) \; dh.\]
In particular, 
\[\int_{G^{\E^+}} f_1(h) (\Delta^{e} f_2)(h) \; dh = \int_{G^{\E^+}} (\Delta^{e} f_1)(h) f_2(h) \; dh.\]

2. Let $e$ be an edge of $\G$. Let $c\in \Path(\G)$ be a path in $\G$ from the finishing point of $e$ to its starting point, such that $c$ does not traverse $e$ nor $e^{-1}$. Let $l$ be the loop $ec$. Let $q:G\to\R$ be a smooth function invariant by conjugation. 

The two functions $h\mapsto (\Delta q)(h(l))$ and $\Delta^{e}\left(h\mapsto q(h(l)) \right)$ are equal and the two functions $h\mapsto (\Delta q)(h(l^{-1}))$ and $\Delta^{e}\left(h\mapsto q(h(l^{-1})) \right)$ are also equal.

In particular, if $e$ is an edge which bounds a face $F$, whether positively or negatively, then for all $t>0$, 
\[ (\Delta Q_{t})(h(\partial F)) = \Delta^{e}\left(h\mapsto Q_{t}(h (\partial F)) \right).\]

3. Let $F$ be a face of $\G$. Let $e$ and $e'$ be two edges which bound $F$, respectively positively and negatively. Let $c$ be the portion of the boundary of $F$ which joins the starting point of $e'$ to the starting point of $e$ (see Figure \ref{face3} below). Let $X$ be an element of $\g$. Then for all $t>0$,
\[\D_{X}^{e} \left(h\mapsto Q_{t}(h(\partial F))\right)=-\D_{X}^{c^{-1},e'}\left(h\mapsto Q_{t}(h(\partial F))\right).\]
More generally, if $d$ is a path which starts from the starting point of $e$, then
\[\D_{X}^{d^{-1},e} \left(Q_{t}(h(\partial F))\right)=-\D_{X}^{(cd)^{-1},e'}\left(Q_{t}(h(\partial F))\right).\]
\end{lemma}

\begin{figure}[h!]
\begin{center}
\includegraphics{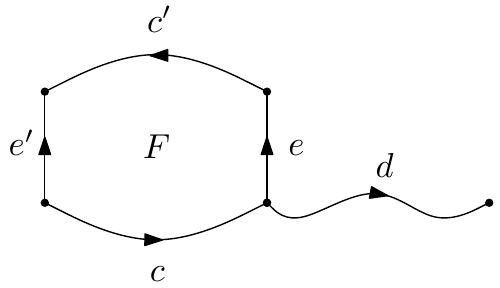}
\caption{\label{face3} The paths involved in the third assertion of Lemma \ref{rules}.}
\end{center}
\end{figure}

\begin{proof} 1. The operator $\D_X^{e}$ satisfies the Leibniz rule. Hence, the first set of assertions is a consequence of the fact that for all smooth function $f:G^{\E^+} \to \R$, one has 
\[\int_{G^{\E^+}} (\D_{X}^{e} f)(h) \; dh=0.\]
This equality in turn follows from the fact that the Haar measure on $G^{\E^+}$ is invariant by the flow of $\D^{e}_{X}$, which is a flow of translations, on the right if $e\in \E^{+}$ and on the left if $e^{-1}\in \E^{+}$.

2. Since $q$ is invariant by conjugation, so is $\Delta q$. We thus have
\begin{align*}
(\Delta q)(h(l))&=(\Delta q)(h(ec))=(\Delta q)(h(c) h(e)) \\
&=\sum_{k=1}^{d} \frac{d^2}{dt^2}_{|t=0} q\left(h(c)h(e)e^{tX_k}\right) =\Delta^{e} (q(h(ec))).
\end{align*}
Similarly,
\begin{align*}
(\Delta q)(h(l^{-1}))&=(\Delta q)(h(e)^{-1} h(c)^{-1}) =\sum_{k=1}^{d} \frac{d^2}{dt^2}_{|t=0} q\left(e^{tX_k} h(e)^{-1}h(c)^{-1}\right) \\
&=\sum_{k=1}^{d} \frac{d^2}{dt^2}_{|t=0} q\left(e^{-tX_k} h(e)^{-1}h(c)^{-1}\right) =\Delta^{e} (q(h(l^{-1}))).
\end{align*}

3. The first assertion follows from the second, by taking $d$ to be the constant path at the starting point of $e$. Let us write $\partial F=ec'(e')^{-1}c$, with $c'$ the appropriate path (see Figure \ref{face3} above). We have
\begin{align*}
\D^{d^{-1},e}_{X} \left(Q_{t}(h(\partial F))\right)  &= \frac{d}{ds}_{|s=0} Q_{t}\left( h(c) h(e')^{-1}h(c')h(e)e^{s\Ad(h(d^{-1})) X} \right)\\
&= \frac{d}{ds}_{|s=0} Q_{t}\left(e^{s\Ad(h((cd)^{-1})) X}h(e')^{-1}h(c')h(e) h(c)\right)\\
&=\D^{e'}_{-\Ad(h((cd)^{-1})) X}\left(Q_{t}(h(\partial F))\right)\\
&= - \D^{(cd)^{-1},e'}_{X}\left(Q_{t}(h(\partial F))\right),
\end{align*}
as expected.
\end{proof}

Let us emphasise that the operator $\D_{X}^{c,e}$ does not satisfy in general a formula of integration by parts analogous to the one satisfied by $\D_{X}^{e}$. 
More precisely, it satisfies such a formula only when it is applied to observables which are invariant under the action of the gauge group. 

Let us describe how the differential operators which we have defined are transformed by the action of the gauge group (see the end of Section \ref{grouploops}).

\begin{lemma} \label{invariance op} Let $f:\Conf^{\G}_{G}\to \C$ be a smooth function. Let $e_{1}$ and $e_{2}$ be two edges of $\G$. Let $c$ be a path joining the starting point of $e_{2}$ to the starting point of $e_{1}$.\\
\indent 1. For all $j\in G^{\V}$ and all $X\in \g$, the following equality holds:
\begin{equation}\label{covariance D}
j\cdot\left(\D^{c,e_{1}}_{X} (j^{-1}\cdot f)\right)=\D^{c,e_{1}}_{\Ad(j(\underline{c})^{-1})X} f.
\end{equation}
\indent 2. The operator $\Delta^{e_{2};c,e_{1}}$ is invariant. In other words, for all $j\in G^{\V}$, the following equality holds:
\[j\cdot\left(\Delta^{e_{2};c,e_{1}} (j^{-1}\cdot f)\right)=\Delta^{e_{2};c,e_{1}} f.\]
\end{lemma}

The proof of this lemma is a straightforward application of the definitions and we leave it to the reader. The formula of integration by parts for the operator $\D^{c,e}$ is the following.

\begin{proposition} Let $f:\Conf^{\G}_{G}\to \C$ be a smooth invariant function. Let $e$ be an edge of $\G$. Let $c$ be a path finishing at the starting point of $e$. Let $X$ be an element of $\g$. The following equality holds:
\[\int_{G^{\E^{+}}} \left(\D^{c,e}_{X}f\right)(h)\; dh=0.\]
\end{proposition}

\begin{proof} We are going to average the equality \eqref{covariance D} over the gauge group, which is a compact Lie group. 

The invariance of $f$ and \eqref{covariance D} imply that for all $j\in G^{\V}$, we have
\[\int_{G^{\E^{+}}} \left(\D^{c,e}_{X}f\right)(h)\; dh=\int_{G^{\E^{+}}} \left(\D^{c,e}_{\Ad(j(\underline{c})^{-1})X}f\right)(j^{-1}\cdot h)\; dh.\]
Since the Haar measure on $G^{\E^{+}}$ is invariant under the action of $G^{\V}$, we can replace $j^{-1}\cdot h$ by $h$ in the right-hand side and, averaging over $j$, we find
\[\int_{G^{\E^{+}}} \left(\D^{c,e}_{X}f\right)(h)\; dh=\int_{G^{\E^{+}}\times G^{\V}} \left(\D^{c,e}_{\Ad(j(\underline{c})^{-1})X}f\right)(h)\; djdh,\]
which by linearity of the map $X\mapsto \D^{c,e}_{X}$, is equal to
\[\int_{G^{\E^{+}}} \left(\D^{c,e}_{Z} f\right)(h)\; dh,\]
where we set $Z=\int_{G^{\V}}\Ad(j(\underline{c})^{-1})X \; dj$. We have $Z=\int_{G}\Ad(x)X\; dx$, which belongs to the centre of the Lie algebra $\g$. Hence, for all $h\in G^{\E^{+}}$, $\D^{c,e}_{Z}f(h)=\D^{e}_{\Ad(h(c))Z}f(h)=\D^{e}_{Z}f(h)$, and we finally find
\[\int_{G^{\E^{+}}} \left(\D^{c,e}_{X}f\right)(h)\; dh=\int_{G^{\E^{+}}} \left(\D^{e}_{Z}f\right)(h)\; dh,\]
which is equal to zero by the first assertion of Lemma \ref{rules}.
\end{proof}

\subsection{Variation of the area in the abstract} \label{sec: variation area}

The main result of this section provides us with an expression of the derivative of $\E_{\YM^{\G}}[f]$ with respect to the area of a face of $\G$, in terms of the differential operators which we introduced in the previous section, and without any assumption on the observable $f$. 

Before we state the result, let us give a more formal description of what we mean by this derivative. Let $\G=(\V,\E,\F)$ be a graph. Let $\F^{b}$ be the set of bounded faces of $\G$. For all $t:\F^{b} \to \R^{*}_{+}$, we define the Yang-Mills measure with areas $t$ on the configuration space $\Conf^{\G}$ by the following formula, analogous to \eqref{def YM}:
\begin{equation}\label{def YM t}
\YM^{\G}_{t}(dh)=\prod_{F\in \F^{b}} Q_{t(F)}(h(\partial F)) \; dh.
\end{equation}
We are interested in the partial derivatives of the mapping $t\mapsto  \E_{\YM^{\G}_{t}}[f]$, where we see $t$ as an element of $(\R^{*}_{+})^{\F^{b}}$. Since we have, up to now, denoted by $|F|$ the area of a face $F$, we will use the notation
\[\frac{d}{d|F|} \E_{\YM_{t}^{\G}}[f]=\frac{\partial}{\partial\, t(F)} \E_{\YM_{t}^{\G}}[f],\]
which is lighter and perhaps clearer.

\begin{proposition}\label{main deriv} Let $\G=(\V,\E,\F)$ be a graph. Let $n\geq 1$ be an integer. Let $F_1,\ldots,F_{n+1}$ be a sequence of faces of $\G$. Assume that $F_{1},\ldots,F_{n}$ are bounded faces of $\G$. For all $r\in \{1,\ldots,n\}$, assume also that the faces $F_{r}$ and $F_{r+1}$ are distinct and adjacent, and let $e_{r}$ be an edge which bounds $F_{r}$ negatively and $F_{r+1}$ positively.
For each $r\in \{2,\ldots,n\}$, denote by $c_r$ the portion of the boundary of $F_r$ which joins the starting point of $e_{r}$ to the starting point of $e_{r-1}$. Finally, let $f:\Conf^{\G}_{G} \to \R$ be a smooth function. Then, if $n\geq 2$, we have the following formula~:

\begin{align}\nonumber
\left(\frac{d}{d|F_1|}-\frac{d}{d |F_2|}\right)\E_{\YM^\G_{t}}[f] &= \E_{\YM^\G_{t}} \left[\frac{1}{2} \Delta^{e_1} f + \sum_{i=2}^n \Delta^{e_i;c_{i}\ldots c_{2},e_{1}} f\right] \\
&\hspace{0.5cm}+\E_{\YM^{\G}_{t}}\left[\left\langle \nabla^{e_{n}}\left(h\mapsto \log Q_{t(F_{n+1})}(h(\partial F_{n+1}))\right),\nabla^{c_{n}\ldots c_{2},e_{1}}f \right\rangle\right]\label{standard deriv}
\end{align}
in which the last term of the right-hand side must be replaced by $0$ in the case where $F_{n+1}=F_\infty$.
If $n=1$ and $F_2$ is not the unbounded face, then the same formula holds after dropping the sum over $i$ in the first expectation and replacing $\nabla^{c_{n}\ldots c_{2},e_{1}}$ by $\nabla^{e_{1}}$ in the second. Finally, if $n=1$ and $F_2$ is the unbounded face of $\G$,  then the formula simply reads
\begin{equation}\label{derivative one face 1}
\frac{d}{d|F_1|}\E_{\YM^\G_{t}}[f] = \E_{\YM^\G_{t}} \left[\frac{1}{2} \Delta^{e_1} f\right].
\end{equation}
\end{proposition}

Let us emphasise that in this proposition, as well as in its forthcoming Corollary \ref{main deriv cor}, the sequence of faces $F_{1},\ldots,F_{n}$ is allowed to contain repetitions. We only assumed that each face is different from the next. Let us also emphasise that the result holds without any assumption of gauge-invariance of the observable $f$.

\begin{figure}[h!]
\begin{center}
\scalebox{1}{\includegraphics{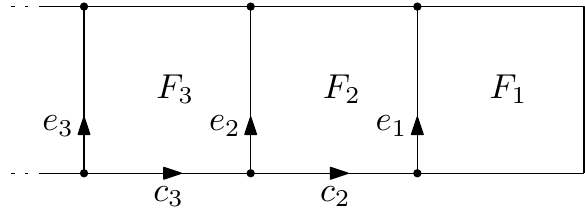}}
\caption{A schematic picture of the paths involved in the differentiation with respect to the area of the face $F_1$.}
\end{center}
\end{figure}

\begin{proof} The proof is a straightforward computation. Let us go through it step by step. 

The first step consists of course in using the heat equation satisfied by the heat kernel $(Q_t)_{t>0}$, together with the second assertion of Lemma \ref{rules}. We find
\begin{align*}
\frac{d}{d|F_1|} \E_{\YM^\G_{t}}[f] &= \frac{d}{d|F_1|} \int_{G^{\E^+}} f(h) \prod_{F\in \F^{b}} Q_{t(F)}(h(\partial F)) \; dh \\
&= \frac{1}{2} \int_{G^{\E^+}} f(h)  (\Delta Q_{t(F_1)})(h(\partial F_1)) \prod_{F\in \F^{b}\setminus\{F_1\}} Q_{t(F)}(h(\partial F))\; dh \\
&= \frac{1}{2} \int_{G^{\E^+}} f(h)  \Delta^{e_1} \left(Q_{t(F_1)}(h(\partial F_1))\right) \prod_{F\in \F^{b}\setminus\{F_1\}} Q_{t(F)}(h(\partial F))\; dh.
\end{align*}
Among all edges which bound $F_1$, we chose the edge $e_1$. Now, we use integration by parts, or equivalently the fact that $\Delta^{e_1}$ is self-adjoint (see Lemma \ref{rules}). If $n=1$ and $F_{2}$ is the unbounded face, then $f(h)$ and $Q_{t(F_{1})}(h(\partial F_{1}))$ are the only factors in the integrand which depend on the edge $e$ and we find
\begin{align*}
\frac{d}{d|F_1|} \E_{\YM^\G}[f] &=
\frac{1}{2} \int_{G^{\E^+}} Q_{t(F_1)}(h(\partial F_1)) \Delta^{e_1} \left( f(h) \prod_{F\in \F^{b}\setminus\{F_1\}} Q_{t(F)}(h(\partial F))\right)  \; dh\\
&=\frac{1}{2} \int_{G^{\E^+}} (\Delta^{e_1} f)(h)  Q_{t(F_1)}(h(\partial F_1)) \prod_{F\in \F^{b}\setminus\{F_1\}} Q_{t(F)}(h(\partial F)) \; dh\\
&=\E_{\YM^{\G}_{t}}\left[\frac{1}{2}\Delta^{e_{1}}f\right],
\end{align*}
which proves the result in this case. If $F_{2}$ is not the unbounded face, then $Q_{|F_{2}|}(h(\partial F_{2}))$ also depends on the edge $e_{1}$, and no other term does. We find, applying the Leibniz rule, 
\begin{align}
\frac{d}{d|F_1|} \E_{\YM^\G_{t}}[f] \nonumber
&=\frac{1}{2} \int_{G^{\E^+}} (\Delta^{e_1} f)(h)  Q_{t(F_1)}(h(\partial F_1)) \prod_{F\in \F^{b}\setminus\{F_1\}} Q_{t(F)}(h(\partial F)) \; dh \\
\nonumber &+ \frac{1}{2} \int_{G^{\E^+}} f(h)  \Delta^{e_1}\left(Q_{t(F_2)}(h(\partial F_2))\right) \prod_{F\in \F^{b}\setminus\{F_2\}} Q_{t(F)}(h(\partial F)) \; dh \\
&+ \sum_{k=1}^{d} \int_{G^{\E^+}} \left(\D_{X_k}^{e_1} f\right)(h) \D^{e_1}_{X_k} \left(Q_{t(F_2)}(h(\partial F_2))\right) \prod_{F\in \F^{b}\setminus\{F_2\}} Q_{t(F)}(h(\partial F)) \; dh. \label{atrans}
\end{align}
We recognised already the first term of this sum as the first part of the first term of the right-hand side of \eqref{standard deriv}. In the second term we recognise, using backwards the second assertion of Lemma \ref{rules}, the derivative of the integral of $f$ with respect to the area of $F_2$. If $n=1$, we thus find
\begin{align*}
\left(\frac{d}{d|F_1|}-\frac{d}{d |F_2|}\right)\E_{\YM^\G_{t}}[f]= &\;  \E_{\YM^\G_{t}} \left[\frac{1}{2} \Delta^{e_1} f\right]\\
&+\E_{\YM^{\G}_{t}}\left[\left\langle \nabla^{e_{1}}\left(h\mapsto \log Q_{t(F_{2})}(h(\partial F_{2}))\right),\nabla^{e_{1}}f \right\rangle\right],
\end{align*}
proving the result in this case.

In order to treat the case where $n>1$, we need to transform the last term of \eqref{atrans}, and we do so by using the third assertion of Lemma \ref{rules}. We find
\begin{align} \nonumber
\frac{d}{d|F_1|} \E_{\YM^\G_{t}}[f] &= \frac{1}{2}\E_{\YM^\G_{t}}[\Delta^{e_1} f ] + \frac{d}{d|F_2|} \E_{\YM^\G_{t}}[f] \\
&- \sum_{k=1}^{d} \int_{G^{\E^+}} \left(\D_{X_k}^{e_1} f\right)(h) \D^{c_{2}^{-1},e_2}_{X_k} \left(Q_{t(F_2)}(h(\partial F_2))\right) \prod_{F\in \F^{b}\setminus\{F_2\}} Q_{t(F)}(h(\partial F)) \; dh, \label{etape F2} 
\end{align}
where we have of course chosen the edge $e_2$ as bounding $F_2$ negatively. Using \eqref{chemin grad grad}, this equality becomes
\begin{align} \nonumber
\frac{d}{d|F_1|} \E_{\YM^\G_{t}}[f] &= \frac{1}{2}\E_{\YM^\G_{t}}[\Delta^{e_1} f ] + \frac{d}{d|F_2|} \E_{\YM^\G_{t}}[f] \\
&- \sum_{k=1}^{N^2} \int_{G^{\E^+}} \left(\D_{X_k}^{c_{2},e_1} f\right)(h) \D^{e_2}_{X_k} \left(Q_{t(F_2)}(h(\partial F_2))\right) \prod_{F\in \F^{b}\setminus\{F_2\}} Q_{t(F)}(h(\partial F)) \; dh. \label{etape F3} 
\end{align}
This form of the equality allows us to move one step forward along the sequence of faces which we are given. For this, we use the first assertion of Lemma \ref{rules} to proceed to an integration by parts with respect to the edge $e_2$, which brings in the term $Q_{t(F_3)}(h(\partial F_3))$:
\begin{align*}
\frac{d}{d|F_1|} \E_{\YM^\G_{t}}[f] &=  \frac{d}{d|F_2|} \E_{\YM^\G_{t}}[f]+\frac{1}{2}\E_{\YM^\G_{t}}[\Delta^{e_1} f ]  \\
&+ \sum_{k=1}^{d} \int_{G^{\E^+}} \left(\D^{e_2}_{X_k} \D_{X_k}^{c_{2},e_1} f\right)(h)  \prod_{F\in \F} Q_{t(F)}(h(\partial F)) \; dh\\
&+\sum_{k=1}^{d} \int_{G^{\E^+}} \left(\D_{X_k}^{c_{2},e_1} f\right)(h) \D^{e_2}_{X_k} \left(Q_{t(F_{3})}(h(\partial F_3))\right) \prod_{F\in \F^{b}\setminus\{F_3\}} Q_{t(F)}(h(\partial F)) \; dh\\
&=  \frac{d}{d|F_2|} \E_{\YM^\G_{t}}[f] + \frac{1}{2}\E_{\YM^\G_{t}}[\Delta^{e_1} f ] + \E_{\YM^\G_{t}}[\Delta^{e_2;c_2,e_1} f ]\\
&-\sum_{k=1}^{d} \int_{G^{\E^+}} \left(\D_{X_k}^{c_{3}c_{2},e_1} f\right)(h) \D^{e_3}_{X_k} \left(Q_{t(F_{3})}(h(\partial F_3))\right) \prod_{F\in \F^{b}\setminus\{F_3\}} Q_{t(F)}(h(\partial F)) \; dh,\\
\end{align*}
where again we used the third assertion of Lemma \ref{rules} and \eqref{chemin grad grad}. 
The last term is similar to the last term of the right-hand side of \eqref{etape F3}, but one step further in the sequence of faces $F_1,\ldots,F_{n+1}$. We can continue this until we reach the end of this sequence: if $F_{n+1}$ is not the unbounded face, a straightforward induction argument finishes the proof. If on the other hand $F_{n+1}=F_\infty$, we still need to observe, as we did in the case where $n=1$, that there are only two terms in the integrand which depend on the edge $e_n$, namely $f(h)$ and $Q_{t(F_n)}(h(\partial F_n))$. Hence, in the last integration by parts, with respect to this edge $e_n$, only one term is produced, which is the integral of $\Delta^{e_n;c_{n}\ldots c_2,e_{1}} f$. This concludes the proof also in this case.
\end{proof}

In the process of computing the derivative of the integral of $f$ with respect to the area of $F_1$, the derivative of the same integral with respect to the area of $F_2$ appeared unexpectedly. We can easily correct this as follows, thus obtaining a generalisation of \eqref{derivative one face 1}.

\begin{corollary} \label{main deriv cor} Recall the notation of Proposition \ref{main deriv}. Let us assume that the face $F_{n+1}$ is the unbounded face of $\G$. Then the following equality holds~:
\begin{align}\label{local derive}
\frac{d}{d|F_1|}\E_{\YM^\G_{t}}[f] &= \E_{\YM^\G_{t}} \left[\frac{1}{2} \sum_{i=1}^{n}\Delta^{e_i} f + \sum_{1\leq i<j\leq n} \Delta^{e_{j};c_{j}\ldots c_{i+1},e_{i}} f\right].
\end{align}
\end{corollary}

\begin{proof} Simply write 
\[\frac{d}{d|F_1|}\E_{\YM^\G_{t}}[f] = \sum_{i=1}^{n-1} \left(\frac{d}{d|F_i|} - \frac{d}{d|F_{i+1}|} \right)\E_{\YM^\G_{t}}[f] + \frac{d}{d|F_n|}\E_{\YM^\G_{t}}[f]\]
and apply Proposition \ref{main deriv} to compute each term.
\end{proof}

\subsection{Outline of the strategy}\label{outline} 
Our goal is to compute $\Phi^{\K}_{N}(l)=\E[\tr (H^{\K}_{N,l})]$ and its large $N$ limit $\Phi(l)$ for each loop $l$ in a certain class to be defined. For this, and given a loop $l$, we will proceed as follows.\\
\indent 1. Consider a graph $\G_{l}$ in which $l$ is traced and see $\Phi^{\K}_{N}(l)$ (or $\Phi(l)$) as a function of the areas of the faces of $\G_{l}$.\\
\indent 2. Find a finite differential system, derivatives being taken with respect to the areas of faces of $\G_{l}$, such that one of the unknown functions of this system is $\Phi^{\K}_{N}(l)$ (or $\Phi(l)$).\\
\indent 3. Solve this differential system and evaluate the solution at the actual areas of the faces of $\G_{l}$.

The first step is easy, provided we choose our loop in an appropriate class. The second step is of course the most delicate one. Our main tool for building a differential system is Corollary \ref{main deriv cor}. We shall apply it to the observable $f:h\mapsto \tr(h(l))$, which is called a Wilson loop. On the right-hand side of \eqref{local derive}, there will appear new observables, which are not necessarily Wilson loops, but polynomials of Wilson loops. We will show that when $f$ is a polynomial in Wilson loops, then the right-hand side of \eqref{local derive} is still a polynomial of Wilson loops. This is a good point, but this does not suffice to ensure that we will be able to write a closed finite differential system. We will achieve this by carefully choosing the sequences of faces which we use in the application of \eqref{local derive}. 

We treat three simple examples in Section \ref{section dwl}. Then, in Sections \ref{wskein} and \ref{wgarland}, we progressively describe the class of observables which is appropriate for our problem, in that it contains Wilson loops and allows us to write finite closed differential systems. We call Wilson garlands the observables of this class, and they are particular polynomials of Wilson loops. Wilson skeins (introduced in Section \ref{wskein}) are an intermediate step on the way to the more complicated Wilson garlands (Section \ref{wgarland}), and they will play an important role in the last sections, where we focus on the function $\Phi$ rather than on its approximations $\Phi^{\K}_{N}$. 

Then, in Section \ref{exp wl}, we write down the differential system, and solve it, thus fulfilling the third step of the program. Solving the system is not at all difficult, because it is a first-order linear differential equation with constant coefficients. However, we need to prescribe some initial value to our solution, and this is something we do in Section \ref{analyticity sn}. In that section, we consider a larger class of observables, called spin networks, for which we prove a result of analyticity of the expectation with respect to the area of the faces. We introduce spin networks in the next section, Section \ref{sec der spin net}, in which we also study how the differential operators which we introduced in Section \ref{op diff config} act on them.

The contents of Sections \ref{section extended gauge} to \ref{sec Kazakov} will be described at the beginning of Section \ref{section extended gauge}.

\subsection{Area derivatives of spin networks}\label{sec der spin net} 

Let us recall the definition of spin networks. Let $\G=(\V,\E,\F)$ be a graph. For each vertex $v$, we define the set $\Out(v)$ as the set of edges issued from $v$:
\[\Out(v)=\{e\in \E : \underline{e}=v\}.\]
In order to define a spin network, we need first to choose a collection $\alpha=(\alpha_{e})_{e\in \E}$ of representations of $G$, acting respectively on the real or complex linear spaces $(V_{e})_{e\in \E}$, and such that for all $e\in \E$, we have $V_{e^{-1}}=V_{e}^{*}$, the dual vector space of $V_{e}$, and $\alpha_{e^{-1}}=\alpha_{e}^{\vee}$, the contragredient representation of $\alpha_{e}$. Once $\alpha$ is chosen, we also need to choose a collection $I=(I_{v})_{v\in \V}$ of tensors such that for all $v\in \V$, the tensor $I_{v}$ belongs to $\bigotimes_{e\in \Out(v)}V_{e}$.

From the data of $\alpha$ and $I$, we build a function $\psi_{\alpha,I}:\Conf^{\G}_{G}\to\C$ as follows. Let us choose an orientation $\E^{+}\subset \E$ of $\G$. Let $h\in \Conf^{\G}_{G}=\M(\Path(\G),G)$ be an element of the configuration space. On one hand, the tensor $\bigotimes_{v\in \V}I_{v}$ belongs to $\bigotimes_{v\in \V}\bigotimes_{e\in \Out(v)} V_{e}$. On the other hand, through the natural identification $\End(V_{e})\simeq V_{e}^{*}\otimes V_{e}$, the tensor $\bigotimes_{e\in \E^{+}} \alpha_{e}(h(e))$ belongs to $\bigotimes_{e\in \E^{+}} V_{e}^{*}\otimes V_{e} \simeq \bigotimes_{e\in \E} V_{e}^{*} \simeq \bigotimes_{v\in \V} \bigotimes_{e\in \Out(v)} V_{e}^{*}$. We define, according to these identifications,
\begin{equation}\label{def spin network}
\psi_{\alpha,I}(h)=\bigg\langle \bigotimes_{e\in \E^{+}} \alpha_{e}(h(e))\; , \; \bigotimes_{v\in \V}I_{v} \bigg\rangle.
\end{equation}
We call spin network on $\G$ any function on $\Conf^{\G}_{G}$ which is of the form above. We shall denote the set of spin networks on $\G$ by $\O_{\G}$. Here, the group $G$ is understood, as well as the choice between real and complex representations of $G$. This should however not cause any confusion.

The set $\O_{\G}$ is a sub-algebra of the algebra $C^{\infty}(\Conf^{\G}_{G})$ of smooth observables. If the structure group $G$ is a compact matrix group, then it is exactly the algebra of polynomial functions on $\Conf^{\G}_{G}$, that is, the algebra of functions which map a configuration $h$ to a polynomial in the entries of the matrices $\{h(e) : e\in \E\}$, or equivalently to a polynomial in the entries of these matrices and their inverses, since $h(e^{-1})=h(e)^{-1}$.

We will need a variant of the definition of a spin network in which not all the pairs $V_{e}^{*}\otimes V_{e}$ which appear in \eqref{def spin network} are contracted. Instead, we contract all the pairs but one, which corresponds to a certain edge $e$. In order to define this properly, let us denote, for each pair $(v,e)\in \V\times \E$ such that $e\in \Out(v)$, by $\Tr_{(v,e)}:V_{e}^{*}\otimes V_{e}\to \C$ the natural contraction. We thus have
\[\psi_{\alpha,I}(h)=\bigg(\bigotimes_{(v,e)}\Tr_{(v,e)}\bigg) \bigg(\bigotimes_{e\in \E^{+}} \alpha_{e}(h(e)) \otimes \bigotimes_{v\in \V}I_{v}\bigg),\]
where the first tensor product is taken over all pairs $(v,e)$ with $e\in \Out(v)$. Given an edge $e\in \E$, we now define
\[\psi_{\alpha,I}^{e}:\Conf^{\G}_{G} \to \End(V_{e})\]
by setting
\[\psi^{e}_{\alpha,I}(h)=\bigg(\bigotimes_{(v,e')\neq (\underline{e},e)}\Tr_{(v,e')}\bigg) \bigg(\bigotimes_{e\in \E^{+}} \alpha_{e}(h(e)) \otimes \bigotimes_{v\in \V}I_{v}\bigg).\]
Similarly, if $e_{1}$ and $e_{2}$ are two distinct edges, we define
\[\psi_{\alpha,I}^{e_{2};e_{1}}:\Conf^{\G}_{G} \to \End(V_{e_{2}})\otimes \End(V_{e_{1}})\]
by not contracting the pair  $V_{e_{2}}^{*}\otimes V_{e_{2}}$ nor the pair $V_{e_{1}}^{*}\otimes V_{e_{1}}$.

We shall also need the following two definitions. Recall that $C_{\g}$ denotes the Casimir element of $\g$, equal to $\sum_{k=1}^{d} X_{k}\otimes X_{k}$ for any choice of an orthonormal basis $\{X_{1},\ldots,X_{d}\}$ of $\g$. Thinking of $C_{\g}$
as an element of the universal enveloping algebra $\mathcal U(\g)$ of $\g$, we can let it act on any representation of $G$. Thus, if $\alpha$ is a representation of $G$ on a vector space $V$, we denote by $\alpha(C_{g})$ the endomorphism $\sum_{k=1}^{d} \alpha(X_{k})^{2}$ of $V$. On the other hand, $C_{\g}$ can also be seen as an element of $\g\otimes \g \subset \mathcal U(\g) \otimes \mathcal U(\g)\simeq\mathcal U(\g\oplus \g)$ and from this point of view it is natural to let it act on any representation of $G\times G$. Accordingly, if $\alpha_{1}$ and $\alpha_{2}$ are two representations of $G$ on $V_{1}$ and $V_{2}$ respectively, we denote by $(\alpha_{2}\otimes \alpha_{1})(C_{\g})$ the endomorphism $\sum_{k=1}^{d} \alpha_{2}(X_{k})\otimes \alpha_{1}(X_{k})$ of $V_{2}\otimes V_{1}$. 

We can now compute the effect on spin networks of the differential operators which we defined in Section \ref{op diff config}. We shall use the notation $\iota_{e}:\End(V_{e})\to \bigotimes_{e'\in \Out(v)} \End(V_{e'})$ and $\iota_{e_{1},e_{2}}:\End(V_{e_{1}})\otimes \End(V_{e_{2}})\to \bigotimes_{e\in \Out(v)} \End(V_{e})$ for the natural operators analogous to $\iota_{i,j}$ defined by \eqref{def iota}.

\begin{proposition}\label{derive spin} Let $\G=(\V,\E,\F)$ be a graph. Let $\psi_{\alpha,I}:\Conf^{\G}_{G}\to \C$ be a spin network. Let $e,e_{1},e_{2}\in \E$ be three edges issued from the same vertex $v$. Assume that $e_{1}\neq e_{2}$. Choose $X\in \g$. The following equalities hold. Firstly,
\[\D^{e}_{X} \psi_{\alpha,I}=\Tr_{V_{e}}\left(\psi_{\alpha,I}^{e} \circ \alpha_{e}(X)\right)=\psi_{\alpha,I'},\]
where $I'_{w}=I_{w}$ for all $w\neq v$, and $I'_{v}=\iota_{e}(\alpha_{e}(X)) (I_{v})$.
Secondly,
\begin{equation}\label{delta e spin}
\Delta^{e}_{X} \psi_{\alpha,I}=\Tr_{V_{e}}\left(\psi_{\alpha,I}^{e} \circ \alpha_{e}(C_{\g})\right)=\psi_{\alpha,I''},
\end{equation}
where $I''_{w}=I_{w}$ for all $w\neq v$, and $I''_{v}=\iota_{e}(\alpha_{e}(C_{\g})) (I_{v})$.
Finally,
\begin{equation}\label{delta spin}
\Delta^{e_{2};e_{1}}\psi_{\alpha,I}=\Tr_{V_{e_{2}}\otimes V_{e_{1}}}\left(\psi^{e_{2};e_{1}}_{\alpha,I} \circ (\alpha_{e_{2}}\otimes \alpha_{e_{1}})(C_{\g})\right)=\psi_{\alpha,I'''},
\end{equation}
where $I'''_{w}=I_{w}$ for all $w\neq v$ and  $I'''_{v}=\iota_{e_{1},e_{2}}((\alpha_{e_{1}}\otimes \alpha_{e_{2}})(C_{\g}))   (I_{v})$.
\end{proposition}
These assertions are illustrated by Figure \ref{spint} below. 

\begin{figure}[h!]
\begin{center}
\scalebox{1}{\includegraphics{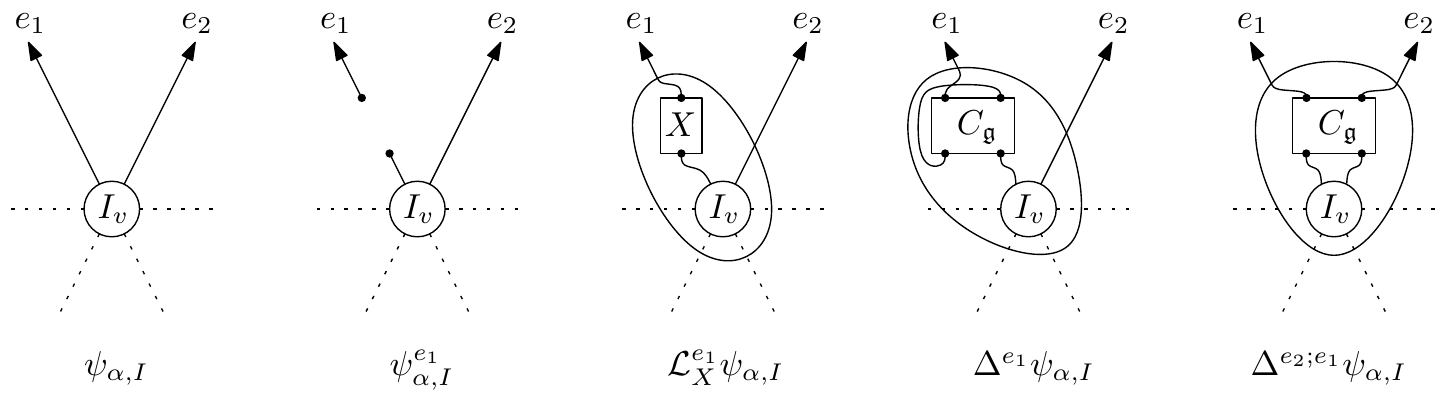}}
\caption{\label{spint} The first picture shows a spin network around the vertex $v$. The representations associated with the edges are not indicated explicitly. In the second picture, the two dots indicate that $\psi_{\alpha,I}^{e_{1}}$ takes its values in $V_{e_{1}}^{*}\otimes V_{e_{1}}$. In the rightmost picture, the Casimir operator $C_{\g}$ acts through the representation $\alpha_{e_{1}}\otimes \alpha_{e_{2}}$.}
\end{center}
\end{figure}

\begin{proof} Let us assume that $\G$ is oriented in such a way that $e\in \E^{+}$. From the definition of $\psi_{\alpha,I}^{e}$ we have, for all $g\in G$ and all $h\in \Conf^{\G}_{G}\simeq G^{\E^{+}}$,
\[\Tr_{V_{e}}\left(\psi_{\alpha,I}^{e} \circ \alpha_{e}(g)\right)=\psi_{\alpha,I}(h'),\]
where $h'\in G^{\E^{+}}$ has all components equal to those of $h$, except for $h'(e)$ which is given by $h'(e)=h(e)g$. Differentiating with respect to $g$ yields the first equality. Differentiating a second time yields the second one. Moreover, for all $X,Y\in \g$, we have
\[\D^{e_{2}}_{Y}\D^{e_{1}}_{X}\psi_{\alpha,I}=\Tr_{V_{e_{2}}\otimes V_{e_{1}}}\left(\psi^{e_{2};e_{1}}_{\alpha,I}\circ (\alpha_{e_{2}}(Y)\otimes \alpha_{e_{1}}(X))\right),\]
from which we deduce the third assertion. 
\end{proof}

Let us generalise our formulas in order to give an expression of $\D^{c,e}_{X}\psi_{\alpha,I}$ and $\Delta^{e_{2};c,e_{1}}\psi_{\alpha,I}$. This requires a construction which is slightly unpleasant to describe verbally, but much more easily explained by a picture: see Figure \ref{spint c} below.

\begin{figure}[h!]
\begin{center}
\scalebox{0.75}{\includegraphics{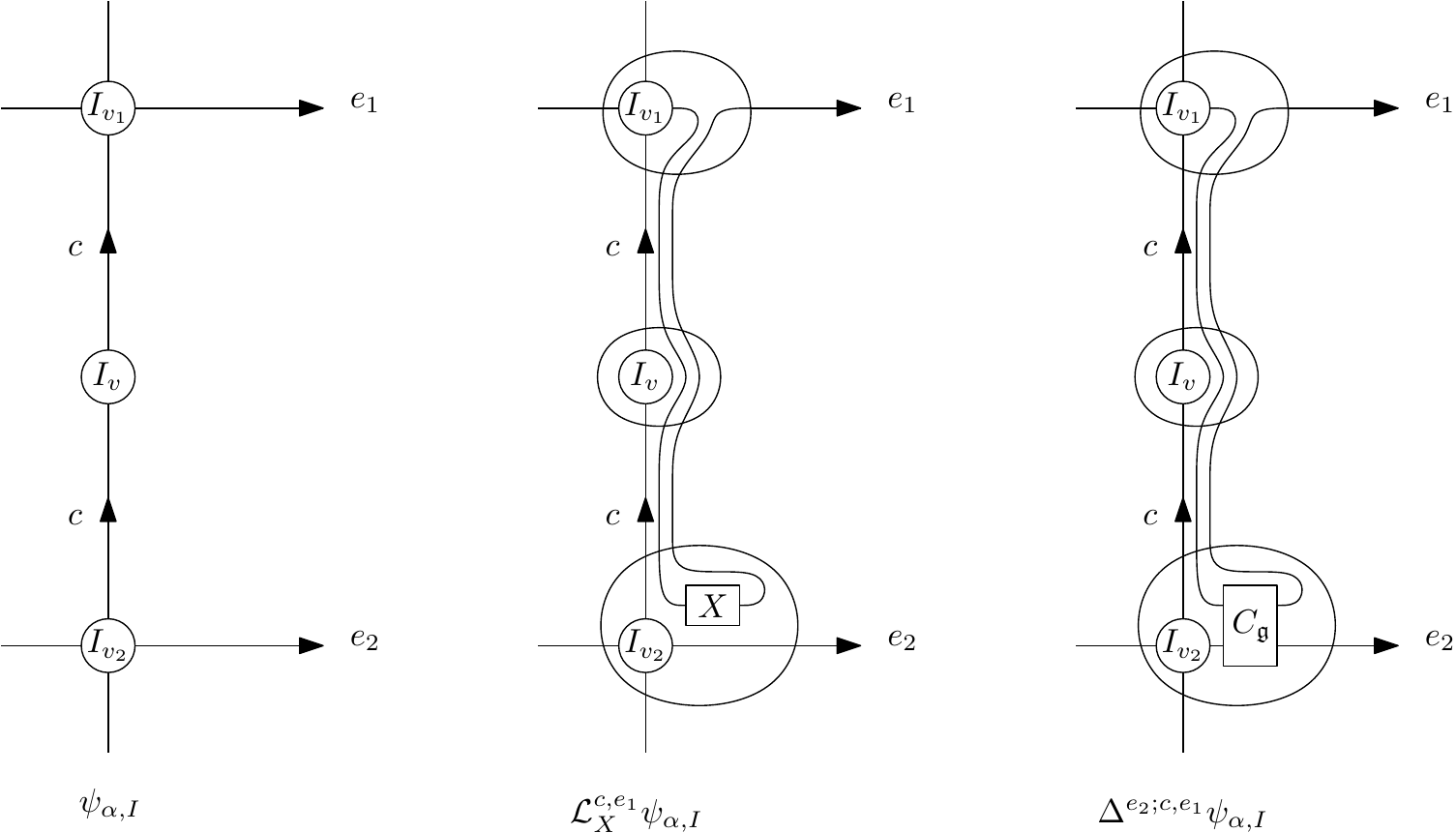}}
\caption{\label{spint c} As in Figure \ref{spint}, we do not indicate the representations explicitly. In this example, the path $c$ is constituted by two edges and crosses a vertex $v$ between $v_{1}$ and $v_{2}$.}
\end{center}
\end{figure}

Let $\psi_{\alpha,I}:\Conf^{\G}_{G}\to \C$ be a spin network. Let $e_{1},e_{2}\in \E$ be two edges, and $c \in \Path(\G)$ be such that $c$ joins the starting point of $e_{2}$ to the starting point of $e_{1}$. Choose $X\in \g$.
For each edge $e\in \E$, let $n^{+}_{e}$ and $n^{-}_{e}$ be the number of times $c$ traverses $e$ and $e^{-1}$ respectively, and set $\alpha'_{e}=\alpha_{e}\otimes (\alpha_{e_{1}}^{\vee}\otimes \alpha_{e_{1}})^{\otimes n^{+}_{e}}\otimes (\alpha_{e_{1}}\otimes \alpha_{e_{1}}^{\vee})^{\otimes n^{-}_{e}}$. We are adding twice as many new factors as the number of times $c$ traverses $e$ or $e^{-1}$, because we are, in a sense, inserting both $c$ and $c^{-1}$ to the spin network. For each vertex $v\in \V$ which is not $\underline{e_{1}}$ nor $\underline{e_{2}}$, let $n_{v}$ be the number of times $c$ visits $v$, and set $I'_{v}=I_{v}\otimes (\id_{V_{e_{1}}^{*}}\otimes \id_{V_{e_{1}}})^{\otimes n_{v}}$, seen as an element of $\bigotimes_{e\in \Out(v)} V_{e}$ in such a way as to connect, for each visit of $c$ and $c^{-1}$, the incoming edge with the outcoming one. Then, set $v_{1}=\underline{e_{1}}$ and $I'_{v_{1}}=I_{v_{1}}\otimes (\id_{V_{e_{1}}^{*}}\otimes \id_{V_{e_{1}}})^{\otimes (n_{v_{1}}-1)} \otimes \id_{V_{e_{1}}}$. In this tensor, the component of $I_{v_{1}}$ in $V_{e_{1}}$ is now seen as a part of the component associated to the last edge of $c$, the $n_{v_{1}}-1$ factors $\id_{V_{e_{1}}^{*}}\otimes \id_{V_{e_{1}}}$ connect the incoming and outcoming strands of $c$ at each visit except the last, and the last factor $\id_{V_{e_{1}}}$ connects the last edge of $c$ to $e_{1}$. Finally, set $v_{2}=\underline{e_{2}}$ and $I'_{v_{2}}=I_{v_{2}}\otimes (\id_{V_{e_{1}}^{*}}\otimes \id_{V_{e_{1}}})^{\otimes (n_{v_{2}}-1)} \otimes \alpha_{e_{1}}(X)$, in which the interpretation of $I_{v_{2}}$ is unchanged, the middle factor connects the strands of $c$, and  $c^{-1}$, at each of their visits but the first, and $\alpha_{e_{1}}(X)$ belongs to the $V_{e_{1}}^{*}\otimes V_{e_{1}}$ part of the representation attached to the first edge of $c$.

Let us also define $\alpha''=\alpha'$, and $I''_{v}=I'_{v}$ for all $v\neq v_{2}$. For $v=v_{2}$, we set $I''_{v_{2}}=c_{e_{2}}(I_{v_{2}}\otimes (\id_{V_{e_{1}}^{*}}\otimes \id_{V_{e_{1}}})^{\otimes (n_{v_{2}}-1)} \otimes (\alpha_{e_{1}}\otimes \alpha_{e_{2}})(C_{\g}))$, where $c_{e_{2}}$ is the contraction of the $V_{e_{2}}$ factor of $I_{v_{2}}$ and the $V_{e_{2}}^{*}$ factor of $(\alpha_{e_{1}}\otimes \alpha_{e_{2}})(C_{\g})$.

We leave the details of the proof of the following proposition to the reader. 

\begin{proposition}\label{derive spin c} Let $\G=(\V,\E,\F)$ be a graph. Let $\psi_{\alpha,I}:\Conf^{\G}_{G}\to \C$ be a spin network. Let $e_{1},e_{2}\in \E$ be two edges, and $c \in \Path(\G)$ be such that $c$ joins the starting point of $e_{2}$ to the starting point of $e_{1}$. Choose $X\in \g$. The following equalities hold:
\begin{align}
\nonumber &\D^{c,e}_{X} \psi_{\alpha,I}=\Tr_{V_{e}}\left(\psi_{\alpha,I}^{e} \circ \alpha_{e}(\Ad(h(c))X)\right)=\psi_{\alpha',I'},\\
&\Delta^{e_{2};c,e_{1}}\psi_{\alpha,I}=\Tr_{V_{e_{2}}\otimes V_{e_{1}}}\left(\psi^{e_{2};e_{1}}_{\alpha,I} \circ (\alpha_{e_{2}}\otimes \alpha_{e_{1}})([\id_{\g}\otimes \Ad(h(c))]C_{\g})\right)=\psi_{\alpha'',I''}. \label{delta spin c}
\end{align}
\end{proposition}

Let us now apply Corollary \ref{main deriv cor} to spin networks. 

\begin{proposition}\label{generator} Let $\G=(\V,\E,\F)$ be a graph. Let $G$ be a compact connected Lie group. For each bounded face $F$ of $\G$, there exists a second-order differential operator $L_{F}$ on $\Conf^{\G}_{G}$ which stabilises $\O_{\G}$ and such that for all $\psi \in \O_{\G}$ and all $t\in (\R^{*}_{+})^{\F^{b}}$, the following relation holds:
\[\frac{d}{d|F|} \E_{\YM^{\G}_{t}}[\psi]=\E_{\YM^{\G}_{t}}[L_{F}\psi].\]
\end{proposition}

\begin{proof} Corollary \ref{main deriv cor} provides us with an expression of the left-hand side which is of the form $\E_{\YM^{\G}_{t}}[L_{F} \psi]$, where $L_{F}$ is a second-order differential operator on $\Conf^{\G}_{G}$. Propositions \ref{derive spin} and \ref{derive spin c} ensure that this differential operator preserves the algebra $\O_{\G}$.
\end{proof}

This proposition may give some indication about the nature of what could be called the infinitesimal generator of the Yang-Mills measure, by analogy with the classical theory of Markov process, and with the role of time played by area. However, this has yet to be given a consistent and substantial form. The operators $L_{F}$ themselves would deserve a more detailed study, which we do not offer in the present work.

\subsection{Analyticity of the expectations of spin networks}\label{analyticity sn} As explained in Section \ref{outline}, 
we are going to consider several classes of observables: Wilson loops, Wilson skeins, Wilson garlands. They are all contained in the class of spin networks, and we continue in the present section to work with spin networks. We establish two results. The first is a result of analyticity of the expectation of a spin network with respect to the areas of the faces. The second, which is a consequence of the first, provides us with what will serve as an initial condition in the resolution of the differential system which we are going to write down in the next sections.

\begin{proposition}\label{analytic} Let $\G=(\V,\E,\F)$ be a graph on $\R^{2}$. Let $G$ be a compact connected Lie group. Let $\psi_{\alpha,I}:\Conf^{\G}_{G}\to \C$ be a spin network. The mapping $(\R^{*}_{+})^{\F^{b}} \to \C$ defined by $t\mapsto \E_{\YM^{\G}_{t}}[\psi_{\alpha,I}]$ is the restriction of an entire function defined on $\C^{\F^{b}}$. More precisely, this function is a linear combination of functions of the form
$t\mapsto \exp\left(\sum_{F\in \F^{b}}\frac{1}{2}c_{F} t(F) \right)$, where the coefficients $c_{F}$ are values of the Casimir operator of $\g$ in certain irreducible representations of $G$.
\end{proposition}

The proof of this result relies crucially on the fact that the graph which we consider is drawn on the plane, and therefore has an unbounded face. There is no reason to expect that such a simple result should hold on a compact surface without boundary.

As a preparation for the proof, let us review the orthogonality properties of spin networks. Let $\widehat G$ denote the set of isomorphism classes of irreducible complex representations of $G$. Let us choose a representation $\rho:G\to \GL(V_{\rho})$ in each class, and identify the representation $\rho$ with its class in $\widehat G$. Let us endow $V_{\rho}$ with a $G$-invariant Hermitian scalar product. Any finite tensor product of the spaces $V_{\rho}$ is thus endowed with a Hermitian structure.

Let us use the notation $\widehat G^{(\E)}=\{\beta \in \widehat G^{\E} : \forall e\in \E, \beta_{e^{-1}}=\beta_{e}^{\vee}\}$. The orthogonality properties of spin networks are entirely expressed by the fact, which is a generalisation of the Peter-Weyl theorem, that the mapping
\begin{align}
\nonumber \bigoplus_{\beta \in \widehat G^{(\E)}}^{\perp}\left( \bigotimes_{v\in \V} \bigotimes_{e\in \Out(v)} V_{\beta_{e}}\right)  & \longrightarrow L^{2}(\Conf^{\G}_{G},dh) \\
\bigoplus_{\beta \in \widehat G^{(\E)}} \left(\bigotimes_{v\in \V} J_{v}^{\beta}\right) &\longmapsto \sum_{\beta \in \widehat G^{(\E)}}\psi_{\beta,J^{\beta}} \label{Fourier}
\end{align}
is an isometry with dense range. Any square-integrable function on $\Conf^{\G}$ can thus be expanded in a Fourier series indexed by $\widehat G^{(\E)}$ and of which the coefficients are the tensors $\bigotimes_{v\in \V} J_{v}^{\beta}$. Spin networks for which the representation attached to each edge of $\G$ is irreducible are called irreducible spin networks. The right-hand side of \eqref{Fourier} is a sum of such irreducible spin networks.

We want to compute the expectation of a spin network under the Yang-Mills measure, which is the integral with respect to the uniform measure on the configuration space of the product of this spin network and the density of the Yang-Mills measure. Let us make two remarks.

Firstly, the constant function on $\Conf^{\G}_{G}$ equal to $1$ is a spin network, namely the spin network $\psi_{\triv,1}$ where $\triv_{e}$ is the trivial representation for each edge $e$, acting on the complex vector space $\C$, and $1$ is just the complex number $1$ seen as an element of $\bigotimes_{v\in \V} \bigotimes_{e\in \Out(v)} \C\simeq \C$. Hence, integrating a function over the configuration space amounts to extracting its Fourier coefficient corresponding to this particular spin network $\psi_{\triv,1}$.

Secondly, in order to integrate a product of functions, we need to understand the structure of algebra of the vector space of spin networks. This structure is governed by the plethysm of $G$, that is, the particular way in which the tensor product of two irreducible representations splits as a sum of irreducible representations. Given two representations $\rho_{1}$ and $\rho_{2}$ of $G$, not necessarily irreducible, and for each irreducible representation $\pi$, let $P_{\rho_{1},\rho_{2}}^{\pi}$ be the orthogonal projection of $V_{\rho_{1}}\otimes V_{\rho_{2}}$ onto its $\pi$-isotypical component, and let $(\rho_{1}\otimes \rho_{2})^{\pi}=P_{\rho_{1},\rho_{2}}^{\pi} \circ (\rho_{1}\otimes \rho_{2})$ be the corresponding $\pi$-isotypical sub-representation. Consider two spin networks $\psi_{\alpha_{1},I_{1}}$ and $\psi_{\alpha_{2},I_{2}}$. For each $\beta\in \widehat G^{(\E)}$, let us denote by $(\alpha_{1}\otimes \alpha_{2})^{\beta}$ the family $((\alpha_{1,e}\otimes \alpha_{2,e})^{\beta_{e}})_{e\in \E}$, and let us set
\[P_{\alpha_{1},\alpha_{2}}^{\beta} = \bigotimes_{v\in \V} \bigotimes_{e\in \Out(v)} P_{\alpha_{1,e},\alpha_{2,e}}^{\beta_{e}} \in \bigotimes_{v\in \V} \bigotimes_{e\in \Out(v)} \End_{G}\left(V_{\alpha_{1,e}}\otimes V_{\alpha_{2,e}}\right).\]
Then the Fourier series of the product $\psi_{\alpha_{1},I_{1}}\psi_{\alpha_{2},I_{2}}$ reads
\[\psi_{\alpha_{1},I_{1}}\psi_{\alpha_{2},I_{2}}=\sum_{\beta\in \widehat G^{(\E)}} \psi_{(\alpha_{1}\otimes \alpha_{2})^{\beta},P_{\alpha_{1},\alpha_{2}}^{\beta}(I_{1}\otimes I_{2})}.\]
The right-hand side of this equality is indeed a series of irreducible spin networks, for the representations $(\alpha_{1}\otimes \alpha_{2})^{\beta}$, although not irreducible, are isotypical, so that for each $\beta$, the spin network $\psi_{(\alpha_{1}\otimes \alpha_{2})^{\beta},P_{\alpha_{1},\alpha_{2}}^{\beta}(I_{1}\otimes I_{2})}$ is equal to $\psi_{\beta,J}$, for some $J$ which we do not need to compute explicitly.

For our present purposes, the most important consequence of these two observations is that two spin networks $\psi_{\alpha_{1},I_{1}}$ and $\psi_{\alpha_{2},I_{2}}$ are orthogonal with respect to the uniform measure as soon as $P_{\alpha_{1},\alpha_{2}}^{\triv}=0$, regardless of $I_{1}$ and $I_{2}$. Moreover, for this equality to hold it suffices that for one single edge $e$ the representation $\alpha_{1,e}\otimes \alpha_{2,e}$ does not contain a copy of the trivial representation. We can now turn to the proof of Proposition \ref{analytic}.

\begin{proof}[Proof of Proposition \ref{analytic}] Let us start by expanding the density of the Yang-Mills measure into a Fourier series of the form \eqref{Fourier}. For each $\rho\in \widehat G$, let $c_{\rho}$ denote the scalar by which the Casimir operator of $\g$ acts on $V_{\rho}$, that is, the complex number such that $\rho(C_{\g})=c_{\rho} \id_{V_{\rho}}$. This number is in fact real and non-positive. Let us also denote by $\chi_{\rho}:G\to \C$ the character of $\rho$. The heat kernel $Q_{t}$ on $G$ can be expressed as the sum of the series
\[Q_{t}=\sum_{\rho\in \widehat G} e^{\frac{1}{2}c_{\rho}t} \dim(V_{\rho})\chi_{\rho}.\]
Let $\gamma=(\gamma_{F})_{F\in \F^{b}}\in \widehat G^{\F^{b}}$ be the data of one irreducible representation of $G$ for each bounded face of the graph $\G$. Recall that $F_{\infty}$ denotes the unbounded face of $\G$ and set $\gamma_{F_{\infty}}$ equal to the trivial representation of $G$. For each edge $e$ of $\G$, let us denote respectively by $F^{L}(e)$ and $F^{R}(e)$ the faces of $\G$ located on the left and on the right of $e$. Note that $F^{L}(e)$ and $F^{R}(e)$ can be equal. Let us define $\beta(\gamma)\in \widehat G^{(\E)}$ by setting, for each edge $e\in \E$, $\beta(\gamma)_{e}=\gamma_{F^{L}(e)}\otimes (\gamma_{F^{R}(e)})^{\vee}$. 
At each vertex of $\G$, the outcoming edges are cyclically ordered by the orientation of $\R^{2}$ (see \cite[Lemma 1.3.16]{LevyAMS} for a proof of this intuitively obvious fact), and each pair of neighbouring edges in this cyclic order determines a face of $\G$, of which we say that it is adjacent to $v$. 
For each vertex $v$, let us define $J_{v}^{\beta(\gamma)}=\bigotimes \id_{V_{\gamma_{F}}}$, where the tensor product is taken over all faces adjacent to $v$. Then for each assignment $t: \F^{b} \to \R^{*}_{+}$ of a positive real number to each bounded face of $\G$, we have for all $h\in \Conf^{\G}_{G}$ the equality
\begin{equation}\label{YM Fourier}
\prod_{F\in \F^{b}} Q_{t(F)}(h(\partial F))=\sum_{\gamma \in \widehat G^{\F^{b}}}  \left(\prod_{F\in \F^{b}} e^{\frac{1}{2} c_{\gamma_{F}}t(F)}\dim(\gamma_{F})\right) \psi_{\beta(\gamma),J^{\beta(\gamma)}}(h),
\end{equation}
which is the Fourier expansion of the density of the Yang-Mills measure. 

Let us now consider an arbitrary spin network $\psi_{\alpha,I}$. Our main claim is that the spin networks $\psi_{\alpha,I}$ and $\psi_{\beta(\gamma),J^{\beta(\gamma)}}$ are orthogonal in $L^{2}(\Conf^{\G}_{G},dh)$ for all but a finite number of $\gamma$ in the set $\widehat G^{\F^{b}}$. We will in fact prove that $P_{\alpha,\beta(\gamma)}^{\triv}=0$ for all but a finite number of $\gamma$. 

Let us recall the notation of Section \ref{grouploops} and consider a spanning tree $\wT$ of the dual graph of $\G$. This spanning tree is rooted at the unbounded face $\hat F_{\infty}$ of $\G$. Let $F$ be a bounded face of $\G$. We will prove by induction on the graph distance in $\wT$ between $\hat F$ and $\hat F_{\infty}$ that there are only finitely many $\pi\in \widehat G$ for which there exists $\gamma\in \widehat G^{\F^{b}}$ such that $\gamma_{F}=\pi$ and $P_{\alpha,\beta(\gamma)}^{\triv}\neq 0$.

It is in the initialisation of this induction argument that we benefit from the presence of the unbounded face. Indeed, let us consider a face $F$ which is adjacent to the unbounded face. Let $e$ be an edge such that $F^{R}(e)=F_{\infty}$ and $F^{L}(e)=F$. For all $\gamma\in \widehat G^{\F^{b}}$, we have $\beta(\gamma)_{e}=\gamma_{F}$. We can then have $P_{\alpha,\beta(\gamma)}^{\triv}\neq 0$ only if $\alpha_{e}\otimes \beta(\gamma)_{e}=\alpha_{e}\otimes \gamma_{F}$ contains a copy of the trivial representation, and this happens if and only if $\gamma_{F}$ is an irreducible sub-representation of $\alpha_{e}^{\vee}$. Since there are finitely many irreducible sub-representations in any finite-dimensional representation of $G$, the property is proved under the assumption $\hat d_{\wT}(\hat F,\hat F_{\infty})=1$.

Let us assume that for some $n\geq 2$, the property has been proved for all faces $F$ such that $\hat d_{\wT}(\hat F,\hat F_{\infty})\leq n-1$ and let us consider a face $F$ at a distance $n$ from the unbounded face. There exists a face $F'$ adjacent to $F$ which is at a distance $n-1$ from the unbounded face. Let $\pi'$ be one of the finitely many irreducible representations of $G$ such that there exists $\gamma\in \widehat G^{\F^{b}}$ such that $\gamma_{F'}=\pi'$ and $P_{\alpha,\beta(\gamma')}^{\triv}\neq 0$. The assertion which we are trying to prove for $F$ will follow from the fact that there are only finitely many $\pi\in \widehat G$ such that there exists $\gamma\in \widehat G^{\F^{b}}$ satisfying $\gamma_{F}=\pi$, $\gamma_{F'}=\pi'$ and $P_{\alpha,\beta(\gamma')}^{\triv}\neq 0$. Indeed, let $e$ be an edge such that $F^{L}(e)=F$ and $F^{R}(e)=F'$. If $\gamma$ satisfies the last three conditions, then $\alpha_{e}\otimes \pi \otimes (\pi')^{\vee}$ contains a copy of the trivial representation of $G$. Hence, $\pi$ is a sub-representation of $\alpha_{e}^{\vee}\otimes \pi'$, and this restricts $\pi$ to a finite subset of $\widehat G$. This concludes the inductive argument.

It follows from our main claim that in the computation of the expectation of any particular spin network $\psi_{\alpha,I}$ with respect to the Yang-Mills measure, only finitely many terms of the right-hand side of \eqref{YM Fourier} contribute. The result of the integration is thus a linear combination of functions of the form $t\mapsto \exp \left(\frac{1}{2} \sum_{F\in \F^{b}}c_{\gamma_{F}}t(F)\right)$, as expected. 
\end{proof}

Since we presented this result of analyticity as a help for dealing with boundary conditions, it is only fair that we give a precise statement about the limit of the expected value of a spin network under the measure $\YM_{t}^{\G}$ as $t$ tends to $0$. We shall do this for spin networks which are invariant under the action of the gauge group. Let us indicate in general how the gauge group acts on spin networks. Let $\psi_{\alpha,I}$ be a spin network. Let $j\in G^{\V}$ be an element of the gauge group. Then for all $h\in \Conf^{\G}_{G}$,
\begin{align*}
(j\cdot \psi_{\alpha,I})(h)=\psi_{\alpha,I}(j^{-1}\cdot h)&=\bigg\langle \bigotimes_{e\in \E^{+}} \alpha_{e}(j(\overline{e}))\alpha_{e}(h(e)) \alpha_{e}(j(\underline{e})^{-1})\; , \; \bigotimes_{v\in \V}I_{v} \bigg\rangle\\
&=\bigg\langle \bigotimes_{e\in \E^{+}} \alpha_{e}(h(e)) \; , \; \bigotimes_{v\in \V} \Bigg[\bigotimes_{e\in \Out(v)}\alpha_{e}(j(v)^{-1})\Bigg] (I_{v}) \bigg\rangle.
\end{align*}
It follows that a spin network $\psi_{\alpha,I}$ is invariant as soon as $I_{v}$ is an invariant tensor for each vertex $v$. Moreover, by averaging the Fourier expansion of an invariant square-integrable observable under the action of the gauge group, one checks that any such observable admits an invariant Fourier expansion. In other words, the restriction of \eqref{Fourier} to invariant tensors on the left-hand side and to invariant square-integrable observables on the right-hand side is also an isometry with dense range.  

Let us turn to the computation of the expectation of an invariant spin network at $t=0$. Let us denote by $1$ the element of the configuration space which is defined by $1(c)=1$, the unit element of $G$, for each path $c$. This notation conflicts with other uses of the symbol $1$ in the present section, but this should not cause any confusion.

\begin{proposition}\label{value at 0} With the notation of Proposition \ref{analytic}, and if $\psi_{\alpha,I}$ is an invariant spin network, then the value at $t=(0,\ldots,0)$ of the mapping $t\mapsto \E_{\YM^{\G}_{t}}[\psi_{\alpha,I}]$ is $\psi_{\alpha,I}(1)$.
\end{proposition}

This proposition follows at once from Proposition \ref{analytic} and the following lemma.

\begin{lemma} Let $f:\Conf^{\G}_{G}$ be a smooth invariant observable. Then
\[\lim_{t\to 0} \E_{\YM^{\G}_{t}}[f]=f(1).\]
\end{lemma}

\begin{proof} By Proposition \ref{YM basis invariant} and since $f$ is invariant,
\[\E_{\YM^{\G}_{t}}[f]=\int_{G^{\F^{b}}} f(\{g_{F}: F\in \F^{b}\},\{1 : e\in \T^{+}\})\; \prod_{F\in \F^{b}}Q_{t(F)}(g_{F}) \,dg_{F}.\]
Using the fact that the measure $Q_{s}(g)\, dg$ converges weakly on $G$, as $s$ tends to $0$, to the Dirac mass at $1$, we find the desired result. 
\end{proof}

\subsection{Area derivatives of Wilson loops I: examples} \label{section dwl} Wilson loops are a particular case of spin networks, so that the results of Section \ref{sec der spin net} tell us how the differential operators introduced in Section \ref{op diff config} act on them. However, we want to interpret \eqref{delta spin} and \eqref{delta spin c} very concretely when the group $G$ is one of the groups $\U(N,\K)$ which we studied in the first part of this work, and ultimately to understand geometrically the right-hand side of \eqref{local derive}.

Let us start by giving a formal definition of Wilson loops. Let $\G$ be a graph, $l\in \Loop(\G)$ a loop, $\chi : G\to \C$ a conjugation-invariant function. The Wilson loop associated to this data is the invariant observable $W_{\chi,l}:\Conf^{\G}_{G}\to \C$ defined by
\[W_{\chi,l}(h)=\chi(h(l)).\]
When $G$ is a matrix group, the function $\chi$ is often taken to be the normalised trace, or the real part of the normalised trace in the case of a quaternionic group. More precisely, if $\K\in \{\R,\C,\H\}$ then we define $W_{l}^{\K,N}:\Conf^{\G}_{\U(N,\K)}\to \C$ by setting 
\[W_{l}^{\K,N}(h)=\tr(h(l)) \mbox{ if } \K\in \{\R,\C\},\; \mbox{ and } W_{l}^{\H,N}(h)=\Re\tr(h(l)).\]
It is straightforward to check that Wilson loops are invariant observables, in the sense that they are invariant under the action of the gauge group.

Before we start developing a general treatment of expectations of Wilson loops, let us study a few simple examples. This will serve as a motivation and an illustration for the content of the next two sections. 

\begin{example}{\em The simplest example is that of a loop $l$ which goes once around a circle of area $t$ (see the left half of Figure \ref{0} below). The graph $\G_{l}$ has a single vertex $v$ and a single unoriented edge $\{e,e^{-1}\}$. 

\begin{figure}[h!]
\begin{center}
\includegraphics{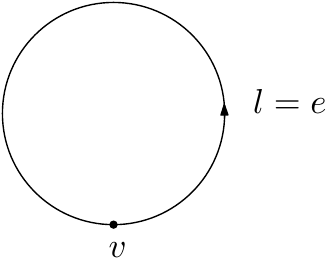}
\caption{\label{0} The simplest elementary loop.}
\end{center}
\end{figure}

Let us express the Wilson loop $W^{\K,N}_{l}$ as a spin network. To this end, let us denote by $\nat$ the natural representation of $\U(N,\K)$, by which we mean the inclusion in $\GL_{N}(\K)$ if $\K\in \{\R,\C\}$, or the mapping $\iota : \U(N,\H) \to \GL_{2N}(\C)$ defined in Section \ref{moments empirique} if $\K=\H$. The natural representation acts thus on $\R^{N}$, $\C^{N}$ or $\C^{2N}$, depending on the value of $\K$.
Let us denote, in all three cases, by $\id$ the identity of this space. Let us finally define $\alpha_{e}=\nat$ and $I_{v}=\id$. Then the spin network $\psi_{\alpha,I}$ is equal to $NW^{\K,N}_{l}$ if $\K\in \{\R,\C\}$ and to $2NW^{\H,N}_{l}$ if $\K=\H$. Equation \eqref{derivative one face 1} yields in this case
\[\frac{d}{dt} \E_{t}\left[W_{l}\right]=\frac{1}{2}\E_{t}\left[\Delta^{e} W_{l}\right]=\frac{1}{2N}\E_{t}\left[\Delta^{e} \psi_{\alpha,I}\right],\]
where $\frac{1}{2N}$ must be replaced by $\frac{1}{4N}$ in the quaternionic case. For the sake of clarity, we are using here a simplified notation where $\E_{t}$ means $\E_{\YM^{\G_{l}}_{t}}$ and the superscripts $\K$ and $N$ are understood in the Wilson loop $W_{l}$.

Equation \eqref{delta e spin} allows us to compute $\Delta^{e} \psi_{\alpha,I}$. We have, for all $h\in \Conf^{\G_{l}}\simeq G^{\{e\}}$, and thanks to Lemma \ref{sum squares},
\[\Delta^{e} \psi_{\alpha,I}(h)=\Tr(\nat(h(e)) \nat(C_{\u(N,\K)}))=\Tr(\nat(h(e)) c_{\u(N,\K)})=c_{\u(N,\K)} \psi_{\alpha,I}(h).\]
Hence, we find
\[\frac{d}{dt} \E_{t}\left[W_{l}\right]=\frac{1}{2}c_{\u(N,\K)}\E_{t}\left[W_{l}\right].\]
By Proposition \ref{value at 0}, we know that $\E_{0}[W_{l}]=1$. Hence, we finally find
\[\E_{t}\left[W_{l}\right]=e^{\frac{t}{2}c_{\u(N,\K)}}.\]
In particular, since $c_{\u(N,\K)}$ tends to $-1$ as $N$ tends to infinity, we have $\Phi(l)=e^{-\frac{t}{2}}$.
}
\end{example}

\begin{example}\label{ex 2}{\em Let us now consider the next simplest example. We consider the loop $l$ depicted in the left part of Figure \ref{1x1} below.

\begin{figure}[h!]
\begin{center}
\includegraphics{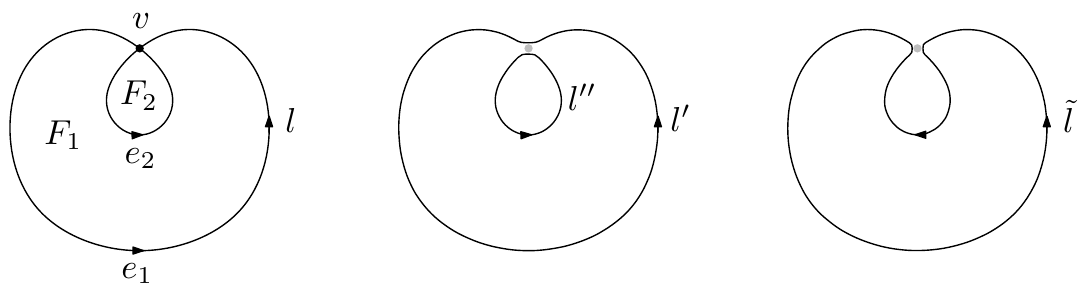}
\caption{\label{1x1} The second simplest elementary loop and the loops which arise in the area derivatives of the expectation of the corresponding Wilson loop.}
\end{center}
\end{figure}

There are now two unoriented edges $e_{1}$ and $e_{2}$, and still a single vertex $v$. Let us define $\alpha_{e_{1}}=\alpha_{e_{2}}=\nat$. Then, let us define $I_{v}\in V_{e_{1}}^{*}\otimes V_{e_{2}}^{*} \otimes V_{e_{1}} \otimes V_{e_{2}}\simeq \End(V_{e_{1}}\otimes V_{e_{2}})$ as the exchange of the factors, that is, $I_{v}(x_{1}\otimes x_{2})=x_{2}\otimes x_{1}$. Then, as in the previous example, the spin network $\psi_{\alpha,I}$ is equal to $NW^{\K,N}_{l}$ if $\K\in \{\R,\C\}$ and to $2NW^{\K,N}_{l}$ if $\K=\H$. The expectation $\E_{t}[W_{l}]$ is a function of the two variables $(t_{1},t_{2})=(t(F_{1}),t(F_{2}))$. We know by Proposition \ref{value at 0} that the value at $(0,0)$ of this function is $1$. By an argument similar to the one we used in the previous example, we find
\[\partial_{t_{1}}\E_{t}\left[W_{l}\right]=\frac{1}{2}c_{\u(N,\K)}\E_{t}\left[W_{l}\right].\]
In order to compute the partial derivative with respect to $t_{2}$, we apply Proposition \ref{main deriv} with the sequence of faces $F_{2},F_{1},F_{\infty}$ and the sequence of edges $e_{2}^{-1},e_{1}^{-1}$. We find
\[\left(\partial_{t_{2}}-\partial_{t_{1}}\right)\E_{t}\left[W_{l}\right]=\E_{t}\left[\frac{1}{2}\Delta^{e_{2}}W_{l}+\Delta^{e_{1};e_{2}}W_{l}\right].\]
The same computation as before yields $\Delta^{e_{2}} W_{l}=c_{\u(N,\K)}W_{l}$. Then, using \eqref{delta spin}, we find
\[\Delta^{e_{1};e_{2}}\psi_{\alpha,I}=\Tr_{V_{e_{1}}\otimes V_{e_{2}}}\left(I\circ \nat\otimes \nat(h(e_{1})\otimes h(e_{2}) \circ C_{\u(N,\K)})\right).\]
Using, depending on the value of $\K$, one of the formulas which we have established on the Casimir operator (see \eqref{casimir o u}, \eqref{iotacasimir so}, \eqref{iotacasimir}, \eqref{P cycles o}, Lemma \ref{tr tr sp}), we find that the right-hand side is a linear combination of $\Tr(\nat(h(e_{1})))\Tr(\nat(h(e_{2})))$ and $\Tr(\nat(h(e_{1}e_{2}^{-1})))$. Hence, with the notation of Figure \ref{1x1} above, the partial derivative with respect to $t_{2}$ of $\E_{t}[W_{l}]$ is a linear combination of $\E_{t}[W_{\tilde l}]$ and $\E_{t}[W_{l'}W_{l''}]$. More explicitly, we have 
\begin{equation}\label{ex2 1}
\left(\partial_{t_{2}}-c_{\u(N,\K)}\right)\E_{t}\left[W_{l}\right]=\left\{\begin{array}{ll}
-\E_{t}\left[W_{l'}W_{l''}\right] +\frac{1}{N}\E_{t}\left[W_{\tilde l}\right] & \mbox{if } \K=\R,\\
-\E_{t}\left[W_{l'}W_{l''}\right]& \mbox{if } \K=\C,\\
-\E_{t}\left[W_{l'}W_{l''}\right] +\frac{1}{-2N}\E_{t}\left[W_{\tilde l}\right] & \mbox{if } \K=\H.
\end{array}\right.
\end{equation}
The loop $\tilde l$ is essentially a simple loop enclosing the union of the faces $F_{1}$ and $F_{2}$. Thus, from our study of the first example, we know that $\E_{t}\left[W_{\tilde l}\right]=e^{\frac{t_{1}+t_{2}}{2} c_{\u(N,\K)}}$. However, we must compute $\E_{t}\left[W_{l'}W_{l''}\right]$, and for this we must start from the beginning again. The product $W_{l'}W_{l''}$ is equal to the spin network $\frac{1}{N^{2}}\psi_{\alpha,J}$, where $\alpha$ is as above and $J_{v}=\id_{V_{1}}\otimes \id_{V_{2}}$. In the quaternionic case, $\frac{1}{N^{2}}$ must of course be replaced by $\frac{1}{4N^{2}}$. The same computation as we did for $\psi_{\alpha,I}$, involving the equality
\[\Delta^{e_{1};e_{2}}\psi_{\alpha,J}=\Tr_{V_{e_{1}}\otimes V_{e_{2}}} \nat\otimes \nat(h(e_{1})\otimes h(e_{2}) \circ C_{\u(N,\K)}),\]
gives us a second differential equation, namely
\begin{equation}\label{ex2 2}
\left(\partial_{t_{2}}-c_{\u(N,\K)}\right)\E_{t}\left[W_{l'}W_{l''}\right]=\left\{\begin{array}{ll}
-\frac{1}{N^{2}}\E_{t}\left[W_{l}\right] +\frac{1}{N^{2}}\E_{t}\left[W_{\tilde l}\right] & \mbox{if } \K=\R,\\[1pt]
-\frac{1}{N^{2}}\E_{t}\left[W_{l}\right]& \mbox{if } \K=\C,\\[1pt]
-\frac{1}{4N^{2}}\E_{t}\left[W_{l}\right] +\frac{1}{4N^{2}}\E_{t}\left[W_{\tilde l}\right] & \mbox{if } \K=\H.
\end{array}\right.
\end{equation}
For each value of $\K$, the system formed by \eqref{ex2 1} and \eqref{ex2 2} can now easily be solved. For example, in the real case, we find
\[\E_{t}\left[W_{l}^{\R,N}\right]=e^{-\frac{s}{2}-t}\left(e^{\frac{t}{N}} \left(\cosh \frac{t}{N}-N\sinh \frac{t}{N}\right) + \frac{4}{3N-1} e^{(1-\frac{1}{2N})s}\left(1-e^{\left(\frac{3}{2}-\frac{1}{2N}\right)t}\right)\right).\]
The symplectic expression is formally given by the equality 
\[\E_{t}\left[W_{l}^{\H,N}\right]=\E_{t}\left[W_{l}^{\R,-2N}\right].\]
Finally, in the complex case, the expression is simpler, as we have
\[\E_{t}\left[W_{l}^{\C,N}\right]=e^{-\frac{s}{2}-t}\left(\cosh \frac{t}{N} - N \sinh \frac{t}{N}\right).\]
By letting $N$ tend to infinity in either of these expressions, we find $\Phi(l)=e^{-\frac{s}{2}-t}(1-t)$.
}
\end{example}

\begin{example}\label{ex 3}{\em Let us consider as a third and slightly more complicated example the loop $l$ depicted in the left half of Figure \ref{2x1} below.

\begin{figure}[h!]
\begin{center}
\includegraphics{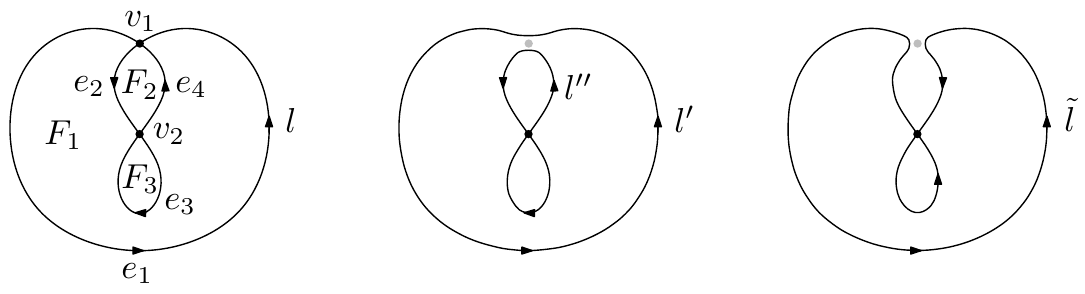}
\caption{\label{2x1} A third example of elementary loop and some of the loops which arise in the area derivatives of the expectation of the corresponding Wilson loop.}
\end{center}
\end{figure}

We set $\alpha_{e_{1}}=\alpha_{e_{2}}=\alpha_{e_{3}}=\alpha_{e_{4}}=\nat$, and let both $J_{v_{1}}:V_{e_{4}}\otimes V_{e_{1}}\to V_{e_{2}}\otimes V_{e_{1}}$ and $J_{v_{2}}:V_{e_{2}}\otimes V_{e_{3}}\to V_{e_{4}}\otimes V_{e_{3}}$ be the exchange of the two factors. Then $\psi_{\alpha,J}$ is equal to $NW_{l}$ if $\K\in \{\R,\C\}$ and $2NW_{l}$ if $\K=\H$.
By the same reasoning as in the other examples, we find
\[\partial_{t_{1}}\E_{t}\left[W_{l}\right]=\frac{1}{2}\E_{t}[\Delta^{e_{1}}W_{l}]=\frac{1}{2}c_{\u(N,\K)}\E_{t}\left[W_{l}\right].\]
Corollary \ref{main deriv cor} applied to the sequence of faces $F_{2},F_{1},F_{\infty}$ yields
\[\partial_{t_{2}}\E_{t}[W_{l}]=\frac{1}{2}\E_{t}[\Delta^{e_{4}^{-1}}W_{l}+\Delta^{e_{1}^{-1}}W_{l}]+\E_{t}[\Delta^{e_{4}^{-1};e_{1}^{-1}}W_{l}].\]
Applied to the sequence of faces $F_{3}, F_{1},F_{\infty}$, it yields
\[\partial_{t_{3}}\E_{t}[W_{l}]=\frac{1}{2}\E_{t}[\Delta^{e_{3}}W_{l}+\Delta^{e_{1}^{-1}}W_{l}]+\E_{t}[\Delta^{e_{1}^{-1};e_{4}^{-1},e_{3}}W_{l}].\]
Using \eqref{delta spin} as we did in the previous example, we find that $\Delta^{e_{4}^{-1};e_{1}^{-1}}W_{l}$ is a linear combination of $W_{l'}W_{l''}$ and $W_{\tilde l}$, where the notation is that of Figure \ref{2x1}. More precisely, we find that \eqref{ex2 1} holds without a change in the present situation.

What is new in this example is that we need to use \eqref{delta spin c} in order to compute $\Delta^{e_{1}^{-1};e_{4}^{-1},e_{3}}W_{l}$. A glance at Figures \ref{spint c} and \ref{2x2} may be helpful at this point. With the notation of Figure \ref{2x2}, we find

\begin{equation}\label{ex3 1}
\left(\partial_{t_{3}}-c_{\u(N,\K)}\right)\E_{t}\left[W_{l}\right]=\left\{\begin{array}{ll}
-\frac{1}{N} \E_{t}[W_{\tilde l_{1}}]+\E_{r}[W_{l'_{1}}W_{l''_{1}}] & \mbox{if } \K=\R,\\
-\frac{1}{N} \E_{t}[W_{\tilde l_{1}}]& \mbox{if } \K=\C,\\
\frac{1}{2N} \E_{t}[W_{\tilde l_{1}}]+\E_{r}[W_{l'_{1}}W_{l''_{1}}] & \mbox{if } \K=\H.
\end{array}\right.
\end{equation}

\begin{figure}[h!]
\begin{center}
\includegraphics{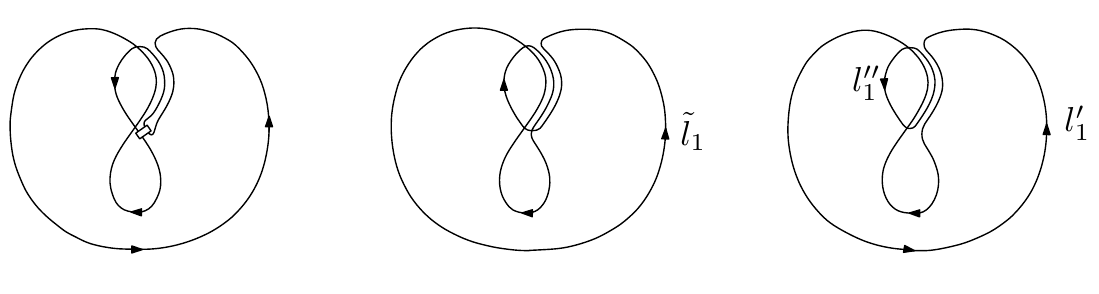}
\caption{\label{2x2} The left-hand side of this picture shows how \eqref{delta spin c} applies in the present context. The white box represents, as on Figure \ref{spint c}, the Casimir operator of $\u(N,\K)$. The two other pictures show the loops which arise in $\partial_{t_{3}}\E_{t}[W_{l}]$. Note that the edge $e_{4}$ is traversed more than once by these loops.}
\end{center}
\end{figure}

At this point, this example appears less promising that the other two which we studied. In order to compute $\E_{t}[W_{l}]$, we have to establish differential relations involving $W_{\tilde l_{1}}$, $W_{l'_{1}}$ and $W_{l''_{1}}$ and this will produce new products of Wilson loops of which we will need to compute the expectations as well. We shall prove that it is possible to do this in such a way as to produce a closed differential system involving products of only finitely many Wilson loops, one of which is $W_{l}$ (see Proposition \ref{matrice M}). We shall however not finish the present computation for the moment. 
}
\end{example}

In the next two sections, we systematise the operations which we were led to apply in the three examples which we have studied. For the sake of clarity, we introduce the formalism in two steps, which correspond to the two levels of complexity illustrated by Examples \ref{ex 2} and \ref{ex 3} respectively. These two steps occupy Sections \ref{wskein} and \ref{wgarland} respectively.

\subsection{Area derivatives of Wilson loops II: Wilson skeins}\label{wskein}

Let us call {\em skein} a finite collection $\sk=\{l_{1},\ldots,l_{r}\}$ of elementary loops (see Section \ref{sec: uniformity}) such that there exists a graph whose skeleton is the union of the ranges of $l_{1},\ldots,l_{r}$ and such that in each pair formed by one edge and its inverse, one edge is traversed exactly once by exactly one of the loops $l_{1},\ldots,l_{r}$, and the other edge is not traversed by any of the loops $l_{1},\ldots,l_{r}$.

The proof of Lemma \ref{graph gl} extends to skeins and for each skein $\sk$ there exists a graph $\G_{\sk}$ which is the least fine graph on which all the loops of $\sk$ can be traced.

To each skein $\sk=\{l_{1},\ldots,l_{r}\}$ traced in a graph $\G$ we associate the observable 
\[W_{\sk}^{\K,N}=W_{l_{1}}^{\K,N}\ldots W_{l_{r}}^{\K,N}\]
on the configuration space $\Conf^{\G}_{\U(N,\K)}$, and we call this observable a {\em Wilson skein}. 

Let $\sk=\{l_{1},\ldots,l_{r}\}$ be a skein and set $\G=\G_{\sk}$. The set of edges traversed by the loops of $\sk$ is an orientation of $\G$ which we denote by $\E^{+}=\{e_{1},\ldots,e_{n}\}$. The skein $\sk$ determines a permutation of $\E^{+}$, which to each edge $e$ associates the edge traversed immediately after $e$ by the unique loop of $\sk$ which traverses $e$. We denote this permutation by $\lambda_{\sk}$. The cycles of $\lambda_{\sk}$ are naturally in bijection with the elements of $\sk$. Through the labelling of $\E^{+}$ which we have chosen, we identify it with the set $\{1,\ldots,n\}$, and $\lambda_{\sk}$ with an element of the symmetric group $\S_{n}$.

Recall the notation of Sections \ref{section:brauer I} and \ref{Brauer II}, in particular the definitions of the morphisms $\rho_{\K}$ (see \eqref{def rho K}). 

\begin{lemma} \label{lem wilson rho}With the notation above, the Wilson skein $W^{\K,N}_{\sk}$ can be written as follows: for each $h\in \Conf^{\G}_{\U(N,\K)}$,
\begin{equation}\label{wilson rho}
W_{\sk}^{\K,N}(h)=\left\{\begin{array}{ll} N^{-r}\Tr^{\otimes n}(\rho_{\K}(\lambda_{\sk})\circ h(e_{1})\otimes \ldots \otimes h(e_{n})) & \mbox{if } \K=\R \mbox{ or } \C,\\[2pt]
(-2N)^{-r}(-2\Re\Tr)^{\otimes n}(\rho_{\H}(\lambda_{\sk})\circ h(e_{1})\otimes \ldots \otimes h(e_{n})) & \mbox{if } \K=\H.\end{array}\right.
\end{equation}
\end{lemma}

\begin{proof} According to the formula \eqref{P cycles o}, we have, in the real and complex cases,
\[N^{-r}\Tr^{\otimes n}(\rho_{\K}(\lambda_{\sk})\circ h(e_{1})\otimes \ldots \otimes h(e_{n}))=\prod_{(i_{1}\ldots i_{s})\preccurlyeq \lambda_{\sk}} \tr(h(e_{i_{s}})\ldots h(e_{i_{1}})).\]
Observe that, since $\lambda_{\sk}$ is a permutation, the signs $\epsilon_{1},\ldots,\epsilon_{n}$ are all equal to $1$. Now, each cycle $(i_{1}\ldots i_{s})$ of $\lambda_{\sk}$ corresponds to a loop $e_{i_{1}}\ldots e_{i_{s}}$ of $\sk$ and the multiplicativity of $h$ (recall \eqref{def mult h}) reads $h(e_{i_{s}})\ldots h(e_{i_{1}})=h(e_{i_{1}}\ldots e_{i_{s}})$. Thus, the right-hand side of the equality above is exactly $W_{\sk}^{\K,N}(h)$. In the quaternionic case, we apply Lemma \ref{tr tr sp} and use the same argument.   
\end{proof}

Let us define two operations on skeins, analogous to the operations $S_{i,j}$ and $F_{i,j}$ which we defined on the Brauer algebra in Section \ref{Brauer III}. Let $\sk=\{l_{1},\ldots,l_{r}\}$ be a skein. Let $e_{1}$ and $e_{2}$ be two edges of  $\G_{\sk}$ issued from the same vertex. Let us assume that $e_{1}$ and $e_{2}$ belong to $\E^{+}$. We can assume that $e_{1}$ is traversed by the loop $l_{1}$. Let us first assume that $e_{2}$ is also traversed by $l_{1}$, and that $l_{1}$ traverses $e_{1}$ before $e_{2}$. We can write $l_{1}=ae_{1}be_{2}c$, where $a,b,c$ are paths in $\G$. Let us define $l'_{1}=e_{1}b$, $l''_{1}=e_{2}ca$ and $\tilde l_{1}=e_{1}b(e_{2} ca)^{-1}$. These are elementary loops in $\G_{\sk}$. We define
\[S^{e_{2};e_{1}}(\sk)=\{l'_{1},l''_{1},l_{2},\ldots,l_{r}\} \mbox{ and } F^{e_{2};e_{1}}(\sk)=\{\tilde l_{1},l_{2},\ldots,l_{r}\}.\]
If $l_{1}$ traverses $e_{2}$ before $e_{1}$, we set $S^{e_{2};e_{1}}(\sk)=S^{e_{1};e_{2}}(\sk)$ and $F^{e_{2};e_{1}}(\sk)=F^{e_{1};e_{2}}(\sk)$.

In the case where $e_{2}$ is not traversed by $l_{1}$, we may assume that it is traversed by $l_{2}$. Let us write the loops as $l_{1}=ae_{1}b$ and $l_{2}=ce_{2}d$, where $a,b,c,d$ are paths. We define $l'=e_{1}bae_{2}dc$ and $\tilde l=e_{1}ba(e_{2}dc)^{-1}$, and set
\[S^{e_{2};e_{1}}(\sk)=\{l',l_{3},\ldots,l_{r}\} \mbox{ and } F^{e_{2};e_{1}}(\sk)=\{\tilde l,l_{3},\ldots,l_{r}\}.\]
One checks easily that in all cases, $S^{e_{2};e_{1}}(\sk)$ and $F^{e_{2};e_{1}}(\sk)$ are skeins traced on the graph $\G_{\sk}$.

\begin{figure}[h!]
\begin{center}
\scalebox{1}{\includegraphics{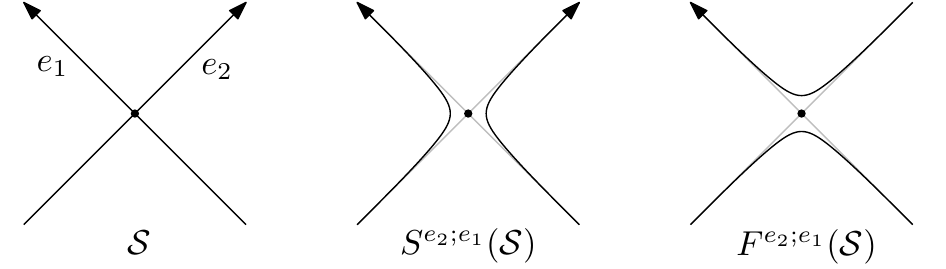}}
\caption{\label{sfs} The operations $S^{e_{2};e_{1}}$ and $F^{e_{2};e_{1}}$ can be understood as acting locally at the common origin of $e_{1}$ and $e_{2}$. From this point of view, the fact that $e_{1}$ and $e_{2}$ are on the same loop or not does not matter. The only difference is, in the case of the operation $F^{e_{2};e_{1}}$, the direction in which the lower strand is traversed. In fact, in this case, the orientation of the loop $\tilde l$ itself is arbitrary. We chose to let it traverse $e_{1}$ positively, but the other choice would not make any difference, since this operation is used only in the real and quaternionic cases, where a Wilson loop is not altered by changing the loop into its inverse.}
\end{center}
\end{figure}

The following proposition shows that the vector space of smooth complex-valued functions on $\Conf^{\G}_{\U(N,\K)}$ spanned by Wilson skeins is stable under the action of the operators $\Delta^{e_{1}}$ and $\Delta^{e_{2};e_{1}}$. 

\begin{proposition}\label{derive skeins} Let $\sk$ be a skein. Set $\G=\G_{\sk}$. Let $\E^{+}$ be the orientation of $\G$ induced by $\sk$. Let $e_{1},e_{2}$ be two distinct edges of $\G$ issued from the same vertex $v$. The following properties hold, with the superscripts $\K,N$ understood for each Wilson skein.\\
\indent 1. $\Delta^{e_{1}}W_{\sk}=c_{\u(N,\K)}W_{\sk}$.\\
\indent2. $\Delta^{e_{2};e_{1}}W_{\sk}=\Delta^{e_{1};e_{2}}W_{\sk}$.\\
\indent3. If $e_{1}\notin \E^{+}$, then $\lambda_{\sk}(e_{1}^{-1}) \in \Out(v)\cap \E^{+}$ and $\Delta^{e_{2};e_{1}}W_{\sk}=-\Delta^{e_{2};\lambda_{\sk}(e_{1}^{-1})}W_{\sk}$.\\
\indent 4. Let us assume that $e_{1}$ and $e_{2}$ belong to $\E^{+}$. If $e_{1}$ and $e_{2}$ are traversed by the same loop of $\sk$, then
\[\Delta^{e_{2};e_{1}}W_{\sk}=\left\{\begin{array}{ll} 
-W_{S^{e_{2};e_{1}}(\sk)}+\frac{1}{N}W_{F^{e_{2};e_{1}}(\sk)} & \mbox{if } \K=\R,\\[2pt]
-W_{S^{e_{2};e_{1}}(\sk)} & \mbox{if } \K=\C,\\[2pt]
 -W_{S^{e_{2};e_{1}}(\sk)}-\frac{1}{2N}W_{F^{e_{2};e_{1}}(\sk)} & \mbox{if } \K=\H.\\
\end{array}\right.\]
If, on the contrary, $e_{1}$ and $e_{2}$ are traversed by distinct loops of $\sk$, then
\[N^{2}\Delta^{e_{2};e_{1}}W_{\sk}=\left\{\begin{array}{ll} 
-W_{S^{e_{2};e_{1}}(\sk)}+W_{F^{e_{2};e_{1}}(\sk)} & \mbox{if } \K=\R,\\[2pt]
-W_{S^{e_{2};e_{1}}(\sk)} & \mbox{if } \K=\C,\\[2pt]
 -\frac{1}{4}W_{S^{e_{2};e_{1}}(\sk)}+\frac{1}{4}W_{F^{e_{2};e_{1}}(\sk)} & \mbox{if } \K=\H.\\
\end{array}\right.\]
\end{proposition}

\begin{proof} 1.  The easiest way to derive this relation is to start from \eqref{wilson rho} and to use the definition of $c_{\u(N,\K)}$ given by \eqref{sum square unif}. In the orthogonal case for example, we have, for all $h\in \Conf^{\G}_{\U(N,\R)}$, and with $d=\frac{N(N-1)}{2}$,
\begin{align*}
\Delta^{e_{1}}W_{\sk}^{\R,N}(h)&=N^{-r}\sum_{k=1}^{d}\frac{d^{2}}{dt^{2}}_{|t=0} \Tr^{\otimes n}(\rho_{\R}(\lambda_{\sk})\circ h(e_{1}) e^{tX_{k}}\otimes \ldots \otimes h(e_{n})) \\
&= N^{-r}\Tr^{\otimes n}(\rho_{\R}(\lambda_{\sk})\circ c_{\u(N,\R)}h(e_{1})\otimes \ldots \otimes h(e_{n}))\\
&=c_{\u(N,\R)}W_{\sk}^{\R,N}(h).
\end{align*}

2. For all $X,Y\in \g$, the operators $\D_{X}^{e_{1}}$ and $\D_{Y}^{e_{2}}$ commute. The equality follows immediately.

3. Let us assume that $e_{1}\notin \E^{+}$. Let us consider $h\in \Conf^{\G}_{\U(N,\K)}$. Then $W^{\K,N}_{\sk}(h)$, according to its initial definition, is a product of traces, one of which involves $h(e_{1}^{-1})$ and $h(\lambda_{\sk}(e_{1}^{-1}))$. Thus, for all $X\in \g$ and thanks to \eqref{dex}, we have
\begin{align*}
\D^{e_{1}}_{X} W^{\K,N}_{\sk}(h)&= \ldots \frac{d}{dt}_{|t=0} \tr(h(\lambda_{\sk}(e_{1}^{-1}))e^{-tX}h(e_{1}^{-1}) \ldots) \ldots \\
&=\D^{\lambda_{\sk}(e_{1}^{-1})}_{-X}W^{\K,N}_{\sk}(h),
\end{align*}
or the same with $\tr$ replaced by $\Re\tr$ if $\K=\H$. This implies the desired equality.

4. The enumeration of $\E^{+}$ which we have chosen here does not play any particular role, and we may assume that it is compatible with our choice of the edges $e_{1}$ and $e_{2}$. Using the definition of the operator $\Delta^{e_{2};e_{1}}$, we find, in close analogy with \eqref{ito general},
\[\Delta^{e_{2};e_{1}}W_{\sk}^{\K,N}=N^{-r}\Tr^{\otimes n}(\iota_{1,2}(C_{\u(N,\K)})\circ \rho_{\K}(\lambda_{\sk})\circ  h(e_{1})\otimes \ldots \otimes h(e_{r}))\]
if $\K=\R$ or $\C$, and
\[\Delta^{e_{2};e_{1}}W_{\sk}^{\H,N}=(-2N)^{-r}(-2\Re\Tr)^{\otimes n}(\iota_{1,2}(C_{\u(N,\H)})\circ \rho_{\H}(\lambda_{\sk}) \circ h(e_{1})\otimes \ldots \otimes h(e_{r}))\]
if $\K=\H$. Note that the operator $\D_{X}^{e_{1}}$ multiplies $h(e_{1})$ on the right by $X$, so that $\rho(\lambda_{\sk})$ is multiplied by $\iota_{1,2}(C_{\u(N,\H)})$ on the left. Note also that these relations could be seen as instances of \eqref{delta spin}.

Thanks to the expressions \eqref{casimir general}, \eqref{casimir o u},  \eqref{iotacasimir so} and \eqref{iotacasimir} of the Casimir operators, we know that
\[N\iota_{1,2}(C_{\u(N,\K)})=\left\{\begin{array}{ll} -\rho_{\R}(1\, 2)+\rho_{\R}\langle 1\, 2\rangle & \mbox{if } \K=\R,\\[1pt]
- \rho_{\C}(1\, 2)  & \mbox{if } \K=\C, \\[1pt]
\frac{1}{2}\rho_{\H}(1\, 2)-\frac{1}{2}\rho_{\H}\langle 1\, 2\rangle & \mbox{if } \K=\H.
\end{array}\right.\]

Since the mappings $\rho_{\K}$ are homomorphisms of algebra, we can perform the computation easily. If $\K=\C$, then
\[\iota_{1,2}(C_{\u(N,\C)}) \rho_{\C}(\lambda_{\sk})=-\frac{1}{N}\rho_{\C}((e_{1}\, e_{2}) \lambda_{\sk}).\]
One verifies easily that $(e_{1}\, e_{2}) \lambda_{\sk}=\lambda_{S^{e_{2};e_{1}}(\sk)}$. Thus, 
\[\Delta^{e_{2};e_{1}}W_{\sk}^{\C,N}=-N^{-r-1}\Tr^{\otimes n}( \rho_{\C}(\lambda_{S^{e_{2};e_{1}}(\sk)})\circ  h(e_{1})\otimes \ldots \otimes h(e_{r})).\]
If $e_{1}$ and $e_{2}$ are traversed by the same loop, then the skein $S^{e_{2};e_{1}}(\sk)$ contains $r+1$ loops and the right-hand side is equal to $-W_{S^{e_{2};e_{2}}(\sk)}^{\C,N}$. If $e_{1}$ and $e_{2}$ are traversed by different loops, then $S^{e_{2};e_{1}}(\sk)$ contains $r-1$ loops and the right-hand side is equal to $-N^{-2}W_{S^{e_{2};e_{2}}(\sk)}^{\C,N}$. In both cases, this is the desired equality.

If $\K=\R$, then
\[\iota_{1,2}(C_{\u(N,\R)}) \rho_{\R}(\lambda_{\sk})=-\frac{1}{N}\rho_{\R}((e_{1}\, e_{2}) \lambda_{\sk})+\frac{1}{N} \rho_{\R}(\langle e_{1}\, e_{2} \rangle \lambda_{\sk}).\]
The first term can be treated as in the complex case. For the second term, using \eqref{P cycles o}, we find that for all $h\in \Conf^{\G}$,
$N^{-r}\Tr^{\otimes n}(\rho_{\R}(\langle e_{1}\, e_{2} \rangle \lambda_{\sk})\circ h(e_{1})\otimes \ldots \otimes h(e_{n}))$ is equal either to $W_{F^{e_{2};e_{1}}(\sk)}^{\R,N}(h)$, if $e_{1}$ and $e_{2}$ are traversed by the same loop of $\sk$, or to $\frac{1}{N}W_{F^{e_{2};e_{1}}(\sk)}^{\R,N}(h)$ if they are not.

Finally, if $\K=\H$, then
\[\iota_{1,2}(C_{\u(N,\R)}) \rho_{\H}(\lambda_{\sk})=-\frac{1}{-2N}\rho_{\H}((e_{1}\, e_{2}) \lambda_{\sk})+\frac{1}{-2N} \rho_{\H}(\langle e_{1}\, e_{2} \rangle \lambda_{\sk}).\]
Lemma \ref{tr tr sp} allows us to conclude the proof as in the orthogonal case.
\end{proof}

\subsection{Area derivatives of Wilson loops III: Wilson garlands}\label{wgarland}

Proposition \ref{derive skeins} shows among other things that the space of observables spanned linearly by Wilson loops is not stable under the action of the differential operators $\Delta^{e_{2};e_{1}}$. This justifies a posteriori the fact that we introduced Wilson skeins. Unfortunately, our fundamental derivation formula \eqref{standard deriv} involves not only these operators but the operators $\Delta^{e_{2};c,e_{1}}$. Wilson skeins form a class of observables which is not stable under the action of these more general operators. This is why we need to enlarge a second time the class of observables which we consider, and define Wilson garlands. This will fortunately be the last enlargement.

Let $\G$ be a graph. Let $\T\subset \E$ be a spanning tree of $\G$.
We say that a collection of reduced loops $\gar=\{l_{1},\ldots,l_{r}\}$ on $\G$ is a {\em garland} on $\G$ with respect to $\T$, or a garland on $(\G,\T)$, if in each pair $\{e,e^{-1}\}$ of edges contained in $\E\setminus \T$, exactly one edge is traversed exactly once by exactly one of the loops $l_{1},\ldots,l_{r}$. The edges of $\T$, on the other hand, can be traversed many times and by more than one loop. 

Given a garland, there exists a least fine graph on which it can be traced. However, the garland does not necessarily determine a spanning tree on this graph. For example, if $\sk$ is a skein, then it is a garland on $\G_{\sk}$ with respect to any spanning tree of $\G_{\sk}$.

Given a garland $\gar=\{l_{1},\ldots,l_{r}\}$ on $(\G,\T)$, we naturally define the {\em Wilson garland} $W_{N,\gar}^{\K}:\Conf^{\G}_{\U(N,\K)}\to \C$ by $W_{N,\gar}^{\K}=W_{N,l_{1}}^{\K}\ldots W_{N,l_{r}}^{\K}$.

The definition of the permutation $\lambda_{\sk}$ which we associated to a skein $\sk$ can be extended to the case of garlands. Let $\gar=\{l_{1},\ldots,l_{r}\}$ be a garland on $(\G,\T)$. The set of edges which are not in $\T$ and which are traversed by the loops of $\gar$ form a partial orientation $(\E\setminus \T)^{+}$ of $\G$.
Besides, the order in which the loops of $\gar$ traverse the edges of $(\E\setminus \T)^{+}$ determines a permutation $\lambda_{\gar}$ of $(\E\setminus \T)^{+}$ and we claim that this permutation suffices to determine the Wilson garland $W^{\K}_{N,\gar}$. 

The most natural way to prove this claim involves the gauge invariance of Wilson loops. We mentioned already that Wilson loops are invariant observables, and it follows immediately that Wilson skeins and Wilson garlands are also invariant functions on $\Conf^{\G}_{G}$.

Let $(e_{i_{1}}\, \ldots e_{i_{r}})$ be the cycle of $\lambda_{\gar}$ corresponding to the loop $l_{1}$. Write $l_{1}=ce_{i_{1}}d$, with $c$ and $d$ two appropriate paths. Then $c^{-1} l_{1} c=e_{i_{1}}dc$ is equivalent to the loop
\[e_{i_{1}} [\overline{e_{i_{1}}},\underline{e_{i_{2}}}]_{\T} e_{i_{2}} \ldots e_{i_{r-1}} [\overline{e_{i_{r-1}}},\underline{e_{i_{r}}}]_{\T} e_{i_{r}} [\overline{e_{i_{r}}},\underline{e_{i_{1}}}]_{\T}.\]
Hence, for each configuration $h\in \Conf^{\G}_{\U(N,\K)}$, we have
\[W_{\N,l_{1}}^{\K}(h)=W_{\N,l_{1}}^{\K}(j_{h,\T}\cdot h)=\tr((j_{h,\T}\cdot h) (e_{i_{n}}) \ldots (j_{h,\T}\cdot h) (e_{i_{1}})),\]
where $j_{h,\T}$ is the gauge transformation defined by \eqref{def jhT}. It follows that  
\begin{equation}\label{wilson gar rho}
W_{N,\gar}^{\K}(h)=\left\{\begin{array}{ll} N^{-r}\Tr^{\otimes n}(\rho_{\K}(\lambda_{\gar})\circ (j_{h,\T}\cdot h)(e_{1})\otimes \ldots \otimes (j_{h,\T}\cdot h)(e_{n})) & \mbox{if } \K=\R \mbox{ or } \C,\\[2pt]
(-2N)^{-r}(-2\Re\Tr)^{\otimes n}(\rho_{\H}(\lambda_{\gar})\circ (j_{h,\T}\cdot h)(e_{1})\otimes \ldots \otimes (j_{h,\T}\cdot h)(e_{n})) & \mbox{if } \K=\H,\end{array}\right.
\end{equation}
which is for garlands what \eqref{wilson rho} was for skeins and gives us an explicit formula for $W_{N,\gar}^{\K}$ in terms of the permutation $\lambda_{\gar}$.

In order to state a result similar to Proposition \ref{derive skeins}, let us extend the operations $S$ and $F$ to garlands. Let $\gar=\{l_{1},\ldots,l_{r}\}$ be a garland on $(\G,\T)$. Let $e_{1}$ and $e_{2}$ be two edges of  $\G$. Let $c$ be a path in $\T$ which joins the starting point of $e_{2}$ to the starting point of $e_{1}$. Let us assume that $e_{1}$ and $e_{2}$ belong to $(\E\setminus T)^{+}$, and that $e_{1}$ is traversed by the loop $l_{1}$. Let us first treat the case where $e_{2}$ is also traversed by $l_{1}$, and $l_{1}$ traverses $e_{1}$ before traversing $e_{2}$. We can write $l_{1}=ae_{1}be_{2}d$, where $a,b,d$ are paths in $\G$. Let us define $l'_{1}$, $l''_{1}$ and $\tilde l_{1}$ respectively as the reduced loops equivalent to $e_{1}bc$, $e_{2}dac^{-1}$ and $e_{1}bc(e_{2} da)^{-1}c$ (see Figure \ref{guirlandes} below). We define
\[S^{e_{2};c,e_{1}}(\gar)=\{l'_{1},l''_{1},l_{2},\ldots,l_{r}\} \mbox{ and } F^{e_{2};c,e_{1}}(\gar)=\{\tilde l_{1},l_{2},\ldots,l_{r}\}.\]
If $l_{1}$ traverses $e_{2}$ before $e_{1}$, then $l=ae_{2}be_{1}d$ for some paths $a,b,d$. We define $l'_{1}$, $l''_{1}$ and $\tilde l_{1}$ as the reduced loops equivalent to $e_{1}dac$, $e_{2}bc^{-1}$ and $e_{1}dac(e_{2}b)^{-1}c$ respectively and use the same definition as above for $S^{e_{2};c,e_{1}}(\gar)$ and $F^{e_{2};c,e_{1}}(\gar)$.

Finally, if $e_{2}$ is not traversed by $l_{1}$, we may assume that it is traversed by $l_{2}$. Let us write the loops as $l_{1}=ae_{1}b$ and $l_{2}=de_{2}f$, where $a,b,d,f$ are paths. We define $l'$ and $\tilde l$ as the reduced loops equivalent to $e_{1}bac^{-1}e_{2}fdc$ and $e_{1}bac^{-1}(e_{2}fd)^{-1}c$ respectively, and set
\[S^{e_{2};c,e_{1}}(\gar)=\{l',l_{3},\ldots,l_{r}\} \mbox{ and } F^{e_{2};c,e_{1}}(\gar)=\{\tilde l,l_{3},\ldots,l_{r}\}.\]
One checks easily in all cases that $S^{e_{2};c,e_{1}}(\gar)$ and $F^{e_{2};c,e_{1}}(\gar)$ are still garlands on $(\G,\T)$.

\begin{figure}[h!]
\begin{center}
\includegraphics{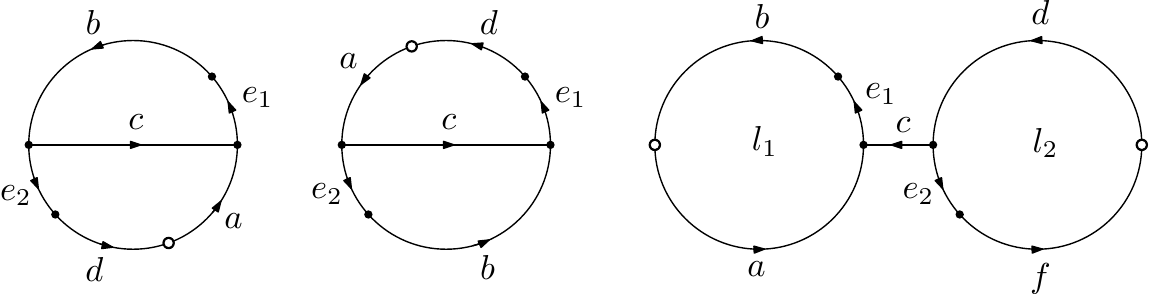}
\caption{\label{guirlandes} The paths involved in the definition of  $S^{e_{2};c,e_{1}}(\gar)$ and $F^{e_{2};c,e_{1}}(\gar)$, in the three successive cases which we considered. The basepoints of loops are indicated by a white vertex.}
\end{center}
\end{figure}

We can now prove that the linear space of smooth complex-valued functions on $\Conf^{\G}_{\U(N,\K)}$ spanned by the Wilson garlands with respect to a given spanning tree $\T$ is stable under the action of the operators $\Delta^{e_{2};c,e_{1}}$, where $e_{1}$ and $e_{2}$ are edges of $\E\setminus \T$ and $c$ is a path in $\T$.

\begin{proposition}\label{derive garlands} Let $\G$ be a graph and $\T$ a  spanning tree of $\G$. Let $\gar$ be a garland on $(\G,\T)$. Let $(\E\setminus \T)^{+}$ be the partial orientation of $\G$ induced by $\gar$. Let $e_{1},e_{2}$ be two distinct edges of $\E\setminus \T$. Let $c$ be a path in $\T$ from the starting point of $e_{2}$ to the stating point of $e_{1}$. The following properties hold.\\
\indent 1. If $e_{1}\notin (\E\setminus \T)^{+}$, then $\lambda_{\gar}(e_{1}^{-1}) \in (\E\setminus \T)^{+}$. Moreover, $\Delta^{e_{2};c,e_{1}}W_{N,\gar}^{\K}=-\Delta^{e_{2};c',\lambda_{\gar}(e_{1}^{-1}))}W_{N,\gar}^{\K}$,
where $c'=c [\underline{e_{1}},\underline{\lambda_{\gar}(e_{1}^{-1})}]_{\T}$.\\
\indent 2. If $e_{2}\notin (\E\setminus \T)^{+}$, then $\lambda_{\gar}(e_{2}^{-1}) \in (\E\setminus \T)^{+}$. Moreover, $\Delta^{e_{2};c,e_{1}}W_{N,\gar}^{\K}=-\Delta^{(\lambda_{\gar}(e_{2}^{-1}))(c'',e_{1})}W_{N,\gar}^{\K}$,
where $c''=[\underline{\lambda_{\gar}(e_{2}^{-1})},\underline{e_{2}}]_{\T}c$.\\
\indent 3. Let us assume that $e_{1}$ and $e_{2}$ belong to $(\E\setminus \T)^{+}$. The fourth assertion of Proposition \ref{derive skeins} holds after substituting everywhere $W_{N,\sk}^{\K}$ by $W_{N,\gar}^{\K}$, $\Delta^{e_{2};e_{1}}$ by $\Delta^{e_{2};c,e_{1}}$, $S^{e_{2};e_{1}}$ by $S^{e_{2};c,e_{1}}$ and $F^{e_{2};e_{1}}$ by $F^{e_{2};c,e_{1}}$.
\end{proposition}

\begin{proof} Let us use \eqref{wilson gar rho} to prove the first assertion. Let us assume that $e_{1}\notin (\E\setminus \T)^{+}$. We proceed as in the proof of the third assertion of Proposition \ref{derive skeins}, and use moreover the fact, granted by Lemma \ref{invariance op}, that all the functions which we consider are invariant. Let $h$ be an element of the configuration space $\Conf^{\G}_{\U(N,\K)}$. We have
\begin{align*} \D^{c,e_{1}}_{X} W^{\K}_{N,\gar}(h)&= \D^{c,e_{1}}_{X} W^{\K}_{N,\gar}(j_{h,T}\cdot h)\\
&= \ldots \frac{d}{dt}_{|t=0} \tr((j_{h,T}\cdot h)(\lambda_{\gar}(e_{1}^{-1}))e^{-t\Ad((j_{h,T}\cdot h)(c))X}(j_{h,T}\cdot h)(e_{1}^{-1}) \ldots) \ldots.
\end{align*}
Since $c$ and $c'$ are paths in $\T$, we have $(j_{h,T}\cdot h)(c)=(j_{h,T}\cdot h)(c')=I_{N}$. Hence, 
\begin{align*}
\D^{c,e_{1}}_{X} W^{\K}_{N,\gar}(j_{h,T}\cdot h)&=\D^{\lambda_{\gar}(e_{1}^{-1})}_{-\Ad((j_{h,T}\cdot h)(c'))X}W^{\K}_{N,\gar}(j_{h,T}\cdot h)\\
&=-\D^{c',\lambda_{\gar}(e_{1}^{-1})}_{X}W^{\K}_{N,\gar}(j_{h,T}\cdot h).
\end{align*}
The proof of the second assertion is very similar.

The proof of the third assertion is an adaptation of the proof of the fourth assertion of Proposition \ref{derive skeins}. To start with, if $\K\in\{\R,\C\}$, then
\begin{equation}\label{delta garland proof}
\Delta^{e_{2};c,e_{1}}W_{N,\gar}^{\K}=N^{-r}\Tr^{\otimes n}(\iota_{1,2}\left((\Ad(h(c))\otimes \id_{\g})C_{\u(N,\K)}\right)\circ \rho_{\K}(\lambda_{\gar})\circ  h(e_{1})\otimes \ldots \otimes h(e_{r})),
\end{equation}
and the same formula holds with the usual replacement of $N^{-r}\Tr^{\otimes n}$ by $(-2N)^{-r}(\Re\Tr)^{\otimes n}$ if $\K=\H$. This equality can be proved directly from the definition or as a consequence of \eqref{delta spin c}.

In order to compute the right-hand side, we write it in a way which allows us to use directly the computations which we made in the proof of Proposition \ref{derive skeins}. Recall from \eqref{def theta} the definition of the operations $\theta^{\pm}_{i}$ of left and right multiplication. We will use two simple identities. The first is valid for any $M\in \Mat_{N}(\K)^{\otimes n}$ and writes
\[\iota_{1,2}\left((\Ad(h(c))\otimes \id_{\g})C_{\u(N,\K)}\right)\circ M=\theta^{+}_{1}(h(c))\cdot\left[ \iota_{1,2}(C_{\u(N,\K)})\circ \left( \theta^{+}_{1}(h(c)^{-1})\cdot M\right)\right],\]
where we recall that the symbol $\circ$ denotes the multiplication in the algebra $\Mat_{N}(\K)^{\otimes n}$. The second identity, which holds for all permutation $\sigma\in \S_{n}$ and all $X\in \Mat_{N}(\K)$, is
\[\theta^{+}_{i}(X) \cdot \left(\rho_{\K}(\sigma) \circ M\right)=\rho_{\K}(\sigma)\circ  \left(\theta^{+}_{\sigma^{-1}(i)}(X)\cdot M\right).\]
Applying the first identity to $M=\rho_{\K}(\lambda_{\gar}) \circ h(e_{1})\otimes \ldots \otimes h(e_{r})$, then the second identity to $i=1$, $\sigma=\lambda_{\gar}$ and $M=h(e_{1})\otimes \ldots \otimes h(e_{r})$, and using the adjunction relation \eqref{adjunction}, we find that the right-hand side of \eqref{delta garland proof} is equal to
\[N^{-r}\Tr^{\otimes n}(\iota_{1,2}\left(C_{\u(N,\K)}\right)\circ \rho_{\K}(\lambda_{\gar})\circ  h(ce_{1})\otimes \ldots \otimes h(e_{j}c^{-1})\otimes \ldots \otimes h(e_{r})),\]
where $e_{j}=\lambda_{\gar}^{-1}(e_{1})$.

The end of the proof is similar to that of the fourth assertion of Proposition \ref{derive skeins}. It suffices to compute
\[N^{-r}\Tr^{\otimes n}(\iota_{1,2}\left(C_{\u(N,\K)}\right)\circ \rho_{\K}((e_{1}\, e_{2})\lambda_{\gar})\circ  h(ce_{1})\otimes \ldots \otimes h(e_{j}c^{-1})\otimes \ldots \otimes h(e_{r}))\]
and
\[N^{-r}\Tr^{\otimes n}(\iota_{1,2}\left(C_{\u(N,\K)}\right)\circ \rho_{\K}(\langle e_{1}\, e_{2} \rangle\lambda_{\gar})\circ  h(ce_{1})\otimes \ldots \otimes h(e_{j}c^{-1})\otimes \ldots \otimes h(e_{r})),\]
and to check that they are respectively equal, up to a power of $N$ which is determined exactly as in the case of skeins, to $W^{\K}_{N,S^{e_{2};c,e_{1}}(\gar)}$ and $W^{\K}_{N,F^{e_{2};c,e_{1}}(\gar)}$.
\end{proof}

\subsection{Expectations of Wilson loops}\label{exp wl} In the present section, we finally write down and solve the differential system announced in Section \ref{outline}. This enables us, at least in principle, to compute the expectation of any Wilson loop, skein, or garland, for any $\K\in \{\R,\C,\H\}$ and any integer $N$.

Given a graph $\G$ and a spanning tree $\T$ of $\G$, let us denote by $\ggar(\G,\T)$ the set of garlands on $(\G,\T)$.

\begin{lemma} Let $\G$ be a graph. Let $\T$ be a spanning tree of $\G$. The set $\ggar(\G,\T)$ is finite.
\end{lemma}

\begin{proof} Since the loops which constitute a garland are assumed to be reduced, a garland $\gar$ on $(\G,\T)$ is almost completely determined by the permutation $\lambda_{\gar}$ which it induces on $(\E\setminus \T)^{+}$. The only information about $\gar$ which is missing in $\lambda_{\gar}$ is the data of the base point of each loop, and this can be chosen in only finitely many different ways.
\end{proof}

The following proposition summarises the results of the last three sections.

\begin{theorem}\label{matrice M} Choose $\K\in \{\R,\C,\H\}$ and an integer $N\geq 1$. Let $\G$ be a graph. Let $\T$ be a spanning tree of $\G$. For each face $F$ of $\G$, there exists a $\ggar(\G,\T)\times \ggar(\G,\T)$ matrix, which depends on $\K,N,\G,\T,F$, and which we simply denote by $M_{F}^{\K,N}$, such that for all $t:\F^{b}\to \R^{*}_{+}$,
\[\left(\frac{d}{d|F|}-M_{F}^{\K,N}\right) \left(\E_{\YM^{\G}_{t}}\left[W_{\gar}^{\K,N}\right] : \gar \in \ggar(\G,\T)\right)=0.\]
\end{theorem}

\begin{proof} Recall from Section \ref{grouploops} the way in which the spanning tree $\T$ determines a spanning tree $\wT$ of the dual graph $\wG$ and, once we have chosen a first neighbour of the dual root, a labelling of the set $\F$ of faces by words of integers. We use this structure on $\F$ to determine the specific sequence of faces to which we shall apply Corollary \ref{main deriv cor}.

If $k_{1}\ldots k_{p}$ is a word of integers corresponding to a face of $\G$, we denote by $c(k_{1}\ldots k_{p})$ the number of children of $k_{1}\ldots k_{p}$ in $\wT$, that is, the largest integer $l$ such that $k_{1}\ldots k_{p}l$ corresponds to a face of $\G$. 
If $l\in \{0,\ldots,c(k_{1}\ldots k_{p})\}$, we define 
\[s(k_{1}\ldots k_{p},l)=\left\{\begin{array}{ll} (k_{1}\ldots k_{p}l,c(k_{1}\ldots k_{p}l))  & \mbox{if } l>0,\\
 (k_{1}\ldots k_{p-1},k_{p}-1)& \mbox{if } l=0.  \\
 \end{array}\right.\] 
Starting from $(\varnothing,c(\varnothing))$ and iterating $s$ until one reaches $(\varnothing,0)$, whose image by $s$ is not defined, corresponds to the exploration of the dual tree by a person who keeps it on her right-hand side. 

\begin{figure}[h!]
\begin{center}
\scalebox{1}{\includegraphics{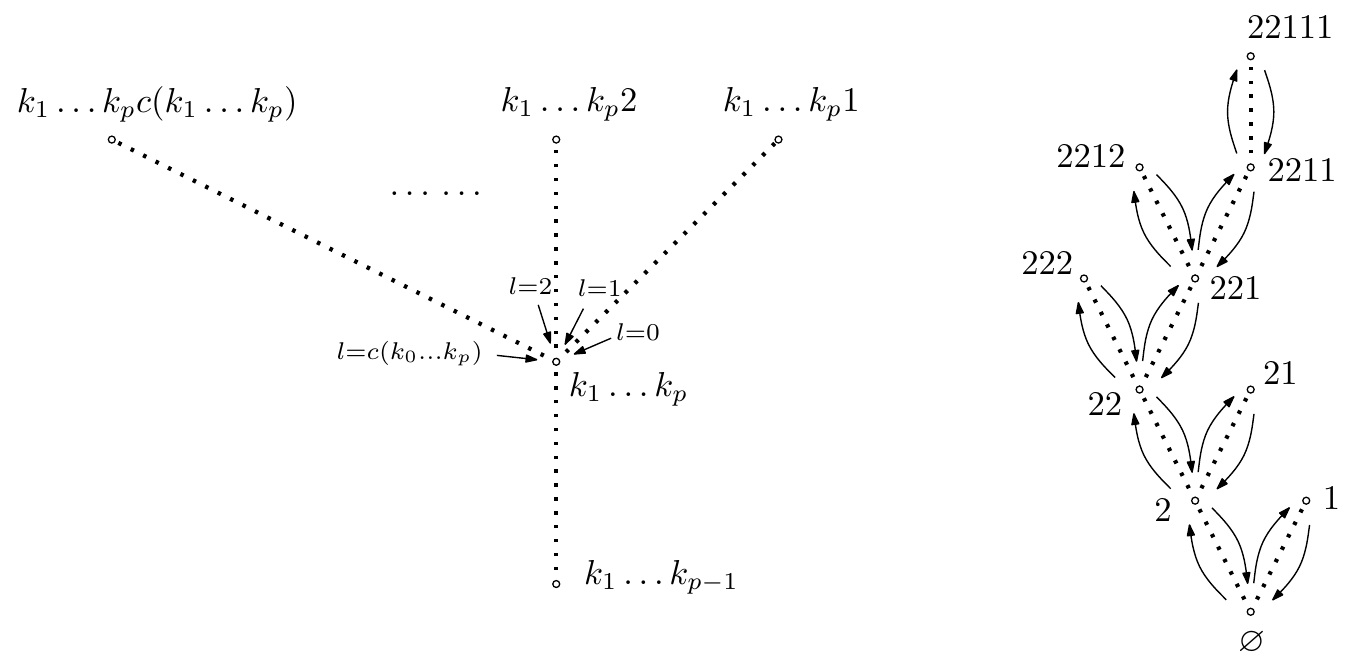}}
\caption{\label{suitefaces} The left-hand side explains the meaning of the integer $l$, namely the location of the explorer around the vertex which is currently visited. The right-hand side shows the trajectory of the left-handed exploration of the dual tree of the example depicted in Figure \ref{ex base}.}
\end{center}
\end{figure}

Let us consider a bounded face $F$ and its label $k_{1}\ldots k_{p}$. Let us construct a sequence of faces by starting from $(k_{1}\ldots k_{p},0)$, iterating $s$ until we reach the unbounded face for the first time, and forgetting the values of $l$ in each term of the sequence obtained. We find a sequence $F=F_{1},F_{2},\ldots,F_{n},F_{n+1}=F_{\infty}$. For example, if we use the graph depicted in Figure \ref{ex base} and start from the face $22$, we find the sequence 
$(22,2,21,2,\varnothing)$.

Each face of the sequence $(F_{1},\ldots,F_{n+1})$ is adjacent to the next, for they correspond to adjacent vertices in the dual spanning tree. For each $r\in \{1,\ldots,n\}$, let $e_{r}$ be the edge of $\G$ such that $(F_{r+1},e_{r},F_{r})$ is the dual edge which joins $F_{r}$ and $F_{r+1}$ in $\wT$.

The fact that we chose the sequence of faces by right-handed exploration of the spanning tree $\wT$ implies that the paths $c_{2},\ldots,c_{n}$ defined in the statement of Proposition \ref{main deriv} are paths in $\T$. Indeed, for each $r\in\{2,\ldots,n\}$, the dual vertex $\hat F_{r+1}$ is immediately followed by $\hat F_{r-1}$ in the cyclic order of the neighbours of $\hat F_{r}$ in $\wT$. Hence, the path $c_{r}$ does not cross any dual edge of $\wT$, and this is equivalent to saying that it is contained in $\T$ (see Figure \ref{ttdual}).

\begin{figure}[h!]
\begin{center}
\scalebox{1}{\includegraphics{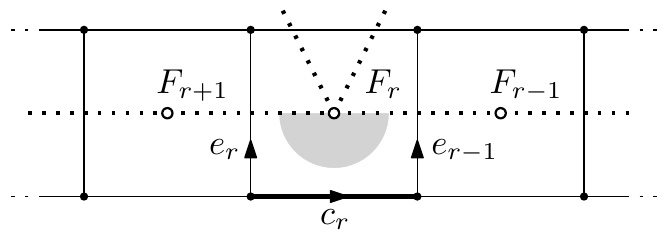}}
\caption{\label{ttdual} By construction of the sequence $F_{1},\ldots,F_{n+1}$, the grey sector does not contain any edge of the dual spanning tree $\wT$. Hence, the path $c_{r}$ stays in the spanning tree $\T$.}
\end{center}
\end{figure}

We now apply Corollary \ref{main deriv cor} to the sequence $F_{1},\ldots,F_{n+1}$, the edges $e_{1},\ldots,e_{r}$, once for each observable of the form $W^{\K}_{N,\gar}$, where $\gar$ spans $\ggar(\G,T)$. By Proposition \ref{derive garlands}, each derivative with respect to the area of $F$ of such an observable can be expressed as a linear combination of Wilson garlands belonging to $\ggar(\G,T)$. The coefficients of these linear combinations form the coefficients of the matrix $M_{F}^{\K,N}$.
\end{proof}

Let us explain how to apply Theorem \ref{matrice M} to the computation of the expectation of an elementary Wilson loop. Let $l$ be an elementary loop. Let $\T$ be a spanning tree of $\G_{l}$, the least fine graph on which $l$ can be traced. Then $\{l\}$ is a garland on $(\G,\T)$. Hence, the expectation of the Wilson loop $W_{l}$ is one of the components of the vector $\left(\E_{\YM^{\G}_{t}}\left[W_{\gar}^{\K,N}\right] : \gar \in \ggar(\G,T)\right)$, evaluated at $t:\F^{b}\to \R^{*}_{+}$ given by $t(F)=|F|$. Let us denote by $\1$ the vector of size $\ggar(\G,T)$ with all its components equal to $1$. By Proposition \ref{matrice M} and Proposition \ref{value at 0}, we have the following equality, which is very much analogous to \eqref{inflate pt}:
\begin{equation}\label{inflate W}
\left(\E_{\YM^{\G}_{t}}\left[W_{\gar}^{\K,N}\right] : \gar \in \ggar(\G,T)\right)=\left(\prod_{F\in \F^{b}}e^{t(F) M_{F}^{\K,N}}\right)\1.
\end{equation}
The order in which the matrices are multiplied does not matter, since the derivatives with respect to the areas of the various faces commute.

This is a formula of the sort we were aiming at: it provides us with a graphical procedure to compute the expectation of products of Wilson loops. 

It is however rather impractical. The number of garlands on a given pair $(\G,\T)$ is large: if all the vertices of $\G$ have degree $4$, which is the generic case, and even after identifying garlands which differ only by the base points of their constituting loops, there are $2^{q}q!$ garlands, where $q$ is the number of bounded faces of $\G$. In the case of Example \ref{ex 3}, where $q=3$, this is already too many for one to expect to be able to write down the full system by hand. In the last sections, we explain how, in the large $N$ limit, this procedure can be greatly simplified.

\subsection{The Makeenko-Migdal equations for the Yang-Mills field}\label{section extended gauge} The reason why 
we were led to introduce the arguably not very natural class of observables which we called Wilson garlands, and ended up with such an impractical procedure as the one which is summarised by the formula \eqref{inflate W}, is that formulas \eqref{standard deriv} and \eqref{local derive} share the following unpleasant feature: in order to express the derivative of the expectation of an observable with respect to the area of a single face, they involve derivatives of the observable with respect to edges which may be located very far from this face, indeed all the edges located on a path from this face to the unbounded face. 

In the remaining sections, we elaborate on previous work of Makeenko and Migdal \cite{MakeenkoMigdal}, Kazakov \cite{Kazakov}, Kazakov and Kostov \cite{KazakovKostov}, and describe a much more efficient way of computing the master field. 

The main discovery of Makeenko and Migdal is that the alternated sum of the derivatives of the expectation of a Wilson loop with respect to the areas of the faces which surround a given vertex can be described by local transformations of the loop at the vertex considered. We will understand this as a consequence of the fact that the differential operators which involve edges located far away cancel out. The original statement of the Makeenko-Migdal relation was essentially pictorial, and its proof was based on an ill-defined path integral with respect to the continuous Yang-Mills measure over the space of gauge fields. In the rest of this paper, we prove a more general and rigorous version of these equations, and apply them to produce an efficient algorithm for computing the master field.

Let us describe briefly the content of each of the next sections. In the present section, we describe a general framework in which cancellations of this sort happen. This turns out to be related with properties of invariance of the observable under consideration with respect to the action of a group larger than the gauge group. Proposition \ref{derive loc inv}, which is the abstract form which we propose for the generalised Makeenko-Migdal equations, is valid for an arbitrary compact connected gauge group $G$. In Section \ref{mm large N}, we apply the general result obtained in Section \ref{section extended gauge} to the specific case of the unitary group, and recover the Makeenko-Migdal equations for the master field. Then, in Section \ref{recur mf}, we prove that the Makeenko-Migdal equations contain enough information to enable one to compute the value of the master field on any elementary loop. Indeed, the Makeenko-Migdal equations give simple expressions for various linear combinations of the area-derivatives of the master field, but it must be shown that the set of linear combinations which is thus made accessible is large enough to generate the space of all area-derivatives. We prove that this is the case, essentially by proving that some finite-dimensional linear mapping is injective. We also show that for finite $N$, the injectivity fails, and there is in general a part of the information missing, so that we cannot give a better algorithm than that encoded in \eqref{inflate W}. Finally, in Section \ref{sec Kazakov}, we use a change of variables due to Kazakov to explicitly produce a left inverse of the linear mapping of which we proved, in the case of the master field, that it is injective. \\

Let $\G=(\V,\E,\F)$ be a graph. Recall from Section \ref{grouploops} that the gauge group $G^\V$ acts on $\M(\Path(\G),G)$ according to the following rule: given $j=(j(v))_{v\in\V} \in G^{\V}$ and a multiplicative function $h$, we have for all path $c$ the equality 
\[(j\cdot h)(c)=j(\overline{c})^{-1}h(c)j(\underline{c}).\]

Let us give an infinitesimal version of the gauge invariance of a function. For each vertex $v$, recall that we defined $\Out(v)=\{e\in \E : \underline{e}=v\}$ as the set of edges issued from $v$.

\begin{lemma}\label{inv infinitesimal} Let $f:\Conf^{\G} \to \R$ be a smooth invariant function. For all vertex $v$ and all $X\in \g$, we have
\[\sum_{e\in \Out(v)} \D_X^{e} f=0.\]
\end{lemma}

\begin{proof} Let $v\in \V$ be a vertex. Choose $X\in \g$. Consider the one-parameter subgroup of gauge transformations $(j_{t})_{t\in \R}$ defined by $j_{t}(v)=e^{tX}$ and $j_{t}(w)=1$ for all vertex $w\neq v$. Differentiating the equality $j_{t}\cdot f=f$ with respect to $t$ and evaluating at $t=0$ yields the desired equality.  \end{proof}

We want to consider invariant functions which are invariant under a larger symmetry group than the gauge group. We shall give natural examples of such functions in a moment.

\begin{definition}\label{extended gauge} Let $f:\Conf^{\G} \to \R$ be a smooth function. Let $v\in \V$ be a vertex. Let $I$ be a subset of $\Out(v)$. We say that $f$ is $I$-invariant at $v$ if for all $X\in \g$ we have the equality
\[\sum_{e\in I} \D_X^{e} f=0.\]
\end{definition}

We have seen that any invariant function is $\Out(v)$-invariant at each vertex $v$. It follows for instance that a smooth invariant function which is $I$-invariant at $v$ is also $(\Out(v)\setminus I)$-invariant. 

The simplest examples of functions which are $I$-invariant at some vertex $v$ for some proper subset $I$ of $\Out(v)$ are provided by Wilson loops. For example, let $l$ be a loop in $\G$ which visits exactly once the vertex $v$. Assume that $l$ arrives at $v$ through the edge $e_{1}^{-1}$ and leaves $v$ through the edge $e_{3}$ (see the left-hand side of Figure \ref{extinv}). Then the Wilson loop $W_{\chi,l}$ is invariant and $\{e_{1},e_{3}\}$-invariant at $v$. 

\begin{center}
\begin{figure}[h!]
\includegraphics{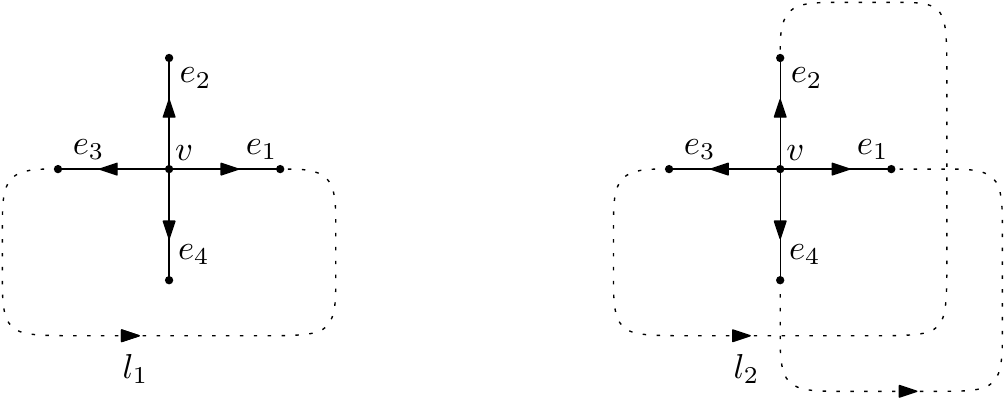}
\caption{\label{extinv} The Wilson loops associated to the loops $l_{1}$ and $l_{2}$ are both $\{e_{1},e_{3}\}$-invariant, and also $\{e_{2},e_{4}\}$-invariant, at $v$.}
\end{figure}
\end{center}

This example is however in a sense trivial, for the Wilson loop we chose does not depend at all on $e_{2}$ nor $e_{4}$. The next simplest example is also the fundamental one with the Makeenko-Migdal equations in mind. It is that of a loop which visits exactly twice the vertex $v$, once arriving through $e_{1}^{-1}$ and leaving through $e_{3}$, and once arriving through $e_{2}^{-1}$ and leaving through $e_{4}$ (see the right-hand side of Figure \ref{extinv}). This loop is also $\{e_{1},e_{3}\}$-invariant at $v$.

The next result shows that an observable which enjoys a property of local invariance as we just defined it also satisfies a local differential relation with respect to the areas of the faces of the graph. Recall that for each edge $e$ of a graph, we denote respectively by $F^{L}(e)$ and $F^{R}(e)$ the faces of the graph which are bounded positively and negatively by $e$. 
In the following statement and its proof, we denote by $\leq$ the cyclic order on $\Out(v)$ induced by the orientation of $\R^{2}$. If $e_{1}$ and $e_{2}$ are two elements of $\Out(v)$, we use the notation $[e_{1},e_{2}]=\{e\in \Out(v) : e_{1}\leq e \leq e_{2}\}$.

\begin{proposition}[Abstract Makeenko-Migdal equations]\label{derive loc inv} Let $\G=(\V,\E,\F)$ be a graph. Let $v\in \V$ be a vertex. Assume that each edge issued from $v$ is adjacent to two distinct faces of $\G$. Assume also that there exists at most one edge issued from $v$ which bounds positively the unbounded face. 

Let $f:\Conf^{\G} \to \C$ be a smooth function. Let $I$ be a non-empty subset of $\Out(v)$. Assume that $f$ is $I$-invariant at $v$. Let $e_{*}$ and $e^{*}$ be two edges in $\Out(v)$ such that $I \subset [e_{*},e^{*}]$.  Then for all $t:\F^{b}\to \R^{*}_{+}$, the following equality holds:
\begin{equation}\label{general mak mig}
\sum_{e\in I}\left(\frac{d}{d|F^{R}(e)|}-\frac{d}{d|F^{L}(e)|}\right) \E_{\YM^\G_{t}}[f]=\sum_{\substack{e_{*}\leq e_{1}<e_{2}\leq  e^{*} \\ e_{1}\in I, e_{2}\notin I}} \E_{\YM^\G_{t}}\left[\Delta^{e_{2};e_{1}}f\right],
\end{equation}
with the convention $\frac{d}{d|F_{\infty}|}\E_{\YM^\G_{t}}[f]=0$.
\end{proposition}

We shall prove in particular that the sum on the right-hand side of \eqref{general mak mig} does not depend on the choice of $e_{*}$ and $e^{*}$ such that $I \subset [e_{*},e^{*}]$. This sum has the least possible number of terms if $e_{*}$ and $e^{*}$ are chosen to belong to $I$, and such that $\Out(v)\setminus [e_{*},e^{*}]$ is a longest possible interval not meeting $I$.

\begin{figure}[h!]
\begin{center}
\includegraphics{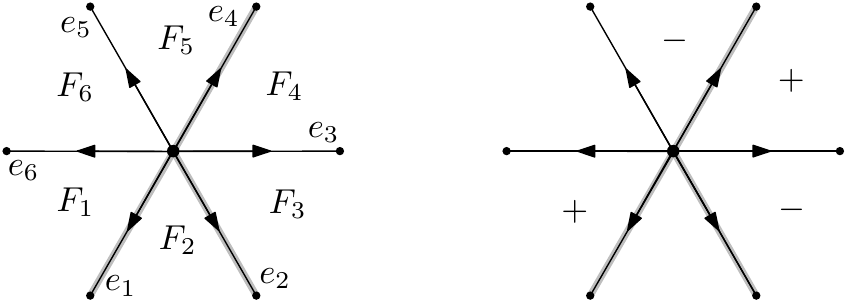}
\caption{\label{difflocal} In this example, $I=\{e_{1},e_{2},e_{4}\}$. The right-hand side of \eqref{general mak mig} is $\E_{\YM^\G_t}\left[\Delta^{e_{3};e_{1}}f+\Delta^{e_{3};e_{2}}f\right]$, which thanks to the $I$-invariance of $f$ is equal to $-\E_{\YM^\G_t}\left[\Delta^{e_{3};e_{4}}f\right]$. Hence, the equality in this case is
\[\left(\frac{d}{d|F_{1}|}-\frac{d}{d|F_{3}|}+\frac{d}{d|F_{4}|}-\frac{d}{d|F_{5}|}\right)\E_{\YM^\G_t}[f]=-\E_{\YM^\G_t}\left[\Delta^{e_{3};e_{4}}f\right].\]
}
\end{center}
\end{figure}

\begin{proof} If $\Out(v)$ contained only one edge, this edge would be adjacent to a unique face, contradicting our assumptions. Hence, $\Out(v)$ contains at least two elements. Moreover, if $I=\Out(v)$, then the equality holds because both sides are equal to $0$. Let us now assume that $I$ is a proper subset of $\Out(v)$. Let us also assume that $v$ is not adjacent to the unbounded face.

Let us enumerate $\Out(v)$ in its cyclic order around $v$, starting from the edge $e_{*}$ and stopping at the edge $e^{*}$, thus finding a sequence of edges $e_{1}=e_{*},e_{2},\ldots,e_{p}=e^{*}$. For each $i\in \{1,\ldots,p\}$, let $F_{i}=F^{R}(e_{i})$ denote the face adjacent to $v$ which is bounded negatively by $e_{i}$ (see Figure \ref{difflocal}). The proof consists in applying Proposition \ref{main deriv}, once for each edge of $I$, which by assumption appears in the sequence $e_{1},\ldots,e_{p}$.
We need to choose an appropriate path in the dual graph of $\G$, from the face $F_{1}$ to the unbounded face. We do this by considering the sequence $F_{1},\ldots, F_{p},F_{p+1}=F^{L}(e_{p})$, which we complete by an arbitrary sequence $F_{p+2},\ldots, F_{n+1}$, where $F_{n+1}$ is the unbounded face and $F_{k+1}\neq F_{k}$ for all $k\in \{p+1,\ldots,n\}$. Observe that the assumption that each edge issued from $v$ is adjacent to two distinct faces implies that no two successive faces are equal in the sequence $F_{1},\ldots,F_{p+1}$. For each $i\in \{1,\ldots,p\}$, we choose the edge $e_{i}$ as bounding $F_{i}$ negatively and $F_{i+1}$ positively. 

Let us choose $i\in \{1,\ldots,p\}$ such that $e_{i}\in I$ and let us apply Proposition \ref{main deriv} to the sequence $F_{i},\ldots, F_{n+1}$. Since the edges $e_{i},\ldots, e_{p}$ are all issued from $v$, the paths $c_{i+1},\ldots,c_{p}$ are constant. The paths $c_{p+1},\ldots,c_{n}$ on the other hand may not be, but they do not depend on $i$. Proposition \ref{main deriv} yields
\begin{equation}\label{une arete de I}
\left(\frac{d}{d|F_{i}|}-\frac{d}{d|F_{i+1}|}\right)  \E_{\YM^\G_{t}}[f] = \frac{1}{2}\E_{\YM^\G_{t}}[\Delta^{e_{i}}f] + \sum_{j=i+1}^{n}  \E_{\YM^\G_{t}}[\Delta^{e_{j};c_{j}\ldots c_{i+1},e_{i}}f].
\end{equation}
Summing the right-hand side over the indices $i$ such that $e_{i}\in I$ and splitting according to the values of $j$ gives
\begin{equation}\label{rhs}
\E_{\YM^\G_{t}}\left[\frac{1}{2} \sum_{e\in I} \Delta^{e}f+\sum_{e_{i}\in I, i<j\leq p}\Delta^{e_{j};e_{i}}f\right]
+ \sum_{k=1}^{d}\sum_{e\in I} \sum_{j=p+1}^{n}\E_{\YM^\G_{t}}\left[\D^{e_{j}}_{X_{k}}\D^{c_{j}\ldots c_{i+1},e}_{X_{k}} f\right].
\end{equation}
The last term of \eqref{rhs} can be rewritten as
\[  \sum_{k=1}^{d}\sum_{j=p+1}^{n}\E_{\YM^\G_{t}}\left[\D^{e_{j}}_{X_{k}} \sum_{e\in I} \D^{c_{j}\ldots c_{p+1},e}_{X_{k}} f\right],\]
and for all $h\in G^{\E^{+}}$, all $k\in \{1,\ldots,d\}$ and all $j\in \{p+1,\ldots,n\}$, we have
\[\left(\sum_{e\in I} \D^{c_{j}\ldots c_{p+1},e}_{X_{k}} f\right)(h)=\sum_{e\in I}\left(\D^{e}_{\Ad(h(c_{j}\ldots c_{p+1}))X_{k}}f\right)(h)=0,\]
thanks to the $I$-invariance of $f$.

Another consequence of the $I$-invariance of $f$ is
\[0=\sum_{k=1}^{d} \left(\sum_{e\in I}\D^{e}_{X_{k}}\right)^{2}f=\sum_{e\in I} \Delta^{e}f+2\sum_{\substack{e_{i},e_{j}\in I\\i<j\leq p}} \Delta^{e_{j};e_{i}}f.\]
It follows that the first term of \eqref{rhs} is equal to
\[\sum_{e_{i} \in I, e_{j} \notin I,i<j\leq p}\E_{\YM^\G_{t}}\left[\Delta^{e_{j};e_{i}}f\right]=\sum_{\substack{e_{*}\leq e_{1}<e_{2}\leq e^{*} \\ e_{1}\in I, e_{2}\notin I}} \E_{\YM^\G_{t}}\left[\Delta^{e_{2};e_{1}}f\right],\]
as expected.

Let us now consider the case where $v$ is adjacent to the unbounded face. Let $e$ be the unique edge issued from $v$ such that $F^{L}(e)=F_{\infty}$.  Let us choose $e_{*}$ as the first element of $I$ which comes strictly after $e$ in the cyclic order of $\Out(v)$ and $e^{*}$ as the last element of $\Out(v)$ which comes strictly before $e_{*}$. We claim that \eqref{general mak mig} holds with this particular choice of $e_{*}$ and $e^{*}$.

Indeed, if $e\notin I$, then neither $e_{*}$ nor $e^{*}$ are adjacent to the unbounded face and the proof above applies verbatim. If on the contrary $e\in I$, then $e=e^{*}$. In this case, let us choose $n=p$ and $F_{p+1}=F_{\infty}$. Then \eqref{une arete de I} is still true for each $i$ such that $e_{i}\in I$, including when $e_{i}=e^{*}$, thanks to our agreement that $\frac{d}{d|F_{\infty}|}=0$. The rest of the proof is not altered.

Let us conclude by proving that the right-hand side of \eqref{general mak mig} does not depend on the choice of $e_{*}$ and $e^{*}$ provided $I\subset [e_{*},e^{*}]$. We do this in general, assuming only that $f$ is a smooth observable which is $I$-invariant at $v$. We shall denote the right-hand side of \eqref{general mak mig} by $S(e_{*},e^{*})$.

To start with, for all $e\notin I$, the $I$-invariance of $f$ implies that $\sum_{e_{1}\in I} \Delta^{e;e_{1}}f=0$, so that we may assume in computing $S(e_{*},e^{*})$ that $e_{*}$ and $e^{*}$ belong to $I$. In this case, $e_{*}$ is the element of $I$ which follows immediately $e^{*}$ in the cyclic order of $\Out(v)$. It suffices thus to prove that for all three consecutive elements $e''$, $e'$ and $e$
of $I$ in the cyclic order around $v$, the equality $S(e',e'')=S(e,e')$ holds. By unfolding the definition of $S$ and using the $I$-invariance of $f$, we find
\begin{align*}
S(e',e'')-S(e,e')&=
\sum_{\substack{e'<e_{2}<e\\ e_{2}\notin I}} \E_{\YM^\G_{t}}\left[\Delta^{e_{2};e'}f\right] -\sum_{\substack{e\leq e_{1} \leq e''\\ e_{1}\in I}}\sum_{\substack{e''<e_{2}<e'\\ e_{2}\notin I}}\E_{\YM^\G_{t}}\left[\Delta^{e_{2};e_{1}}f\right] \\
&=\sum_{e_{2}\notin I} \E_{\YM^\G_{t}}\left[\Delta^{e_{2};e'}f\right]\\
&=\sum_{e_{2} \in \Out(v)} \E_{\YM^\G_{t}}\left[\Delta^{e';e_{2}}f\right].
\end{align*}
We used the symmetry property $\Delta^{e';e_{2}}=\Delta^{e_{2};e'}$ as well as the $I$-invariance of $f$ in the last step. Now if $f$ was assumed to be invariant, we could conclude by Lemma \ref{inv infinitesimal} that this quantity is equal to $0$. In fact, it is $0$ even if $f$ is not invariant. Indeed, we claim that for all $X\in \g$ we have
\[\sum_{e_{2}\in \Out(v)}\E_{\YM^{\G}_{t}}\left[\D^{e_{2}}_{X}f\right]=0.\]
The reason for this equality is that the Yang-Mills measure is invariant under the action of the gauge group. Indeed, the uniform measure on $\Conf^{\G}$ is invariant, as well as the density of $\YM^{\G}_{t}$, defined by \eqref{def YM t}. If we let $(j_{s})_{s\in \R}$ be the same one-parameter group of gauge transformations as in the proof of Lemma \ref{inv infinitesimal}, then
\[\sum_{e_{2}\in \Out(v)}\E_{\YM^{\G}_{t}}\left[\D^{e_{2}}_{X}f\right]=\frac{d}{ds}_{|s=0} \int_{\Conf^{\G}} f(h) \, (\YM^{\G}_{t}
\circ j_{s}^{-1})(dh)=0.\]
Thus, $S(e',e'')=S(e,e')$ and the proof is finished.
\end{proof}

We mentioned before stating Proposition \ref{derive loc inv} that the main situation where we intend to apply it is at a point of self-intersection of a Wilson loop, or at the intersection point of two Wilson loops. First of all, let us state and prove the extended gauge-invariance properties of Wilson loops, indeed of Wilson skeins.

\begin{lemma}\label{wilson gi} Let $\sk=\{l_{1},\ldots,l_{r}\}$ be a skein. Let $\G_{\sk}$ be the underlying graph, with its orientation $\E^{+}$. Let $\lambda_{\sk}$ be the permutation of $\E^{+}$ induced by $\sk$. Let $\chi:G\to \C$ be a central function. For each $e\in \E^{+}$, the Wilson skein $W_{\chi,\sk}=W_{\chi,l_{1}}\ldots W_{\chi,l_{r}}$ is $\{e,(\lambda_{\sk}^{-1}(e))^{-1}\}$-invariant at $\underline{e}$.
\end{lemma}

\begin{proof} Without loss of generality, we may assume that $e$ is traversed by $l_{1}$. Setting $e'=\lambda_{\sk}^{-1}(e)$, we have $l_{1}=ae'eb$ with appropriate paths $a$ and $b$. For all $h\in \Conf^{\G_{\sk}}$, we thus have $W_{\chi,l_{1}}(h)=\chi(h(b) h(e)h(e')h(a))$. Since $l_{1}$ is an elementary loop, the paths $a$ and $b$ do not traverse $e$ nor $e'$, in either direction. Thus, for all $X\in \g$, we have
\begin{align*}
\big(\big(\D_{X}^{e}+\D_{X}^{(e')^{-1}})W_{\chi,l_{1}}\big)(h)&=\frac{d}{dt}_{|t=0}\left( \chi(h(b)h(e)e^{-tX}h(e')h(a))
+ \chi(h(b)h(e)e^{tX}h(e')h(a))\right)\\
&=0.
\end{align*}
On the other hand, the product $W_{\chi,l_{2}}\ldots W_{\chi,l_{r}}$ does not depend on $h(e)$ nor $h(e')$, so that for all $X\in \g$,
\[\D_{X}^{e}(W_{\chi,l_{2}}\ldots W_{\chi,l_{r}})=\D_{X}^{(e')^{-1}}(W_{\chi,l_{2}}\ldots W_{\chi,l_{r}})=0.\]
An application of the Leibniz rule completes the proof.
\end{proof}

Combining the extended invariance properties of Wilson skeins (Lemma \ref{wilson gi}), the local differential relation which this entails for their expectation (Proposition \ref{derive loc inv}), and our understanding of the way in which the differential operators which appear in \eqref{general mak mig} act on Wilson skeins (Proposition \ref{derive skeins}), we can finally prove the Makeenko-Migdal equations in their original version.

In order to formulate the result, we need to give a precise definition of the fact that a skein $\sk$ has a crossing at a certain vertex of $\G_{\sk}$. Let thus $\sk$ be a skein. Recall that the graph $\G_{\sk}$ carries a natural orientation, which we denote by $\E^{+}$, and that $\sk$ determines a permutation $\lambda_{\sk}$ of $\E^{+}$. Let $v$ be a vertex of degree $4$ of $\G_{\sk}$. Among the four edges of $\Out(v)$, exactly two belong to $\E^{+}$. If they are not consecutive in the cyclic order around $v$, we say that $\sk$ has a frontal dodge at $v$. If they are consecutive, let us label them $e_{1}$ and $e_{2}$ in such a way that $e_{2}$ follows immediately $e_{1}$. Let $e_{3}=(\lambda_{\sk}^{-1}(e_{1}))^{-1}$ and $e_{4}=(\lambda_{\sk}^{-1}(e_{2}))^{-1}$ be the other two edges of $\Out(v)$. If the cyclic order around $v$ is $(e_{1},e_{2},e_{4},e_{3})$, we say that $\sk$ has a lateral dodge at $v$. If the cyclic order is $(e_{1},e_{2},e_{3},e_{4})$, we say that $\sk$ has a crossing at $v$ (see Figure \ref{collisions}). 

\begin{figure}[h!]
\begin{center}
\includegraphics{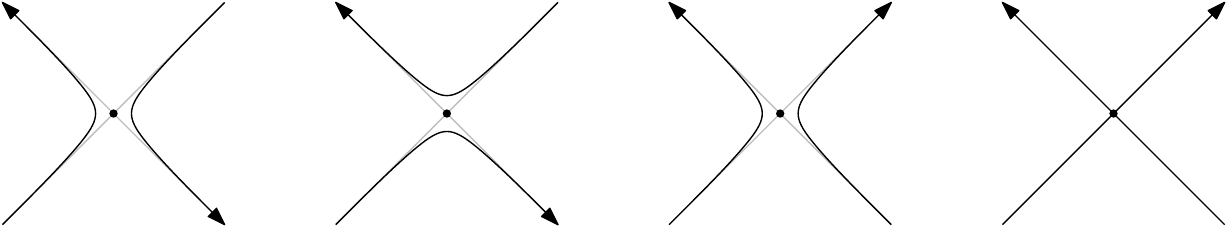}
\caption{\label{collisions} The four possible situations at a vertex of degree $4$ of a skein. From left to right, they are the two frontal dodges, the lateral dodge and the crossing.}
\end{center}
\end{figure}

\begin{proposition}[The Makeenko-Migdal equations] \label{prop: mmeq} Let $\sk=\{l_{1},\ldots,l_{r}\}$ be a skein. Let $\E^{+}$ be the orientation of $\G_{\sk}$ induced by $\sk$. Let $v$ be a vertex of $\G$ of degree $4$ at which $\sk$ has a crossing. Let $e_{1}$ and $e_{2}$ be the two consecutive edges of $\E^{+}$ which start at $v$. Let $F_{1},F_{2},F_{3},F_{4}$ be the faces adjacent to $v$, listed in cyclic order and starting by the face $F^{R}(e_{1})$. Then for all $t:\F^{b}\to \R^{*}_{+}$,
\begin{equation} \label{mm}
\left(\frac{d}{d|F_{1}|}-\frac{d}{d|F_{2}|}+\frac{d}{d|F_{3}|}-\frac{d}{d|F_{4}|}\right) \E_{\YM^{\G}_{t}}\left[W_{N,\sk}^{\K} \right]=\E_{\YM^{\G}_{t}}\left[\Delta^{e_{2};e_{1}}W_{N,\sk}^{\K}\right],
\end{equation}
where $\Delta^{(e_{2})(e_{1})}W_{N,\sk}^{\N}$ is given by Proposition \ref{derive skeins}. If one of the faces $F_{1},\ldots,F_{4}$ is the unbounded face, then \eqref{mm} still holds with the convention $\frac{d}{d|F_{\infty}|}=0$.

In particular, assume that $\K=\C$. If $e_{1}$ and $e_{2}$ are traversed by the same loop of $\sk$, then
\begin{equation} \label{mmx}
\left(-\frac{d}{d|F_{1}|}+\frac{d}{d|F_{2}|}-\frac{d}{d|F_{3}|}+\frac{d}{d|F_{4}|}\right) \E_{\YM^{\G}_{t}}\left[W_{N,\sk}^{\C} \right]=\E_{\YM^{\G}_{t}}\left[W_{N,S^{e_{2};e_{1}}(\sk)}^{\C}\right].
\end{equation}
On the other hand, if $e_{1}$ and $e_{2}$ are not traversed by the same loop of $\sk$, then the left-hand side of \eqref{mmx} is equal to $\frac{1}{N^{2}}$ times the right-hand side of $\eqref{mmx}$.
\end{proposition}

\begin{figure}[h!]
\begin{center}
\includegraphics{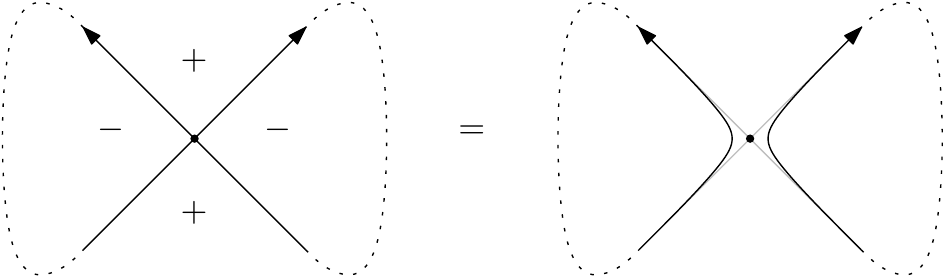}
\caption{\label{MMeq} This picture illustrates the original and most important instance of the Makeenko-Migdal equations, corresponding to \eqref{mmx}. One of the faces adjacent to the vertex considered is allowed to be the unbounded face.}
\end{center}
\end{figure}

\begin{proof} Let $e_{3}$ and $e_{4}$ denote the other two outgoing edges at $v$, in such a way that $e_{1},e_{2},e_{3},e_{4}$ are cyclically ordered in this way around $v$. By Lemma \ref{wilson gi}, the function $W_{N,\sk}^{\K}$ is $\{e_{1},e_{3}\}$-invariant at $v$. Hence, Proposition \ref{derive loc inv} applied with $e_{*}=e_{1}$ and $e^{*}=e_{3}$ yields literally \eqref{mm}.

The two assertions in the case where $\K=\C$ follow from \eqref{mm} and the fourth assertion of Proposition \ref{derive skeins}.
\end{proof}

\subsection{The Makeenko-Migdal equations for the master field}\label{mm large N}

The Makeenko-Migdal equations take a particularly simple form in the limit where $N$ tends to infinity. Before we state and prove them, and since this is first time in Section \ref{variation area} that we consider the master field $\Phi$ itself, rather than its approximations $\Phi^{\K,N}$, it is appropriate to make a few preliminary remarks.  

Let $l$ be an elementary loop. Recall that $\Phi(l)$ is the limit as $N$ tends to infinity of $\Phi^{\K,N}(l)$, for each $\K\in \{\R,\C,\H\}$. Since it is in the case $\K=\C$ that most formulas take their simplest form (see in particular Proposition \ref{derive skeins}), we shall always choose to see the master field as the large $N$ limit of the $\U(N,\C)$-valued Yang-Mills field.

Considering the approach which we have taken throughout this work, it is only natural that we consider $\Phi(l)$ as a function of the areas of the bounded faces of $\G_{l}$. Just as we extended the definition \eqref{def YM} of the discrete Yang-Mills measure by allowing in \eqref{def YM t} the areas of the faces to be prescribed, let us define, for all $t:\F^{b}\to \R^{*}_{+}$, 
\[\Phi^{\C,N}_{t}(l)=\E_{\YM^{\G_{l}}_{t}}\left[W^{\C,N}_{l}\right].\]
Recall from Proposition \ref{analytic} that $\Phi^{\C,N}_{t}(l)$ thus defined is the restriction of an entire mapping on $\C^{\F^{b}}$. 

\begin{proposition}\label{phit convergence} 1. Let $l$ be an elementary loop. Let $\G_{l}$ be the associated graph. As $N$ tends to infinity, the sequence of entire functions $\Phi_{t}^{\C,N}(l)$ of $t$ converges uniformly on every compact subset of $\C^{\F^{b}}$ towards an entire function $t\mapsto \Phi_{t}(l)$. The same convergence holds for all partial derivatives of these functions.

2. Let $l$ be an elementary loop. Let $\G_{l}$ be the associated graph. As $N$ tends to infinity, we have for all $t:\F^{b}\to \R^{*}_{+}$
\[\Var_{\YM^{\G_{l}}_{t}}\left(W^{\C,N}_{l}\right)=O(N^{-2}).\]

3. Let $\sk=\{l_{1},\ldots,l_{r}\}$ be a skein. Let $\G_{\sk}$ be the associated graph. We have for all $t:\F^{b}\to \R^{*}_{+}$
\begin{equation}\label{factorisation}
\lim_{N\to \infty} \E_{\YM^{\G_{\sk}}_{t}}\left[W^{\C,N}_{\sk}\right]=\Phi_{t}(l_{1})\ldots \Phi_{t}(l_{r}).
\end{equation}
\end{proposition}

In the case where for each bounded face $F$ of $\G_{l}$ we have $t(F)=|F|$, the statement of convergence in the first assertion is contained in Proposition \ref{convergence graph group}, and the second assertion follows from Theorem \ref{conv unif long}. We make sure that the same results hold with an arbitrary choice of $t$, and that the dependence in $t$ is analytic.

The third assertion is the property of factorisation which will eventually enable us to effectively compute the master field.

\begin{proof} 1. We use the expression of $\Phi^{\C,N}_{t}(l)$ provided by \eqref{inflate W}. This equation applies because, as explained at the beginning of Section \ref{wgarland}, $\{l\}$ is a garland with respect to any spanning tree on $\G_{l}$. By Proposition \ref{derive garlands}, and more specifically by the particular form of the fourth assertion of \ref{derive skeins}, the entries of the matrices $M_{F}^{\C,N}$ are polynomials of degree $2$ in $\frac{1}{N}$, actually affine functions of $\frac{1}{N^{2}}$. Thus, for each bounded face $F$, the matrix $M_{F}^{\C,N}$ converges, as $N$ tends to infinity, to the matrix $M_{F}$ whose entries are the constant terms of $M_{F}^{\K,N}$. This implies the convergence of the mapping $t\mapsto \Phi_{t}^{\C,N}(l)$ towards the mapping
\[t\mapsto \bigg(\prod_{F\in \F^{b}}e^{t(F) M_{F}}\bigg)\1,\]
uniformly on every compact subset of $\C^{\F^{b}}$. Since we are considering entire functions, this implies the uniform convergence on every compact subset of all partial derivatives of these functions.

2. In the proof of Theorem \ref{conv unif long}, the function $t:\F^{b}\to (\R^{*}_{+})$ is given by $t(F)=|F|$. This proof relies on one hand on Theorem \ref{main estim}, which is valid for an arbitrary $t$, and on Propositions \ref{quant basis RL} and \ref{NCA w l}, which depend on a specific relation between lengths and areas, through the Banchoff-Pohl inequality. However, it is easy to compare the values of $\Aa_{t}(w)$ associated with two different values of $t$. More precisely, setting
\[\alpha(t)= \max_{F\in \F^{b}}\frac{t(F)}{|F|},\]
the proof of Theorem \ref{conv unif long} yields
\[\Var_{\YM^{\G_{l}}_{t}}\left(W^{\C,N}_{l}\right) \leq \frac{1}{N^{2}}\alpha(t)\ell(l)^{2}e^{\alpha(t)\ell(l)^{2}},\]
which is even more precise than what we need.

3. Since the random variables $W_{l_{1}}^{\C,N},\ldots,W_{l_{r}}^{\C,N}$ are bounded by $1$, this follows immediately from the previous assertion.
\end{proof}

For all skein $\sk=\{l_{1},\ldots,l_{r}\}$ on a graph $\G$ and all $t:\F^{b} \to \R^{*}_{+}$, we shall use the notation \[\Phi_{t}(\sk)=\lim_{N\to\infty} \E_{\YM^{\G_{\sk}}_{t}}\left[W^{\C,N}_{\sk}\right].\]
The third assertion of the last proposition can be reformulated by saying that 
\begin{equation}\label{factorisation true}
\Phi_{t}(\sk)=\Phi_{t}(l_{1})\ldots\Phi_{t}(l_{r}).
\end{equation}
We can now formulate the Makeenko-Migdal in the large $N$ limit.

\begin{theorem}[The Makeenko-Migdal equations for the master field]\label{statement mmmf} Let $\sk$ be a skein. Let $F$ be a bounded face of $\G_{\sk}$ adjacent to the unbounded face. Then
\begin{equation}\label{mmmf1}
\frac{d}{d|F|} \Phi_{t}(\sk)=-\frac{1}{2}\Phi_{t}(\sk).
\end{equation}
Moreover, let $v$ be a vertex of degree $4$ of $\G_{\sk}$ at which $\sk$ has a crossing. Let us use the notation of Proposition \ref{prop: mmeq}. 
\begin{equation} \label{mmmf2}
\left(-\frac{d}{d|F_{1}|}+\frac{d}{d|F_{2}|}-\frac{d}{d|F_{3}|}+\frac{d}{d|F_{4}|}\right) \Phi_{t}(\sk)=\left\{\begin{array}{ll}\Phi_{t}(S^{e_{2};e_{1}}(\sk)) & \mbox{ if } e_{1} \mbox{ and } e_{2} \mbox{ are traversed} \\ & \hspace{1.5cm}\mbox{ by the same loop,}\\ 0 & \mbox{ otherwise.}\end{array}\right.
\end{equation}
\end{theorem}

\begin{proof} 
Let $e$ be an edge which is adjacent both to $F$ and to the unbounded face. For each $N\geq 1$, an application of \eqref{derivative one face 1} and the first assertion of Proposition \ref{derive skeins} yields
\[\frac{d}{d|F_{1}|}\Phi_{t}^{\C,N}(\sk)=\frac{1}{2}\E_{\YM^{\G_{\sk}}_{t}}\left[\Delta^{e}W^{\C,N}_{\sk}\right]=-\frac{1}{2}\E_{\YM^{\G_{\sk}}_{t}}\left[W^{\C,N}_{\sk}\right]=-\frac{1}{2}\Phi^{\C,N}_{t}(\sk),\]
because $c_{\u(N,\C)}=-1$. Since the derivatives of $\Phi_{t}^{\C,N}(\sk)$ converge to those of $\Phi_{t}(\sk)$, letting $N$ tend to infinity yields \eqref{mmmf1}.

Let $v$ be a vertex of degree $4$ of $\G_{\sk}$. Applying Proposition \ref{prop: mmeq} and the fourth assertion of Proposition \ref{derive skeins}, we find
\[\left(-\frac{d}{d|F_{1}|}+\frac{d}{d|F_{2}|}-\frac{d}{d|F_{3}|}+\frac{d}{d|F_{4}|}\right) \E_{\YM^{\G_{l}}_{t}}\left[W^{\C,N}_{\sk}\right]=\E_{\YM^{\G_{\sk}}_{t}}\left[W^{\C,N}_{S^{e_{2};e_{1}}(\sk)}\right]\]
if $e_{1}$ and $e_{2}$ are traversed by the same loop, and $\frac{1}{N^{2}}\E_{\YM^{\G_{\sk}}_{t}}\left[W^{\C,N}_{S^{e_{2};e_{1}}(\sk)}\right]$ if they are not. Thanks to the second assertion of Proposition \ref{phit convergence}, this yields the expected equalities in the limit when $N$ tends to infinity.
\end{proof}

\subsection{The recursive computation of the master field} \label{recur mf}
The Makeenko-Migdal equations convey in a nice and practical way a lot of information on Wilson loop and Wilson skein expectations, and on the master field. In this section, we shall determine exactly how much information. The main question is: given a skein $\sk$, is it the case that the linear combinations of area-derivatives of $\Phi_{t}(\sk)$ given by the equations \eqref{mmmf1} and \eqref{mmmf2}, applied at all possible places on the graph $\G_{\sk}$, suffice to determine all area-derivatives of $\Phi_{t}(\sk)$ ?

This is a purely algebraic graph theoretic problem: is a function on the set of faces of $\G_{\sk}$ determined by its values on the faces which are adjacent to the unbounded face and by the alternated sum of its values on the faces located around each vertex ?

We shall prove that the answer is positive exactly when $\sk$ consists in one single loop. As long as we are computing the master field, the factorisation property \eqref{factorisation true} holds, and we can break down any skein to its constituent loops, so that the Makeenko-Migdal equations suffice to compute the master field. However, if we try to compute the functions $\Phi^{\K}_{N}$, then some supplementary information is needed and it seems that we are sent back to \eqref{inflate W}.

Let us start with a graph $\G$ together with an orientation $\E^{+}$. Let us denote by $\ssk(\G,\E^{+})$ the set of skeins $\sk$ such that $\G_{\sk}=\G$ and $\E^{+}$ is the orientation induced by $\sk$. Let us define 
\[\ES={\rm Span}_{\Z}\left(t\mapsto \Phi_{t}(\sk) : \sk \in \ssk(\G,\E^{+})\right),\]
the $\Z$-module of entire functions on $\C^{\F^{b}}$ spanned by the value of the master field on skeins belonging to $\ssk(\G,\E^{+})$.
The space $\ES$ depends on $\G$ and $\E^{+}$, but the context will always make clear which graph we are considering.

Given a smooth real-valued function $\phi:(\R^{*}_{+})^{\F_{b}}\to \R$ of the areas of the faces of $\G$, let us denote by $\area \phi$ the area gradient of $\phi$, that is, the vector
\[\area \phi=\left(\frac{d}{d|F|}\phi : F\in \F\right) : (\R^{*}_{+})^{\F_{b}}\to \R^{\F}.\]
Note that we include the derivative with respect to the area of the unbounded face, which by convention is $0$. In the cases which we will consider, where $\phi$ belongs to $\ES$, both $\phi$ and $\area \phi$ extend to entire functions on $\C^{\F^{b}}$.

The Makeenko-Migdal equations ensure that for each $\phi\in \ES$, certain linear combinations of the components of $\area \phi$ belong to $\ES$. In order to express exactly which linear image of $\area\phi$ we have access to, let us introduce two discrete differential operators associated with a skein $\sk$. We consider the graph $\G=\G_{\sk}$, endowed with its natural orientation.

The first operator is the usual discrete gradient on the dual graph $\wG$, followed by the identification of dual and primal edges: it is the operator $d^{\wedge}:\R^{\F}\to \R^{\E^{+}}$ defined by setting, for all $u\in \R^{\F}$ and all $e\in \E^{+}$,
\[(d^{\wedge}u)(e)=u(F^{R}(e))-u(F^{L}(e)).\]
The second operator depends crucially on the skein $\sk$. It is the discrete derivative in the direction of $\sk$. We define $d_{\sk}:\R^{\E^{+}}\to \R^{\E^{+}}$ by setting, for all $\alpha\in \R^{\E^{+}}$ and all $e\in \E^{+}$,
\[(d_{\sk}\alpha)(e)=\alpha(e)-\alpha(\lambda_{\sk}^{-1}(e)).\]
We finally set $\mu_{\sk}=d_{\sk}\circ d^{\wedge} : \R^{\F}\to \R^{\E^{+}}$, the Makeenko-Migdal operator.

Before we state the first important result, let us define the number of self-crossings of a skein. Let $\sk=\{l_{1},\ldots,l_{r}\}$ be a skein  without triple point, that is, such that each vertex of $\G_{\sk}$ has either degree $2$ or $4$.  Let $v$ be a vertex of $\G_{\sk}$. We set $\nc_{v}(\sk)=1$ if $\sk$ has a crossing at $v$, in the sense explained before Figure \ref{collisions}, and if the two edges of $\E^{+}$ issued from $v$ are traversed by the same loop. Otherwise, we set $\nc_{v}(\sk)=0$. We then define the number of self-crossings of $\sk$ by
\[\nc(\sk)=\sum_{v\in \V} \nc_{v}(\sk).\]
The main observation is that, with the notation of Proposition \ref{prop: mmeq}, and if $e_{1}$ and $e_{2}$ are traversed by the same loop, then $\cn(S^{e_{2};e_{1}}(\sk))=\cn(\sk)-1$.

\begin{proposition}\label{mu skein skein} Let $\G$ be a graph. Let $\E^{+}$ be an orientation of $\G$. Assume that each vertex of $\G$ has degree $2$ or $4$. Then for each skein $\sk\in \ssk(\G,\E^{+})$, each component of  $\mu_{\sk} \left(\area\Phi_{t}(\sk) \right)$ is either $0$ or $\pm \Phi_{t}(\sk')$ for some skein $\sk'$ such that $\cn(\sk')=\cn(\sk)-1$.

In particular, for all $\phi\in \ES$, each component of  $\mu_{\sk} \left(\area \phi \right)$ belongs to $\ES$.
\end{proposition}

\begin{proof} Let us choose $\sk$ in $\ssk(\G,\E^{+})$ and $e\in \E^{+}$. Set $v=\underline{e}$ and $e'=\lambda_{\sk}^{-1}(e)$. If $v$ has degree $2$, then $F^{L}(e)=F^{L}(e')$ and $F^{R}(e)=F^{R}(e')$, so that $\mu_{\sk}(\area\Phi_{t}(\sk) )(e)=0$. 

Let us now assume that $v$ has degree $4$. Let us start by assuming that $\sk$ has no crossing at $v$. In this case, $e$ and $(e')^{-1}$ are consecutive in the cyclic order of $\Out(v)$. Let us assume that $e$ precedes immediately $(e')^{-1}$, the other case being similar. Then $F^{L}(e)=F^{L}(e')$ and 
\[\mu_{\sk}(\area\Phi_{t}(\sk))(e)=\frac{d}{d|F^{R}(e)|}\Phi_{t}(\sk)-\frac{d}{d|F^{R}(e')|}\Phi_{t}(\sk).\]
In order to compute these derivatives, let us apply \eqref{standard deriv} to the Wilson skein $W^{\C,N}_{\sk}$ and let $N$ tend to infinity. Let us assume first that neither $F^{R}(e)$ nor $F^{R}(e')$ is the unbounded face. For the first derivative, let us take $F_{1}=F^{R}(e)$, $e_{1}=e$ and $F_{2}=F^{L}(e)$. For the second, let us take $F_{1}=F^{R}(e')$, $e_{1}=(e')^{-1}$ and $F_{2}=F^{R}(e')=F^{R}(e)$. Since both $W^{\C,N}_{\sk}$ and $h\mapsto Q_{t(F_{2})}(h(\partial F_{2}))$ are $\{e,(e')^{-1}\}$-invariant at $v$, a short computation shows that the two derivatives are equal. Hence, $\mu_{\sk}(\area \Phi_{t}(\sk) )(e)=0$. The same equality holds if both $F^{L}(e)$ and $F^{L}(e')$ are the unbounded face. Finally, let us assume that one of these faces is the unbounded face, and not the other. Let us for example assume that $F^{L}(e')=F_{\infty}$. Then $F^{R}(e')$ is not the unbounded face. We apply \eqref{standard deriv} with $F_{1}=F^{L}(e)$, $F_{2}=F^{R}(e)$ and $F_{3}=F_{\infty}$, and $e_{1}=$

Let us finally assume that $\sk$ has a crossing at $v$. In this case, $\mu_{\sk}(\area \Phi_{t}(\sk) )(e)$ is, up to a sign, the left-hand side of \eqref{mmmf2}, of which we know that it is either $0$ or $\Phi_{t}(S^{e_{2};e_{1}}(\sk))$. The skein $S^{e_{2};e_{1}}(\sk)$ is on one hand an element of $\ssk(\G,\E^{+})$ and it satisfies on the other hand $\cn(S^{e_{2};e_{1}}(\sk))=\cn(\sk)-1$. This concludes the proof.
\end{proof}

Since we are interested in $\area \Phi_{t}(\sk)$, our next task is to invert the Makeenko-Migdal operator $\mu_{\sk}$. We analyse its kernel and image, and prove that it is invertible provided the skein $\sk$ has a single loop. 

For this, let us introduce the following notation. For each edge $e\in \E^{+}$, let $\delta_{e}$ denote the vector of $\R^{\E^{+}}$ whose components are all equal to $0$ except the $e$ component, which is equal to $1$. Let us associate a vector of $\R^{\E^{+}}$ to each loop $l_{i}$ of $\sk$ on one hand, and to each vertex $v$ of $\G_{\sk}$ on the other hand, by setting
\[\delta_{l_{i}}=\sum_{\substack{e \in \E^{+} \\ l_{i}\mbox{\scriptsize{ traverses }} e}}\!\!\! \delta_{e} \;  \mbox{ and } \; \star_{v}=\sum_{\substack{e\in \E^{+}\\ e\in \Out(v)}} \delta_{e}.\]
Finally, let $\1^{\E^{+}}$ be the vector of $\R^{\E^{+}}$ whose components are all equal to $1$. Note that $\sum_{i=1}^{r} \delta_{l_{i}}=\sum_{v\in \V}\star_{v}=\1^{\E^{+}}$. Let us endow $\R^{\E^{+}}$ with the scalar product for which $(\delta_{e})_{e\in \E^{+}}$ is an orthonormal basis.

In the following statement, we think of $\R^{\F}$ as the vector space of functions on $\R^{2}$ which are locally constant on the complement of the skeleton of $\G$. Accordingly, we see the functions $\n_{l_{1}},\ldots,\n_{l_{r}}$, defined in Section \ref{maa}, as elements of $\R^{\F}$.

\begin{lemma}\label{ker mu} Assume that none of the loops $l_{1},\ldots,l_{r}$ is constant. \\
1. The kernel of $\mu_{\sk}$ is spanned by the linearly independent vectors $\1^{\F},\n_{l_{1}},\ldots,\n_{l_{r}}$.\\
2. The intersection of the two subspaces of $\R^{\E^{+}}$ spanned respectively by $\{\delta_{l_{1}},\ldots,\delta_{l_{r}}\}$ and by $\{\star_{v}:v\in \V\}$ is equal to the line spanned by $\1^{\E^{+}}$.\\
3. The image of $\mu_{S}$ is the orthogonal complement of the sum of the two subspaces of $\R^{\E^{+}}$ spanned respectively by $\{\delta_{l_{1}},\ldots,\delta_{l_{r}}\}$ and $\{\star_{v}:v\in \V\}$.
\end{lemma}

\begin{proof}
1. The identity $d^{\wedge}\1^{\F}=0$ shows that $\1^{\F}$ is in the kernel of $\mu_{\sk}$. For each $i\in \{1,\ldots,r\}$, one has $d^{\wedge} \n_{l_{i}}=-\delta_{l_{i}}$ and $d_{\sk}\delta_{l_{i}}=0$, so that $\n_{l_{1}},\ldots,\n_{l_{r}}$ also lie in the kernel of $\mu_{\sk}$.

Let now $u$ be an element of the kernel of $\mu_{\sk}$. Set $\alpha=u_{F_{\infty}}$, so that $u-\alpha \1^{\F}$ vanishes on the unbounded face $F_{\infty}$. The equality $\mu_{\sk}(u-\alpha \1^{\F})=0$ means that for each $i\in \{1,\ldots,r\}$, the jump of $u$ across any two consecutive edges of $l_{i}$ are equal, so that $u$ varies by a certain fixed quantity $\beta_{i}$ when one crosses an edge of $l_{i}$. The function $u-\alpha \1^{\F}-\sum_{i=1}^{r} \beta_{i}\n_{l_{i}}$ vanishes on the unbounded face of $\G$, and does not vary when one crosses one of the edges of $\G$. It is thus identically equal to zero.

The fact that $\1,\n_{l_{1}},\ldots,\n_{l_{r}}$ are linearly independent follows from the fact that given a function of the form $\alpha \1+\sum_{i=1}^{r} \beta_{i}\n_{l_{i}}$, $\alpha$ can be recovered as the value of this function on the unbounded face. Then, the equality $d^{\wedge} \sum_{i=1}^{r} \beta_{i}\n_{l_{i}} = -\sum_{i=1}^{r}\beta_{i}\delta_{l_{i}}$ allows one to recover $\beta_{1},\ldots,\beta_{r}$.

2. Consider an equality $w=\sum_{i=1}^{r} \alpha_{i} \delta_{l_{i}} = \sum_{v\in \V} \beta_{v} \star_{v}$ in $\R^{\E^{+}}$. The first expression of $w$ shows that it has the same value on any two edges of a same loop of $\sk$. The second expression shows that it has the same value on the edges of any two loops which visit a common vertex. Since $\G$ is connected, $w$ has the same value on each edge, that is, $w$ is a multiple of $\1^{\E^{+}}$. 

3. One checks without difficulty that the range of $\mu_{\sk}$ is orthogonal to each $\delta_{l_{i}}, i\in \{1,\ldots,r\}$ and to each $\star_{v}, v\in \V$. On the other hand, by the first assertion, the range of $\mu_{\sk}$ has dimension $|\F|-r-1$. We conclude the proof by observing, thanks to Euler's relation, that $|\F|-r-1=|\E^{+}|-(|\V|+r-1)$. 
\end{proof}

For a skein $\sk$ consisting of $r$ loops, the kernel of $\mu_{\sk}$ has thus dimension $r+1$. If we are to recover $\area\Phi_{t}(\sk)$ from $\mu(\area\Phi_{t}(\sk))$, we need some additional information. Two supplementary relations are always available. The first is given by the fact that the component of $\area \Phi_{t}(\sk)$ corresponding to the unbounded face is $0$. The second is given by \eqref{mmmf1} applied to a face adjacent to the unbounded face. If $r=1$, that is, for a skein which consists in a single loop, this is enough to recover $\area\Phi_{t}(\sk)$.

\begin{proposition}\label{inject one loop} Let $l$ be a non-constant elementary loop. Let $\G_{l}$ be the underlying graph. Let $F$ be a face of $\G_{l}$ which shares a bounding edge with the unbounded face. Then the mapping
\begin{align*}
\R^{\F} & \longrightarrow  \R^{\E^{+}}\times \R^{2}\\
u & \longmapsto  \left(\mu_{\{l\}}(u),u_{F_{\infty}},u_{F}\right)
\end{align*}
is injective. 
\end{proposition}

\begin{proof} Assume that $u$ lies in the kernel of this mapping. Then $\mu_{\{l\}}(u)=0$ and, by Lemma \ref{ker mu}, $u$ is a linear combination of $\1$ and $\n_{l}$. Since $u_{F_{\infty}}=u_{F}=0$, $u$ must be equal to $0$.
\end{proof}

Thanks to the factorisation property of the master field, expressed by \eqref{factorisation true}, it is enough to be able to treat the case of one single loop. However, had we tried to apply the same analysis to the expectations of Wilson skeins, our approach would now fail. Indeed, it is not difficult to prove the analogue of Proposition \ref{mu skein skein} for Wilson skeins expectations. However, even if we start with a single loop, and unless this loop has no self-intersections, the area derivatives of the expectation of the corresponding Wilson loop involve expectations of Wilson skeins with two loops, for which we are not able to determine all area derivatives using local relations. Figure \ref{insuff} below shows an example of a skein for which indeed the local relations which we have at our disposal do not suffice to determine all area derivatives.

\begin{figure}[h!]
\begin{center}
\includegraphics{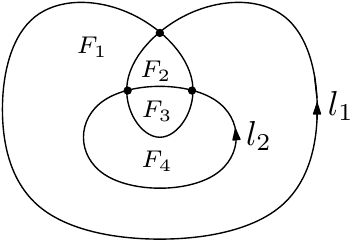}
\caption{\label{insuff} In this example, the area derivatives available through the Makeenko-Migdal equations are $2\frac{d}{d|F_{1}|}-\frac{d}{d|F_{2}|}$ and $\frac{d}{d|F_{1}|}-\frac{d}{d|F_{2}|}+\frac{d}{d|F_{3}|}-\frac{d}{d|F_{4}|}$. Moreover, the derivative $\frac{d}{d|F_{1}|}$ is given by \eqref{mmmf1}. This does however not suffice to determine the area gradient of the expectation of this Wilson skein.}
\end{center}
\end{figure}

Proposition \ref{inject one loop} allows us to prove the following effective version of \eqref{inflate W}.

\begin{proposition}\label{recursion} Let $\G$ be a graph. Assume that each vertex of $\G$ has degree $2$ or $4$. Let $\E^{+}$ be an orientation of $\G$. For each bounded face $F$ of $\G$ there exists a $\ssk(\G,\E^{+})\times \ssk(\G,\E^{+})$ real matrix $M_{F}$ such that on one hand, for all $t:\F^{b}\to \R^{*}_{+}$,
\[\left(\frac{d}{d|F|}-M_{F}\right) \left(\Phi_{t}(\sk) : \sk \in \ssk(\G,\E^{+})\right)=0,\]
and on the other hand, for all $\sk_{1},\sk_{2}\in \ssk(\G,\E^{+})$ such that $(M_{F})_{\sk_{1},\sk_{2}}\neq 0$, one either has $\sk_{2}=\sk_{1}$ or $\nc(\sk_{2})<\nc(\sk_{1})$.
\end{proposition}

\begin{proof} Let $\sk$ be an element of $\ssk(\G,\E^{+})$. By Proposition \ref{mu skein skein}, each component of the vector $\mu_{\sk}(\area \Phi_{t}(\sk))$ is a linear combination of functions of $t$ of the form $t\mapsto \Phi_{t}(\sk')$ with $\cn(\sk')<\cn(\sk)$.

Using a left inverse of the mapping of which Proposition \ref{inject one loop} grants the injectivity, this implies that each component of $\area\Phi_{t}(\sk)$ is a real linear combination of $t\mapsto \Phi_{t}(\sk)$ and the same functions as above. For each bounded face $F$, the matrix $M_{F}$ whose row indexed by $\sk$ contains the coefficient of the linear combination corresponding to $\frac{d}{d|F|}\Phi_{t}(\sk)$ has the desired property.
\end{proof}

We shall prove in the next section that the entries of the matrices $M_{F}$ are half-integers. For the time being, Proposition \ref{recursion} gives a triangular differential system which can effectively be solved, by induction on the number of crossings of the skeins. The following lemma takes care of the case where the number of crossings is zero. It is a slightly sophisticated version of the fact that the master field evaluated on a loop which surrounds once a domain of area $s$ equals $e^{-\frac{s}{2}}$.

\begin{lemma}\label{zero crossing} Let $\G$ be a graph. Let $\E^{+}$ be an orientation of $\G$. Assume that each vertex of $\G$ has degree $2$ or $4$. Let $\sk=\{l_{1},\ldots,l_{r}\}$ be a skein in $\ssk(\G,\E^{+})$. Assume that $\cn(\sk)=0$. For each $i\in \{1,\ldots,r\}$ and each face $F$ of $\G$, write $F\subset l_{i}$ if $\n_{l_{i}}$ is not zero on $F$. Then
\[\Phi_{t}(\sk)=\exp \left(-\frac{1}{2}\sum_{i=1}^{r}\sum_{F\subset l_{i}} t(F)\right).\]
\end{lemma} 

\begin{proof}  By definition of $\Phi_{t}(\sk)$, it suffices to prove that for each $i\in \{1,\ldots,r\}$, we have $\Phi_{t}(l_{i})=\exp {-\frac{1}{2}\sum_{F\subset l_{i}}t(F)}$. It suffices thus to prove the result when $r=1$. In this case, $l=l_{1}$ is an elementary loop on $\G$ which satisfies $\cn(\{l\})=0$. 

Let us start by working under the assumption that for each face $F$, the number $t(F)$ is the actual Lebesgue area of $F$.

By deforming $l$ around each vertex of $\G$ which it visits twice, we can produce a sequence of rectifiable Jordan curves $(l_{n})_{n\geq 1}$ which converges towards $l$. Then the sequence $(\Phi(l_{n}))_{n\geq 1}$ converges to $\Phi(l)$ by Theorem \ref{main mf}. For each $n\geq 1$, $\Phi(l_{n})$ is equal to $e^{-\frac{1}{2}a_{n}}$, where $a_{n}$ is the area enclosed by the Jordan curve $l_{n}$. This area can be written as
\[a_{n}=\int_{\R^{2}}|\n_{l_{n}}(x)|\; dx\]
and since the uniform convergence of a sequence of loops implies the pointwise convergence of the index function, the dominated convergence theorem implies that 
\[\lim_{n\to \infty} a_{n}=\int_{\R^{2}}|\n_{l}(x)|\; dx.\]
There remains to verify that this is equal to $\sum_{F\subset l}t(F)$. Equivalently, we must prove that $|\n_{l}|$ takes only the values $0$ and $1$. Once again, this follows from the pointwise convergence of the sequence $(\n_{l_{n}})_{n\geq 1}$ towards the sequence $\n_{l}$.

In order to complete the proof, it remains to reduce the general case to the case where $t(F)$ is the Lebesgue area of $F$. This is taken care of by the next Lemma.
\end{proof}

Let $\G=(\V,\E,\F)$ and $\G'=(\V',\E',\F')$ be two graphs. By a combinatorial isomorphism between $\G$ and $\G'$ we mean a bijection $i:\E \to \E'$ which is compatible with inversion and preserves the cyclic order of the
 edges around each vertex. It can be shown that the existence of such a combinatorial isomorphism guarantees the existence of a homeomorphism of the whole plane $\R^{2}$ which sends $\G$ to $\G'$ and realises the bijection between $\E$ and $\E'$ (see Section 1.3.1 of \cite{LevySMF} and the references therein). In particular, a combinatorial isomorphism induces a bijection $i:\F\to \F'$ between the faces of $\G$ and the faces of $\G'$, and a bijection $i:\Conf^{\G'} \to \Conf^{\G}$ between the configuration spaces attached to $\G$ and $\G'$. We use the same letter to denote all the maps induced by the isomorphism.

\begin{lemma} Let $\G$ be a graph. Let $t:\F^{b}\to \R^{*}_{+}$ be a function. There exists a graph $\G'$ and a
combinatorial isomorphism $i$ between $\G$ and $\G'$ such that for all bounded face $F$ of $\G$, the Lebesgue area of $i(F)$ is $t(F)$. In particular, the following equality holds:
\[\YM^{\G}_{t} \circ i^{-1}=\YM^{\G'}.\]
\end{lemma}

\begin{proof} Let us prove the result by induction on the number of bounded faces of $\G$. If it is $0$ or $1$, the result is true, for $\G'$ can be taken as the image of $\G$ by a homothecy. Let us assume that the result has been proved for any graph which has strictly less faces than $\G$. By removing from $\G$ one edge which is adjacent to the unbounded face, we reduce by $1$ the number of its bounded faces. We can thus apply the induction argument, and then add the missing edge to $\G'$ in such a way that the area of the face which it closes has the correct value.
\end{proof}

As an example of application of the algorithm presented in this section, we give, in Appendix \ref{table mf}, the value of $\Phi$ on the $28$ simplest elementary loops, namely those which have no more than three points of self-intersection.

\subsection{The Kazakov basis}\label{sec Kazakov}

In \cite{Kazakov}, Kazakov considered a particular basis of the space of functions defined on the set of faces of the graph $\G_{\sk}$, which allows one to explicitly find an inverse to the mapping which we introduced in Proposition \ref{inject one loop} and of which we proved that it is injective.

Let $l$ be an elementary loop which is generic in the sense that each vertex of the graph $\G_{l}$ has degree $4$ and is a crossing in the sense of Figure \ref{collisions}, except for the vertex $l(0)$ which is of degree $2$. 
We shall think of the module $\Z^{\F^{b}}$ as the space of locally constant integer-valued functions on the complement of the range of $l$ which vanish at infinity.

Let us write $\F^{b}=\{F_{1},\ldots,F_{q}\}$. To start with, the module $\Z^{\F^{b}}$ admits the canonical basis indexed by $\F^{b}$, which we simply denote by $\{F_{1},\ldots,F_{q}\}$. 
In order to define the second basis, let us make the assumption that the vertex $v_{0}$ of $\G$ which is the base point of $l$ is located on the boundary of the unbounded face. The orientation of $l$ determines an orientation $\E^{+}$ of $\G_{l}$, and it determines an order on $\V$, which is the order of first visit starting from $v_{0}$. Thus, $\V=\{v_{0},v_{1},\ldots,v_{q-1}\}$. There are indeed $q$ elements in $\V$, as a consequence of Euler's relation.

For each $i\in \{1,\ldots,q-1\}$, let us denote by $l_{i}$ the sub-loop of $l$ which starts at the first visit at $v_{i}$ and finishes at the second visit at this same vertex. Observe in particular that $l_{i}$ does not visit $v_{0}$. Let us also set $l_{0}=l$. For each $i\in\{0,\ldots,q-1\}$, the winding number $\n_{l_{i}}$ (see Section \ref{maa}) is an element of $\Z^{\F^{b}}$. 

Let us define, for each bounded face $F$ and each vertex $v$ of $\G$, an integer $\inc(F,v)$ between $-2$ and $2$, of which we think as an algebraic incidence number. Consider a vertex $v\neq v_{0}$. Let $e$ be the first edge traversed by $l$ after its first visit at $v$. Set $e'=\lambda_{\{l\}}^{-1}(e)$ and define, for each bounded face $F$,
\[\inc(F,v)=\1_{F=F^{L}(e)}-\1_{F=F^{R}(e)}+\1_{F=F^{R}(e')}-\1_{F=F^{L}(e')}.\]
Define also
\[\inc(F,v_{0})=\1_{F=F^{L}(e)}-\1_{F=F^{R}(e)}.\]
Note that since $l$ is elementary, every edge of $\G$ is adjacent to two distinct faces, so that the non-zero terms of this sum always have the same sign. In particular, $\inc(F,v)=0$ if and only if $v$ is not adjacent to $F$.
 
The main statement which underlies Kazakov's approach is the following.

\begin{proposition}\label{Kazakov basis} The set $\{\n_{l_{0}},\ldots,\n_{l_{q-1}}\}$ is a basis of $\Z^{\F^{b}}$. Moreover, for all $i\in \{1,\ldots,q\}$, we have the equality
\[F_{i}=\sum_{j=0}^{q-1} \inc(F_{i},v_{j}) \n_{l_{j}}.\]
In particular, the two matrices $\Nn=(\n_{l_{i-1}}(F_{j}))_{i,j=1\ldots q}$ and $\inc=(\inc(F_{i},v_{j-1}))_{i,j=1\ldots q}$ belong to $\SL_{q}(\Z)$ and are each other's inverse.
\end{proposition}

We propose to understand this result in terms of a certain positive quadratic form on the module $\Z^{\F^{b}}$, which we shall identify with a sub-module of $\Z^{\E^{+}}$, the module of formal linear combinations of edges of $\E^{+}$.

In order to make this identification, we associate to each loop $e_{1}^{\epsilon_{1}}\ldots e_{r}^{\epsilon_{r}}$ in $\G$, with $e_{1},\ldots,e_{r}\in \E^{+}$, the element $\epsilon_{1}e_{1}+\ldots +\epsilon_{r}e_{r}$ of $\Z^{\E^{+}}$. This element depends on the loop only up to equivalence and change of the base point. Hence, for each $i\in \{1,\ldots,q\}$, the image of $\partial F_{i}$ by this mapping is well defined and we also denote it by $\partial F_{i}$. The mapping from $\Z^{\F^{b}}$ to $\Z^{\E^{+}}$ which for each $i\in \{1,\ldots,q\}$ sends $F_{i}$ to $\partial F_{i}$ is injective. Its image can be characterised as the subspace formed by linear combinations such that at each vertex, the sum of the coefficients of the incoming edges equals the sum of the coefficients of the outgoing edges. This subspace is usually denoted by $H^{1}(\G;\Z)$. We have thus the mappings
\begin{align*}
\Z^{\F^{b}} & \build{\longrightarrow}_{}^{\sim} H^{1}(\G;\Z) \subset \Z^{\E^{+}}\\
F & \longmapsto \partial F.
\end{align*}

We will define a symmetric bilinear form on $\Z^{\F^{b}}$ as the restriction of a symmetric bilinear form $\langle\cdot,\cdot \rangle$ on $\Z^{\E^{+}}$ which is the sum over $\V$ of a form $\langle\cdot,\cdot \rangle_{v}$ for each vertex $v$. Let us choose a vertex $v$ and describe $\langle\cdot,\cdot \rangle_{v}$. We are going to distinguish several cases, but in all cases, any edge which is not adjacent to $v$ is in the kernel of $\langle\cdot,\cdot \rangle_{v}$. 

Let us now assume that $v\neq v_{0}$. The loop $l$ visits $v$ twice. Let $e^{\rm in}_{1}$ and $e^{\rm out}_{1}$ be the edges of $\E^{+}$ through which $l$ respectively arrives at $v$ and leaves $v$ at its first visit. Let $e^{\rm in}_{2}$ and $e^{\rm out}_{2}$ be the analogously defined edges for the second visit of $l$ at $v$. Note that among the four edges which we have defined, the only two which can be equal are $e^{\rm out}_{1}$ and $e^{\rm in}_{2}$.

Depending on whether $e^{\rm out}_{1}\neq e^{\rm in}_{2}$ or $e^{\rm out}_{1}= e^{\rm in}_{2}$, the form $\langle \cdot ,\cdot \rangle_{v}$ has the following matrix, respectively in the basis $(e^{\rm in}_{1},e^{\rm out}_{1},e^{\rm in}_{2},e^{\rm out}_{2})$ and in the basis $(e^{\rm in}_{1},e^{\rm out}_{1}=e^{\rm in}_{2},e^{\rm out}_{2})$:
\[\left(\begin{array}{rrrr} 0&0& -\frac{1}{2} &\frac{1}{2} \\[1.5pt] 0&0& \frac{1}{2} & -\frac{1}{2} \\[1.5pt] -\frac{1}{2}&\frac{1}{2}&0 &0 \\[1.5pt] \frac{1}{2} & -\frac{1}{2} & 0&0 \end{array}\right) \;\; \mbox{ or } \;\; \left(\begin{array}{rrr} 0 & -\frac{1}{2} & \frac{1}{2} \\[1.5pt] -\frac{1}{2} & 1  & -\frac{1}{2} \\[1.5pt] \frac{1}{2} & -\frac{1}{2} &0 \end{array}\right).\]
Let us treat the case where $v=v_{0}$. Let $e^{\rm in}$ and $e^{\rm out}$ be respectively the last and the first edge traversed by $l$. If they are equal, we set $\langle e^{\rm in},e^{\rm out}\rangle_{v_{0}}=1$. Otherwise, $\langle \cdot ,\cdot \rangle_{v_{0}}$ has the following matrix in the basis $(e^{\rm in},e^{\rm out})$:
\[\left(\begin{array}{rr} 0& \frac{1}{2}\\[1.5pt]\frac{1}{2} & 0 \end{array}\right).\]
As announced, we define a bilinear form on $\Z^{\E^{+}}$ by setting
\[\langle \cdot , \cdot \rangle= \sum_{i=0}^{q-1} \langle \cdot , \cdot \rangle_{v_{i}}.\]
By restriction, this defines a bilinear form on $H^{1}(\G;\Z)\simeq \Z^{\F^{b}}$. Note that, except in the case where $l$ is a simple loop, none of the forms $\langle \cdot , \cdot \rangle_{v_{i}}$ are non-degenerate, nor even semi-positive. Nevertheless, we have the following result, which immediately implies Proposition \ref{Kazakov basis}.

\begin{proposition} The bilinear form $\langle \cdot, \cdot\rangle$ is symmetric, positive and $\Z$-valued on $\Z^{\F^{b}}$ and the family $\{\n_{l_{0}},\ldots,\n_{l_{q-1}}\}$ is an orthonormal basis of $\Z^{\F^{b}}$. Moreover, for each bounded face $F$ of $\G$ an all $i\in \{0,\ldots,q-1\}$, we have $\langle F,\n_{l_{i}}\rangle=\inc(F,v_{i})$.
\end{proposition}

\begin{proof} Let us choose $i,j\in \{0,\ldots,q-1\}$ and compute $\langle \n_{l_{i}},\n_{l_{j}}\rangle$. According to our identification of $\Z^{\F^{b}}$ with a sub-module of $\Z^{\E^{+}}$, $\n_{l_{i}}$ (resp. $\n_{l_{j}}$) is the sum of the edges of $\E^{+}$ traversed by $l_{i}$ (resp. $l_{j}$). 

We claim that $\n_{l_{i}}$ is in the kernel of $\langle \cdot ,\cdot \rangle_{v}$ for each $v\neq v_{i}$. To start with, if $v= v_{0}$, then $v_{i}\neq v_{0}$, so that $l_{i}$ does not visit $v_{0}$ and the claim is true. Let us consider $v\notin\{v_{0},v_{i}\}$. At each visit at $v$, $l_{i}$ arrives and leaves through the same strand of $l$. Hence, $\n_{l_{i}}$ is equal, modulo the sub-module generated by the edges which are not adjacent to $v$, to $0$, or to one of the vectors $e^{\rm in}_{1}+e^{\rm out}_{1}$ and $e^{\rm in}_{2}+e^{\rm out}_{2}$, or to their sum. If $e^{\rm out}_{1}=e^{\rm in}_{2}$, then $v$ is either equal, modulo the same sub-module, to $0$ or to $e^{\rm in}_{1}+e^{\rm out}_{1}+e^{\rm out}_{2}$. In all cases, it is in the kernel of $\langle \cdot ,\cdot \rangle_{v}$.

This implies that  $\langle \n_{l_{i}},\n_{l_{j}}\rangle=0$ if $i\neq j$. Let us assume that $i=j$. At $v_{i}$, $\n_{l_{i}}$ is congruent, modulo the sub-module generated by the edges not adjacent to $v_{i}$, to $e^{\rm out}_{1}+e^{\rm in }_{2}$, or to $e^{\rm out}_{1}$ in the case where $e^{\rm out}_{1}=e^{\rm in}_{2}$, or $e^{\rm in}+e^{\rm out}$ if $i=0$. In all cases, $\langle \n_{l_{i}},\n_{l_{i}}\rangle=1$.

The family $\{\n_{l_{0}},\ldots,\n_{l_{q-1}}\}$ is thus orthonormal. This implies that it is a basis of $\R^{\F^{b}}$ and that the bilinear form $\langle \cdot, \cdot\rangle$ is a scalar product.

Let $F$ be a bounded face. Choose $i\in \{0,\ldots,q-1\}$. The element $\partial F$ of $\Z^{\E^{+}}$ is a sum of edges adjacent to $F$. Hence, $\partial F$ is in the kernel of $\langle \cdot ,\cdot \rangle_{v}$ for any vertex $v$ which is not adjacent to $F$. Hence, $\langle \partial F,\n_{l_{i}}\rangle=0$ if $v_{i}$ is not adjacent to $F$. Let us assume that $v_{i}$ is adjacent to $F$ and that $e^{\rm out}_{1}\neq e^{\rm in}_{2}$ at $v$. If $\inc(F,v_{i})=1$, then $F$ is either $F^{L}(e^{\rm out}_{1})$ or $F^{R}(e^{\rm in}_{1})$. Hence, $\partial F$ is congruent to one of the vectors $e^{\rm out}_{1}+e^{\rm in}_{2}$, $e^{\rm out}_{1}-e^{\rm out}_{2}$, $-e^{\rm in}_{1}+e^{\rm in}_{2}$ or $-e^{\rm in}_{1}-e^{\rm out}_{2}$ modulo edges not adjacent to $v_{i}$. All these vectors are congruent modulo the kernel of $\langle \cdot ,\cdot \rangle_{v_{i}}$ and in all cases we have $\langle \partial F,\n_{l_{i}}\rangle=\langle \partial F,\n_{l_{i}}\rangle_{v_{i}}=1$. For all other values of $\inc(F,v_{i})$ one checks in the same way that $\langle \partial F,\n_{l_{i}}\rangle=\inc(F,v_{i})$. The same equality holds if $e^{\rm out}_{1}= e^{\rm in}_{2}$, and also if $v=v_{0}$. 
 
Since $\inc(F,v)$ is an integer and $\{F_{1},\ldots,F_{q}\}$ is a basis of $\Z^{\F^{b}}$, the equality which we just proved implies that the scalar product is integer-valued on $\Z^{\F^{b}}$. This finishes the proof. 
\end{proof}

It may be helpful to give a more informal description of the scalar product we introduced on $H^{1}(\G;\Z)$. Since each element of $H^{1}(\G;\Z)$ can be written, although non uniquely, as a linear combination of reduced loops in $\G$, it suffices to understand the bilinear form evaluated on two loops. Let $l_{1}$ and $l_{2}$ be two reduced loops on $\G$. The number $\langle l_{1},l_{2}\rangle$ is the sum of local contributions, one for each pair formed by a visit of $l_{1}$ at a vertex of $\G$ and a visit of $l_{2}$ at the same vertex. At each vertex of $\G$, two strands of $l$ cross each other and we say that a loop which visits this vertex {\em turns} during this visit if it arrives along one strand and leaves it along the other. 
The number $\langle l_{1},l_{2}\rangle$ is the sum of the following contributions:
\begin{itemize}
\item $+1$ for each pair of visits of $l_{1}$ and $l_{2}$ in the same direction at the vertex $v_{0}$,
\item $-1$ for each pair of visits of $l_{1}$ and $l_{2}$ in opposite directions at the vertex $v_{0}$,
\item $+1$ for each pair of visits of $l_{1}$ and $l_{2}$ at a vertex distinct from $v_{0}$, such that both $l_{1}$ and $l_{2}$ turn during this visit, and such that $l_{1}$ and $l_{2}$ arrive along the same strand of $l$,
\item $-1$ for each pair of visits of $l_{1}$ and $l_{2}$ at a vertex distinct from $v_{0}$, such that both $l_{1}$ and $l_{2}$ turn during this visit, and such that $l_{1}$ and $l_{2}$ arrive along distinct strands of $l$.
\end{itemize}

We are now able to explicitly invert the Makeenko-Migdal operator $\mu$. The last ingredient we need is a sign, which we denote by $\epsilon_{0}$, which is equal to $-\inc(F_{\infty},v_{0})$. Thus, $\epsilon_{0}=1$ if and only if the first edge traversed by $l$ bounds positively the unique bounded face to which it is adjacent.

\begin{proposition}\label{derive face loops} For each vertex $v\in \G_{l}$, let $e^{\rm out}_{1}(v)$ be the edge traversed by $l$ immediately after its first visit at $v$. For each bounded face $F$, the following equality holds:
\[\frac{d}{d|F|}\Phi_{t}(l)=-\frac{\epsilon_{0}}{2}\n_{l}(F) \Phi_{t}(l) - \sum_{i=1}^{q-1} \n_{l_{i}}(F) \mu_{\{l\}}(\area \Phi_{t}(l))(e^{\rm out}_{1}(v_{i})).\]
\end{proposition}

\begin{proof} For each $i\in \{1,\ldots, q-1\}$, we have
\begin{equation}\label{mu inc}
\mu_{\{l\}}(\area \Phi_{t}(l))=-\sum_{j=1}^{q} \inc(F_{j},v_{i}) \frac{d}{d|F_{j}|}\Phi_{t}(l).
\end{equation}
Thus, given a bounded face $F$, and since $\mu_{\{l\}}(\area \Phi_{t}(l))(e^{\rm out}_{1}(v_{0}))=0$, we have
\begin{align*}
\sum_{i=1}^{q-1} \n_{l_{i}}(F) \mu_{\{l\}}(\area \Phi_{t}(l))(e^{\rm out}_{1}(v_{i})) &= -\sum_{i=0}^{q-1}\sum_{j=1}^{q} \n_{l_{i}}(F)\inc(F_{j},v_{i}) \frac{d}{d|F_{j}|}\Phi_{t}(l)\\
& \hspace{3cm} +\n_{l}(F)\sum_{j=1}^{q}\inc(F_{j},v_{0}) \frac{d}{d|F_{j}|}\Phi_{t}(l).
\end{align*}
To compute the first term, we use the fact that the matrices $\Nn$ and $\inc$ are each other's inverse. To compute the second term, we observe that the only non-zero contribution to the sum comes from the bounded face to which $v_{0}$ is adjacent and that for this face, the incidence number is $\epsilon_{0}$ and the area derivative equals $-\frac{1}{2}\Phi_{t}(l)$. Hence, we find
\begin{align*}
\sum_{i=1}^{q-1} \n_{l_{i}}(F) \mu_{\{l\}}(\area \Phi_{t}(l))(e^{\rm out}_{1}(v_{i})) &= -\frac{d}{d|F|}\Phi_{t}(l)-\frac{\epsilon_{0}}{2}\n_{l}(F)\Phi_{t}(l),
\end{align*}
which is the expected equality.
\end{proof}

The insight of Kazakov is to use on $\R^{\F^{b}}$, rather than the usual coordinates, the coordinates $a=(a_{0},\ldots,a_{q-1})$ given by
\[\forall i\in \{0,\ldots,q-1\}, \; a_{i}(t)=\sum_{j=1}^{q} \n_{l_{i}}(F_{j}) t(F_{j}).\]
If $t$ is the Lebesgue measure, then $a_{i}$ is the algebraic area enclosed by the loop $l_{i}$. Kazakov's main claim is that the alternated sum of derivatives with respect to the areas of faces around the vertex $v_{i}$, which appears in the Makeenko-Migdal equations, is the derivative with respect to $a_{i}$. Given our previous results, this follows from the fact that for all smooth function $\phi:\R^{\F^{b}}\to \R$ and all $i\in \{0,\ldots,q-1\}$, we have
\begin{equation*}
\frac{\partial \phi}{\partial a_{i}}(t_{1},\ldots,t_{q})=\sum_{j=1}^{q}\inc(F_{j},v_{i}) \frac{d}{d|F_{j}|}\phi(t_{1},\ldots,t_{q}).
\end{equation*}
For $\phi(t)=\Phi_{t}(l)$, we find, using \eqref{mu inc},
\begin{equation}\label{t a}
\frac{\partial}{\partial a_{i}}\Phi_{t}(l)=\left\{\begin{array}{ll} -\frac{\epsilon_{0}}{2} \Phi_{t}(l) & \mbox{if } i=0,\\
-\mu_{\{l\}}(\area \phi)(e^{\rm out}_{1}(v_{i})) & \mbox{if } i\in \{1,\ldots,q-1\}.\end{array}\right.
\end{equation}

This relation has the following consequence. For all $i\in \{1,\ldots,q-1\}$, let us denote by $\tilde l_{i}$ the loop obtained from $l$ by erasing the sub-loop $l_{i}$. Let us also denote $\epsilon_{i}=1$ if, at $v_{i}$, the edges $e^{\rm out}_{2}$ follows immediately $e^{\rm out}_{1}$ in the cyclic order, and $\epsilon_{i}=-1$ otherwise.

\begin{proposition}\label{algo kazakov} For all $i\in \{0,\ldots,q-1\}$, the following equality holds:
\[\frac{\partial}{\partial a_{i}} \Phi(l)=\left\{\begin{array}{ll} -\frac{\epsilon_{0}}{2}\Phi(l) & \mbox{if } i=0, \\ \epsilon_{i}\Phi(l_{i})\Phi(\tilde l_{i}) & \mbox{if } i\in \{1,\ldots,q-1\}.\end{array}\right.\] 
\end{proposition}

\begin{proof} This follows from \eqref{t a} and \eqref{mmmf2}. \end{proof}

We recover here the possibility of computing the master field by a recursive algorithm. The differential system is simpler than that given by Proposition \ref{recursion}, but this is to the price of the use of more complicated coordinates. In his paper, Kazakov proposes, for loops of a special kind which he calls {\em planar}, a formula for $\Phi(l)$. Planar loops can be characterised recursively in a way which is reminiscent of the definition of non-crossing partitions. With the notation of this section, a generic elementary loop $l$ is planar if it is a simple loop or if there exists $i\in \{1,\ldots q-1\}$ such that $l_{i}$ is a simple loop and $\tilde l_{i}$ is planar. However, we were not yet able to analyse deeply enough Kazakov's formula to improve it or let it fit into the present work. For examples of non planar loops, see the last two loops in the table of the master field which we give at the very end of this work (Appendix \ref{table mf}).

Let us conclude this work by the following consequence of our analysis of the Kazakov basis.

\begin{proposition}\label{mf poly} Let $l$ be an elementary loop. Assume that each vertex of $\G_{l}$ has degree $2$ or $4$. Write $\F^{b}=\{F_{1},\ldots,F_{q}\}$. There exists a real polynomial $P$ in $2q$ variables such that for all $t\in \C^{\F^{b}}$, one has
\[\Phi_{t}(l)=P\left(t(F_{1}),e^{-\frac{t(F_{1})}{2}},\ldots,t(F_{q}),e^{-\frac{t(F_{q})}{2}}\right).\]
\end{proposition}

\begin{proof} We prove the result by induction on the number of crossings of $l$. If this number is $0$, then the result follows from Lemma \ref{zero crossing}. If the result has been proved for all loops with strictly less than $n$ crossings and $l$ has $n$ crossings, then Proposition \ref{derive face loops}, the Leibniz rule and the induction hypothesis imply that for each bounded face $F$,
\[\frac{d}{d|F|}\Phi_{t}(l) + \frac{\epsilon_{0}}{2}\n_{l}(F) \Phi_{t}(l)\]
is a polynomial function of the variables $t(F_{i})$ and $e^{-\frac{t(F_{i})}{2}}$. It follows that $\Phi_{t}(l)$ is a polynomial function of $t(F_{i})$, $e^{-\frac{t(F_{i})}{2}}$ and $e^{\frac{t(F_{i})}{2}}$. However, since $t\mapsto \Phi_{t}(l)$ is bounded on $(\R^{*}_{+})^{\F^{b}}$, it cannot involve any monomials of the form $t(F_{i})^{a}e^{b\frac{t(F_{i})}{2}}$ for integers $a$ and $b$ such that $a\geq 0$ and $b>0$. Hence, the result holds for $l$.
\end{proof}

\part*{Appendix}
\renewcommand \thesection{A}
\section{Asymptotic freeness results}
\setcounter{subsection}{0} 
\setcounter{theorem}{0}
\label{appendice}
In the course of our study, more precisely in Section \ref{free proba}, we used a result of asymptotic freeness for large rotationally invariant matrices. In the unitary case, this is a classical result of Voiculescu \cite{VDN}. In the orthogonal and symplectic cases, this is a result which was proved by B. Collins and P. \'Sniady in \cite{CollinsSniady}.

Despite the fact that these results are proved and well proved, we found it distressingly difficult, in the symplectic case, to sort out the signs in the definition of the action of the Brauer algebra $\Br_{n,-2N}$ on $\H^{\otimes n}$ and to arrive at the definitions \eqref{def Srep} and \eqref{def rho H}. This may have had several reasons.  One of them is our choice to view the symplectic group as a group of quaternionic matrices, whereas it is virtually always considered as a group of complex matrices in the literature. Yet another reason is that the symplectic case is often treated as a slight variation of the orthogonal case, and typically given less attention, as the sentence of Brauer quoted at the beginning of this paper illustrates (see the footnote \ref{citation Brauer}). Let us emphasise that Collins and \'Sniady did not define this action, since in fact they did not use the multiplicative structure of the Brauer algebra at all in \cite{CollinsSniady}.

In this appendix, heavily inspired by \cite{CollinsSniady}, we review the main arguments of the proof of the asymptotic freeness result in the three cases which we use, including all details in the symplectic case. We also take this opportunity to write down a formula for the action of the Brauer algebra $\Br_{n,-2N}$ on $(\C^{2N})^{\otimes n}$, derived from our definition of $\rho_{\H}$.

\subsection{Unitary case} Let us start by the unitary case. Let $N$ and $n$ be positive integers. Recall from \eqref{def rho C} the definition of the representation $\rho_{\C}:\C[\S_{n}]\to \End((\C^{N})^{\otimes n})$, which we then still simply called $\rho$. Consider the endomorphism $P$ acting on the vector space $\End((\C^{N})^{\otimes n})$ according to
\[\forall A\in \End((\C^{N})^{\otimes n}) \; , \;\; P(A)=\int_{\U(N)} U^{\otimes n} \circ A \circ  (U^{-1})^{\otimes n} \; dU.\]
The invariance by translation of the Haar measure implies that $P$ is a projection on a subspace of $\End((\C^{N})^{\otimes n})$ which is contained in the commutant of the action of $\U(N)$ on $(\C^{N})^{\otimes n}$. The fundamental assertion of Schur-Weyl duality in this context is that the range of $P$ is thus contained in the range of $\rho_{\C}$ (this is Theorem 4.2.10 of \cite{GoodmanWallach}). 

Collins and \'Sniady gave in \cite{CollinsSniady} an expression of $P(A)$ which makes this inclusion manifest. Let us consider the element $\sum_{\sigma \in \S_{n}} N^{\ell(\sigma)} \sigma$, which is equal to $N^{n}(\id+O(N^{-1}))$, and hence is invertible in $\C[\S_{n}]$ for $N$ large enough, indeed for $N\geq n$ as can be shown by a more detailed analysis. Its inverse is called the Weingarten function, and it is denoted by $\Wg=\sum_{\sigma\in \S_{n}}\Wg(\sigma)\sigma$. A more explicit formula for $\Wg$ can be obtained thanks to the Jucys-Murphy elements $X_{1},\ldots,X_{n}\in \C[\S_{n}]$. They are defined by $X_{1}=0$ and, for all $i\in\{2,\ldots,n\}$, by $X_{i}=(1\, i)+\ldots+(i-1\, i)$. Using the classical notation $h$ for the complete symmetric functions, we have the equalities
\begin{equation}\label{weingarten u}
\Wg=\prod_{i=1}^{n} (N+X_{i})^{-1} = N^{-n}\sum_{k\geq 0} \frac{(-1)^{k}}{N^{k}} h_{k}(X_{1},\ldots,X_{n}).
\end{equation}
Since for all $i\geq 2$ the spectrum of $X_{i}$ in $\C[\S_{n}]$ is contained $\{-i+1,\ldots,i-1\}$ (see for example \cite{OkounkovVershik}), it is now apparent that $\Wg$ is well defined for $N\geq n$. We shall henceforward assume that $N\geq n$.

Consider the endomorphism $Q$ of $\End((\C^{N})^{\otimes n})$ defined by
\[\forall A\in \End((\C^{N})^{\otimes n}) \; , \;\; Q(A)=\rho_{\C}(\Wg) \sum_{\sigma\in \S_{n}} \Tr(A\circ \rho_{\C}(\sigma^{-1})) \rho_{\C}(\sigma).\]
It satisfies $Q(\id)=\id$ and $Q(A\circ \rho_{\C}(\sigma))=Q(A) \circ \rho_{\C}(\sigma)$ for all permutation $\sigma$, so that it is also a projection. For all $A\in \End((\C^{N})^{\otimes n})$, the endomorphism $Q(P(A))$ is on one hand equal to $Q(A)$, by definition of $P$ and $Q$ and because $\rho_{\C}(\sigma)$ and $U^{\otimes n}$ commute for all $\sigma\in \S_{n}$ and all $U\in \U(N)$. On the other hand, $Q(P(A))$ is equal to $P(A)$, because the range of $P$ is contained in the range of $\rho_{\C}$, hence in the range of $Q$. Altogether, the representation $\rho_{\C}$ being understood on the right-hand side, we have for all $A\in \End((\C^{N})^{\otimes n})$ the formula
\begin{equation}\label{proj U}
\int_{\U(N)} U^{\otimes n} \circ A \circ  (U^{-1})^{\otimes n} \; dU=\Wg \sum_{\sigma\in \S_{n}} \Tr(A\circ \sigma^{-1}) \sigma.
\end{equation}
From this equation, and using the fact that $\Wg$, being central in $\C[\S_{n}]$, satisfies $\Wg(\sigma)=\Wg(\sigma^{-1})$, it follows that for all $A_{1},\ldots,A_{n}$ and $B_{1},\ldots,B_{n}$ in $\Mat_{N}(\C)$, one has \begin{align}
\nonumber \int_{\U(N)} \tr(UA_{1}U^{-1}B_{1} \ldots UA_{n}U^{-1}B_{n}) \; dU &=\\
\nonumber &\hspace{-4cm} =\frac{1}{N} \int_{\U(N)} \Tr(U^{\otimes n} \circ A_{1}\otimes \ldots \otimes A_{n} \circ (U^{-1})^{\otimes n} \circ B_{1}\otimes \ldots \otimes B_{n} \circ (n\ldots 1))\; dU\\
&\hspace{-4cm} =\frac{1}{N}\sum_{\sigma,\tau\in \S_{n}} \Wg(\sigma\tau(n\ldots 1)) \Tr(A_{1}\otimes \ldots \otimes A_{n} \circ \sigma^{-1}) \Tr(B_{1}\otimes \ldots \otimes B_{n} \circ \tau^{-1}). \label{som}
\end{align}

Let $\sigma\in \S_{n}$ be a permutation with cycle lengths $m_{1},\ldots,m_{r}$. Let us denote by $|\sigma|=n-\ell(\sigma)=n-r$  the distance between $\sigma$ and the identity in the Cayley graph generated by all transpositions.
From the second characterisation of $\Wg$ given by \eqref{weingarten u}, it is possible to deduce that 
\[\Wg(\sigma)=\frac{(-1)^{|\sigma|}}{N^{n+|\sigma|}} \prod_{i=1}^{r} C_{m_{i}-1} +O(N^{-n-|\sigma|-1}),\]
where $C_{m}=\frac{1}{m+1}\binom{2m}{m}$ is the $m$-th Catalan number, characterised by the relations $C_{0}=1$ and $C_{m+1}=\sum_{k=0}^{m} C_{k}C_{m-k}$. Thus, using the notation $d(\sigma,\tau)=|\sigma^{-1} \tau|$ for the distance on the Cayley graph of $\S_{n}$, we find the highest power of $N$ which appears in the generic term of the sum \eqref{som} above to be
\[\ell(\sigma)+\ell(\tau)-1-n-|\sigma\tau(n\ldots 1)|=d(\id,(1\ldots n))-d(\id, \sigma)-d(\sigma,\sigma\tau)-d(\sigma\tau,(1\ldots n)).\]
This power is non-positive, and it is zero if and only if $\id$, $\sigma$, $\sigma\tau$ and $(1\ldots n)$ are located in this order on a geodesic. We shall use the notation $\sigma_{1}\preccurlyeq \sigma_{2}$ to indicate that $\sigma_{1}$ is located on a geodesic from $\id$ to $\sigma_{2}$. To the highest order, the sum \eqref{som} is thus restricted to the sublattice of $\S_{n}$ formed by the permutations $\sigma$ such that $\id\preccurlyeq \sigma \preccurlyeq (1\ldots n)$. This lattice is isomorphic to the lattice $\NC_{n}$ of non-crossing partitions (see \cite{Biane2}). Moreover, for $\sigma_{1},\sigma_{2}$ in this lattice, with $\sigma_{1}\preccurlyeq\sigma_{2}$, the M\"obius function $\mu(\sigma_{1},\sigma_{2})$ is equal to $(-1)^{d(\sigma_{1},\sigma_{2})}\prod_{i=1}^{r} C_{m_{i}-1}$, where the product runs over the cycles of $\sigma_{2}\sigma_{1}^{-1}$, of length $m_{1},\ldots,m_{r}$ (see \cite{Speicher}).

Then, using the notation
\[p_{\sigma}(A_{1},\ldots,A_{n})=\prod_{\substack{c \mbox{ \scriptsize cycle of }\sigma\\ c=(i_{1}\ldots i_{r})}} \tr(A_{i_{1}}\ldots A_{i_{r}}) \mbox{ and } \kappa_{\sigma}(A_{1},\ldots,A_{n})=\prod_{\substack{c \mbox{ \scriptsize cycle of }\sigma\\ c=(i_{1}\ldots i_{r})}} \kappa_{r}(A_{i_{1}},\ldots, A_{i_{r}}), \]
where $\kappa_{r}$ denotes the free cumulant of order $r$ in the non-commutative probability space $(\Mat_{N}(\C),\tr)$, we have
\begin{align*}
\int_{\U(N)} \tr(UA_{1}U^{-1}B_{1} \ldots UA_{n}U^{-1}B_{n}) \; dU &=\\
&\hspace{-4cm} =\sum_{\substack{\sigma \preccurlyeq (1\ldots n)\\ \tau \preccurlyeq \sigma^{-1}(1\ldots n)}} \mu(\tau,\sigma^{-1}(1\ldots n))
p_{\sigma}(A_{1},\ldots,A_{n}) p_{\tau}(B_{1},\ldots,B_{n}) +O(N^{-1})\\
&\hspace{-4cm} =\sum_{\sigma \preccurlyeq (1\ldots n)} p_{\sigma}(A_{1},\ldots,A_{n}) \kappa_{\sigma^{-1}(1\ldots n)}(B_{1},\ldots,B_{n})+O(N^{-1})\\
&\hspace{-4cm} =\sum_{\beta\in \NC_{n}} \tau_{\beta}(A_{1},\ldots,A_{n}) \kappa_{\beta^{\vee}}(B_{1},\ldots,B_{n})+O(N^{-1}),
\end{align*}
where in the last line we used the classical notation $\tau_{\beta}$ for the non-commutative moments and $\beta^{\vee}$ for the Kreweras complement of a non-crossing partition $\beta$.
The last equation which we have obtained implies classically the asymptotic freeness of the families $\{UA_{1}U^{-1},\ldots,UA_{n}U^{-1}\}$ and $\{B_{1},\ldots,B_{n}\}$.

\subsection{Orthogonal case} In the orthogonal case, things are slightly different but the proof starts as in the unitary case. Let us define  the endomorphism $P$ of $\End((\R^{N})^{\otimes n})$ by
\[\forall A\in \End((\R^{N})^{\otimes n}) \; , \;\;P(A)=\int_{\SO(N)} R^{\otimes n} \circ A \circ  (R^{-1})^{\otimes n} \; dR.\]
Instead of $\rho_{\C}$, we shall naturally use the homomorphism of algebras $\rho_{\R}:\Br_{n,N} \to \End((\R^{N})^{\otimes n})$ (see \eqref{def rho} for the definition of $\rho_{\R}$). From the first fundamental theorem of invariant theory in this case, stated in  \cite{GoodmanWallach} as Theorem 5.3.3, it follows by elementary algebraic manipulations that the range of $P$ is contained in the range of $\rho_{\R}$. We shall give more details on these manipulations in the symplectic case.

The appropriate definition of $Q$ is slightly different from that used in the unitary case. For all $A\in \End((\R^{N})^{\otimes n})$ we define an element $Q_{0}(A)$ of $\Br_{n,N}$ by setting 
\[Q_{0}(A)=\sum_{\pi\in \B_{n}} \Tr(A \circ \rho_{\R}(\t{\pi})) \pi.\]
We shall prove in a moment that for $N$ large enough, $Q_{0}$ restricts to a bijection between the range of $\rho_{\R}$ and $\Br_{n,N}$. It is proved in \cite{CollinsSniady} that this is true as soon as $N\geq n$. We denote by $\Wg_{N}$ the reciprocal bijection, so that
\[\Wg_{N}=({Q_{0}}_{|\rho_{\R}(\Br_{n,N})})^{-1}: \Br_{n,N} \to \rho_{\R}(\Br_{n,N}) \subset \End((\R^{N})^{\otimes n}).\]
We shall use the notation $\Wg(\pi)=\sum_{\pi'\in \B_{n}}\Wg_{N}(\pi,\pi') \rho_{\R}(\pi')$.

Now for all endomorphism $A$, we have on one hand $(\Wg_{N}\circ Q_{0})(P(A))=(\Wg\circ Q_{0})(A)$, because $\rho_{\R}(\pi)$ and $R^{\otimes n}$ commute for all $\pi\in \B_{n}$ and all $R\in \SO(N)$. On the other hand, we have $(\Wg_{N}\circ Q_{0})(P(A))=P(A)$ because $P(A)$ belongs to the range of $\rho_{\R}$. Hence, the formula corresponding to \eqref{proj U} in the orthogonal case is
\begin{equation}\label{proj O}
\int_{\SO(N)} R^{\otimes n} \circ A \circ  (R^{-1})^{\otimes n} \; dU=\sum_{\pi\in \B_{n}} \Tr(A\circ \t{\pi}) \Wg_{N}(\pi).
\end{equation}

As in the unitary case, it follows from this equation that for all $A_{1},\ldots,A_{n}$ and $B_{1},\ldots,B_{n}$ in $\Mat_{N}(\R)$, one has \begin{align}
\nonumber \int_{\SO(N)} \tr(RA_{1}R^{-1}B_{1} \ldots RA_{n}R^{-1}B_{n}) \; dR &=\\
&\hspace{-4cm}=\frac{1}{N}\sum_{\pi,\pi'\in \B_{n}}\Wg_{N}(\pi,\pi')\Tr(A_{1}\otimes \ldots \otimes A_{n} \circ \t{\pi})\Tr(B_{1}\otimes \ldots \otimes B_{n} \circ (n \ldots 1) \pi').\label{som o}
\end{align}
In order to analyse the asymptotic behaviour of this formula, we need to determine $\Wg_{N}(\pi,\pi')$ to the highest order in $N$. The key point is that the set $\B_{n}$ of Brauer diagrams is endowed with a natural distance which plays here the role played in the unitary case by the distance in the Cayley graph of $\S_{n}$. The distance on $\B_{n}$ can be defined in several equivalent ways, and we pause briefly to gather some of these definitions.

We defined the elements of $\B_{n}$ combinatorially, as the partitions of $\{1,\ldots,2n\}$ by pairs, but there are other natural ways to define them. In particular, given an element $\pi$ of $\B_{n}$, there is a unique element $i_{\pi}$ of $\S_{2n}$ whose cycles are the pairs of $\pi$, and the correspondence $\pi\mapsto i_{\pi}$ is a bijection between $\B_{n}$ and the set $\mathcal I_{2n}$ of fixed point free involutions of $\{1,\ldots,2n\}$, which is a conjugacy class in $\S_{2n}$. The group $\S_{2n}$ acts on $\B_{n}$ through its natural action on $\{1,\ldots,2n\}$, a permutation $\alpha$
 transforming a partition $\pi=\{\{i,j\},\ldots\}$ into $\alpha\cdot \pi=\{\{\alpha(i),\alpha(j)\},\ldots\}$. It also acts by conjugation on $\mathcal I_{2n}$ and the map $i:\B_{n}\to \mathcal I_{2n}$ is equivariant, in the sense that $i_{\alpha\cdot \pi}=\alpha i_{\pi} \alpha^{-1}$. In $\B_{n}$, there is a distinguished element $\id$, which is for all $\lambda$ the unit of the algebra $\Br_{n,\lambda}$, and which satisfies $i_{\id}=(1\, n+1) \ldots (n\, 2n)$. The stabiliser of $\id$ under the action of $\S_{2n}$ on $\B_{n}$ is the hyperoctahedral group $H_{n}$ (see Section \ref{Brauer II}). The group $H_{n}$ is also the centraliser of $i_{\id}$. The choice of $\id\in \B_{n}$ thus determines a bijection between $\B_{n}$ and the set $\S_{2n}/H_{n}$ of left $H_{n}$-cosets in $\S_{2n}$. We denote by $\pi\mapsto C_{\pi}$ this correspondence. 

Recall that we defined in Section \ref{sec : orth 1M} the number $\ell(\pi)$. Recall also the operations $S_{a,b}$ and $F_{a,b}$ which we defined in Section \ref{Brauer III}. Finally, let us denote by ${}^{t}\pi$ the pairing $i_{\id}\cdot \pi$, obtained by flipping the box which represents $\pi$ upside down. 

\begin{lemma}\label{d bn} Let $\pi$ and $\pi'$ be two elements of $\B_{n}$. The following numbers are equal.\\
\indent 1. The minimal length of a chain $\pi=\pi_{0},\pi_{1},\ldots,\pi_{r}=\pi'$ such that each element is obtained from the previous one by an operation $S_{a,b}$.\\
\indent 2. The smallest distance in $\S_{2n}$ between the identity and an element $\alpha$ such that $\alpha \cdot \pi=\pi'$.\\
\indent 3. The smallest distance in $\S_{2n}$ between an element of $C_{\pi}$ and an element of $C_{\pi'}$.\\
\indent 4. The number $n-\ell(\t{\pi} \pi')$.\\
\indent 5. One half of the distance in $\S_{2n}$ between $i_{\pi}$ and $i_{\pi'}$.

We denote these five numbers by $d(\pi,\pi')$. The function $d$ is a distance on $\B_{n}$, which makes it a metric space of diameter $n-1$. The action of $\S_{2n}$ on $\B_{n}$ induced by its natural action on $\{1,\ldots,2n\}$ and the actions of $\S_{n}$ by left and right multiplication on $\B_{n}\subset B_{n,\lambda}$ preserve the distance $d$. The inclusion $\S_{n}\subset \B_{n}$ is an isometry. Finally, any shortest path in $\B_{n}$ between two elements of $\S_{n}$ stays in $\S_{n}$.
\end{lemma}

\begin{proof} Let us denote by $d_{1},\ldots,d_{5}$ the five numbers as they are defined in the statement.

The equality $d_{1}=d_{2}$ follows from the identity $S_{a,b}(\pi)=(a\, b)\cdot \pi$, which hold for all $\pi\in \B_{n}$ and all $a,b\in \{1,\ldots,2n\}$.

Let us think of the set $C_{\pi}$ as the set $\{\alpha \in \S_{2n} : \alpha \cdot i_{\id}=\pi\}$. From the equality 
\[\{\alpha \in \S_{2n} : \alpha \cdot \pi =\pi'\}=\{\sigma_{2}\sigma_{1}^{-1} : \sigma_{1}\in C_{\pi}, \sigma_{2}\in C_{\pi'}\}\]
of subsets of $\S_{2n}$, it follows that $d_{2}$, which is the distance of $\id$ to the subset on the left-hand side, is equal to $d_{3}$, which is the distance of $\id$ to the subset on the right-hand side.

The number $\ell(\t{\pi}\pi')$ is the number of loops formed by the superposition in one single box of the diagrams of $\pi$ and $\pi'$. In this picture, each loop contains an even number of edges, and since there are $2n$ edges altogether, there are at most $n$ loops. Moreover, there are $n$ loops only if each loop has length $2$, and this happens only if $\pi=\pi'$. Now let us prove that $d_{2}\leq d_{4}$. If $d_{4}=0$, this follows from our last remark. If $d_{4}>0$, then $\pi\neq \pi'$ and there is at least one loop of length at least $4$. There exists $i,j,k,l\in\{1,\ldots,2n\}$ such that $\{i,j\}$ and $\{k,l\}$ belong to $\pi$ and $\{j,k\}$ belongs to $\pi'$. Then $\ell(\t{(S_{j,l}(\pi))},\pi')=\ell(\t{\pi}\pi')-1$. Iterating this argument, we find that we can go from $\pi$ to $\pi'$ in $d_{4}$ applications of an operator $S_{a,b}$. Hence, $d_{2}\leq d_{4}$. On the other hand, it is even easier to check that the application of an operator $S_{a,b}$ cannot increase or decrease $\ell(\t{\pi}\pi')$ by more than $1$. Following a minimal chain of applications of the operators $S_{a,b}$ leading from $\pi$ to $\pi'$, we find $d_{4}\leq d_{2}$. Finally, $d_{2}=d_{4}$.

The permutations $i_{\pi}$ and $i_{\pi'}$ are involutions, so that their distance in $\S_{2n}$ is equal to $2n$ minus the number of cycles of their product $i_{\pi}i_{\pi'}$. The image of an integer $j$ by this product is easily computed on the diagram formed by the superposition of those of $\pi$ and $\pi'$, by following first the edge of $\pi'$ issued from $j$, thus arriving at an integer $k$, and then following the other edge issued from $k$. The permutation $i_{\pi}i_{\pi'}$ has thus exactly twice as many cycles as the superposition of the diagrams of $\pi$ and $\pi'$. Hence, $d_{4}=d_{5}$. 
\end{proof}

Note that, up to the multiplicative factor involved in the definition of $F_{a,b}$, the effect of an operator $F_{a,b}$ on a diagram can always be obtained by the action of an operator $S_{a,b}$. Thus, in the definition of the first number above, we could have replaced $S$ by $F$.\\

Lemma \ref{d bn} implies the equality of matrices
\[\left(\Tr^{\otimes n}(\rho_{\R}(\pi_{1}) \rho_{\R}(\t{\pi_{2}}))\right)_{\pi_{1},\pi_{2} \in \B_{n}}=N^{n}(N^{-d(\pi,\pi')})_{\pi,\pi'\in \B_{n}}.\]
In particular, $N^{-n}$ times this matrix tends to the identity matrix as $N$ tends to infinity, so that it is invertible for $N$ large enough. This implies that, for $N$ large enough, the family $\{\rho_{\R}(\pi) : \pi\in \B_{n}\}$ is linearly independent. Moreover, the same matrix is the matrix of the restriction of $Q_{0}$ to $\rho_{\R}(\Br_{n,N})$ with respect to the bases $\{\rho_{\R}(\pi) : \pi\in \B_{n}\}$ and $\B_{n}$. Thus, we have proved that this restriction of $Q_{0}$ is invertible for $N$ large enough, as promised. Finally, this equality implies
\begin{equation}\label{wg}
\Wg_{N}(\pi,\pi')=N^{-n} \sum_{r\geq 0} (-1)^{r} \sum_{\pi_{1},\ldots,\pi_{r-1}\in \B_{n}} N^{-d(\pi,\pi_{1})-d(\pi_{1},\pi_{2})-\ldots-d(\pi_{r-1},\pi')},
\end{equation}
where the sum is taken over all the chains $\pi,\pi_{1},\ldots,\pi_{r-1},\pi'$ in which each term is different from the next. The term of highest order is provided by chains for which the exponent of $N$ is $-d(\pi,\pi')$. The highest power of $N$ which appears in the generic term of \eqref{som o} is thus 
\[\ell(\pi)+\ell((n\ldots 1)\pi')-n-1-d(\pi,\pi')=d(\id,(1\ldots n))-d(\id,\pi)-d(\pi,\pi')-d(\pi',(1\ldots n)),\]
using the invariance of $d$ under left multiplication by $(1\ldots n)$. This power is non-positive and equal to $0$ only if $\pi$ and $\pi'$ are located in this order on a same geodesic from $\id$ to $(1\ldots n)$. Lemma \ref{d bn} asserts that a necessary condition for this is that $\pi$ and $\pi'$ belong to $\S_{n}$. We must then have $\id\preccurlyeq \pi \preccurlyeq \pi' \preccurlyeq (1\ldots n)$. Moreover, in this case, we recognise in the expression of the term of highest order of $\Wg_{N}(\pi,\pi')$ given by \eqref{wg} the value  $\mu(\pi,\pi')$ of the M\"obius function of the lattice $\NC_{n}$. We thus obtain
\begin{align*}
\int_{\SO(N)} \tr(RA_{1}R^{-1}B_{1} \ldots RA_{n}R^{-1}B_{n}) \; dR &=\\
&\hspace{-4cm} =\sum_{\sigma \preccurlyeq \sigma' \preccurlyeq (1\ldots n)} \mu(\sigma,\sigma')
p_{\sigma}(A_{1},\ldots,A_{n}) p_{(\sigma')^{-1}(1\ldots n)}(B_{1},\ldots,B_{n}) +O(N^{-1})\\
&\hspace{-4cm} =\sum_{\sigma' \preccurlyeq (1\ldots n)} \kappa_{\sigma'}(A_{1},\ldots,A_{n}) p_{(\sigma')^{-1}(1\ldots n)}(B_{1},\ldots,B_{n})+O(N^{-1})\\
&\hspace{-4cm} =\sum_{\beta\in \NC_{n}} \kappa_{\beta}(A_{1},\ldots,A_{n}) \tau_{\beta^{\vee}}(B_{1},\ldots,B_{n})+O(N^{-1}),
\end{align*}
and conclude as in the unitary case to the asymptotic freeness of the families $\{RA_{1}R^{-1},\ldots,RA_{n}R^{-1}\}$ and $\{B_{1},\ldots,B_{n}\}$.

\subsection{Symplectic case} Let us finally treat the symplectic case. As in the unitary and orthogonal cases, we shall use an instance of the first fundamental theorem of invariant theory. This theorem is usually stated for the symplectic group $\Sp(N)$ seen as a subgroup of $\GL(2N,\C)$, indeed of $\U(2N)$, and acting on tensor powers of $\C^{2N}$. We start by stating and proving the quaternionic version of the first fundamental theorem. Recall from \eqref{def rho H} the definition of the homomorphism of algebras $\rho_{\H}:\Br_{n,-2N} \to \Mat_{N}(\H)^{\otimes n}$.

\begin{theorem}\label{fft sp} Consider the action of the group $\Sp(N)=\U(N,\H)$ on the real algebra $\Mat_{N}(\H)^{\otimes n}$ given by $S\cdot (M_{1}\otimes \ldots \otimes M_{n})=SM_{1}S^{-1} \otimes \ldots \otimes SM_{n}S^{-1}$.\\
1. For all $\pi\in \B_{n}$, the element $\rho_{\H}(\pi)$ of  $\Mat_{N}(\H)^{\otimes n}$ is invariant under the action of $\Sp(N)$.\\
2. The real linear subspace of  $\Mat_{N}(\H)^{\otimes n}$ consisting of the invariant tensors under the action of $\Sp(N)$ is spanned over $\R$ by the tensors $\{\rho_{\H}(\pi) : \pi \in \B_{n}\}$.
\end{theorem}

\begin{proof}
1. A direct computation does not seem to be the simplest way to check this assertion. Instead, it follows from Lemma \ref{tr tr sp}, which implies that for all $M_{1},\ldots,M_{n}\in \Mat_{N}(\H)$, all $\pi\in \B_{n}$ and all $S\in \Sp(N)$, the equality
\[(-2\Re\Tr)^{\otimes n}(S^{\otimes n} \rho_{\H}(\pi)(S^{-1})^{\otimes n}\circ  M_{1}\otimes \ldots \otimes M_{n})=(-2\Re\Tr)^{\otimes n}(\rho_{\H}(\pi)\circ M_{1}\otimes \ldots \otimes M_{n})\]
holds. It suffices then to observe that the bilinear form $(R,T) \mapsto (-2\Re\Tr)^{\otimes n}(R \circ T)$ is non-degenerate on $\Mat_{N}(\H)^{\otimes n}$.

2. We shall deduce this assertion from the classical description of the tensor invariants of $\Sp(N)$, as given by Theorem 5.3.3 in \cite{GoodmanWallach}. In this theorem, the group $\Sp(N)$ is defined as the subgroup of $\U(2N)$ which preserves the antisymmetric bilinear form $\omega$ on $\C^{2N}$ whose matrix in the canonical basis is $J=\begin{pmatrix} 0 & -I_{N} \\ I_{N} & 0 \end{pmatrix}$. 

Recall that in Section \ref{moments empirique} we considered a homomorphism of algebras $M\mapsto \tilde M$, which we now denote by $\imath:\Mat_{N}(\H)\to \Mat_{2N}(\C)$. A tensor of $\Mat_{N}(\H)^{\otimes n}$ is invariant under the action of $\Sp(N)\subset \Mat_{N}(\H)$ if and only if its image by $\iota^{\otimes n}$ is invariant under the action of $\Sp(N)\subset\Mat_{2N}(\C)$. The classical invariant theory describes the invariant tensors in $\Mat_{2N}(\C)^{\otimes n}$.

Theorem 5.3.3 of \cite{GoodmanWallach} describes in fact the space of invariant tensors in $(\C^{2N})^{\otimes 2n}$ rather than $\Mat_{2N}(\C)^{\otimes n}$, and asserts that this space is spanned by the orbit of a certain tensor $\theta$, which we shall characterise in the next paragraph, under the natural action of the symmetric group $\S_{2n}$ on $(\C^{2N})^{\otimes 2n}$.

In order to describe $\theta$, let us consider the isomorphism  $v\mapsto \omega(\cdot, v)$ from $\C^{2N}$ to $(\C^{2N})^{*}$ determined by the non-degenerate bilinear form $\omega$. This isomorphism allows us to build another isomorphism $\gamma:\C^{2N} \otimes \C^{2N} \to \C^{2N} \otimes (\C^{2N})^{*} \to \Mat_{2N}(\C)$ which can be described matricially as sending $X\otimes Y$ to $-X\,{}^{t}YJ$, or in coordinates, if $(e_{1},\ldots,e_{2N})$ denotes the canonical basis of $\C^{2N}$, as sending $e_{i}\otimes e_{j}$ to $-E_{ij}J$. The tensor $\theta$ is characterised by the equality $\gamma^{\otimes n}(\theta)=I_{2N}^{\otimes n}$.

In order to apply the theorem of \cite{GoodmanWallach}, it remains to understand the action of $\S_{2n}$ on $\Mat_{2N}(\C)^{\otimes n}$ inherited through $\gamma^{\otimes n}$ from the natural action on $(\C^{2N})^{\otimes 2n}$. Since $\gamma(e_{i}\otimes e_{j})=-E_{ij}J$, the action of a permutation $\sigma\in \S_{2n}$ is determined by the fact that, for all $i_{1},\ldots,i_{2n}\in \{1,\ldots,2N\}$, 
\[\sigma\cdot ((E_{i_{1}i_{2}}J) \otimes \ldots \otimes (E_{i_{2n-1}i_{2n}}J))=(E_{i_{\sigma^{-1}(1)}i_{\sigma^{-1}(2)}}J) \otimes \ldots \otimes (E_{i_{\sigma^{-1}(2n-1)}i_{\sigma^{-1}(2n)}}J).\]
Let us define, for all $\pi\in \B_{n}$, the tensor $\eta(\pi)\in\Mat_{2N}(\C)^{\otimes n}$ by
\begin{equation}\label{def eta}
\eta(\pi)=\sum_{i_{1},\ldots,i_{2n}=1}^{2N}\bigg( \prod_{\substack{\{k,l\}\in \pi\\ k<l}} J_{i_{k}i_{l}} \bigg) E_{i_{n+1}i_{1}} \otimes \ldots \otimes E_{i_{2n}i_{n}} \circ J^{\otimes n}.
\end{equation}
One checks easily on one hand that $\eta(\id)=I_{2N}^{\otimes n}$ and on the other hand, for all $\sigma\in \S_{2n}$ and all $\pi\in \B_{n}$, that $\eta(\sigma \cdot \pi)=\sigma \cdot \eta(\pi)$. It follows that the space of invariant tensors in $\Mat_{2N}(\C)^{\otimes n}$ is spanned over $\C$ by $\{\eta(\pi) : \pi \in \B_{n}\}$. 

We claim that for all $\pi\in \B_{n}$, $\eta(\pi)$ belongs to the range of $\iota^{\otimes n}$ and more precisely that, with an indeterminacy on the sign which we shall lift later,
\begin{equation}\label{eta rho pm}
\eta(\pi)=\pm\iota^{\otimes n}(\rho_{\H}(\pi)).
\end{equation}
This follows from two observations, of which we leave the verification to the reader. The first observation is that for all Brauer diagrams $\pi_{1}$ and $\pi_{2}$, one has either $\eta(\pi_{1}\pi_{2})=\eta(\pi_{1})\eta(\pi_{2})$ or $\eta(\pi_{1}\pi_{2})=-\eta(\pi_{1})\eta(\pi_{2})$. This is best understood graphically by representing $\eta(\pi)$ by a box with $n$ edges as we did for $\pi$ in Section \ref{section:brauer I}, with the additional features that each edge carries a matrix $J$, and that there is a box representing $J^{\otimes n}$ below the box representing $\pi$ (see Figure \ref{rhodepi}). The second observation, which is checked by direct computation, is that for all distinct $i,j\in \{1,\ldots,n\}$, we have $\iota^{\otimes n}(\rho_{\H}(i\, j))=-\eta((i\, j))$ and $\iota^{\otimes n}(\rho_{\H}\langle i\, j\rangle)=-\eta(\langle i\, j\rangle)$. Since transpositions and contractions generate the Brauer algebra, our claim is proved. 

\begin{figure}[h!]
\begin{center}
\scalebox{0.8}{\includegraphics{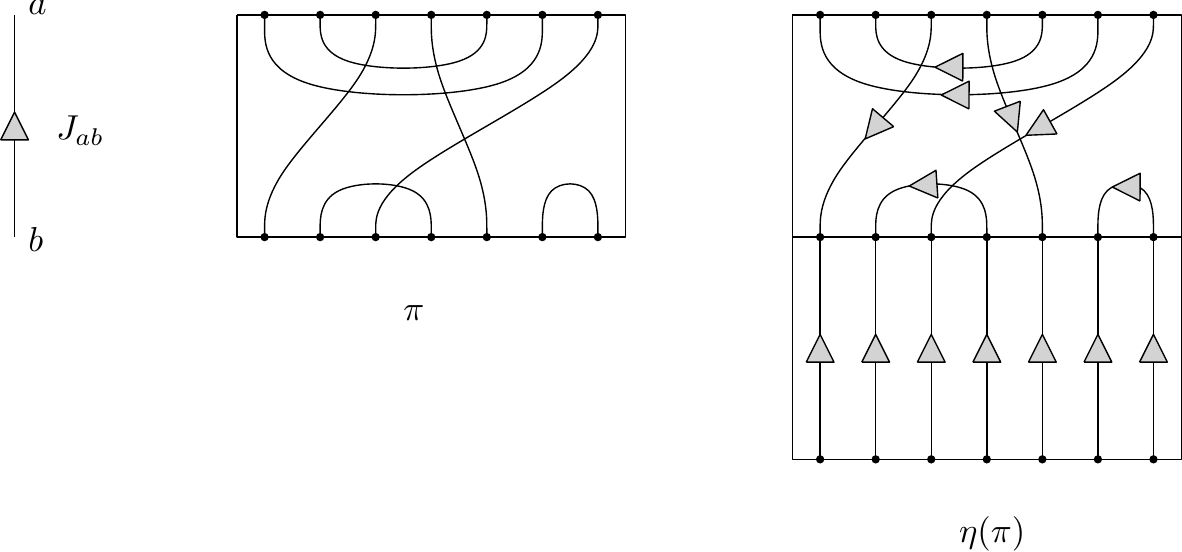}}
\caption{\label{rhodepi}  The diagram on the right represents $\eta(\pi)$. Since $\t{J}=-J$, the orientation of the triangles representing $J$ matters. The way in which we ordered the pairs $\{k,l\}$ in the definition of $\rho(\pi)$ is reflected in the fact that we orient the vertical edges downwards and the horizontal edges leftwards. An elementary parity check suffices to convince oneself that $\eta(\pi_{1}\pi_{2})=\pm\eta(\pi_{1})\eta(\pi_{2})$.}
\end{center}
\end{figure}

We can now consider an invariant tensor $T$ in $\Mat_{N}(\H)^{\otimes n}$. We know that $\iota^{\otimes n}(T)$ is a linear combination with complex coefficients of the tensors $\eta(\pi), \pi \in \B_{n}$. It remains to prove that the coefficients of this linear combination can be taken to be real.

For this, we prove that ${\rm Im}(\iota^{\otimes n}) \cap \ii {\rm Im}(\iota^{\otimes n})= \{0\}$. Indeed, the range of $\iota$ in $\Mat_{2N}(\C)$ is the real subspace $\{M\in \Mat_{2N}(\C) : -J\overline{M}J=M\}$. Thus, any tensor $R$ in the range of $\iota^{\otimes n}$ satisfies $J^{\otimes n} \overline R J^{\otimes n}=(-1)^{n}R$. Since $J^{\otimes n} \overline{ \ii R} J^{\otimes n}=(-1)^{n+1} \ii R$, the tensor $\ii R$ does not belong to the range of $\iota^{\otimes n}$, unless $R=0$. 

Let us assume that $\iota^{\otimes n}(T)=\sum_{\pi\in \B_{n}} c_{\pi} \eta(\pi)$, for some complex coefficients $c_{\pi}$ which we write $a_{\pi}+\ii b_{\pi}$ with $a_{\pi}$ and $b_{\pi}$ real. Then $\sum_{\pi\in \B_{n}} \ii b_{\pi} \eta(\pi)=\iota^{\otimes n}(T)-\sum_{\pi\in \B_{n}} a_{\pi} \eta(\pi)$ belongs to ${\rm Im}(\iota^{\otimes n}) \cap \ii {\rm Im}(\iota^{\otimes n})$ and hence, according to our last observation, vanishes. Using \eqref{eta rho pm} and the fact that $\iota^{\otimes n}$ is injective, it follows that $T=\sum_{\pi\in \B_{n}} \pm a_{\pi} \rho_{\H}(\pi)$, which concludes the proof of the theorem.
\end{proof}

Having reached this point, and although this is not strictly necessary for our purpose, we will indulge in taking the time to determine the exact sign which appears in \eqref{eta rho pm}. For this, we use again the characterisation of $\rho_{\H}(\pi)$ given by Lemma \ref{tr tr sp}. Since $\iota$ commutes with the adjunctions and satisfies the equality $\Tr(\iota(M))=2\Re\Tr(M)$, it follows from Lemma \ref{tr tr sp} that, for all $M_{1},\ldots,M_{n}\in \Mat_{2N}(\C)$ and all $\pi \in \B_{n}$, we have
\[(-1)^{n}\Tr^{\otimes n}(\iota^{\otimes n}(\rho_{\H}(\pi)) \circ M_{1} \otimes \ldots \otimes M_{n}) = (-1)^{\ell(\pi)}\prod_{(i_{1}\ldots i_{s}) \preccurlyeq \sigma_{\pi}} \Tr(M_{i_{s}}^{*_{i_{s}}} \ldots M_{i_{1}}^{*_{i_{1}}}),\]
where $M_{i}^{*_{i}}$ equals $M_{i}$ if $\epsilon_{\pi}(i)=1$ and $M_{i}^{*}$ if $\epsilon_{\pi}(i)=-1$.

Let $\pi\in \B_{n}$ be a Brauer diagram. We shall now compute $\Tr^{\otimes n}(\eta(\pi)\circ M_{1} \otimes \ldots \otimes M_{n})$, and for this we need to define a few more integers depending on $\pi$. Recall from Section \ref{Brauer II} the cycle structure on $\{1,\ldots,2n\}$ which we attached to $\pi$ and the way in which we oriented it. Let us call vertical edge a primary edge which joins a point at the bottom of the box to a point at the top of the box which represents $\pi$. Let us call horizontal edge a primary edge which is not vertical. Let us define $v_{-}(\pi)$ as the number of vertical edges which are oriented downwards and $h_{-}(\pi)$ the number of horizontal edges which are oriented leftwards. Let us also define $n_{-}(\pi)$ as the number of indices $i\in \{1,\ldots,n\}$ such that $\epsilon_{\pi}(i)=-1$. Then as a variation on \eqref{P cycles o}, we have, for all $M_{1},\ldots,M_{n}\in \iota(\Mat_{N}(\H))$ and all $\pi \in \B_{n}$,
\[\Tr^{\otimes n}(\eta(\pi)\circ M_{1} \otimes \ldots \otimes M_{n})=(-1)^{v_{-}(\pi)+h_{-}(\pi)+n_{-}(\pi)} \prod_{(i_{1}\ldots i_{s}) \preccurlyeq \sigma_{\pi}} \Tr(M_{i_{s}}^{*_{i_{s}}} \ldots M_{i_{1}}^{*_{i_{1}}}).\]
Let us emphasise that this formula holds as such only for complex matrices $M_{i}$ which are in the range of $\iota$, and hence satisfy $J\t{M}J=-M^{*}$.

The sign which appears can be slightly simplified as follows. Let $h(\pi)$ and $h_{+}(\pi)$ denote respectively the number of horizontal edges and the number of horizontal edges oriented towards the right. Then plainly $h(\pi)=h_{+}(\pi)+h_{-}(\pi)$. Moreover, $n_{-}(\pi)=v_{-}(\pi)+\frac{1}{2}h(\pi)$. Indeed, the bottom end of each vertical edge oriented downwards carries a sign $\epsilon_{\pi}$ equal to $-1$, as does exactly one end of each horizontal edge at the bottom of the box. It remains to observe that there are as many horizontal edges at the bottom and at the top of the box. Finally, $v_{-}(\pi)+h_{-}(\pi)+n_{-}(\pi)$ has the same parity as $\frac{1}{2}(h_{+}(\pi)-h_{-}(\pi))$.

Altogether, using again the fact that the bilinear form $(R,T)\mapsto \Tr^{\otimes n}(RT)$ is non-degenerate, we find
the equality $\iota^{\otimes n}(\rho_{\H}(\pi))=(-1)^{n-\ell(\pi)+\frac{1}{2}\left(h_{+}(\pi)-h_{-}(\pi)\right)} \eta(\pi)$. We have thus proved the following.

\begin{proposition} Let $n,N$ be two positive integers. The mapping $\rho_{\H}^{\C}:\B_{n}\to\Mat_{2N}(\C)^{\otimes n}$ defined, for all $\pi \in \B_{n}$, by
\[\rho_{\H}^{\C}(\pi)=(-1)^{n-\ell(\pi)+\frac{1}{2}\left(h_{+}(\pi)-h_{-}(\pi)\right)}\sum_{i_{1},\ldots,i_{2n}=1}^{2N}\bigg( \prod_{\substack{\{k,l\}\in \pi\\ k<l}} J_{i_{k}i_{l}} \bigg) E_{i_{n+1}i_{1}} \otimes \ldots \otimes E_{i_{2n}i_{n}} \circ J^{\otimes n}\]
extends by linearity to a homomorphism of algebras $\rho_{\H}^{\C}:\Br_{n,-2N}\to \Mat_{2N}(\C)^{\otimes n}$.
\end{proposition}

This digression is now over and we come back to our main problem of asymptotic freeness. To start with, let us define the endomorphism $P$ of the real algebra $\Mat_{N}(\H)^{\otimes n}$ by setting
\[\forall A\in \Mat_{N}(\H)^{\otimes n} \; , \;\;P(A)=\int_{\Sp(N)} S^{\otimes n} \circ A \circ  (S^{-1})^{\otimes n} \; dS.\]
By Theorem \ref{fft sp}, the range of $P$ is contained in the range of $\rho_{\H}$. In order to make this inclusion explicit, let us define, for all $A\in \Mat_{N}(\H)^{\otimes n}$, an element $Q_{0}(A)$ of $\Br_{n,-2N}$ by setting 
\[Q_{0}(A)=\sum_{\pi\in \B_{n}} (-2\Re\Tr)^{\otimes n}(A \circ \rho_{\H}(\t{\pi})) \pi.\]

Since $\rho_{\H}$ is a homomorphism of algebras, and thanks to Lemma \ref{tr tr sp} and Lemma \ref{d bn}, we have
\[\left((-2\Re\Tr)^{\otimes n}(\rho_{\H}(\pi_{1}) \rho_{\H}(\t{\pi_{2}}))\right)_{\pi_{1},\pi_{2} \in \B_{n}}=(-2N)^{n}((-2N)^{-d(\pi,\pi')})_{\pi,\pi'\in \B_{n}}.\]
Just as in the orthogonal case, this matrix is invertible for $N$ large enough, and so is the restriction of $Q_{0}$ to $\rho_{\H}(\Br_{n,2N})$. We denote by $\Wg_{-2N}$ its inverse and shall use the notation $\Wg_{-2N}(\pi)=\sum_{\pi'\in \B_{n}}\Wg_{-2N}(\pi,\pi') \rho_{\R}(\pi')$.

Consider $A\in \Mat_{N}(\H)^{\otimes n}$. On one hand, $(\Wg_{-2N}\circ Q_{0})(P(A))=(\Wg_{-2N}\circ Q_{0})(A)$, because $\rho_{\H}(\pi)$ and $S^{\otimes n}$ commute for all $\pi\in \B_{n}$ and all $S\in \Sp(N)$. On the other hand, $(\Wg_{-2N}\circ Q_{0})(P(A))=P(A)$ because $P(A)$ belongs to the range of $\rho_{\H}$, as we know by Theorem \ref{fft sp}. Hence, the formula in the symplectic case is
\[\int_{\Sp(N)} S^{\otimes n} \circ A \circ  (S^{-1})^{\otimes n} \; dS=\sum_{\pi\in \B_{n}} (-2\Re\Tr)^{\otimes n}(A \circ \rho_{\H}(\t{\pi})) \Wg_{-2N}(\pi).\]
From this we deduce, for all $A_{1},\ldots,A_{n}$ and $B_{1},\ldots,B_{n}$ in $\Mat_{N}(\H)$,
\begin{align}
\nonumber \int_{\Sp(N)} \Re\tr(SA_{1}S^{-1}B_{1} \ldots SA_{n}S^{-1}B_{n}) \; dS =& \\
&\hspace{-4cm}-\frac{1}{2N}\sum_{\pi,\pi'\in \B_{n}}\Wg_{-2N}(\pi,\pi')(-2\Re\Tr)^{\otimes n}(A_{1}\otimes \ldots \otimes A_{n} \circ \rho_{\H}(\t{\pi})) \nonumber\\
&\hspace{0.7cm}(-2\Re\Tr)^{\otimes n}(B_{1}\otimes \ldots \otimes B_{n} \circ  \rho_{\H}((n \ldots 1)\pi')).\label{som sp}
\end{align}
The same computation as in the orthogonal case, with $N$ replaced by $-2N$ shows that the highest order of $N$ in $\Wg_{-2N}(\pi,\pi')$ is $-d(\pi,\pi')$. The dominant terms of \eqref{som sp} are thus of order $0$ in $N$, so that the constant $-2$ disappears, and the coefficients are, for the same reason as in the orthogonal case, given by the M\"obius function of the lattice $\NC_{n}$. We thus find
\begin{align*}
 \int_{\Sp(N)} \Re\tr(SA_{1}S^{-1}B_{1} \ldots SA_{n}S^{-1}B_{n}) \; dS =& \\
&\hspace{-6cm}\sum_{\sigma \preccurlyeq \sigma' \preccurlyeq (1\ldots n)}\mu(\sigma,\sigma') \prod_{(i_{1}\ldots i_{r})\preccurlyeq \sigma} \Re\tr(A_{i_{1}} \ldots  A_{i_{r}})  \prod_{(j_{1}\ldots j_{s})\preccurlyeq (n\ldots 1)\sigma'} \Re\tr(B_{j_{s}} \ldots B_{j_{1}}) +O(N^{-1}).\label{som sp}
\end{align*}
Let us modify the definition of $p_{\sigma}$ to suit the symplectic case, by setting
\[p_{\sigma}(A_{1},\ldots,A_{n})=\prod_{\substack{c \mbox{ \scriptsize cycle of }\sigma\\ c=(i_{1}\ldots i_{r})}} \Re\tr(A_{i_{1}}\ldots A_{i_{r}}).\]
The cumulants are defined by the usual relation $\kappa_{\sigma}=\sum_{\sigma'\preccurlyeq \sigma} \mu(\sigma',\sigma) p_{\sigma'}$.
We finally have
\begin{align*}
\int_{\Sp(N)} \Re\tr(SA_{1}S^{-1}B_{1} \ldots SA_{n}S^{-1}B_{n}) \; dS =& \\
&\hspace{-5cm} = \sum_{\sigma\preccurlyeq \sigma' \preccurlyeq (1\ldots n)}\mu(\sigma,\sigma') p_{\sigma}(A_{1},\ldots,A_{n}) p_{(\sigma')^{-1}(1\ldots n)}(B_{1},\ldots,B_{n})+O(N^{-1})\\
&\hspace{-5cm} = \sum_{\sigma' \preccurlyeq (1\ldots n)} \kappa_{\sigma'}(A_{1},\ldots,A_{n}) p_{(\sigma')^{-1}(1\ldots n)}(B_{1},\ldots,B_{n})+O(N^{-1})\\
&\hspace{-5cm} = \sum_{\beta \in \NC_{n}} \kappa_{\beta}(A_{1},\ldots,A_{n}) \tau_{\beta^{\vee}}(B_{1},\ldots,B_{n})+O(N^{-1})
\end{align*}
and conclude as before to the asymptotic freeness of the families $\{SA_{1}S^{-1},\ldots,SA_{n}S^{-1}\}$ and $\{B_{1},\ldots,B_{n}\}$.

\renewcommand \thesection{B}
\section{Table of the master field}\label{table mf}

We give a table of the values of the function $\Phi$ on all elementary loops with no more than three points of self-intersection. We start with the unique loop without self-intersection and the two loops with one point of self-intersection. Each face is labelled by the letter which denotes its area.\\

\begin{center}
\begin{tabular}{|cccc|}
\hline $l$ &&&$\Phi(l)$ \\
\hline \includegraphics[height=2cm]{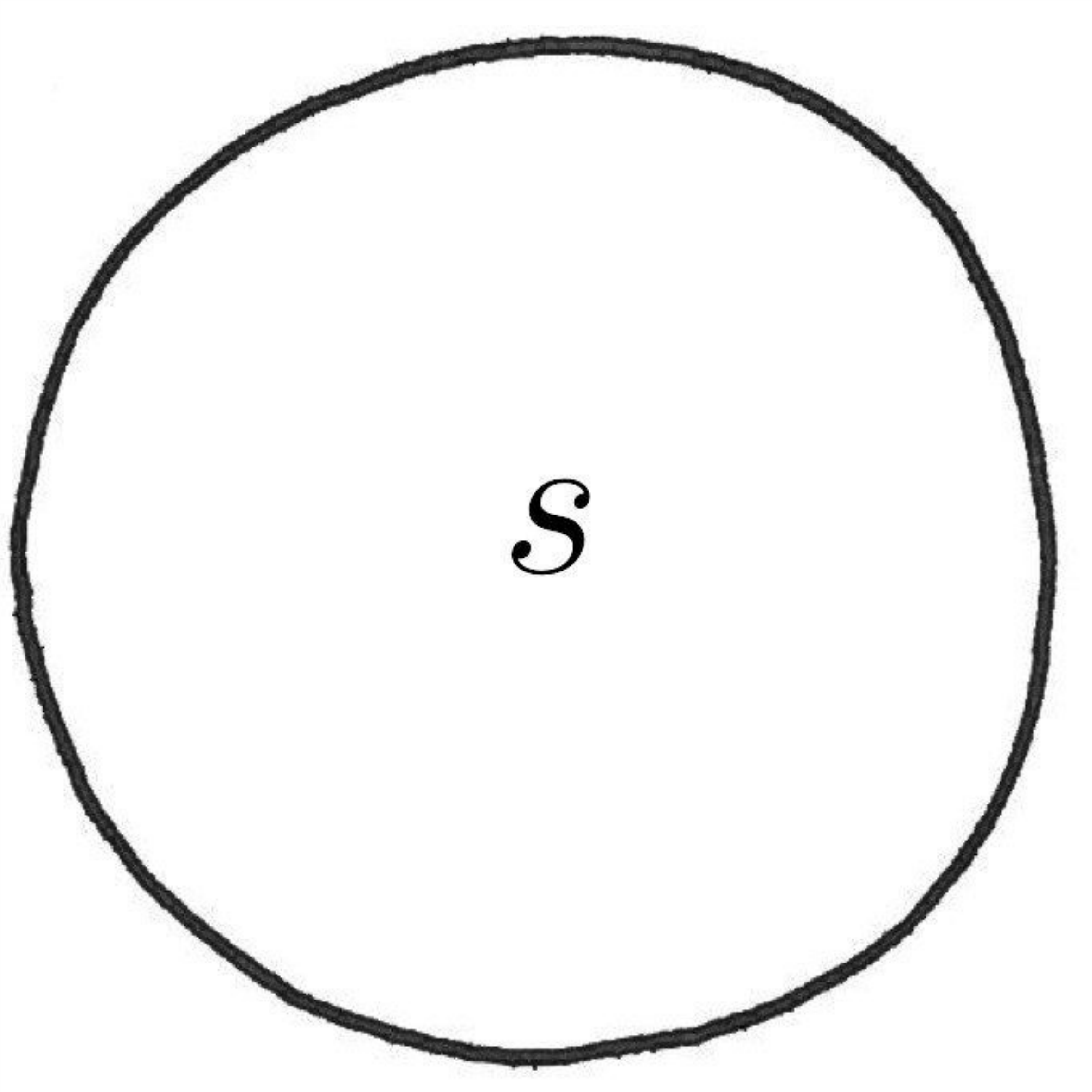} &&& \raisebox{1cm}{$e^{-\frac{s}{2}}$} \\
\hline
\end{tabular}
\hspace{1cm}
\begin{tabular}{|cccl|}
\hline $l$ &&&$\Phi(l)$ \\
\hline \includegraphics[height=2cm]{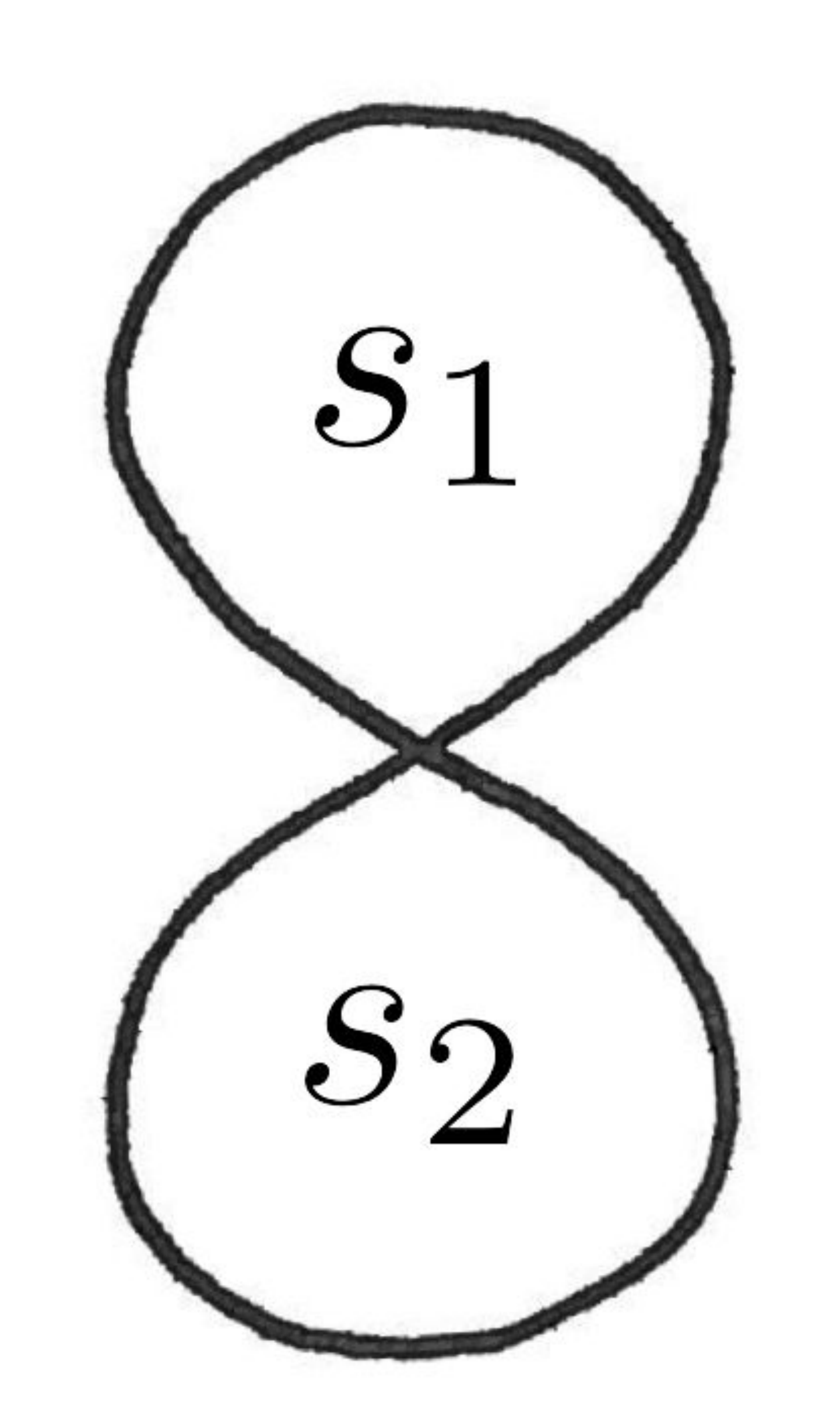} &&& \raisebox{1cm}{$e^{-\frac{1}{2}(s_{1}+s_{2})}$} \\
\includegraphics[height=2cm]{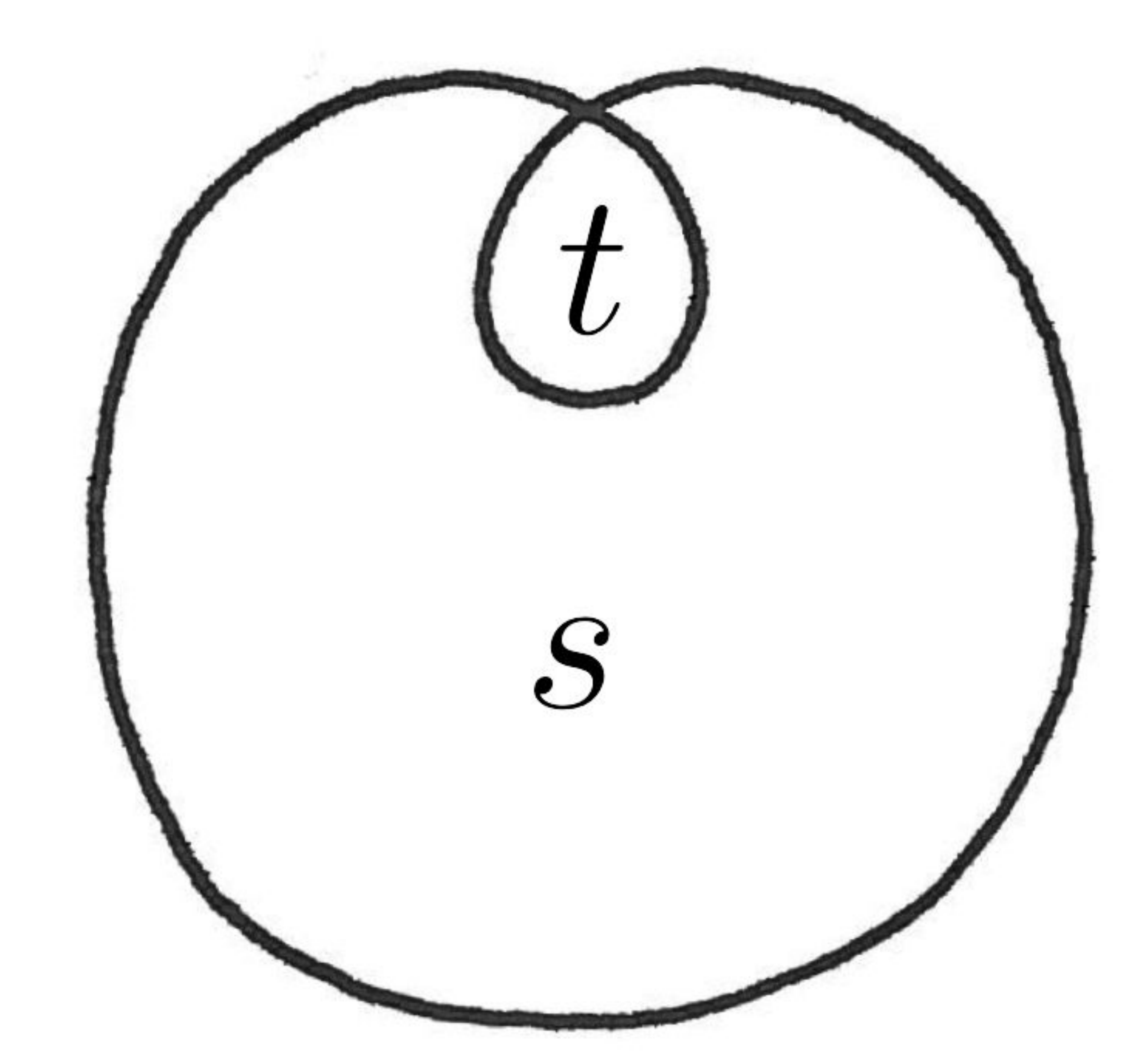} &&& \raisebox{1cm}{$e^{-\frac{s}{2}-t}(1-t)$}\\ \hline 
\end{tabular}
\bigskip
\end{center}

We continue with the five loops, up to isotopy, with two points of self-intersection.\\

\begin{center}
\begin{tabular}{|cccl|}
\hline $l$ &&&$\Phi(l)$ \\
\hline \includegraphics[height=2cm]{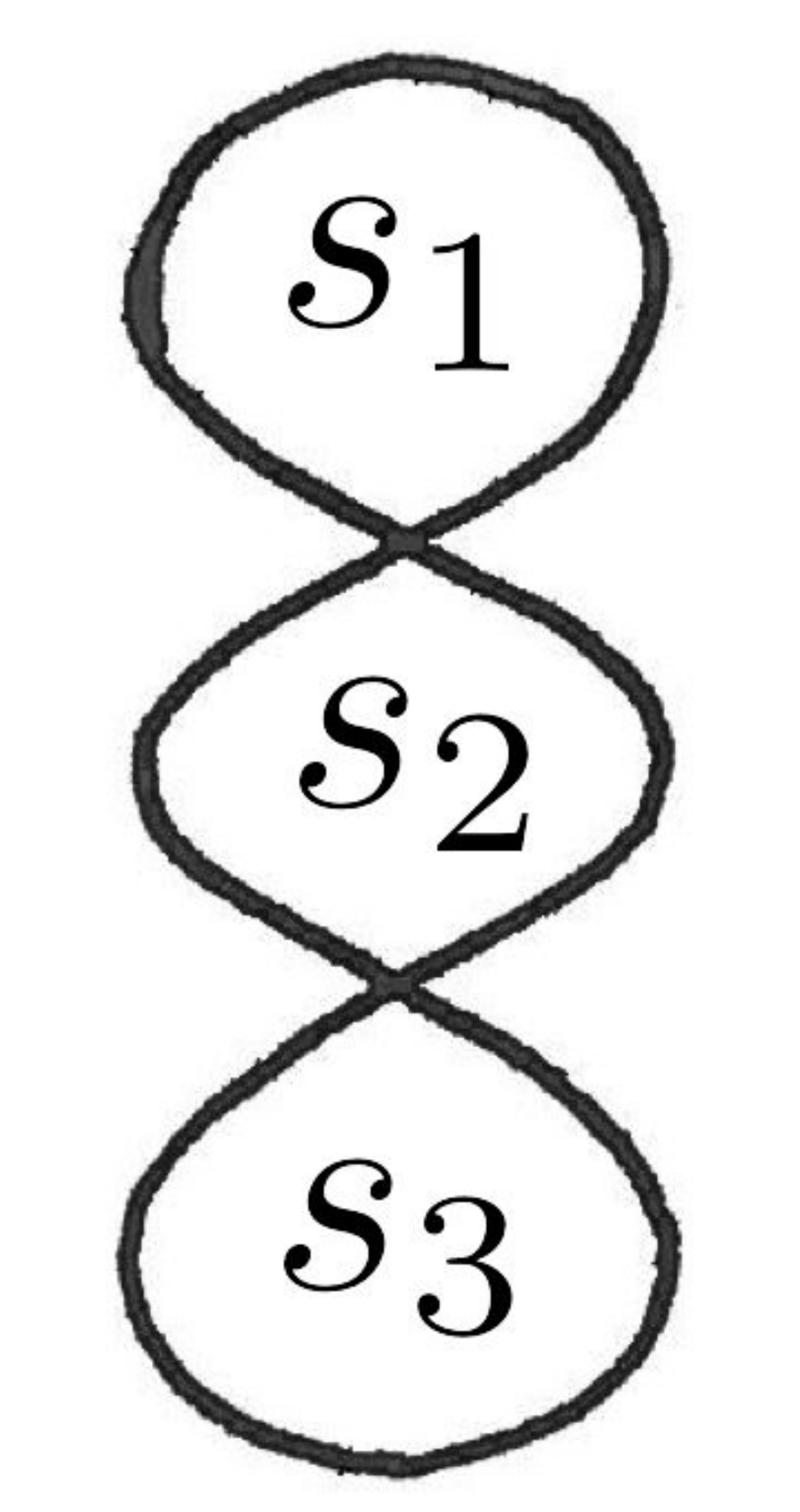} &&& \raisebox{1cm}{$e^{-\frac{1}{2}(s_{1}+s_{2}+s_{3})}$} \\
\includegraphics[height=2cm]{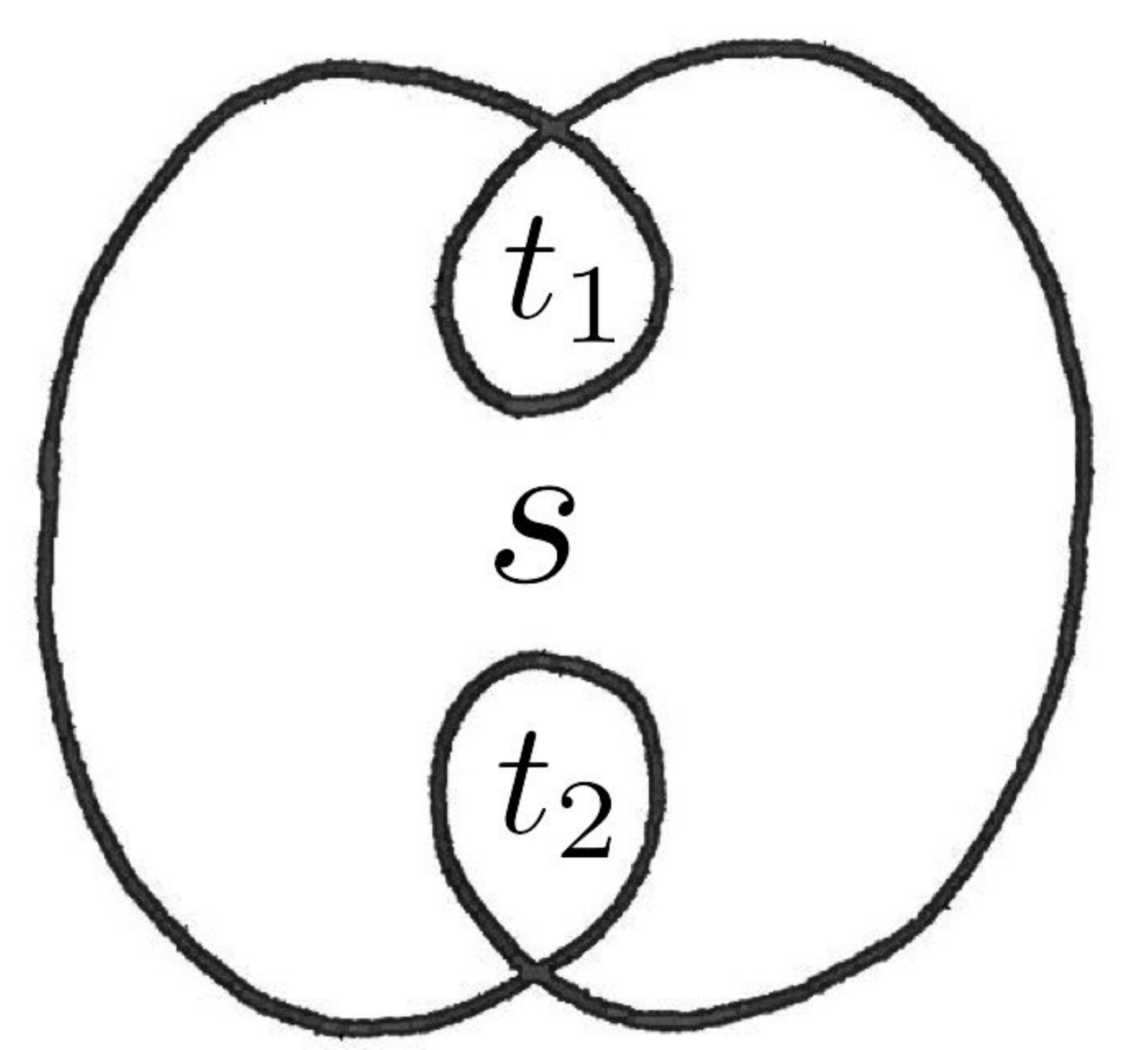} &&& \raisebox{1cm}{$e^{-\frac{s}{2}-t_{1}-t_{2}}(1-t_{1})(1-t_{2})$}\\
\includegraphics[height=2cm]{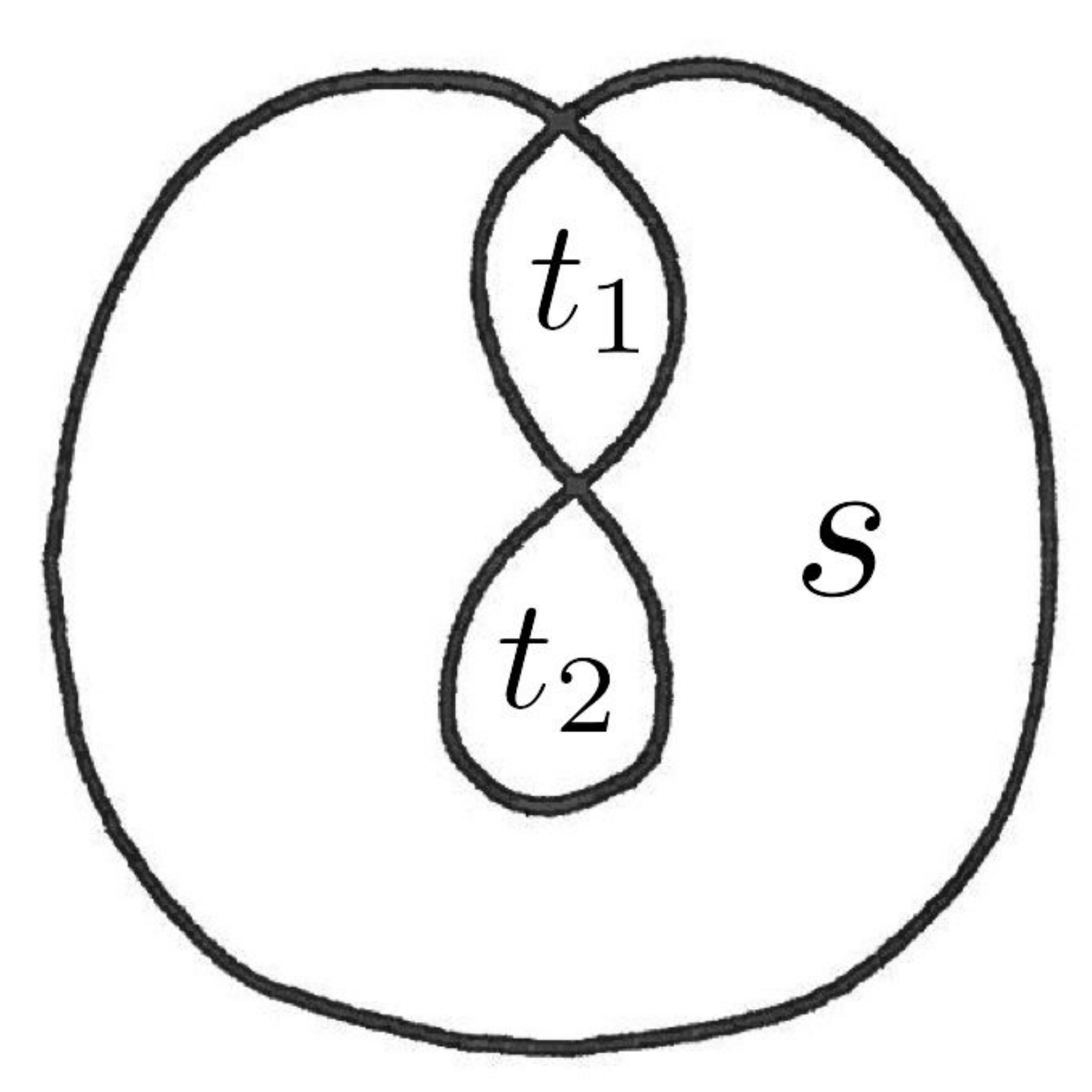} &&& \raisebox{1cm}{$e^{-\frac{s}{2}-t_{1}}(1-t_{1}e^{-t_{2}})$}\\
\includegraphics[height=2cm]{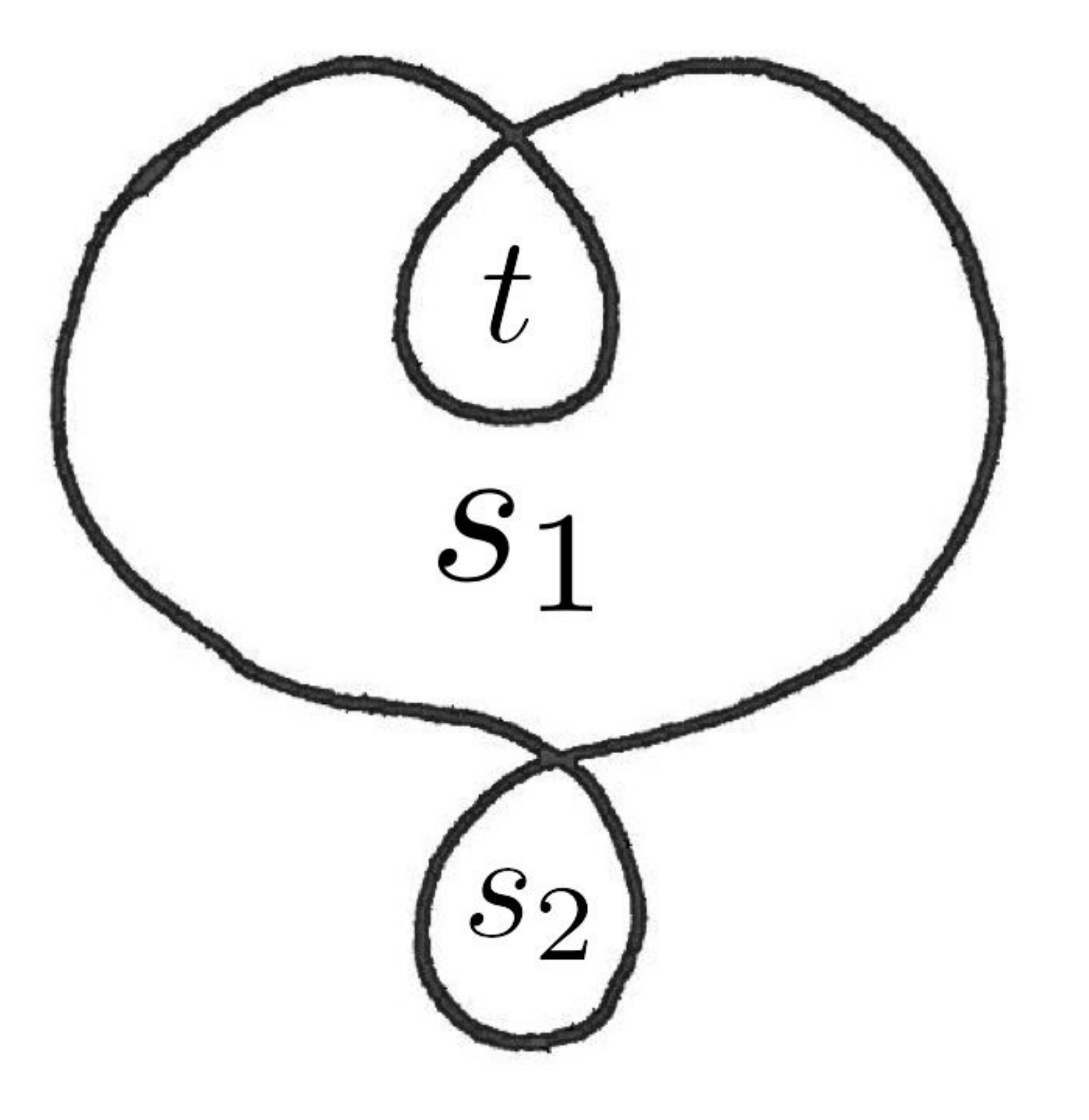} &&& \raisebox{1cm}{$e^{-\frac{1}{2}(s_{1}+s_{2})-t}(1-t)$}\\
\includegraphics[height=2cm]{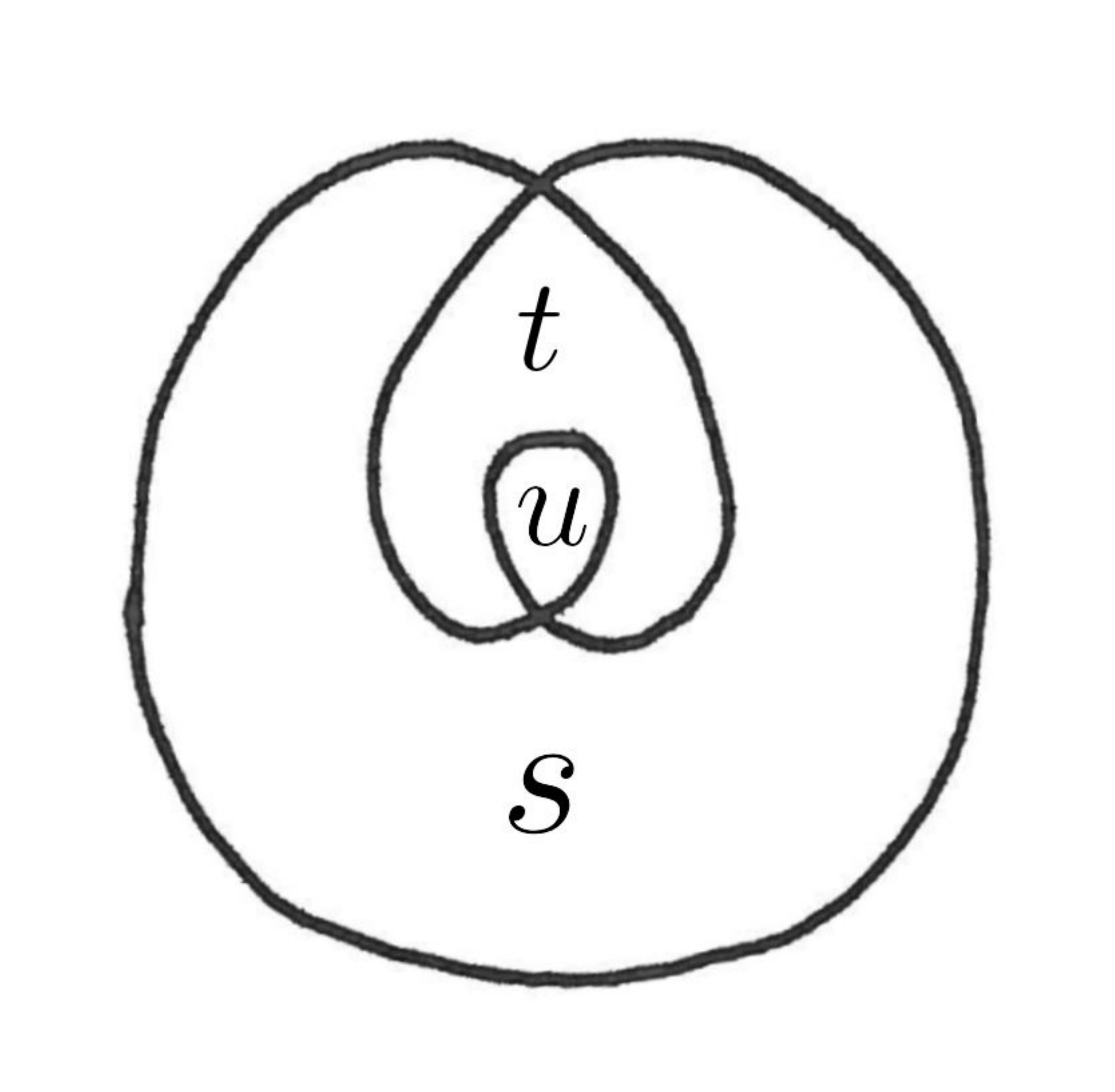} &&& \raisebox{1cm}{$e^{-\frac{s}{2}-t-\frac{3u}{2}}\left(1-3u+\frac{3}{2}u^{2}-t(1-u)\right)$}\\ \hline 
\end{tabular}
\bigskip
\end{center}

And finally the twenty loops, up to isotopy, with three points of self-intersection.\\

\begin{center}
\begin{tabular}{|cccl|}
\hline $l$ &&&$\Phi(l)$ \\
\hline \includegraphics[height=2cm]{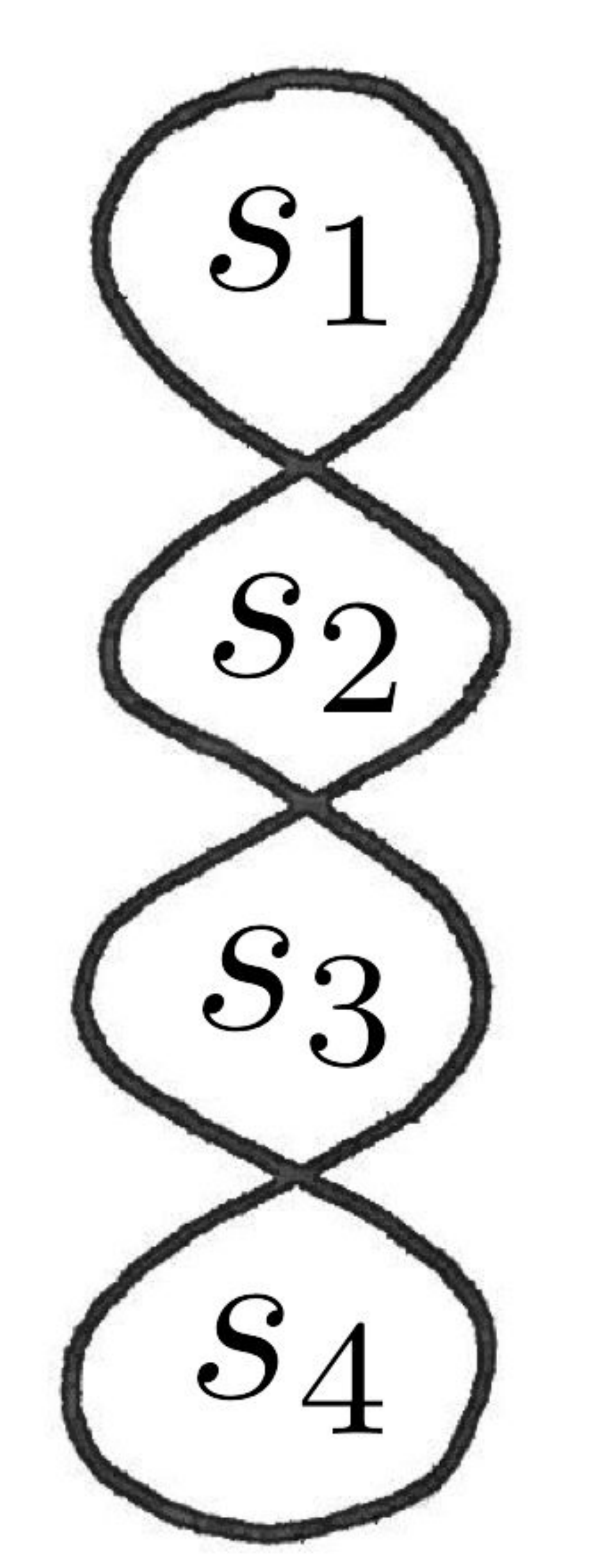} &&& \raisebox{1cm}{$e^{-\frac{1}{2}(s_{1}+s_{2}+s_{3}+s_{4})}$} \\
\includegraphics[height=2cm]{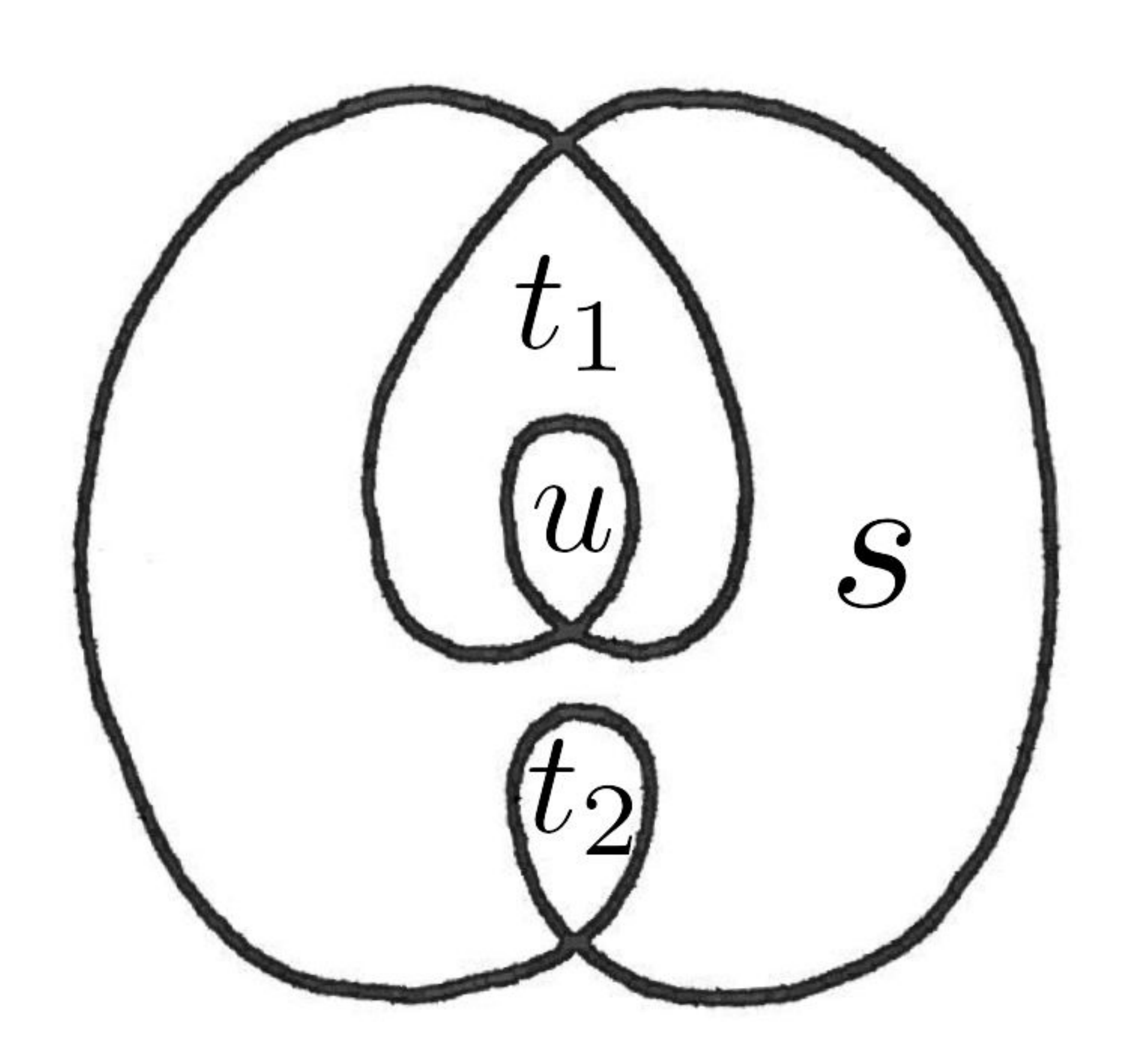} &&& \raisebox{1cm}{$e^{-\frac{s}{2}-t_{1}-t_{2}-\frac{3u}{2}}(1-3u+\frac{3}{2}u^{2}-t_{1}(1-u))(1-t_{2})$}\\
\includegraphics[height=2cm]{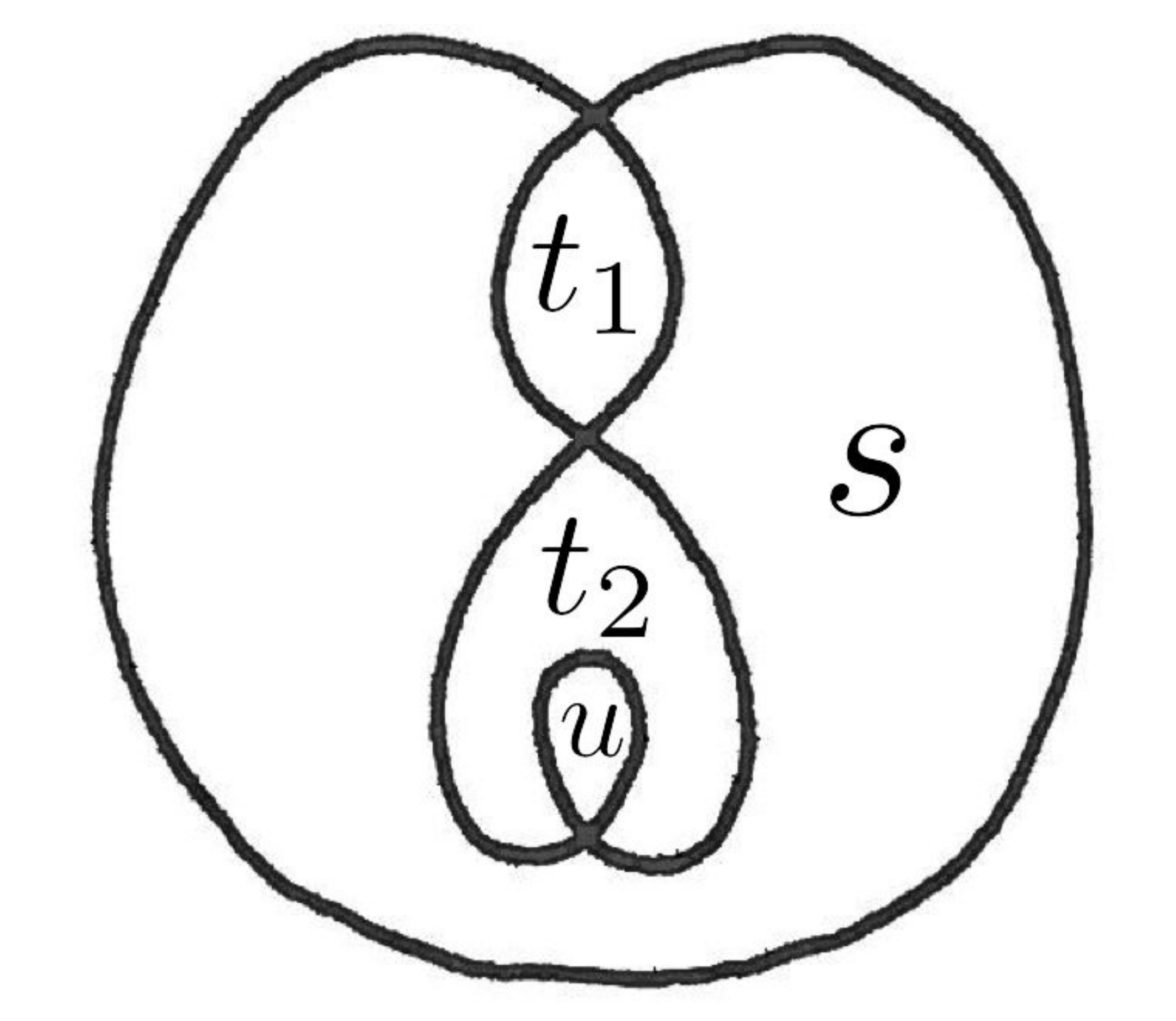} &&& \raisebox{1cm}{$e^{-\frac{s}{2}-t_{1}-\frac{u}{2}}(1-t_{1}e^{-t_{2}-u}(1-u))$}\\
\includegraphics[height=2cm]{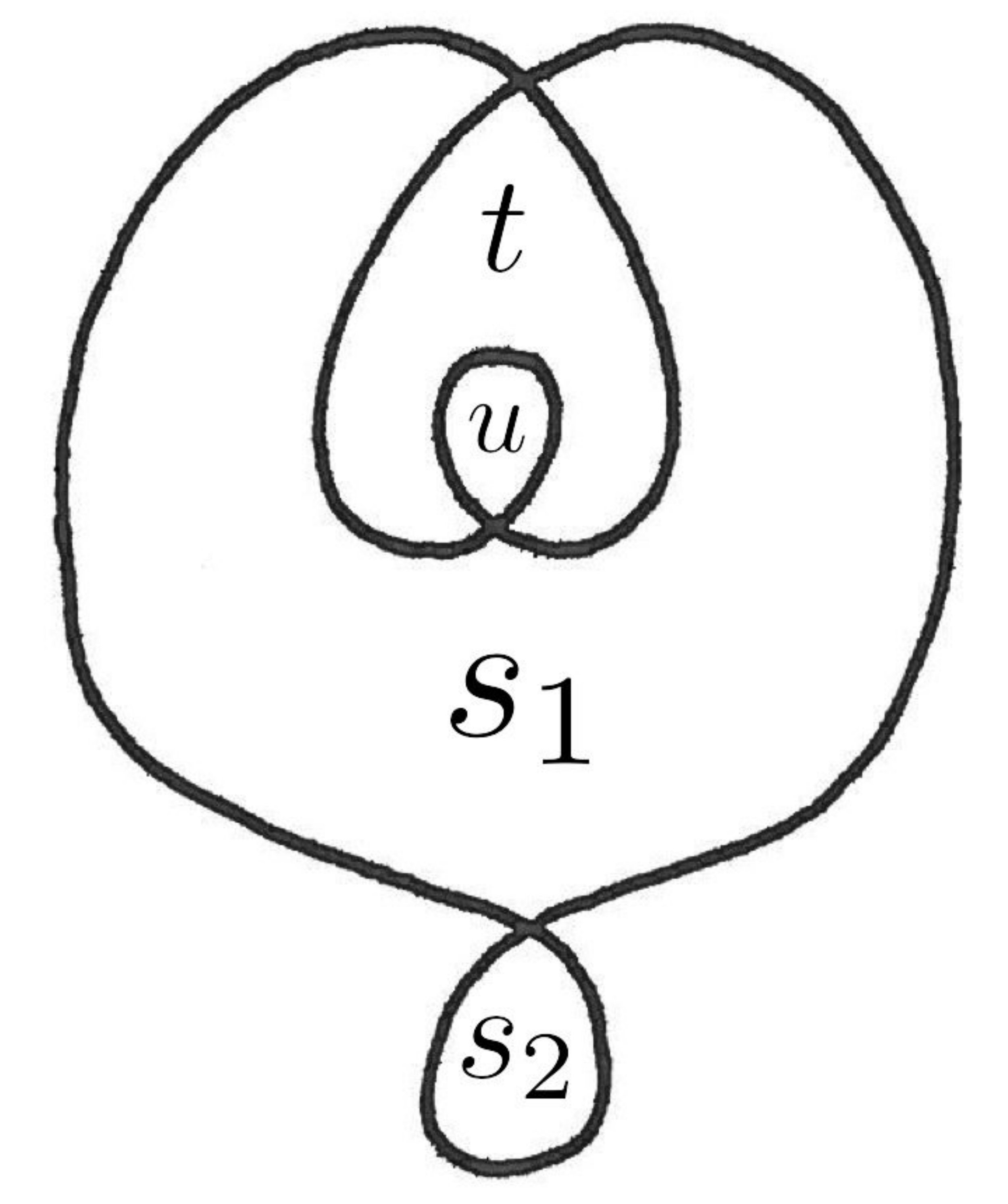} &&& \raisebox{1cm}{$e^{-\frac{1}{2}(s_{1}+s_{2})-t-\frac{3u}{2}}(1-3u+\frac{3}{2}u^{2}-t(1-u))$}\\ 
\includegraphics[height=2cm]{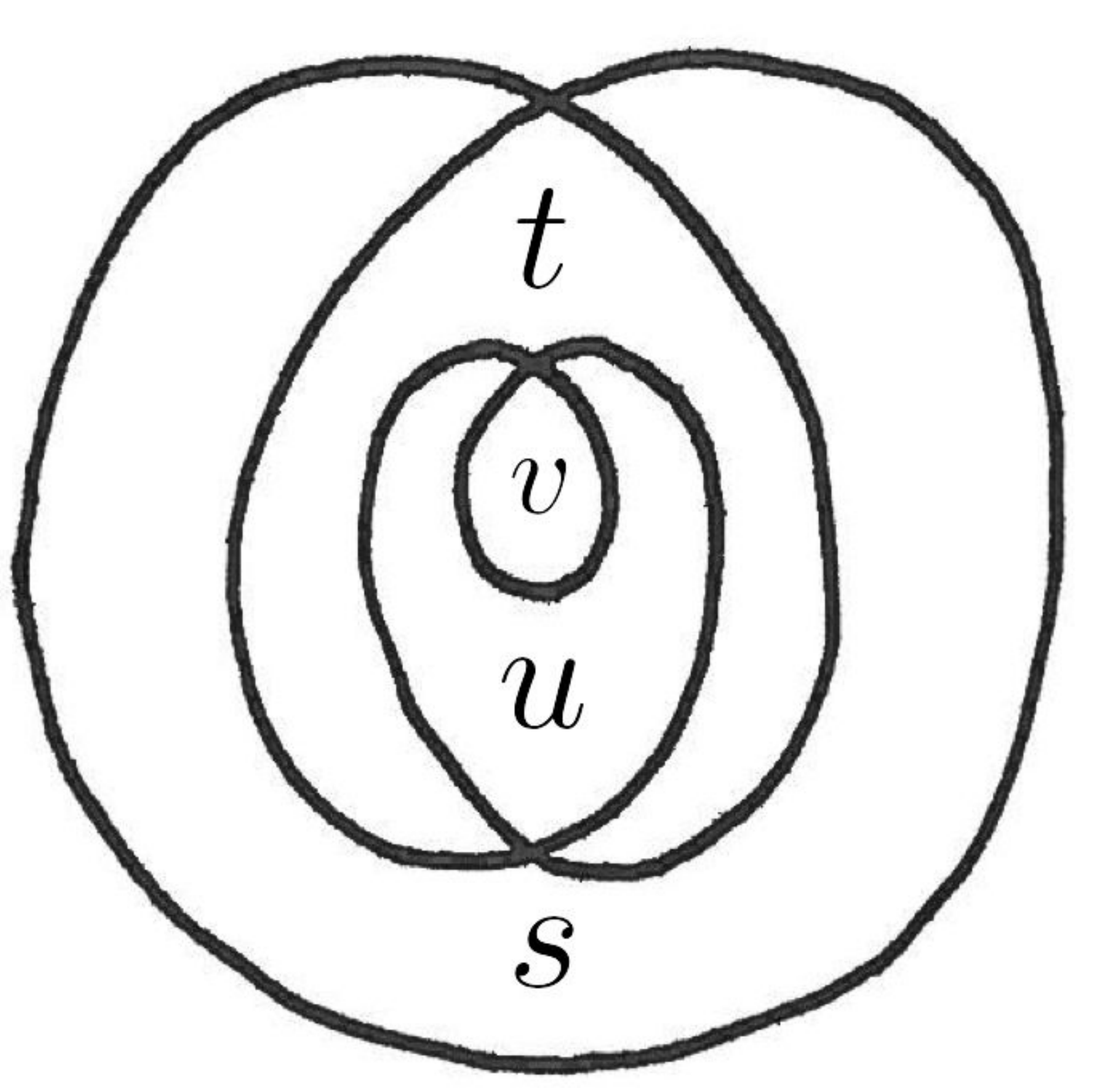} &&& \raisebox{1.2cm}{$e^{-\frac{s}{2}-t-\frac{3u}{2}-2v}\left(1-3u+\frac{3}{2}u^{2}-t(1-u)-6v+8v^{2}-\frac{8}{3}v^{3}\right.$}

\raisebox{0.6cm}{\hspace{-6.5cm}$\left.+8uv-\frac{3}{2}u^{2}v-4uv^{2}-tuv-\frac{3}{2}tv^{2}+3tv\right)$} \\
\includegraphics[height=2cm]{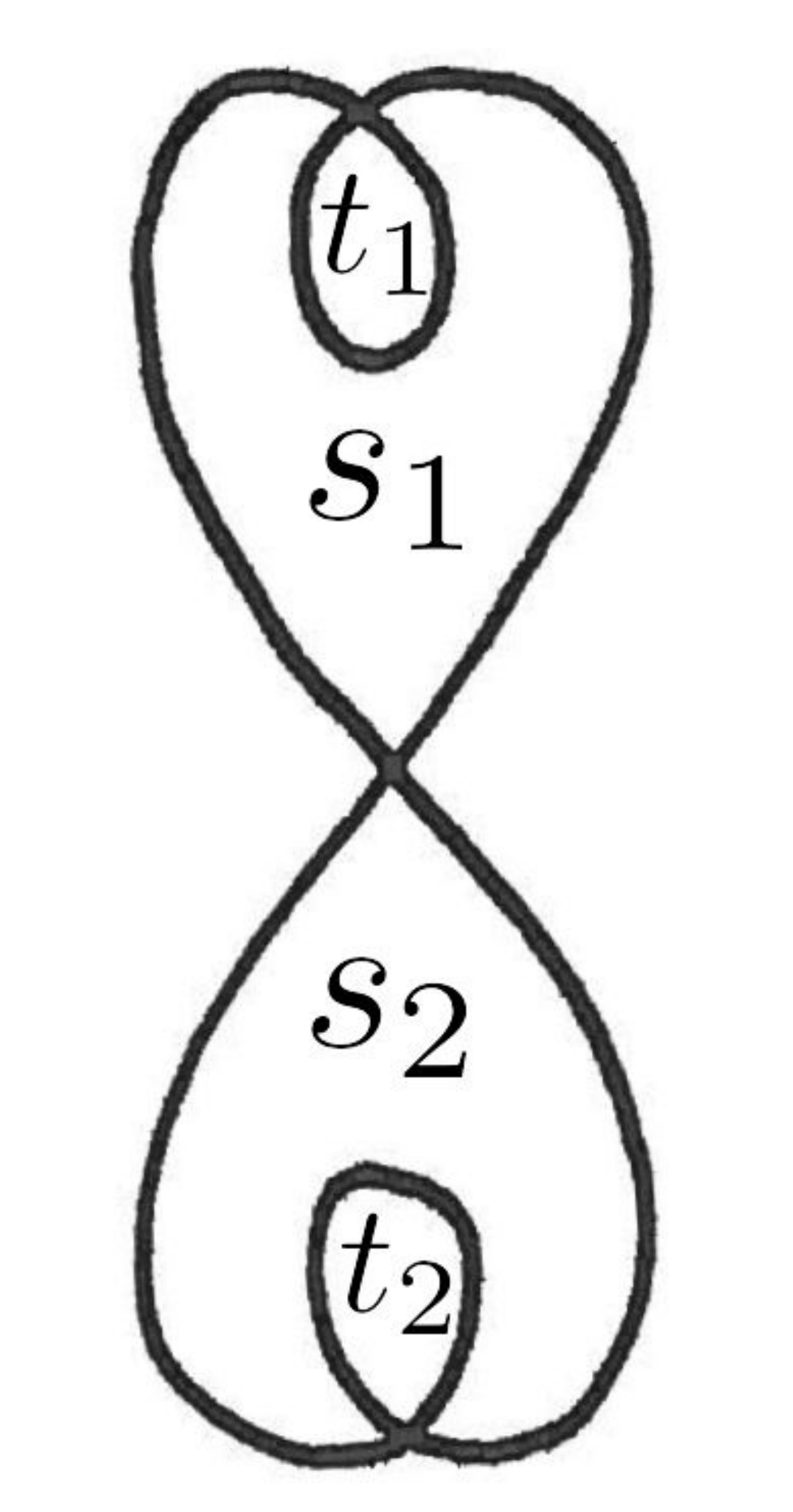} &&& \raisebox{1cm}{$e^{-\frac{1}{2}(s_{1}+s_{2})-t_{1}-t_{2}}(1-t_{1})(1-t_{2})$}\\
\includegraphics[height=2cm]{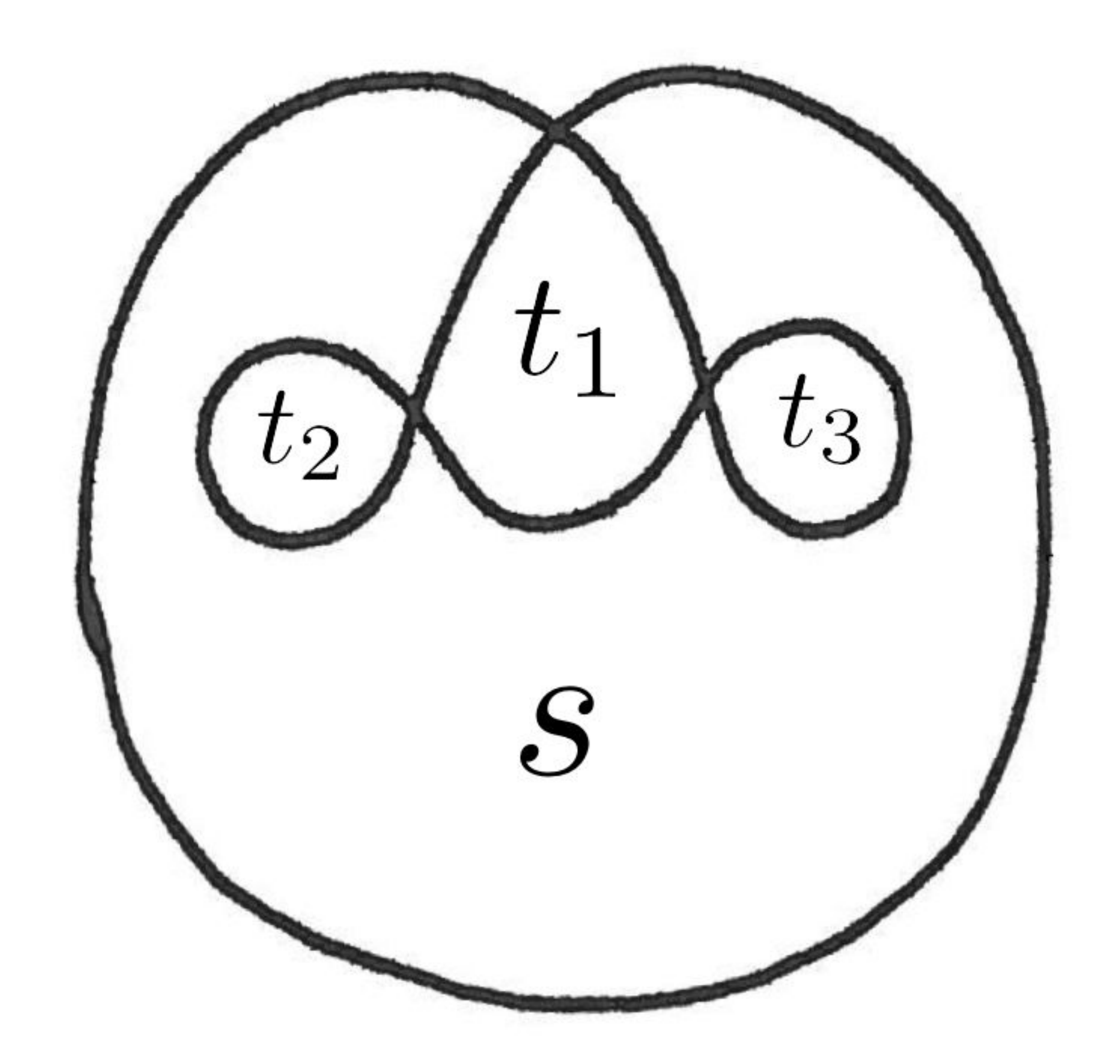} &&& \raisebox{1cm}{$e^{-\frac{s}{2}-t_{1}}(e^{-t_{2}}+e^{-t_{3}}-(1+t_{1})e^{-t_{2}-t_{3}})$}\\
\includegraphics[height=2cm]{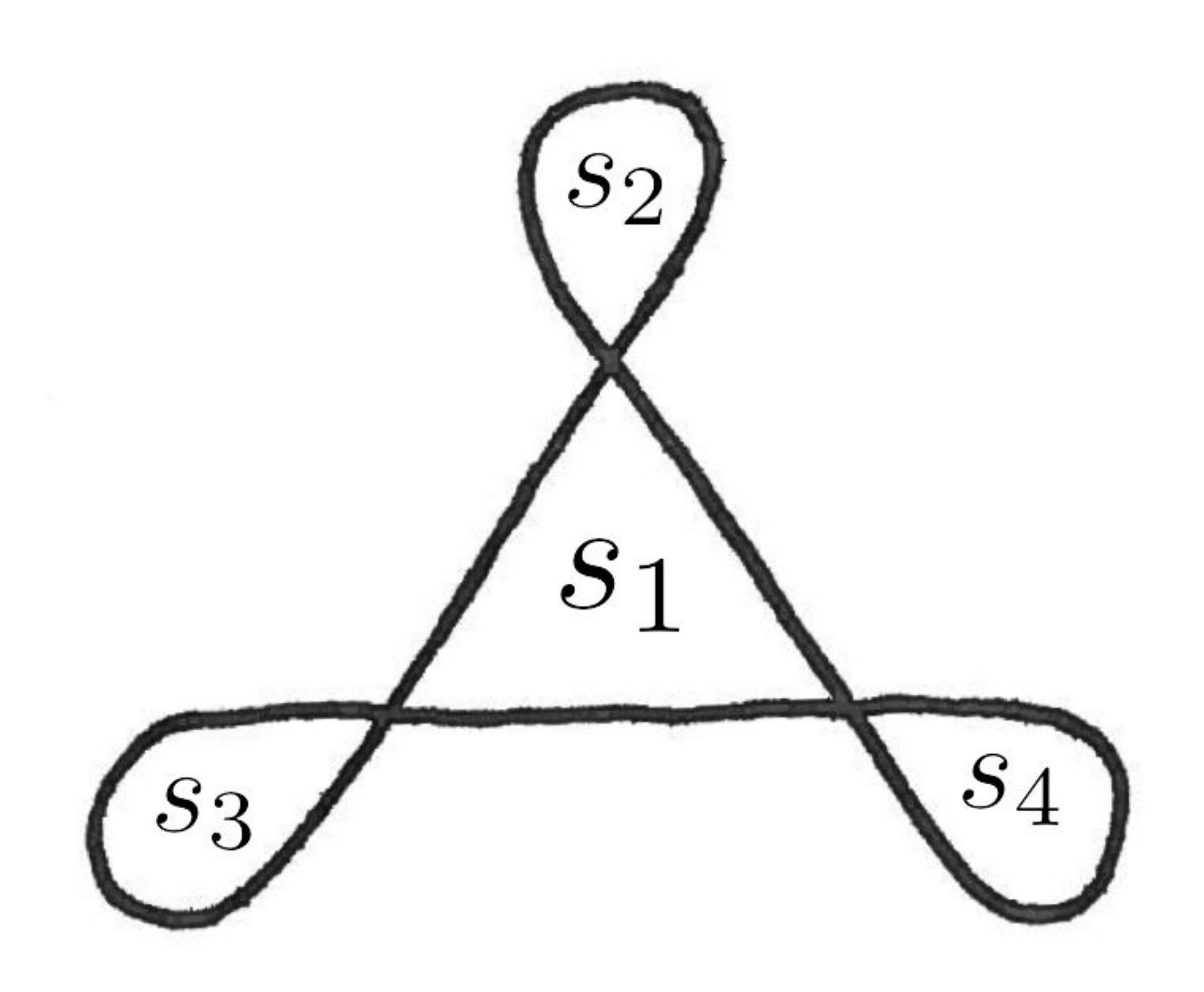} &&& \raisebox{1cm}{$e^{-\frac{1}{2}(s_{1}+s_{2}+s_{3}+s_{4})}$}\\
\includegraphics[height=2cm]{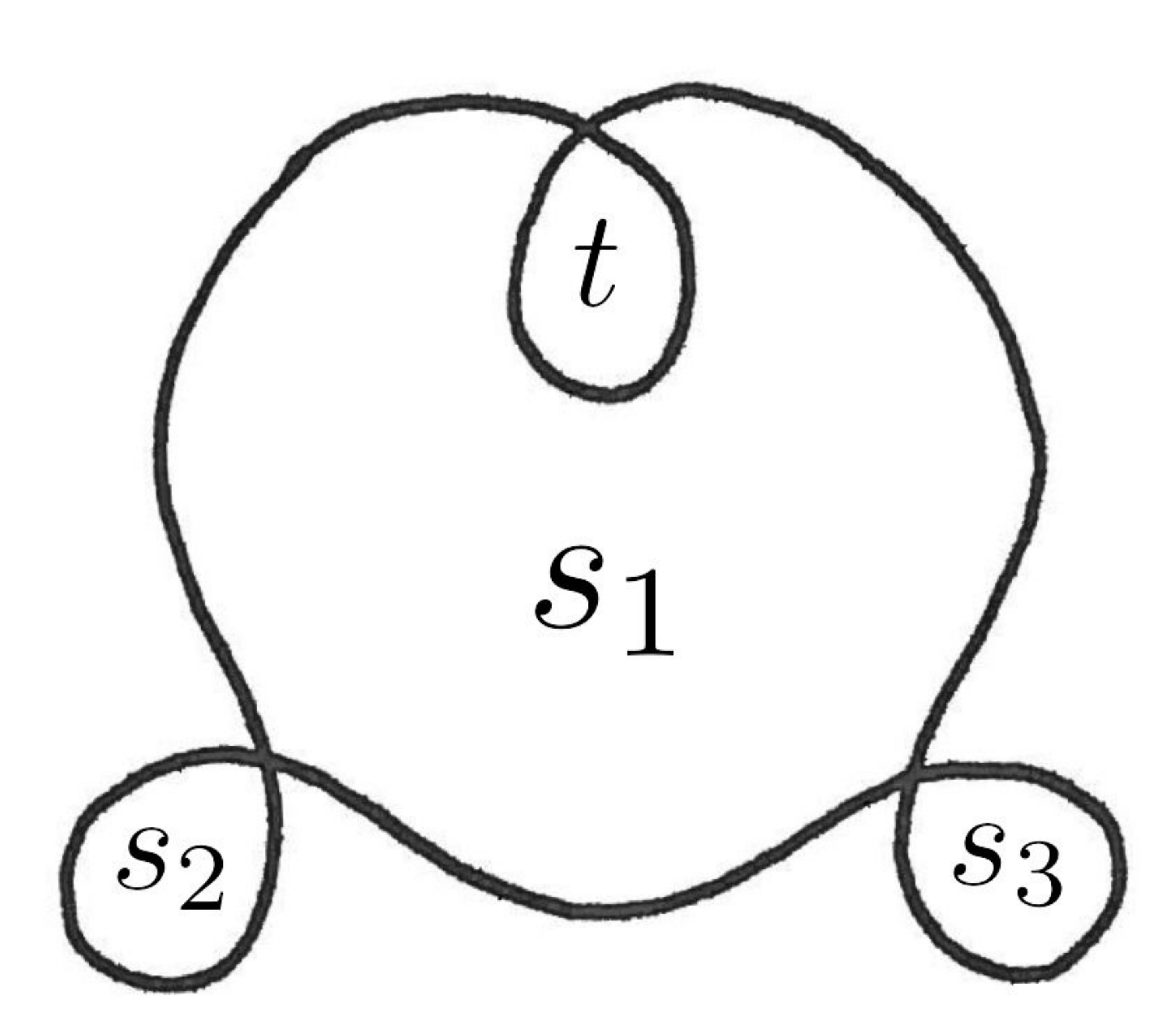} &&& \raisebox{1cm}{$e^{-\frac{1}{2}(s_{1}+s_{2}+s_{3})-t}(1-t)$}\\
\includegraphics[height=2cm]{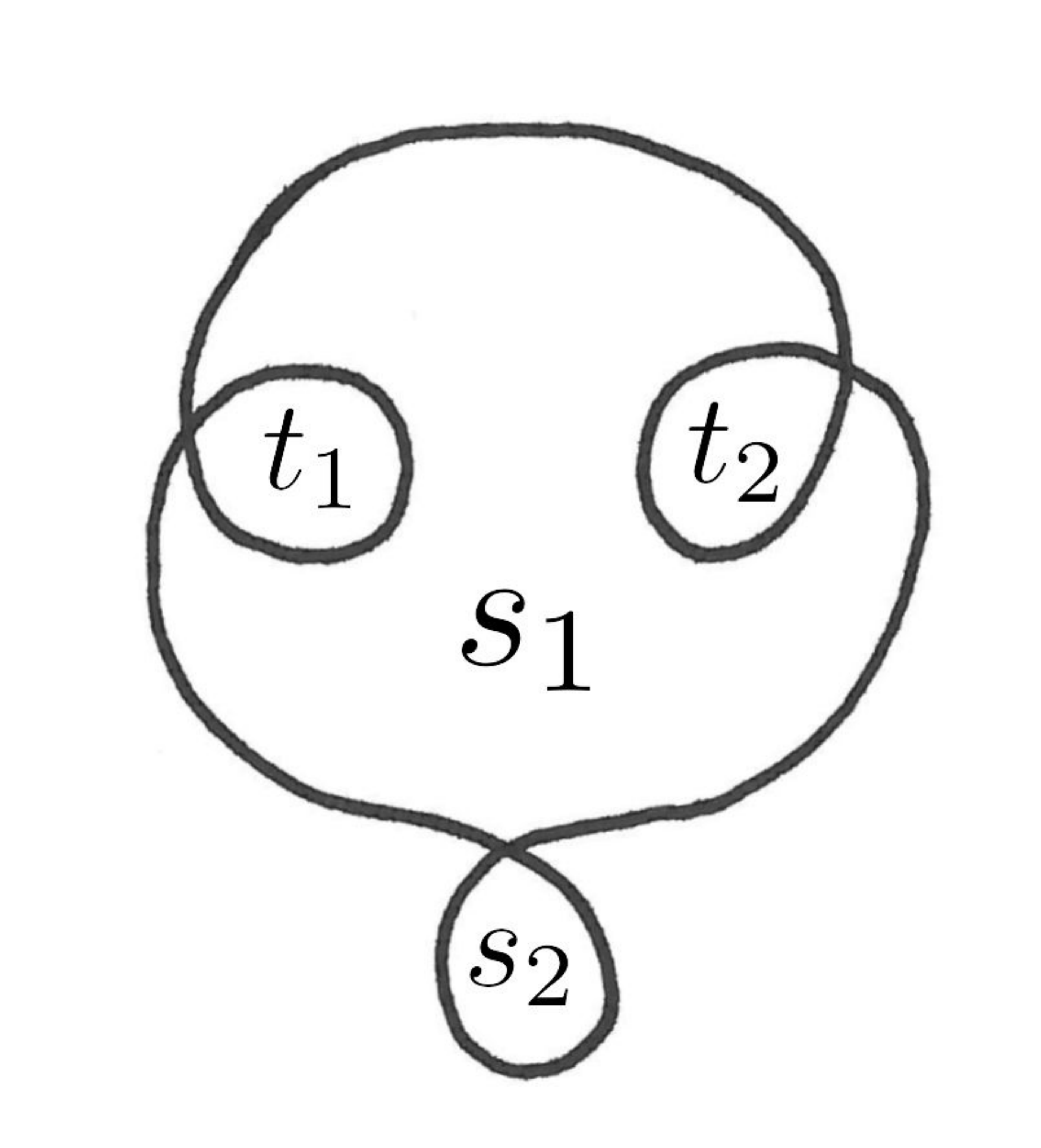} &&& \raisebox{1cm}{$e^{-\frac{1}{2}(s_{1}+s_{2})-t_{1}-t_{2}}(1-t_{1})(1-t_{2})$}\\
\hline
\end{tabular}
\end{center}

\begin{center}
\begin{tabular}{|cccl|}
\hline $l$ &&&$\Phi(l)$ \\\hline
\includegraphics[height=2cm]{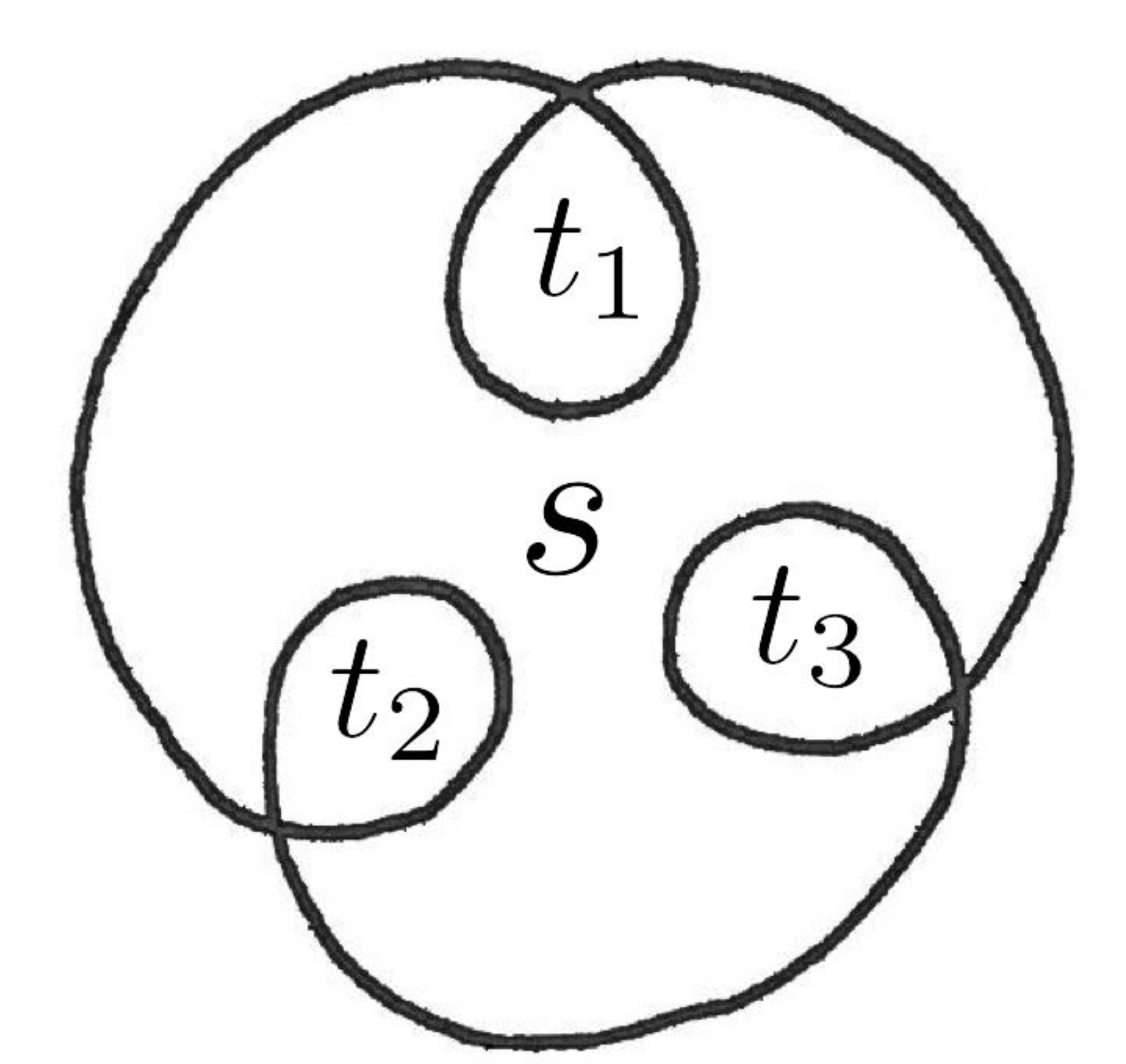} &&& \raisebox{1cm}{$e^{-\frac{s}{2}-t_{1}-t_{2}-t_{3}}(1-t_{1})(1-t_{2})(1-t_{3})$}\\
\includegraphics[height=2cm]{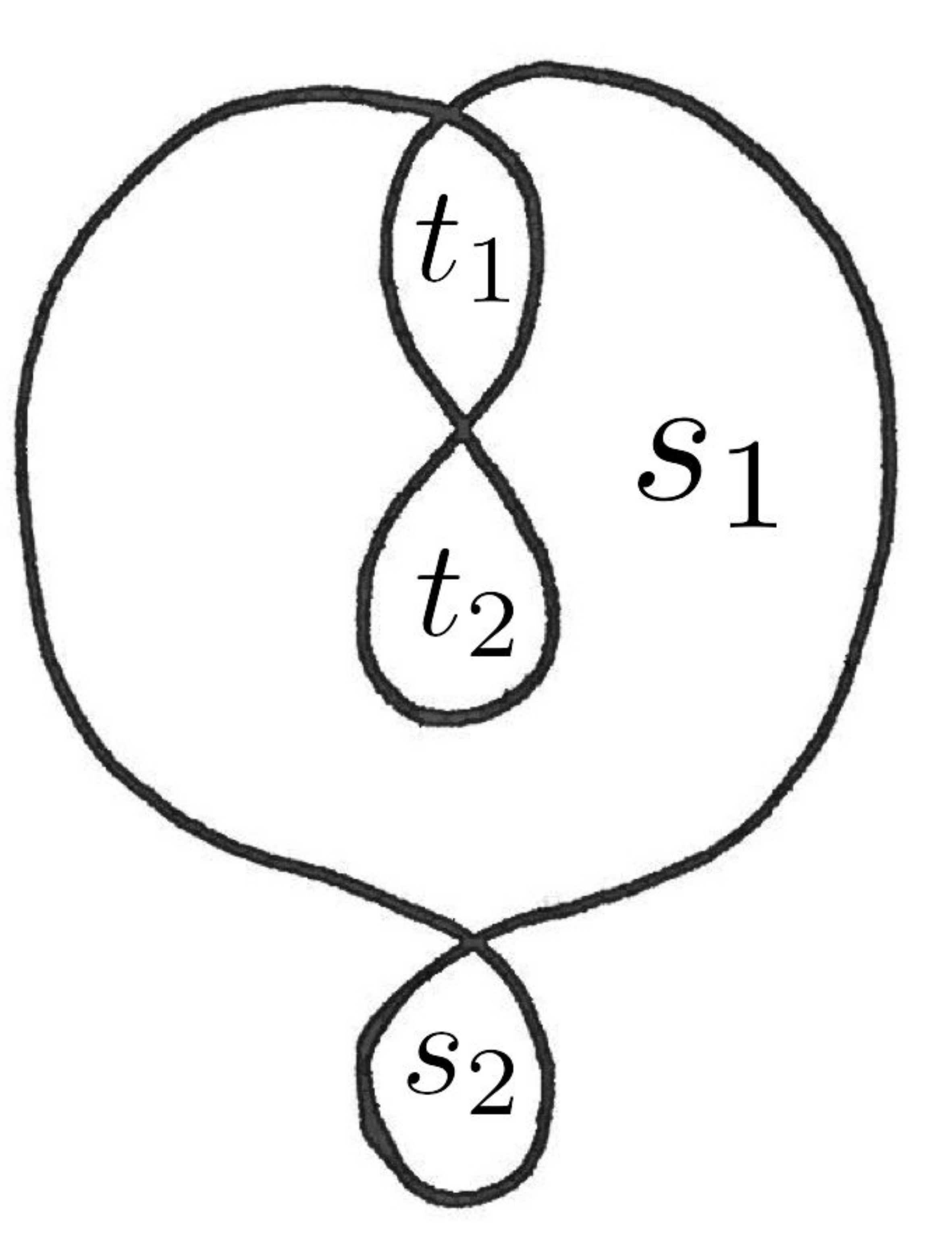} &&& \raisebox{1cm}{$e^{-\frac{1}{2}(s_{1}+s_{2})-t_{1}}(1-t_{1}e^{-t_{2}})$}\\
\includegraphics[height=2cm]{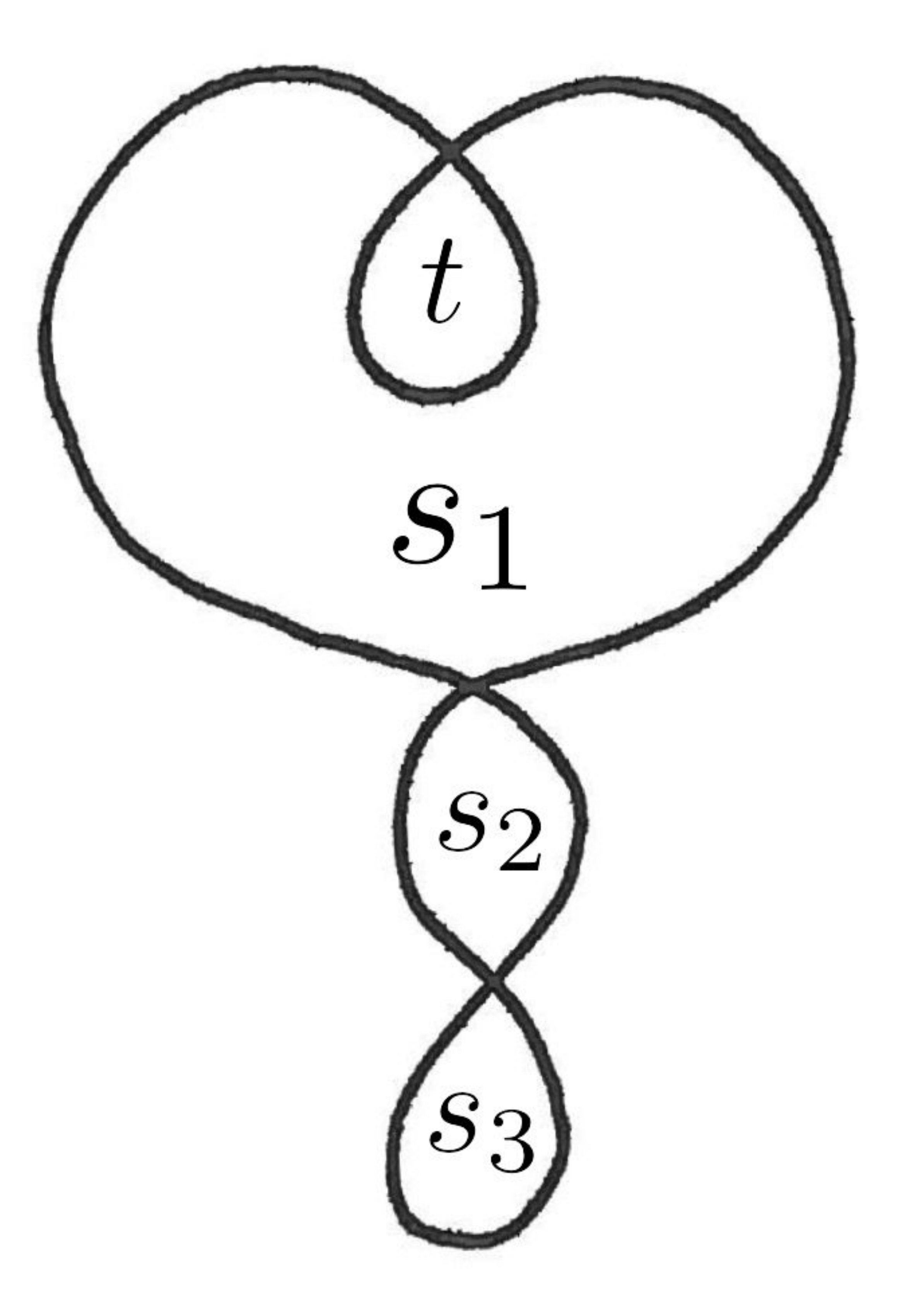} &&& \raisebox{1cm}{$e^{-\frac{1}{2}(s_{1}+s_{2}+s_{3})-t}(1-t)$}\\
\includegraphics[height=2cm]{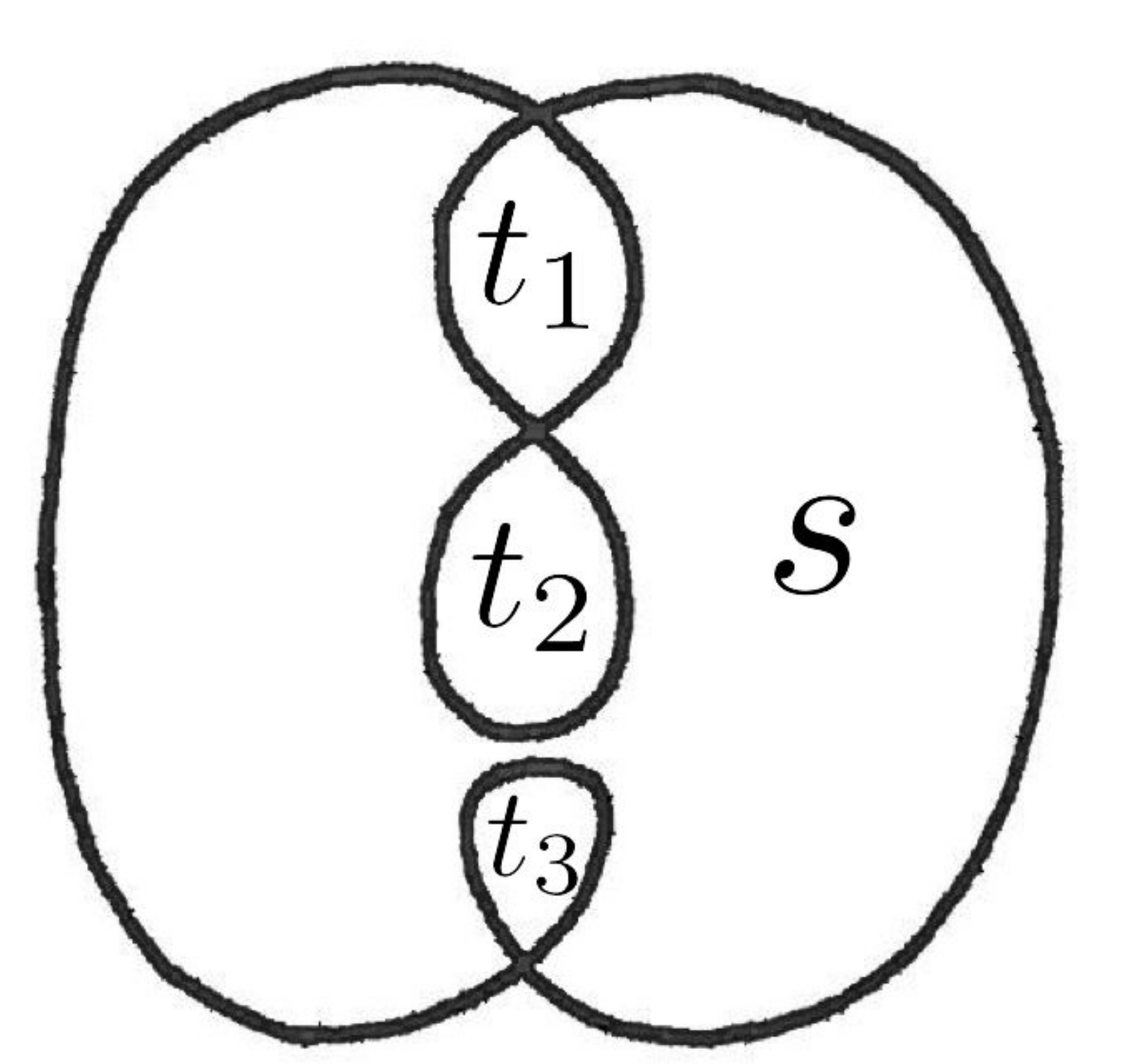} &&& \raisebox{1cm}{$e^{-\frac{s}{2}-t_{1}-t_{3}}(1-t_{1}e^{-t_{2}})(1-t_{3})$}\\ 
 \includegraphics[height=2cm]{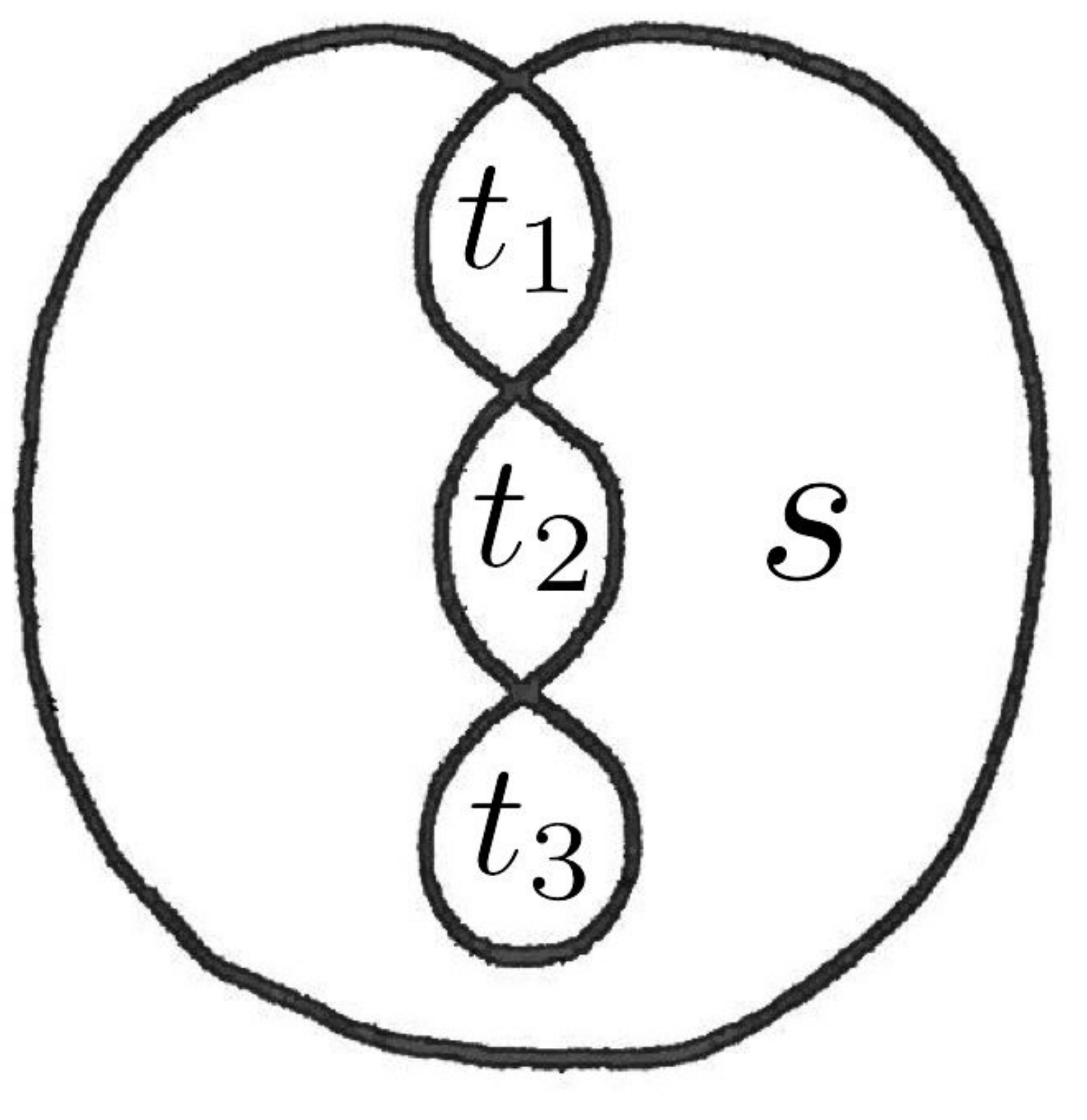} &&& \raisebox{1cm}{$e^{-\frac{s}{2}-t_{1}-t_{3}}(1-t_{1}e^{-t_{2}}-t_{3})$} \\
\includegraphics[height=2cm]{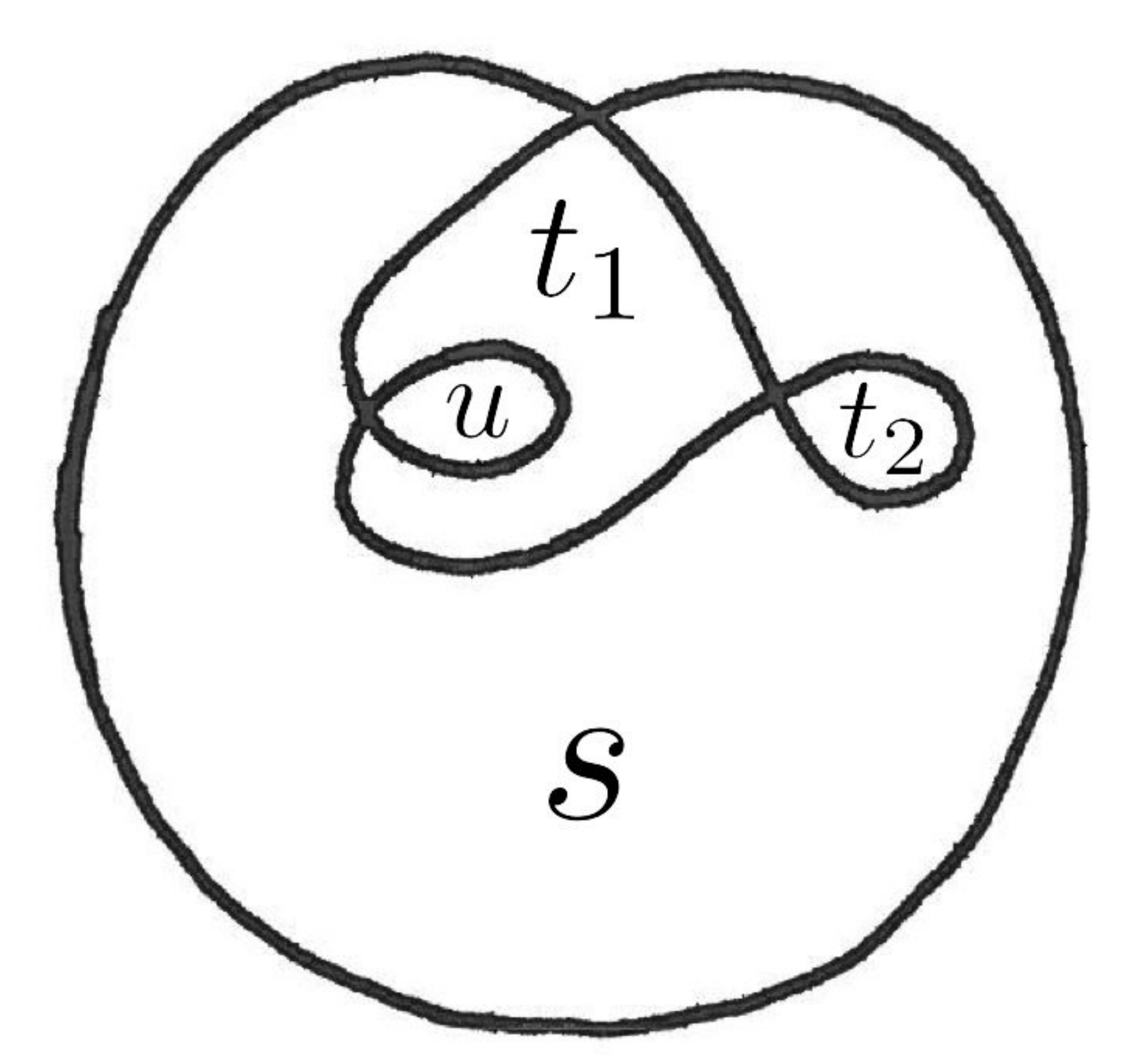} &&& \raisebox{1cm}{$e^{-\frac{s}{2}-t_{1}-\frac{3u}{2}}\left(e^{-t_{2}}(1-3u+\frac{3}{2}u^{2}-(1+t_{1})(1-u))+1-u\right)$}\\
\includegraphics[height=2cm]{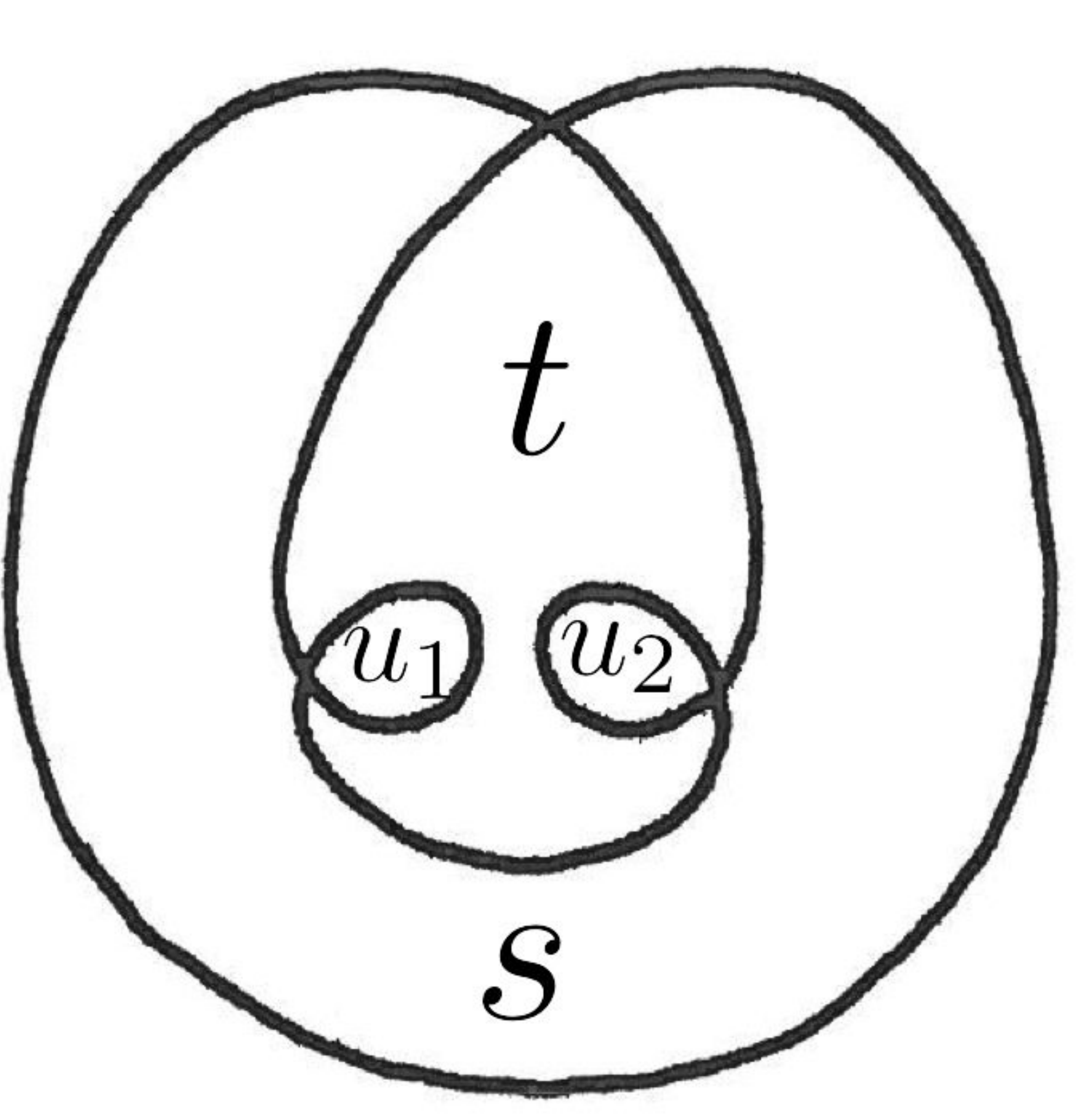} &&& \raisebox{1cm}{$e^{-\frac{s}{2}-t-\frac{3}{2}(u_{1}+u_{2})}\left(1-3(u_{1}+u_{2})+\frac{3}{2}(u_{1}+u_{2})^{2}\right.$}
\raisebox{0.6cm}{\hspace{-6cm}$\left.-t(1-(u_{1}+u_{2}))+u_{1}u_{2}(2-t-\frac{3}{2}(u_{1}+u_{2}))\right)$}\\
\includegraphics[height=2cm]{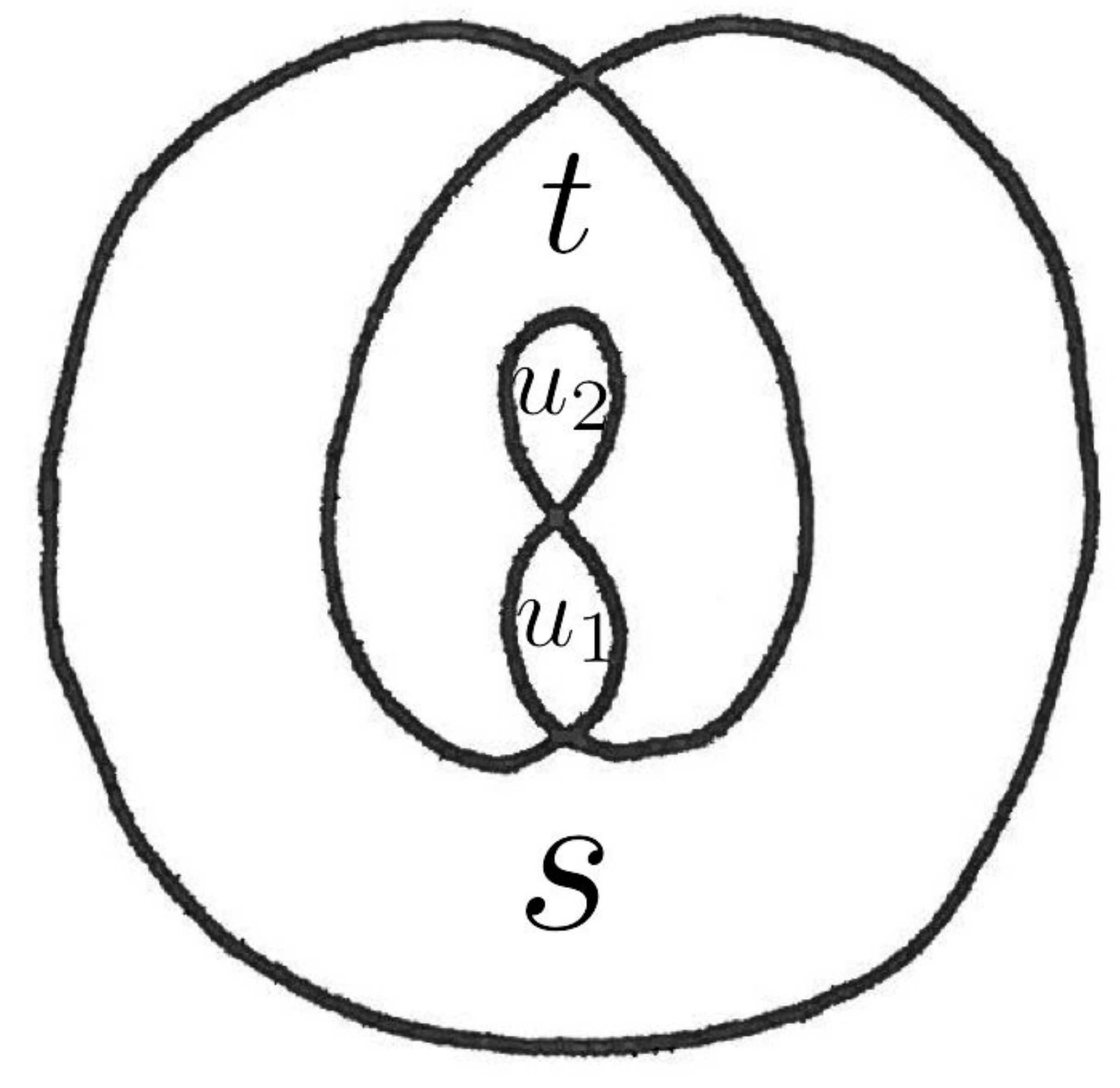} &&& \raisebox{1cm}{$e^{-\frac{s}{2}-t-\frac{3u_{1}}{2}-\frac{u_{2}}{2}}\left(e^{-u_{2}}(u_{1}(t+u_{2}-1)+\frac{3}{2}u_{1}^{2})+1-t-2u_{1}\right)$}\\
\includegraphics[height=2cm]{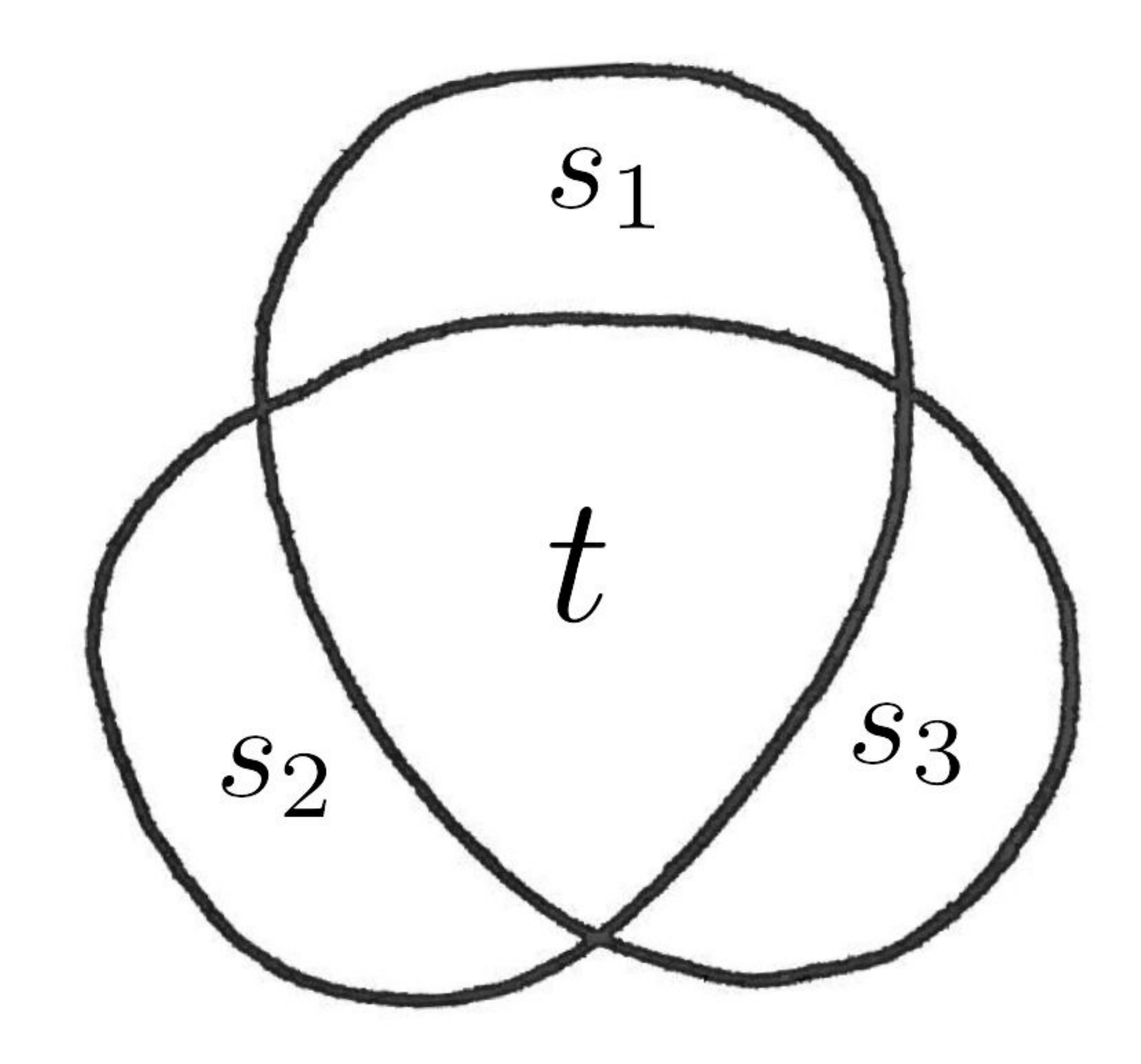} &&& \raisebox{1cm}{$e^{-\frac{1}{2}(s_{1}+s_{2}+s_{3})-t}(1-t)$}\\
\includegraphics[height=2cm]{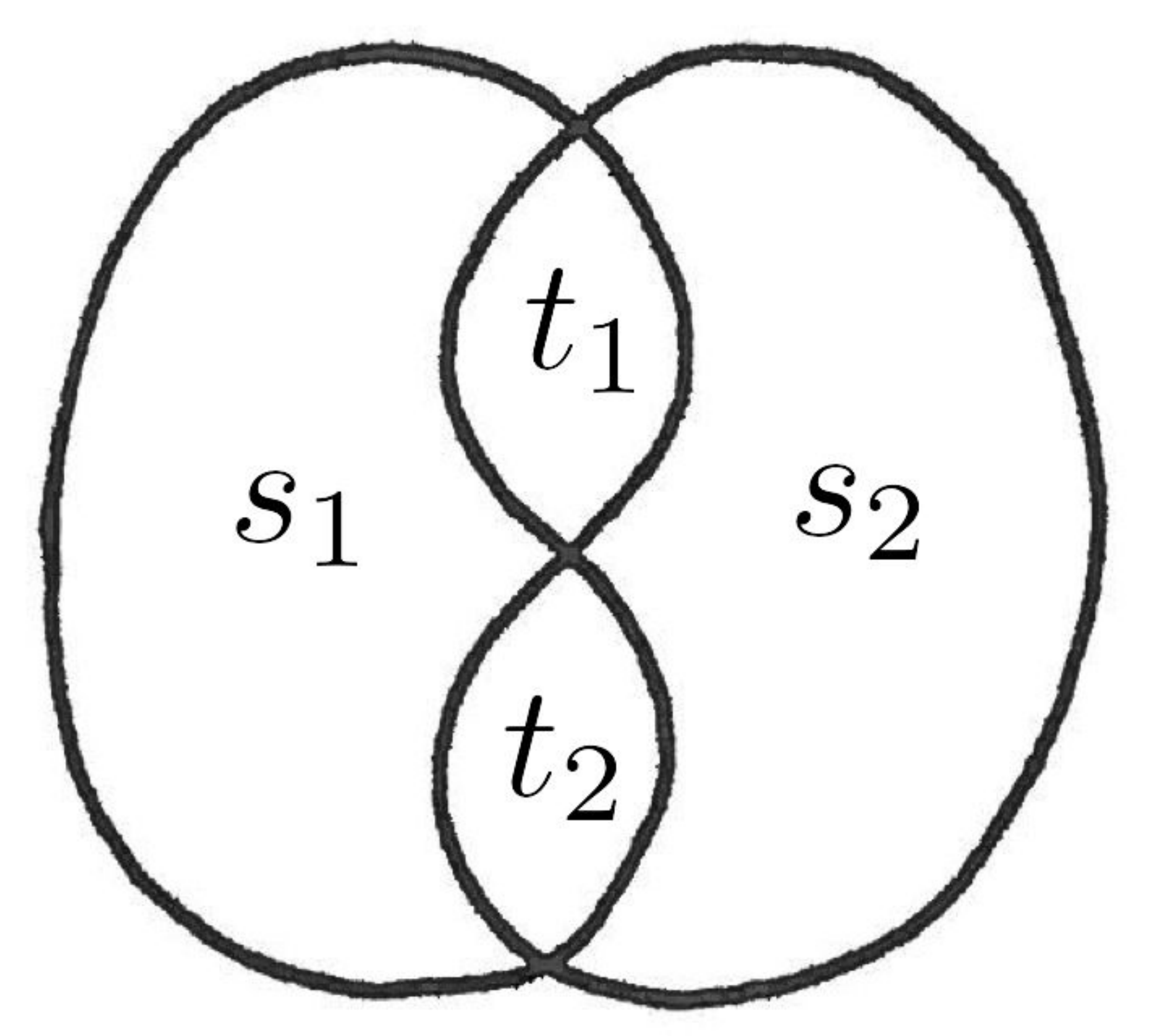} &&& \raisebox{1cm}{$e^{-\frac{1}{2}(s_{1}+s_{2})}(e^{-t_{1}}+e^{-t_{2}}-e^{-t_{1}-t_{2}})$}\\
\hline
\end{tabular}
\bigskip
\end{center}

\bibliographystyle{plain}
\bibliography{MasterField}

\def\cprime{$'$}
\begin{thebibliography}{10}

\bibitem{AnshelevitchSengupta}
Michael Anshelevitch and Ambar~N. Sengupta.
\newblock Quantum free yang-mills on the plane.
\newblock {\em Preprint}, 2011.

\bibitem{BanchoffPohl}
Thomas~F. Banchoff and William~F. Pohl.
\newblock A generalization of the isoperimetric inequality.
\newblock {\em J. Differential Geometry}, 6:175--192, 1971/72.

\bibitem{Biane}
Philippe Biane.
\newblock Free {B}rownian motion, free stochastic calculus and random matrices.
\newblock In {\em Free probability theory ({W}aterloo, {ON}, 1995)}, volume~12
  of {\em Fields Inst. Commun.}, pages 1--19. Amer. Math. Soc., Providence, RI,
  1997.

\bibitem{Biane2}
Philippe Biane.
\newblock {Some properties of crossings and partitions.}
\newblock {\em Discrete Math.}, 175(1-3):41--53, 1997.

\bibitem{Bleecker}
David Bleecker.
\newblock {\em Gauge theory and variational principles}, volume~1 of {\em
  Global Analysis Pure and Applied Series A}.
\newblock Addison-Wesley Publishing Co., Reading, Mass., 1981.

\bibitem{Bowditch}
B.~H. Bowditch.
\newblock Treelike structures arising from continua and convergence groups.
\newblock {\em Mem. Amer. Math. Soc.}, 139(662):viii+86, 1999.

\bibitem{Brauer}
Richard Brauer.
\newblock On algebras which are connected with the semisimple continuous
  groups.
\newblock {\em Ann. of Math. (2)}, 38(4):857--872, 1937.

\bibitem{CollinsSniady}
Beno\^{\i}t Collins and Piotr {\'S}niady.
\newblock Integration with respect to the {H}aar measure on unitary, orthogonal
  and symplectic group.
\newblock {\em Comm. Math. Phys.}, 264(3):773--795, 2006.

\bibitem{Hawaii}
Bart de~Smit.
\newblock {The fundamental group of the Hawaiian earring is not free.}
\newblock {\em Int. J. Algebra Comput.}, 2(1):33--38, 1992.

\bibitem{Eilenberg}
Samuel Eilenberg.
\newblock {Sur les transformations continues d'espaces m\'etriques compacts.}
\newblock {\em Fundam. Math.}, 22:292--296, 1934.

\bibitem{Fort}
Marion~K. Fort, Jr.
\newblock Mappings on {$S^{1}$} into one-dimensional spaces.
\newblock {\em Illinois J. Math.}, 1:505--508, 1957.

\bibitem{GambiniPullin}
Rodolfo Gambini and Jorge Pullin.
\newblock {\em Loops, knots, gauge theories and quantum gravity}.
\newblock Cambridge Monographs on Mathematical Physics. Cambridge University
  Press, Cambridge, 1996.

\bibitem{GoodmanWallach}
Roe Goodman and Nolan~R. Wallach.
\newblock {\em {Symmetry, representations, and invariants.}}
\newblock {Graduate Texts in Mathematics 255. New York, NY: Springer. xx,
  716~p. EUR~64.15 }, 2009.

\bibitem{GopakumarGross}
Rajesh Gopakumar and David~J. Gross.
\newblock Mastering the master field.
\newblock {\em Nuclear Physics B}, 451:379, 1995.

\bibitem{GrossMatytsin}
David~J. Gross and Andrei Matytsin.
\newblock Some properties of large-{$N$} two-dimensional {Y}ang-{M}ills theory.
\newblock {\em Nuclear Phys. B}, 437(3):541--584, 1995.

\bibitem{GrossTaylor2}
David~J. Gross and Washington Taylor.
\newblock Twists and {W}ilson loops in the string theory of two dimensional
  {QCD}.
\newblock {\em Nuclear Physics B}, 403:395, 1993.

\bibitem{GrossTaylor}
David~J. Gross and Washington Taylor.
\newblock {Two-dimensional QCD is a string theory}.
\newblock {\em Nucl. Phys.}, B400:181--210, 1993.

\bibitem{GrossKingSengupta}
Leonard Gross, Christopher King, and Ambar~N. Sengupta.
\newblock Two-dimensional {Y}ang-{M}ills theory via stochastic differential
  equations.
\newblock {\em Ann. Physics}, 194(1):65--112, 1989.

\bibitem{HamblyLyons}
Ben Hambly and Terry Lyons.
\newblock Uniqueness for the signature of a path of bounded variation and the
  reduced path group.
\newblock {\em Ann. of Math. (2)}, 171(1):109--167, 2010.

\bibitem{HiaiPetz}
Fumio Hiai and D{\'e}nes Petz.
\newblock {\em The semicircle law, free random variables and entropy},
  volume~77 of {\em Mathematical Surveys and Monographs}.
\newblock American Mathematical Society, Providence, RI, 2000.

\bibitem{Hurewicz}
Witold Hurewicz and Henry Wallman.
\newblock {\em Dimension {T}heory}.
\newblock Princeton Mathematical Series, v. 4. Princeton University Press,
  Princeton, N. J., 1941.

\bibitem{Kazakov}
Vladimir~A. Kazakov.
\newblock Wilson loop average for an arbitrary contour in two-dimensional
  {U}(${N}$) gauge theory.
\newblock {\em Nuclear Phys. B}, 179(2):283--292, 1981.

\bibitem{KazakovKostov}
Vladimir~A. Kazakov and Ivan~K. Kostov.
\newblock Nonlinear strings in two-dimensional {${\rm U}(\infty )$} gauge
  theory.
\newblock {\em Nuclear Phys. B}, 176(1):199--215, 1980.

\bibitem{Kuratowski}
Casimir Kuratowski.
\newblock {\em Topologie. {II}. {E}spaces compacts, espaces connexes, plan
  euclidien}.
\newblock Monografie Matematyczne, vol. 21. Warszawa-Wroc\l aw, 1950.

\bibitem{LevyAMS}
Thierry L{\'e}vy.
\newblock Yang-{M}ills measure on compact surfaces.
\newblock {\em Mem. Amer. Math. Soc.}, 166(790):xiv+122, 2003.

\bibitem{LevyJGP}
Thierry L{\'e}vy.
\newblock Wilson loops in the light of spin networks.
\newblock {\em J. Geom. Phys.}, 52(4):382--397, 2004.

\bibitem{LevyAIM}
Thierry L{\'e}vy.
\newblock Schur-{W}eyl duality and the heat kernel measure on the unitary
  group.
\newblock {\em Adv. Math.}, 218(2):537--575, 2008.

\bibitem{LevySMF}
Thierry L{\'e}vy.
\newblock Two-dimensional {M}arkovian holonomy fields.
\newblock {\em Ast\'erisque}, 329:vi+172, 2010.

\bibitem{LevySym}
Thierry L{\'e}vy.
\newblock Prefixes of minimal factorisations.
\newblock {\em Preprint}, 2011.

\bibitem{Liao}
Ming Liao.
\newblock {\em L\'evy processes in {L}ie groups}, volume 162 of {\em Cambridge
  Tracts in Mathematics}.
\newblock Cambridge University Press, Cambridge, 2004.

\bibitem{MakeenkoMigdal}
Yuri Makeenko and Alexander~A. Migdal.
\newblock Exact equation for the loop average in multicolor {Q}{C}{D}.
\newblock {\em Phys. Lett. B}, 88B:135, 1979.

\bibitem{NicaSpeicher}
Alexandru Nica and Roland Speicher.
\newblock {\em Lectures on the combinatorics of free probability.}
\newblock London Mathematical Society Lecture Note Series 335. Cambridge:
  Cambridge University Press. xv, 417~p., 2006.

\bibitem{OkounkovVershik}
Andrei Okounkov and Anatoly Vershik.
\newblock A new approach to representation theory of symmetric groups.
\newblock {\em Selecta Math. (N.S.)}, 2(4):581--605, 1996.

\bibitem{Polyakov}
A.~M. Polyakov.
\newblock Gauge fields as rings of glue.
\newblock {\em Nuclear Phys. B}, 164(1):171--188, 1980.

\bibitem{SenguptaAMS}
Ambar~N. Sengupta.
\newblock Gauge theory on compact surfaces.
\newblock {\em Mem. Amer. Math. Soc.}, 126(600):viii+85, 1997.

\bibitem{SenguptaFN}
Ambar~N. Sengupta.
\newblock {The large-$N$ Yang-Mills field on the plane and free noise.}
\newblock {Kielanowski, Piotr (ed.) et al., Geometric methods in physics.
  Proceedings of the xxvii workshop on geometric methods in physics, Bia\l
  owie$\dot{\rm z}$a, Poland, 29 June -- 5 July 2008. Melville, NY: American
  Institute of Physics (AIP). AIP Conference Proceedings 1079, 121-134
  (2008).}, 2008.

\bibitem{SenguptaMF}
Ambar~N. Sengupta.
\newblock Traces in two-dimensional {QCD}: the large-{$N$} limit.
\newblock In {\em Traces in geometry, number theory and quantum fields (edited
  by Sergio Albeverio, Matilde Marcolli, Sylvie Paycha, and Jorge Plazas)}.
  Vieweg, 2008.

\bibitem{Singer}
I.~M. Singer.
\newblock On the master field in two dimensions.
\newblock In {\em Functional analysis on the eve of the 21st century, {V}ol.\ 1
  ({N}ew {B}runswick, {NJ}, 1993)}, volume 131 of {\em Progr. Math.}, pages
  263--281. Birkh\"auser Boston, Boston, MA, 1995.

\bibitem{Speicher}
Roland Speicher.
\newblock {Multiplicative functions on the lattice of non-crossing partitions
  and free convolution.}
\newblock {\em Math. Ann.}, 298(4):611--628, 1994.

\bibitem{tHooft}
Gerard 't~Hooft.
\newblock {A planar diagram theory for strong interactions}.
\newblock {\em Nucl. Phys.}, B72:461, 1974.

\bibitem{VDN}
Dan~V. Voiculescu, K.J. Dykema, and A.~Nica.
\newblock {\em Free random variables. A noncommutative probability approach to
  free products with applications to random matrices, operator algebras and
  harmonic analysis on free groups.}
\newblock CRM Monograph Series. 1. Providence, RI: American Mathematical
  Society (AMS). v, 70 p., 1992.

\end{thebibliography}
\end{document}